\documentclass[%
 reprint,
unsortedaddress,
11pt,
preprint,
 nofootinbib,
 amsmath,amssymb,
 aps,
]{revtex4-2}

\usepackage{graphicx}
\usepackage{dcolumn}
\usepackage{bm}
\usepackage{hyperref}
\usepackage{color}
\usepackage[normalem]{ulem}
\usepackage{mhchem}
\usepackage{url}
\usepackage{amssymb}
\usepackage[bbgreekl]{mathbbol}
\usepackage{amsthm}
\usepackage{comment}
\usepackage{todonotes}
\usepackage{lineno}

\newcommand{\Prob}{\mathbb{P}}

\newcommand{\bbGamma}{{\mathpalette\makebbGamma\relax}}
\newcommand{\makebbGamma}[2]{%
  \raisebox{\depth}{\scalebox{1}[-1]{$\mathsurround=0pt#1\mathbb{L}$}}%
}

\newcommand{\stoiMatrix}{\mathbb{S}}
\newcommand{\stoiVector}{s}

\newcommand{\flux}{j}

\newcommand{\Vc}[1]{\boldsymbol{#1}}
\newcommand{\dd}{\mathrm{d}}
\newcommand{\dt}{\dd t}
\newcommand{\defeq}{:=}

\newcommand{\TJK}[1]{#1}
\newcommand{\kcoef}{k}
\newcommand{\Transpose}{T}
\newcommand{\X}{\mathcal{X}}
\newcommand{\Y}{\mathcal{Y}}
\newcommand{\Z}{\mathcal{Z}}
\newcommand{\Zeta}{\mathfrak{Z}}
\newcommand{\cmMatrix}{\bbGamma} 
\newcommand{\cmvector}{\gamma} 
\newcommand{\B}{B}

\newcommand{\Polytope}{\mathcal{P}}
\newcommand{\BD}{\mathcal{D}}

\newcommand{\tf}{f}

\newcommand{\frenecy}{\omega}

\newcommand{\Fspace}{\mathcal{F}}
\newcommand{\Jspace}{\mathcal{J}}
\newcommand{\EPR}{\dot{\Sigma}}
\newcommand{\pEPR}{\dot{\Sigma}^{p}}

\newcommand{\FIN}{\mathbb{I}_{F}}
\newcommand{\gLap}{\mathcal{L}}
\newcommand{\metric}{M}

\newcommand{\FEnergy}{\mathcal{F}}

\newcommand{\Grad}{\mathrm{grad}}
\newcommand{\Div}{\mathrm{div}}
\newcommand{\Curl}{\mathrm{curl}}

\newcommand{\FM}{G}
\newcommand{\diag}{\mathrm{diag}}
\newcommand{\identityM}{I}

\newcommand{\Real}{\mathbb{R}}
\newcommand{\Integer}{\mathbb{Z}}

\newcommand{\Tan}{\mathcal{T}}

\newcommand{\Dissp}{\Psi}
\newcommand{\Pfunc}{\Phi}
\newcommand{\pfunc}{\phi}
\newcommand{\dissip}{\psi}

\newcommand{\Variety}{\mathcal{M}}
\newcommand{\Manifold}{\mathcal{M}}
\newcommand{\cManifold}{\mathcal{C}}

\newcommand{\Img}{\mathrm{Im}}
\newcommand{\Ker}{\mathrm{Ker}}
\newcommand{\Dim}{\mathrm{dim}}

\newcommand{\IncMatrix}{\mathbb{B}}
\newcommand{\HIncMatrix}{\mathbb{S}}
\newcommand{\cycMatrix}{\mathbb{V}}
\newcommand{\consMatrix}{\mathbb{U}}

\newcommand{\eqnref}[1]{Eq. \ref{#1}}
\newcommand{\secref}[1]{Sec. \ref{#1}}
\newcommand{\defref}[1]{Def. \ref{#1}}
\newcommand{\fgref}[1]{Fig. \ref{#1}}
\newcommand{\propref}[1]{Prop. \ref{#1}}
\newcommand{\thmref}[1]{Thm. \ref{#1}}
\newcommand{\lmmref}[1]{Lemma \ref{#1}}
\newcommand{\corref}[1]{Cor. \ref{#1}}
\newcommand{\exref}[1]{Ex. \ref{#1}}

\newcommand{\com}[1]{#1}

\newcommand{\Graph}{\mathbb{G}}
\newcommand{\HGraph}{\mathbb{H}}
\newcommand{\node}{\mathbb{v}}
\newcommand{\molX}{\mathbb{X}}
\newcommand{\edge}{\mathbb{e}}
\newcommand{\cons}{\mathbb{l}}

\newcommand{\cycle}{\mathbb{z}}

\newcommand{\chain}{C}
\newcommand{\boundary}{\delta}

\usepackage{mathtools}

\theoremstyle{plain}
\newtheorem{thm}{Theorem}
\newtheorem*{thm*}{Theorem}

\newtheorem{lmm}{Lemma}
\newtheorem*{lmm*}{Lemma}

\newtheorem{prop}{Proposition}
\newtheorem*{prop*}{Proposition}

\newtheorem{cor}{Corollary}
\newtheorem*{cor*}{Corollary}

\theoremstyle{definition}
\newtheorem{dfn}{Definition}

\theoremstyle{remark}
\newtheorem{rmk}{Remark}

\theoremstyle{definition}
\newtheorem{ex}{Example}

\begin{document}

\preprint{APS/123-QED}

\title{Information Geometry of Dynamics on Graphs and Hypergraphs}
\author{Tetsuya J. Kobayashi}
\email[E-mail me at:]{tetsuya@mail.crmind.net}
\affiliation{Institute of Industrial Science, The University of Tokyo, 4-6-1, Komaba, Meguro-ku, Tokyo 153-8505 Japan}
\affiliation{Department of Mathematical Informatics, Graduate School of Information Science and Technology, The
University of Tokyo, Tokyo 113-8654, Japan}
\affiliation{Universal Biology Institute, The University of Tokyo, 7-3-1, Hongo, Bunkyo-ku, 113-8654, Japan.}
\author{Dimitri Loutchko}
\affiliation{Institute of Industrial Science, The University of Tokyo, 4-6-1, Komaba, Meguro-ku, Tokyo 153-8505 Japan}
\author{Atsushi Kamimura}
\affiliation{Institute of Industrial Science, The University of Tokyo, 4-6-1, Komaba, Meguro-ku, Tokyo 153-8505 Japan}
\author{Shuhei A. Horiguchi}
\affiliation{Department of Mathematical Informatics, Graduate School of Information Science and Technology, The
University of Tokyo, Tokyo 113-8654, Japan}
\author{Yuki Sughiyama}
\affiliation{Institute of Industrial Science, The University of Tokyo, 4-6-1, Komaba, Meguro-ku, Tokyo 153-8505 Japan}
\date{\today}

\
\date{\today}

\begin{abstract}
We introduce a new information-geometric structure associated with the dynamics on discrete objects such as graphs and hypergraphs. 
The presented setup consists of two dually flat structures built on the vertex and edge spaces, respectively.
The former is the conventional duality between density and potential, e.g., the probability density and its logarithmic form induced by a convex thermodynamic function. 
The latter is the duality between flux and force induced by a convex and symmetric dissipation function, which drives the dynamics of the density.
These two are connected topologically by the homological algebraic relation induced by the underlying discrete objects.
The generalized gradient flow in this doubly dual flat structure is an extension of the gradient flows on Riemannian manifolds, which include Markov jump processes and nonlinear chemical reaction dynamics as well as the natural gradient and mirror descent.  
The information-geometric projections on this doubly dual flat structure lead to information-geometric extensions of  the Helmholtz-Hodge decomposition and the Otto structure in $L^{2}$ Wasserstein geometry.
 The structure can be extended to non-gradient nonequilibrium flows, from which we also obtain the induced dually flat structure on cycle spaces. 
 This abstract but general framework can extend the applicability of information geometry to various problems of linear and nonlinear dynamics.
\end{abstract}

\maketitle

\onecolumngrid

\setcounter{tocdepth}{-1}
\tableofcontents
\newpage

\section{Introduction}\label{sec:Int_Int}
Information geometry is finding and establishing a firm position as a geometric language in various scientific disciplines\cite{amari2016,ay2018}.
Information geometry enables us to gain an intuitive understanding of the structures behind complicated problems of inference and estimation, for which Euclidean or Riemannian geometry is not sufficient. 
In addition, it can provide ways to devise new solutions and approaches for the problems\cite{amari2016}.
While information geometry was originally developed for statistics, its applicability now reaches far beyond statistical problems. 
Whenever the notions of probability, information, or positive density appear in a problem, it is natural to consider its information-geometric structure.

\subsection{Information geometry of dynamics}\label{sec:Int_IG_dynamics}
Dynamical systems and phenomena can be naturally analyzed with information geometric methods, as conventionally one considers the dynamics of probability distributions\cite{risken1996,horsthemke2006,gardiner2010}, e.g., via the Fokker-Planck equations (FPE) and the Master equation, or those of positive densities, e.g, via population dynamics, epidemic models, diffusion dynamics on networks, and chemical reaction dynamics\cite{murray2007,beard2008,feinberg2019}.
Although the application of information geometry to dynamical systems has been attempted almost since its birth, information geometry for dynamics is much less organized and principled compared with those for static problems in statistics, optimization, and others\cite{amari2016}.
In connection with statistical inference, information geometry was employed by Amari and others to investigate Gaussian time series and autoregressive moving average (ARMA) models by representing their power spectrum as parametric manifolds\cite{amari1987,ravishanker1990J.TimeSer.Anal.,tanaka2011SankhyaA}. 
This idea was also used to investigate linear systems\cite{amari1987Math.SystemsTheory}.
Markov jump processes on finite states\footnote{Also known as Markov chains} were investigated information-geometrically by considering the hierarchical structure of joint or conditional probabilities at different time points, e.g., $\Prob_{\Vc{\theta}}(x_{1},x_{2},\cdots, x_{t})$\cite{amari2001IEEETrans.Inf.Theory}. 
or by introducing exponential families of Markov kernels (transition matrices), $\mathbb{T}_{\Vc{\theta}}(x|x')$, via exponential tilting of the kernels\cite{nakagawa1993IEEETrans.Inf.Theory,takeuchi1998Proc.1998IEEEInt.Symp.Inf.TheoryCatNo98CH36252,nagaoka2017,takeuchi20072007IEEEInt.Symp.Inf.Theory,hayashi2016Ann.Stat.,wolfer2021Info.Geo.,pistone2013AnnInstStatMath}.
Furthermore, information geometry was applied to studies of random walks, nonlinear diffusion equations of porous media , and networks\cite{obata1992Phys.Rev.A,ohara2009Eur.Phys.J.B,ohara2021Geom.Sci.Inf.}.
In relation to mechanics, integrable systems were associated with the dualistic gradient flow of information geometry in the seminal works\cite{nakamura1993JapanJ.Indust.Appl.Math.,fujiwara1995PhysicaD:NonlinearPhenomena}, and other connections of information geometry with Lagrangian or Hamiltonian mechanics have been pursued\cite{felice2018,goto2018J.Phys.A:Math.Theor.,chirco2022Int.J.Geom.MethodsMod.Phys.}.

\subsection{Information measures for dynamics}\label{sec:Int_Inf_measure}
Concurrently with and almost independently of these attempts within the community of information geometry, information measures relevant to information geometry have been employed in various problems of dynamical systems and stochastic processes in information theory\cite{ihara1993}, filtering theory\cite{brigo1999Bernoulli,newton2018Inf.Geom.ItsAppl.}, control theory\cite{fleming1982Stochastics,todorov2009Proc.Natl.Acad.Sci.,theodorou20122012IEEE51stIEEEConf.Decis.ControlCDC}, and non-equilibrium physics and chemistry\cite{jaynes1957Phys.Rev.,frieden2004,sagawa2012}.
The Kullback-Leibler (KL) divergence\cite{kullback1951Ann.Math.Stat.} for probabilities and positive densities was shown to be a Lyapunov function of Markov jump processes (MJP)\cite{gardiner2010}, FPE\cite{lebowitz1957AnnalsofPhysics,risken1996}, deterministic chemical reaction networks (CRN)\cite{shear1967JournalofTheoreticalBiology,horn1972Arch.RationalMech.Anal.}, and other dynamical systems\cite{goh1977Am.Nat.,figueiredo2000PhysicsLettersA}, the origin of which can be dated back to Gibbs' H-theorem\cite{gibbs2010}.  
Among those topics, since the establishment of chemical thermodynamics by Gibbs\cite{gibbs2010} and chemical kinetics by Guldberg and Waage\cite{waage1986J.Chem.Educ.a},  CRN has played the role of a seedbed for cultivating the theory between dynamics and divergence owing to its close connection with thermodynamics\cite{ge2016ChemicalPhysics,rao2016Phys.Rev.X,sughiyama2021ArXiv211212403Cond-MatPhysicsphysics}.
More recently, it was also clarified that the divergences and information geometry are fundamental in stochastic thermodynamics \cite{seifert2012Rep.Prog.Phys.,ito2018Phys.Rev.Lett.,kolchinsky2021Phys.Rev.X,yoshimura2021Phys.Rev.Research,ohga2021ArXiv211211008Cond-Mat}.

In addition to the KL divergence, the Fisher-information-like quantity 
\begin{align}
    \FIN[p] \defeq \int p(\Vc{r})(\nabla_{\Vc{r}} \ln p(\Vc{r}))^{2}\mathrm{d}\Vc{r} \in \Real_{\ge 0} \label{eq:FisherInfoNum}
\end{align}
was also revealed to play an important role in characterizing dynamics for densities on a continuous space, e.g., Gaussian convolution, diffusion processes, and FPE\cite{stam1959InformationandControl,plastino1998PhysicsLettersA,wibisono20172017IEEEInt.Symp.Inf.TheoryISIT}. 
Various governing equations in physics were claimed to be derived in a unified way from this quantity\cite{frieden2004}.
The quantity $\FIN$ looks like the Fisher information\cite{fisher1922Philos.Trans.R.Soc.Lond.Ser.Contain.Pap.Math.Phys.Character} but is different from the conventional Fisher information matrix \cite{rao1958Scand.Actuar.J.,papaioannou2005Commun.Stat.-TheoryMethodsa,kharazmi2018Braz.J.Probab.Stat.} because the derivative $\nabla_{\Vc{r}} \ln p(\Vc{r})$ is not for the parameters but for the base space variable of $p(\Vc{r})$\footnote{The relation between the two forms of Fisher information has been explained in multiple ways. For example, they are related as the shift of the base space via parameters\cite{johnson2004,frieden2004}. The Fisher information number was introduced by Rao\cite{rao1958Scand.Actuar.J.}.}. 
Because $\FIN$ is a scalar, we follow \cite{papaioannou2005Commun.Stat.-TheoryMethodsa} and call it Fisher information number. 
The Fisher information number $\FIN$ is related to the KL divergence in additive Gaussian channels\cite{stam1959InformationandControl} and other systems\cite{yamano2013Eur.Phys.J.B,wibisono20172017IEEEInt.Symp.Inf.TheoryISIT}, which is known as the De Bruijn identity\cite{stam1959InformationandControl}. 
In addition, the logarithmic Sobolev inequality also provides a relation between the Fisher information number and the KL divergence (or Shannon information)\cite{gross1975Am.J.Math.,gross1975DukeMath.J.}.
These results have recently been associated with the formal Riemannian geometric structure induced by the $L^{2}$-Wasserstein geometry\cite{otto2001Commun.PartialDiffer.Equ.,villani2003}.

\subsection{Information geometry and dynamics in machine learning}\label{sec:Int_ML}
On top of these traditional trends, information geometry is now playing a pivotal role in machine learning for designing and evaluating online optimization algorithms (dynamics) in the space of model parameters such as natural gradient\cite{amari1998NeuralComputation} and mirror descent\cite{beck2003OperationsResearchLetters,raskutti2015IEEETrans.Inf.Theory} as well as evolutionary computation (information-geometric optimization) \cite{ollivier2017J.Mach.Learn.Res.}. 
Geometric interpretation allows to understand the behaviors and efficiency of algorithms and their dynamics more intuitively in a principled manner\cite{hino2022,raskutti2015IEEETrans.Inf.Theory,ollivier2017J.Mach.Learn.Res.}.

\subsection{Aim and contributions of this work}
Despite the wide applicability and the long history of information geometry, we still lack a solid theoretical framework to unify these outcomes that spread across different fields from the viewpoint of information geometry.
In this work, we introduce a new information geometric structure for the dynamics of probability and positive densities. 
In this structure, we consider not only the single dually flat structure built on the space of densities as in \cite{nakamura1993JapanJ.Indust.Appl.Math.,fujiwara1995PhysicaD:NonlinearPhenomena} but also another structure constructed on the space of fluxes.
These two structures are linked algebraically and topologically via the continuity equation and the gradient equation as illustrated in \fgref{fig:Intro}. 

Under this doubly dual flat structure, we can consider the dynamics of densities as a generalized flow, and various previous results can be unified in this framework. 
We exclusively consider dynamics of densities on finite-dimensional discrete manifolds, i.e., finite graphs or hypergraphs,  because the structure introduced here can be explicitly manifested in this setup and also because we do not need the mathematically elaborated setup for infinite-dimensional information geometry on a smooth manifold\cite{pistone2021ProgressinInformationGeometry:TheoryandApplications}.
For the case of FPE in a continuous state space, the dually flat structure built on the flux space can be reduced to the formal Riemannian geometric structure of $L^{2}$ Wasserstein geometry where the convex functions that induce the dually flat structure become quadratic. 
Our structure generalizes the linear inner product on the tangent and cotangent spaces with the nonlinear Legendre transform, thereby requiring information geometry.
By elucidating this information geometric structure, we can easily see that some quantities such as the bilinear product, convex thermodynamic potential functions, the Fisher information matrix, and the Fisher information number are consolidated into one quantity for FPE with the quadratic convex functions (see \secref{sec:dissp_CRN_LDG} and \secref{sec:FPE_functions}).
Therefore, our structure provides a way to unify the dualistic gradient flow mentioned in \secref{sec:Int_IG_dynamics} and also the information-number related topics in \secref{sec:Int_Inf_measure}.

\begin{figure}[h]
\includegraphics[bb=0 320 1024 768, width=\linewidth]{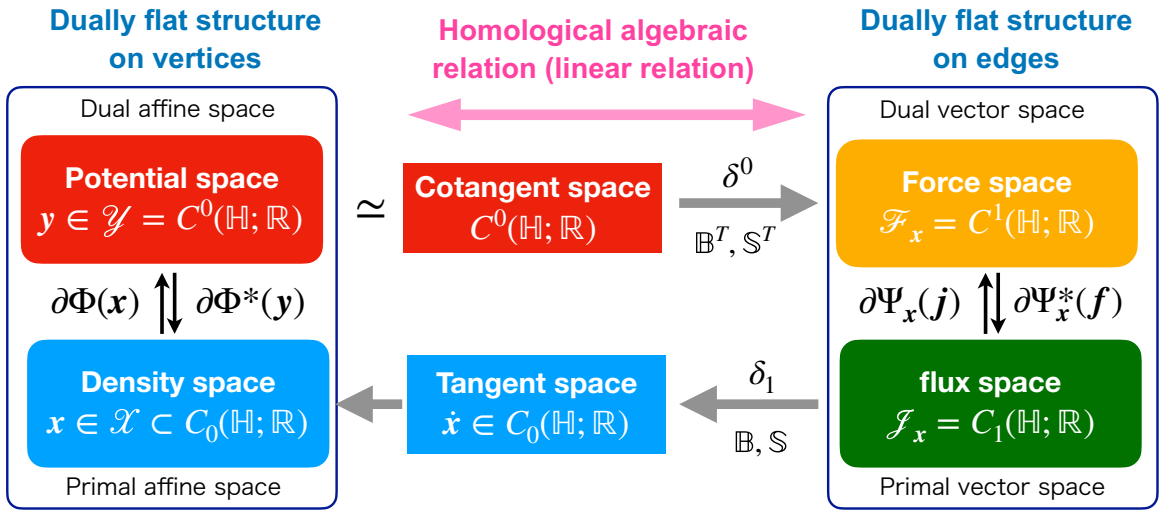}
\caption{Diagrammatic illustration of the doubly dual flat structure established on the vertex affine spaces (left) and the edge spaces (right), which are topologically related via the underlying graph or hypergraph.}
\label{fig:Intro}
\end{figure}

From the viewpoint of homological algebra, the structure we work on is a modification of the chain and cochain complexes of graphs or hypergraphs, which replace the usual inner product duality\cite{grady2010a} on each pair of chains and cochains with Legendre duality.
Moreover, the dually flat space built on the flux space is linked to a finite-dimensional version of Orlicz spaces\cite{musielak1983}, which have been employed for constructing infinite-dimensional information geometry\cite{pistone2021ProgressinInformationGeometry:TheoryandApplications}.
From the nice properties of the doubly dual flat structures, we can obtain information-geometric extensions of the Helmholtz-Hodge-Kodaira (HHK) decomposition (\thmref{thm:HHK}), the Otto calculus (\thmref{thm:inducedHessiangeometry1}), and its induction to cycle spaces(\thmref{thm:inducedHessiangeometry2}).

Our construction of an information geometry for dynamics is heavily based on the idea of using Legendre duality for the force and flux relation, proposed in the recent work of large deviations theory and the macroscopic fluctuation theorem for MJP and CRN led by A.Mieleke, R.I.A.Petterson, M.A.Peletier, D.R.M. Renger, J.Zimmer, and others\footnote{Ordered alphabetically.}\cite{mielke2014PotentialAnala,mielke2017SIAMJ.Appl.Math.,renger2018Entropy,kaiser2018JStatPhys,patterson2021ArXiv210314384Math-Ph,peletier2022Calc.Var.,renger2021DiscreteContin.Dyn.Syst.-S,peletier2022}.
We clarified its information-geometric aspects in the context of CRN and thermodynamics in our previous work\cite{kobayashi2022Phys.Rev.Researcha}.
We also concurrently elucidated the intimate link of equilibrium chemical thermodynamics and information geometry on the density state space \cite{kobayashi2022Phys.Rev.Research,sughiyama2022Phys.Rev.Research,sughiyama2022Phys.Rev.Researchb}.
In light of those, the contribution of this work is three-fold. 
First, we integrate these results in terms of information geometry, which clarifies the underlying geometric nature of the problem, provides transparent interpretations for known results, and leads to new information geometric results and insights (\thmref{thm:HHK}--\thmref{thm:inducedHessiangeometry2});
Second, this structure substantially extends the applicability of information geometry to a wide variety of dynamical problems; Lastly, the structure links information geometry to algebraic graph theory, discrete calculus, and homological algebra, which were not fully appreciated yet but provides a versatile way to consider the topology of the base manifold in information geometry.

\subsection{Organization of this paper}
This work is organized as follows: 
In \secref{sec:graph_hypergraph}, we introduce a range of models of dynamics on graphs and hypergraphs.
In \secref{sec:algebraic_homology}, we outline the homological algebra of graphs and hypergraphs.
In \secref{sec:doubly_dual_flat_spaces}, we abstractly introduce the doubly dual flat structures on the density and flux spaces and define the generalized flow associated with these structures.
In \secref{sec:explict_form}, we clarify that the introduced structures include a wide class of dynamics on graph and hypergraph.
In \secref{sec:orthogonality} and \secref{sec:central_affine}, we further define information-geometric objects and quantities, which naturally appear from this setup and play an integral role in the subsequent analysis of dynamics.
In \secref{sec:InfoGeoEquilibrium} and \secref{sec:InfoGeoNonEquilibrium}, we derive several results for equilibrium and nonequilibrium flows, respectively.
Finally, we provide a summary and prospects of our work in \secref{sec:summary}.
The notations and symbols are listed in the appendix.


\section{Classes of models for density dynamics on graph and hypergraph}\label{sec:graph_hypergraph}
In this work, we focus on linear and nonlinear dynamics defined on graphs\cite{godsil2013} and hypergraphs\cite{bretto2013}.

The linear dynamics of densities on graphs (LDG) includes Markov jump processes (MJP)\cite{meyn2009}, monomolecular chemical reaction networks\cite{schnakenberg1976Rev.Mod.Phys.},  and others\cite{godsil2013}.
We consider an extension of LDG to hypergraphs and nonlinear dynamics, common instances of which are chemical reaction networks (CRN) with the law of mass action (LMA) kinetics \cite{feinberg2019} and polynomial dynamical systems (PDS)\cite{craciun2019}. 
Because the extension we deal with in this work is a subclass of nonlinear dynamical systems on hypergraphs, we use CRN to designate this subclass.

In the following subsections, LDG and CRN are introduced using the language of algebraic graph theory\cite{biggs1997Bull.Lond.Math.Soc.,godsil2013}.
Then, we also give a brief and formal introduction of the Fokker-Planck equation (FPE)\cite{risken1996}, a linear dynamics of probability densities defined in Euclidean space.
We use the FPE throughout this paper only to contrast our results with the previous ones obtained for the FPE.

\subsection{Reversible Linear Dynamics of Densities on Graphs}
\begin{dfn}[Edge-weighted finite graph $\Graph_{\Vc{\kcoef}^{\pm}}$]\label{dfn:graph}
A finite graph $\Graph\defeq(\{\node_{i}\},\{\edge_{e}\},\IncMatrix)$ consists of $N_{\node} \in \Integer_{> 0}$ vertices, $\{\node_{i}\}_{i\in[1,N_{\node}]}$, and $N_{\edge} \in \Integer_{> 0}$ oriented edges, $\{\edge_{e}\}_{e\in[1,N_{\edge}]}$, each of which connects two different vertices\footnote{This means that we exclude self loops.} (\fgref{fig:LDG_CRN} (a)).  
The incidence relation is represented by the incidence matrix $\IncMatrix \in \{0,\pm 1\}^{N_{\node}\times N_{\edge}}$ where, for $\IncMatrix=(b_{i,e})$, 
\begin{align}
b_{i,e}&\defeq+1 && \mbox{if $\node_{i}$ is the head of edge $\edge_{e}$}, \notag \\
b_{i,e}&\defeq-1 && \mbox{if $\node_{i}$ is the tail of edge $\edge_{e}$}, \notag \\
b_{i,e}&\defeq0 && \mbox{otherwise}. \notag
\end{align}
An edge-weighted finite graph $\Graph_{\Vc{\kcoef}^{\pm}}\defeq(\{\node_{i}\},\{\edge_{e}\},\IncMatrix, \{\kcoef_{e}^{\pm}\})$ has two positive weighting parameters $\kcoef_{e}^{\pm}=(\kcoef_{e}^{+},\kcoef_{e}^{-})\in \Real_{> 0}$ for each edge $\edge_{e}$. 
Parameters $\kcoef_{e}^{+}$ and $\kcoef_{e}^{-}$ are denoted as forward and reverse rates or weights of edge $\edge_{e}$, respectively.
\end{dfn}
\begin{figure}[h]
\includegraphics[bb=0 400 1024 768, width=\linewidth]{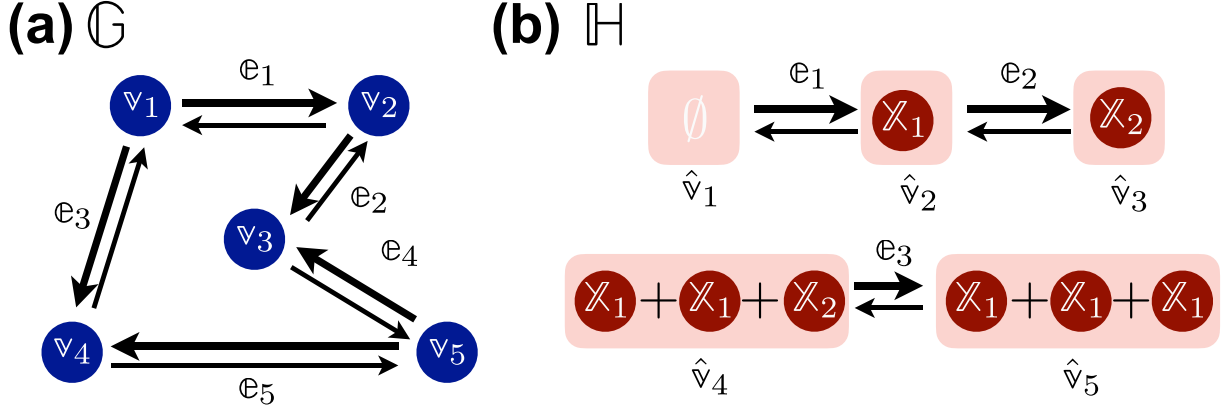}
\caption{Schematic diagrams of the reversible finite graph $\Graph$ (a) and CRN-hypergraph $\HGraph$ (b). Each pair of thick and thin arrows represents the pair of forward and reverse orientations of the corresponding edge. \TJK{The CRN-hypergraph $\HGraph$ in (b) corresponds to the simplified Brusselator reaction. The hypervertex $\hat{\node}_{1}$ contains no vertices.}}
\label{fig:LDG_CRN}
\end{figure}

A reversible linear dynamics (rLDG) on graphs is defined on the edge-weighted finite graph $\Graph_{\Vc{\kcoef}^{\pm}}$:
\begin{dfn}[Reversible linear dynamics of density on graph $\Graph_{\Vc{\kcoef}^{\pm}}$]\label{dfn:rLDG}
The reversible linear dynamics of the non-negative density $\Vc{x}(t)=(x_{1}(t), \cdots, x_{N_{\node}}(t))^{\Transpose}\in \Real_{\ge 0}^{N_{\node}}$ on $\Graph_{\Vc{\kcoef}^{\pm}}$ is defined by the continuity equation 
\begin{align}
    \dot{\Vc{x}}&=-\IncMatrix \Vc{\flux}(\Vc{x})=-\IncMatrix [\Vc{\flux}^{+}(\Vc{x})-\Vc{\flux}^{-}(\Vc{x})], \label{eq:LDM1}
\end{align}
and linear forward and reverse one-way fluxes $\Vc{\flux}^{\pm}(\Vc{x})=(\flux^{\pm}_{1}(\Vc{x}),\cdots, \flux^{\pm}_{N_{\edge}}(\Vc{x}))^{\Transpose}\in \Real^{N_{\edge}}_{\ge 0}$ with the following specific functional form\footnote{We may consider other functional forms for $\Vc{\flux}(\Vc{x})$, which can induce nonlinear dynamics on the graph. In this work, we focus mainly on the linear case.}:
\begin{align}
    \Vc{\flux}^{\pm}(\Vc{x})&=\Vc{\kcoef}^{\pm}\circ (\IncMatrix^{\pm})^{\Transpose}\Vc{x} ,\label{eq:LDM2}
\end{align}
where $\Vc{\flux}(\Vc{x})\defeq \Vc{\flux}^{+}(\Vc{x})-\Vc{\flux}^{-}(\Vc{x})$ is the total flux, the symbol $\circ$ denotes the component-wise product of two vectors\footnote{Also known as the Hadamard product or Schur product of vectors.}, and $\IncMatrix^{+}$ and $\IncMatrix^{-}$ are the head and tail incidence matrices defined respectively as $\IncMatrix^{+}\defeq\max[\IncMatrix,0]$ and $\IncMatrix^{-}\defeq\max[-\IncMatrix,0]$.
The incidence matrix $\IncMatrix$ in \eqnref{eq:LDM1} is often regarded as the discrete divergence operator on a graph \cite{grady2010a} and denoted also by $\Div_{\IncMatrix}=\IncMatrix$ to emphasize this interpretation in this work\footnote{This interpretation is because \eqnref{eq:LDM1} is associated with the continuity equation on a Euclidean space or on a Riemannian manifold where we have the divergence operator $\nabla \cdot$ instead of $\IncMatrix$. 
However, divergence on a Riemannian manifold implicitly includes the information of the metric via the Hodge operator.
On the contrary, $\IncMatrix$ does not.
From the viewpoint of homological algebra, $\IncMatrix$ should be regarded as the adjoint (transpose) of the discrete exterior derivative operator $\boundary^{0}\defeq\IncMatrix^{\Transpose}$, which is also often called a discrete gradient operator\cite{grady2010a}.
For a Euclidean space, they are the same.}.
\end{dfn}

Reversible Markov jump processes (rMJP) are a representative class of the rLDG describing random jumps of noninteracting particles on $\Graph_{\Vc{\kcoef}^{\pm}}$\footnote{We use the word 'reversible' in this work to mean that each edge allows both forward and reverse jumps while reversible Markov jump processes sometimes mean that the detailed balancing condition is satisfied. We introduce the notion of equilibrium later to designate the detailed balanced situation. }.
The weighting parameter $\kcoef_{e}^{+}$ is interpreted as the forward jump rate from the head of the oriented edge $\edge_{e}$ to its tail, whereas $\kcoef_{e}^{-}$ is the reverse jump rate from the tail to the head of $\edge_{e}$\footnote{If we allow $\kcoef_{e}^{\pm}$ to be $0$, we can include the irreversible MJP and also LDG in this formulation. We leave this extension for future work because it should require additional assumptions on the Legendre duality introduced in the subsequent sections.}.
For infinitely many such particles, we consider $p_{i}(t)\in[0,1]$, the fraction of particles on vertex $\node_{i}$ at time $t$, which is a non-negative density on vertices. 
Then, the forward and reverse one-way fluxes on the $e$th edge defined by \eqnref{eq:LDM2} are represented as
\begin{align}
\flux^{+}_{e}(\Vc{p})&=\kcoef_{e}^{+}p_{\node^{+}_{e}}\in \Real_{\ge 0}, & 
\flux^{-}_{e}(\Vc{p})&=\kcoef_{e}^{-}p_{\node^{-}_{e}}\in \Real_{\ge 0}, \label{MJP_each_flux}
\end{align}
where $\node^{+}_{e}$ and $\node^{-}_{e}$ are the head and tail vertices of edge $\edge_{e}$\footnote{Here, we have abused the notation $\node^{+}_{e}$ to indicate the index of the vertex $\node^{+}_{e}$. \eqnref{eq:LDM2} is reduced to \eqnref{MJP_each_flux} because $b^{\pm}_{i,e}=+1$ only when $i$ is the index of the head (tail) vertex $\node^{\pm}_{e}$ of $\edge_{e}$ and $0$ otherwise.}.
The linearity of $\flux^{\pm}_{e}(\Vc{p})$ with respect to $\Vc{p}$ comes from the independence of particles on the graph.
Then, the continuity equation (\eqnref{eq:LDM1}) with the state vector $\Vc{p}(t)\defeq (p_{1}(t), \cdots, p_{N_{\node}}(t))^{\Transpose}\in \Real^{N_{\node}}_{\ge 0}$ is reduced to the master equation: $\dot{\Vc{p}}=-\IncMatrix \Vc{\flux}(\Vc{p})$.

\begin{dfn}[Weighted asymmetric graph Laplacian\cite{biggs1997Bull.Lond.Math.Soc.,chung1997}]\label{dfn:wagLaplacian}
For $\Graph_{\Vc{\kcoef}^{\pm}}$, the corresponding weighted asymmetric graph Laplacian is defined by 
\begin{align}
\gLap_{\Vc{\theta}}\defeq \IncMatrix \left[\diag[\Vc{\kcoef}^{+}] (\IncMatrix^{+})^{\Transpose}-\diag[\Vc{\kcoef}^{-}] (\IncMatrix^{-})^{\Transpose} \right], \label{eq:graphLaplacian}
\end{align}
where $\Vc{\theta}\defeq (\Vc{\kcoef}^{+},\Vc{\kcoef}^{-})$ and $\diag[\Vc{\kcoef}^{+}]$ is the diagonal matrix whose diagonal elements are $\Vc{\kcoef}^{+}$.
Using $\gLap_{\Vc{\theta}}$, \eqnref{eq:LDM1} and \eqnref{eq:LDM2} are represented as 
\begin{align}
    \dot{\Vc{x}}&=-\IncMatrix \Vc{\flux}(\Vc{x})=-\IncMatrix \left[\diag[\Vc{\kcoef}^{+}] (\IncMatrix^{+})^{\Transpose}-\diag[\Vc{\kcoef}^{-}] (\IncMatrix^{-})^{\Transpose} \right]\Vc{x}=-\gLap_{\Vc{\theta}}\Vc{x}. \label{eq:LDMall}
\end{align}
\end{dfn}
The operator $\gLap_{\Vc{\theta}}$ is reduced to the weighted symmetric graph Laplacian if $\Vc{\kcoef}^{+}=\Vc{\kcoef}^{-}$ and also to the conventional graph Laplacian if $\Vc{\kcoef}^{+}=\Vc{\kcoef}^{-}=\Vc{1}$\cite{biggs1997Bull.Lond.Math.Soc.,chung1997}. 
\eqnref{eq:LDMall} can also cover linear transport on graphs, linear electric circuits \cite{dorfler2018Proc.IEEE}, consensus dynamics on graphs \cite{saber2003Proc.2003Am.ControlConf.2003}, and other linear dynamics on graphs \cite{godsil2013,veerman2020}\footnote{For some of these applications, the relevant state space is $\Real^{N_{\node}}$ instead of $\Real^{N_{\node}}_{\ge 0}$.}.

\subsection{Chemical reaction network and polynomial dynamical systems on hypergraphs}
Next, we introduce a class of nonlinear dynamics on hypergraphs, which includes the rLDG (\eqnref{eq:LDM1} and \eqnref{eq:LDM2}) as a special case.
The most common instance is deterministic chemical reaction networks (CRN) with the law of mass action (LMA) kinetics\cite{waage1986J.Chem.Educ.a,beard2008,qian2021a,feinberg2019}, and this class is sometimes referred to as polynomial dynamical systems (PDS). 
Because the major part of the PDS theory has been developed for CRN, we use CRN to introduce and specify this class in this work. 

\begin{dfn}[Reversible edge-weighted CRN hypergraph $\HGraph_{\Vc{\kcoef}^{\pm}}$]\label{dfn:rCRNhyperG}
The reversible CRN hypergraph $\HGraph\defeq (\{\molX_{i}\},\{\edge_{e}\}, \cmMatrix, \IncMatrix)$ consists of a finite number of vertices $\{\molX_{i}\}_{i\in[1,N_{\molX}]}$ and hyperedges $\{\edge_{e}\}_{e\in[1,N_{\edge}]}$ where $N_{\molX}, N_{\edge} \in \Integer_{>0}$ (\fgref{fig:LDG_CRN} (b)).
Each hyperedge $\edge_{e}$ connects two different hypervertices $\hat{\node}^{+}_{e}$ and $\hat{\node}^{-}_{e}$ where $\hat{\node}^{+}_{e} \neq \hat{\node}^{-}_{e}$\footnote{This means that we exclude self loop hyperedges. However, the head and tail hypervertices are allowed to contain the same vertices as long as $\hat{\node}^{+}_{e} \neq \hat{\node}^{-}_{e}$ holds.}.
The hypervertices are multisets of vertices $\{\molX_{i}\}_{i\in[1,N_{\molX}]}$, each of which is defined as $\hat{\node}_{\ell}=\sum_{i=1}^{N
_{\molX}}\cmvector_{i,\ell}\molX_{i}$ where $\cmvector_{i,\ell}\in \Integer_{\ge 0}$ is the number of the $i$th vertex included in the $\ell$th hypervertex\footnote{Our definition of CRN hypergraph differs in a couple of aspects from the conventional definition because of the additional information required to define CRN. For example, while the definition of edges is usually extended from those of graphs\cite{bretto2013}, our definition extends vertices instead. }. 
Thus, the nonnegative integer vector $\Vc{\cmvector}_{\ell}\defeq (\cmvector_{1,\ell}, \cdots, \cmvector_{N_{\molX},\ell})^{\Transpose}\in\Integer_{\ge 0}^{N_{\molX}}$ defines the $\ell$th hypervertex. 
Let $N_{\hat{\node}} \in \Integer_{>0}$ be the total number of the hypervertices and $\cmMatrix\defeq (\Vc{\cmvector}_{1}, \cdots, \Vc{\cmvector}_{N_{\hat{\node}}})\in\Integer_{\ge 0}^{N_{\molX}\times N_{\hat{\node}}}$ be the hypervertex matrix.
The matrix $\IncMatrix \in \{0,\pm 1\}^{N_{\hat{\node}}\times N_{\edge}}$ is the incidence matrix encoding the incidence relations among the hypervertices and the hyperedges.
The hypergraph incidence matrix $\HIncMatrix \in \Integer^{N_{\molX} \times N_{\edge}}$ is then defined as 
\begin{align}
    \HIncMatrix\defeq \cmMatrix \IncMatrix.\label{eq:CRN_stoiMatrix}
\end{align}
If $\cmMatrix=\identityM$ where $\identityM$ is the identity matrix, then $\HGraph=(\{\molX_{i}\}_{i\in[1,N_{\molX}]},\{\edge_{e}\}_{e\in[1,N_{\edge}]}, \cmMatrix, \IncMatrix)$ is reduced to $\Graph=(\{\node_{\ell}\}_{\ell\in [1,N_{\molX}]},\{\edge_{e}\}_{e\in[1,N_{\edge}]},\IncMatrix)$ where $\node_{\ell}=\molX_{\ell}$.
An edge-weighted CRN hypergraph $\HGraph_{\Vc{\kcoef}^{\pm}}\defeq (\{\molX_{i}\},\{\edge_{e}\}, \cmMatrix, \IncMatrix, \{\kcoef_{e}^{\pm}\})$ has forward and reverse rates $\kcoef_{e}^{\pm}> 0$ as weights  of edge $\edge_{e}$.
\end{dfn}

In the context of CRN theory, the vertices $\{\molX_{i}\}$ correspond to the molecular species involved in a CRN, and each hyperedge $\edge_{e}$ represents a pair of forward and reverse reactions: 
\begin{align}
\gamma_{1,e}^{+}\molX_{1}+\cdots+\gamma_{N_{\molX},e}^{+}\molX_{N_{\molX}}
\xrightleftharpoons[]{} \gamma^{-}_{1,e}\molX_{1}+\cdots+\gamma^{-}_{N_{\molX},e}\molX_{N_{\molX}},\label{eq:mthReaction}
\end{align}
where the forward and reverse reactions are from left to right and from right to left, respectively. 
Head and tail hypervertices $\hat{\node}_{e}^{+}\defeq (\gamma^{+}_{1,e}\molX_{1}+\cdots+\gamma^{+}_{N_{\molX},e}\molX_{N_{\molX}})$ and $\hat{\node}_{e}^{-}\defeq(\gamma^{-}_{1,e}\molX_{1}+\cdots+\gamma^{-}_{N_{\molX},e}\molX_{N_{\molX}})$ in \eqnref{eq:mthReaction} are the sets of reactants and products of the $e$th forward reaction, respectively. 
More specifically, $\gamma^{+}_{i,e}\in \Integer_{\ge 0}$ and $\gamma^{-}_{i,e}\in \Integer_{\ge 0}$ are the numbers of the molecule $\molX_{i}$ involved as the reactants and products of the $e$th forward reaction, respectively.
For the reverse reaction, $\hat{\node}_{e}^{-}$ and $\hat{\node}_{e}^{+}$ are the reactants and products.
Some head and tail hypervertices are overlapping among different reactions (hyperedges) as in \fgref{fig:LDG_CRN} (b).
As a result, $\{\hat{\node}_{\ell}\}_{\ell \in N_{\hat{\node}}}$ is the union of the head and tail hypervertices, $\{\hat{\node}_{\ell}\}_{\ell \in N_{\hat{\node}}}=\bigcup_{e \in N_{\edge}}\{\hat{\node}_{e}^{+}, \hat{\node}_{e}^{-}\}$.

The hypervertices are called complexes in CRN theory \cite{feinberg2019}\footnote{Because we use the term {\it complex} for the cell complex in homological algebra, we use hypervertices to indicate the complexes of the CRN theory.}.
From $\{\gamma^{+}_{i,e}\}$ and $\{\gamma^{-}_{i,e}\}$, we can define
\begin{align}
\Vc{\stoiVector}_{e}\defeq(\gamma^{+}_{1,e}-\gamma^{-}_{1,e},\cdots, \gamma^{+}_{N_{\molX},e}-\gamma^{-}_{N_{\molX},e})^{\Transpose}\in \Integer^{N_{\molX}}, \label{eq:HIncM}
\end{align}
where $\mp \Vc{\stoiVector}_{e}$ specify the change in the number of molecules induced when the $e$th forward or reverse reaction occurs just once, respectively.
The hypergraph incidence matrix $\HIncMatrix$ defined in \eqnref{eq:CRN_stoiMatrix} is represented as $\HIncMatrix= (\Vc{\stoiVector}_{1},\cdots, \Vc{\stoiVector}_{N_{\edge}})\in \Integer^{N_{\molX} \times N_{\edge}}$.
 In chemistry, the negative of $\Vc{\stoiVector}_{e}$ and $\HIncMatrix$, i.e., $-\Vc{\stoiVector}_{e}$ and $-\HIncMatrix$, are called the stoichiometric vector and matrix, respectively\cite{feinberg2019}.

\begin{rmk}[]
To define a reversible CRN hypergraph, the hypergraph matrix $\HIncMatrix$ is not sufficient. 
If the head and tail hypervertices of a hyperedge contain the same vertex (molecule), the corresponding element in $\HIncMatrix$ of such a shared vertex becomes $0$ by canceling out. 
Thus, the existence of shared vertices (molecules) is invisible in $\HIncMatrix$, and the pair $(\cmMatrix, \IncMatrix)$ is required to define $\HGraph$.
Such shared molecules are called catalysts in CRN. 
\end{rmk}
For a CRN hypergraph, the continuity equation for CRN is defined:
\begin{dfn}[CRN continuity equation]
Let a vector of nonnegative densities $\Vc{x}=(x_{1},\cdots,x_{N_{\molX}})^{\Transpose}\in \Real_{\ge 0}^{N_{\molX}}$ represents the concentration of molecules $\{\molX_{i}\}$. 
The CRN continuity equation is defined as 
\begin{align}
    \dot{\Vc{x}}=-\HIncMatrix \Vc{\flux}(\Vc{x})= -\Div_{\HIncMatrix} \Vc{\flux}(\Vc{x}),\label{CRN_rateEq}
\end{align}
where $\flux_{e}^{+}(\Vc{x})\in \Real_{\ge 0}$ and $\flux_{e}^{-}(\Vc{x})\in \Real_{\ge 0}$ are the one-way fluxes of the $e$th forward and reverse reactions, $\Vc{\flux}^{\pm}(\Vc{x})\defeq(\flux_{1}^{\pm}(\Vc{x}),\cdots,\flux_{N_{\edge}}^{\pm}(\Vc{x}))^{\Transpose}\in \Real_{\ge 0}^{N_{\edge}}$ are their vector representations, and $\Vc{\flux}(\Vc{x})\defeq \Vc{\flux}^{+}(\Vc{x})-\Vc{\flux}^{-}(\Vc{x})\in \Real^{N_{\edge}}$ is the total reaction flux\cite{beard2008,qian2021a,feinberg2019}.
The hypergraph divergence operator $\Div_{\HIncMatrix} \defeq \HIncMatrix$ is defined accordingly.
\end{dfn}
To define the dynamics of a CRN, the functional form of $\flux_{e}^{\pm}(\Vc{x})$ is required\footnote{Even though the functional form of $\flux_{e}^{\pm}(\Vc{x})$ is automatically determined in the case of LDG because of the linearity, we have multiple possibilities to define nonlinear $\flux_{e}^{\pm}(\Vc{x})$. }. 
Before introducing specific forms, we define two important properties of the fluxes and also other functions defined on edges: 
\begin{dfn}[Consistency of fluxes $\Vc{\flux}^{\pm}(\Vc{x})$ with hypergraph $\HGraph$]
One-way fluxes $\Vc{\flux}^{\pm}(\Vc{x})$ are consistent with the hypergraph $\HGraph$ if, for all $e\in[1,N_{\edge}]$, $\flux^{\pm}_{e}(\Vc{x})$ becomes $0$ when $x_{i}=0$ where $\molX_{i}$ is any reactant of $\flux^{\pm}_{e}(\Vc{x})$, respectively.
In other words, $\flux^{\pm}_{e}(\Vc{x})$ satisfies $\cmvector_{i,e}^{\pm}\flux_{e}^{\pm}(\Vc{x})=0$ if $x_{i}=0$ for any $i\in[1,N_{\molX}]$.
\end{dfn}
\begin{dfn}[Locality of function on edges over $\HGraph$]
A vector function $\Vc{g}(\Vc{x})\in \Real^{N_{\edge}}$ defined on edges is local on $\HGraph$ if, for all $e\in[1,N_{\edge}]$, $g_{e}(\Vc{x})$ is a function only of the elements of $\Vc{x}$ incident to the edge $\edge_{e}$ on $\HGraph$, i.e., $g_{e}(\Vc{x})=g_{e}(\bar{\gamma}_{1,e}^{+} x_{1},\cdots, \bar{\gamma}_{N_{\molX},e}^{+}x_{N_{\molX}},\bar{\gamma}_{1,e}^{-} x_{1},\cdots, \bar{\gamma}_{N_{\molX},e}^{-} x_{N_{\molX}})$ where $\bar{\gamma}_{i,e}^{\pm}\defeq \min[1, \gamma_{i,e}^{\pm}] \in \{0,1\}$.
\end{dfn}

The consistency condition is indispensable to prohibit a reaction that can decrease $x_{i}$ from occurring when $x_{i}=0$. 
For $\Vc{\flux}^{\pm}(\Vc{x})$, the locality means that the fluxes of the $e$th reaction depend only on the concentrations of their reactants and products.
The local flux is determined solely by the information stored on the vertices incident to the edge and plays a crucial role when we regard the structure introduced in this work as an extension of differential forms on continuous manifolds to graphs and hypergraphs.
When we work on specific forms of fluxes in this work, we consider only local fluxes consistent with the given hypergraph $\HGraph$.

In chemistry, we have a variety of candidates for the functional form of flux, e.g., the Michaelis-Menten function, Hill's function, and others\cite{beard2008,keener2008}. 
Among others, the LMA kinetics is the most basic and well-established one.

\begin{dfn}[Waage-Guldberg’s law of mass action kinetics (LMA kinetics)]
A CRN follows the LMA kinetics if, for all $e\in[1,N_{\edge}]$, the $e$th forward and reverse reaction fluxes are represented as
\begin{align}
    \flux^{\pm}_{e}(\Vc{x})=\kcoef^{\pm}_{e} \prod_{j=1}^{N_{\molX}}x_{j}^{\gamma^{\pm}_{j,e}}=\kcoef^{\pm}_{e} \sum_{\ell=1}^{N_{\hat{\node}}}b_{\ell,e}^{\pm}\prod_{j=1}^{N_{\molX}}x_{j}^{\gamma_{j,\ell}},\label{CNR_each_flux}
\end{align}
where $\kcoef_{e}^{+}\in \Real_{> 0}$ and $\kcoef_{e}^{-}\in \Real_{> 0}$ are the reaction rate constants of the $e$th forward and reverse reactions, respectively.
The fluxes under LMA kinetics can be compactly represented as 
\begin{align}
    \Vc{\flux}^{\pm}_{\mathrm{MA}}(\Vc{x})=\Vc{\kcoef}^{\pm}\circ (\IncMatrix^{\pm})^{\Transpose}\Vc{x}^{\cmMatrix^{\Transpose}},\label{eq:CRN_rate}
\end{align}
where $\Vc{x}^{\Vc{\cmvector}}\defeq \prod_{j=1}^{N_{\molX}} x_{j}^{\cmvector_{j}}\in \Real_{\ge 0}$ and $\Vc{x}^{\cmMatrix^{\Transpose}}\defeq (\Vc{x}^{\Vc{\cmvector}_{1}}, \cdots, \Vc{x}^{\Vc{\cmvector}_{N_{\hat{\node}}}})^{\Transpose}\in \Real_{\ge 0}^{N_{\hat{\node}}}$\footnote{We should note an important relation, $(\IncMatrix^{\pm})^{\Transpose}\Vc{x}^{\cmMatrix^{\Transpose}}=\Vc{x}^{(\cmMatrix\IncMatrix^{\pm})^{\Transpose}}$, which holds because every column vector of $\IncMatrix^{\pm}$ contains only one $+1$ and the others are $0$.}.
We use the subscript $\mathrm{MA}$ as in $\Vc{\flux}^{\pm}_{\mathrm{MA}}(\Vc{x})$ to discriminate this specific form of the fluxes from others.
We can easily observe that $\Vc{\flux}^{\pm}_{\mathrm{MA}}(\Vc{x})$ is consistent and local with respect to $\HGraph$. 
Furthermore, $\Vc{\flux}^{\pm}_{\mathrm{MA}}(\Vc{x})$ is specified by the edge-weighted CRN hypergraph 
$\HGraph_{\Vc{\kcoef}^{\pm}}\defeq (\{\molX_{i}\},\{\edge_{e}\}, \cmMatrix, \IncMatrix, \{\kcoef_{e}^{\pm}\})$.
\end{dfn}

\begin{rmk}[Algebraic aspect of LMA kinetics]
Because $\Vc{x}^{\cmMatrix^{\Transpose}}$ is a vector of monomials of $\Vc{x}$,  each one-way flux, $\flux^{\pm}_{e}(\Vc{x})$, is a monomial of $\Vc{x}$ under \eqnref{eq:CRN_rate} and thus the total flux $\flux_{e}(\Vc{x})=\flux^{+}_{e}(\Vc{x})-\flux^{-}_{e}(\Vc{x})$ is a binomial.
This fact links the real algebraic geometry of the toric varieties\cite{sottile2008arXiv:math/0212044,cox2011} to CRN\cite{craciun2009JournalofSymbolicComputation,kobayashi2021ArXiv211214910Phys.} as it does in algebraic statistics\cite{pachter2005,rapallo2007AISM}.
\end{rmk}
\begin{rmk}[Extended LMA kinetics]
While we mainly work on the normal LMA kinetics, we can extend it.
The extended LMA kinetics defined on $\HGraph$ is defined as 
\begin{align}
    \Vc{\flux}^{\pm}_{\mathrm{eMA}}(\Vc{x})=\Vc{\kcoef}^{\pm}\circ \Vc{g}(\Vc{x})\circ (\IncMatrix^{\pm})^{\Transpose}\Vc{x}^{\cmMatrix^{\Transpose}},\label{eq:CRN_rate_eLMA}
\end{align}
where $\Vc{g}(\Vc{x}) \in \Real_{> 0}^{N_{\edge}}$ and is local with respect to $\HGraph$\footnote{There exists another type of extension known as generalized LMA kinetics where the monomials $\Vc{x}^{\cmMatrix^{\Transpose}}$ are replaced with fractional monomials, i.e., powers of $\Vc{x}$ with nonnegative real-valued exponents\cite{muller2012SIAMJ.Appl.Math.}.}.
An example of the extended LMA kinetics is reversible Michaelis-Menten kinetics\cite{noor2013FEBSLetters}.
\end{rmk}

By combining the continuity equation (\eqnref{CRN_rateEq}) and the LMA kinetics (\eqnref{eq:CRN_rate}), we have the following chemical rate equation:
\begin{align}
    \dot{\Vc{x}}&=-\HIncMatrix \Vc{\flux}_{\mathrm{MA}}(\Vc{x})=-\cmMatrix\IncMatrix \left[\diag[\Vc{\kcoef}^{+}] (\IncMatrix^{+})^{\Transpose}-\diag[\Vc{\kcoef}^{-}] (\IncMatrix^{-})^{\Transpose} \right]\Vc{x}^{\cmMatrix^{\Transpose}}=-\cmMatrix \gLap_{\Vc{\theta}} \Vc{x}^{\cmMatrix^{\Transpose}}, \label{CRN_rateEqall}
\end{align}
where $\gLap_{\Vc{\theta}}$ is the weighted asymmetric graph Laplacian defined as in \eqnref{eq:graphLaplacian}.
Now, we can see that CRN contains rLDG (\eqnref{eq:LDMall}) as a special case if $\cmMatrix=\identityM$.
Owing to this inclusion relation, CRN with LMA kinetics is a mathematically sound generalization of rLDG.
Because LDG has been used in various fields of social science, network science, machine learning, and so on, CRN theory is potentially important for extending the results there.

\TJK{\begin{ex}[Simplified Brusselator CRN\cite{feinberg2019,yoshimura2022}]\label{ex:BrusselatorEq}
The Brusselator is a representative CRN, which can generate non-trivial dynamic behaviors such as oscillations. 
We use a reversible CRN version of the simplified Brusselator\cite{feinberg2019,yoshimura2022}, whose CRN-hypergraph depicted in  \fgref{fig:LDG_CRN} (b) has the following structural information: 
\begin{align}
\cmMatrix&=\bordermatrix{
&\hat{\node}_{1} & \hat{\node}_{2} & \hat{\node}_{3} & \hat{\node}_{4} & \hat{\node}_{5}\cr
\molX_{1} &0 & 1  & 0  & 2 & 3 \cr
\molX_{2} &0 & 0  & 1  & 1 & 0 \cr
},&
\IncMatrix&=\bordermatrix{
& \edge_{1} & \edge_{2} & \edge_{3} \cr
\hat{\node}_{1} & +1 & 0 & 0 \cr
\hat{\node}_{2} & -1 & +1 & 0 \cr
\hat{\node}_{3} & 0 & -1 & 0 \cr
\hat{\node}_{4} & 0 & 0 & +1 \cr
\hat{\node}_{5} & 0 & 0 & -1 \cr
},& 
\HIncMatrix&=\bordermatrix{
& \edge_{1} & \edge_{2} & \edge_{3} \cr
\molX_{1} & -1 & +1 & -1 \cr
\molX_{2} & 0 & -1 &+1 \cr
}
\end{align}, 
\end{ex}
The rate equation (\eqnref{CRN_rateEqall}) can be represented as
\begin{align}
\frac{\dd}{\dt}\begin{pmatrix}x_{1}\\x_{2}\end{pmatrix}
=- \overbrace{\begin{pmatrix}
 -1 & +1 & -1 \\
 0 & -1 &+1 
\end{pmatrix}}^{\HIncMatrix}
\left[
\overbrace{
\begin{pmatrix}
\kcoef_{1}^{+} \\ \kcoef_{2}^{+}x_{1} \\ \kcoef_{3}^{+} x_{1}^{2}x_{2}
\end{pmatrix}}^{\Vc{\flux}^{+}(\Vc{x})}
-
\overbrace{\begin{pmatrix}
\kcoef_{1}^{-}x_{1} \\ \kcoef_{2}^{-}x_{2} \\ \kcoef_{3}^{-} x_{1}^{3}
\end{pmatrix}}^{\Vc{\flux}^{-}(\Vc{x})}
 \right].
\end{align}
}

\subsection{Fokker Planck Equations}
While our main focus is the dynamics on graphs and hypergraphs, we use FPE as a representative class of density dynamics on a continuous Euclidean space.
Specifically, we use FPE only to demonstrate the relation of our results with previous ones obtained for FPE in various contexts.
Because FPE is infinite-dimensional, we treat it here only formally.

Let $\Vc{r}\in \Real^{d}$ be a vector in a $d$ dimensional Euclidean space.
We consider infinitely many noninteracting particles randomly walking in the space and describe the dynamics by a probability density $p_{t}(\Vc{r})\in \Real_{\ge 0}$ of the particles.
The continuity equation for $p_{t}(\Vc{r})$ is  
\begin{align}
    \partial_{t}p_{t}(\Vc{r})&=-\nabla\cdot \Vc{\flux}_{\mathrm{FP}}[p_{t}(\Vc{r})]
\end{align}
where $\Vc{j}_{\mathrm{FP}}[p_{t}(\Vc{r})]\in \Real^{d}$ is the probability flux, $\nabla \defeq (\partial/\partial r_{1}, \cdots, \partial/\partial r_{d})^{\Transpose}$ is the gradient operator on the Euclidean space, and $(\nabla\cdot): \nabla\cdot\Vc{F}(\Vc{r})\defeq \sum_{i=1}^{d}\partial F_{i}(\Vc{r})/\partial r_{i} \in \Real$ is the divergence. 
The flux of the FPE is defined as
\begin{align}
     \Vc{\flux}_{\mathrm{FP}}[p(\Vc{r})]&= \left[\Vc{F}(\Vc{r})p(\Vc{r}) - D_{0} \nabla  p(\Vc{r})\right]
     , \label{eq:flux_FPE}
\end{align}
where $\Vc{F}(\Vc{r})\in \Real^{d}$ is the drift force, and $D_{0} \in \Real_{>0}$ is the diffusion constant.

\section{Discrete Calculus and Homological Algebra of Graphs and Hypergraphs}\label{sec:algebraic_homology}
The algebraic and topological structure of the dynamics on graphs and hypergraphs can be explicitly and abstractly treated using the language of discrete calculus and homological algebra. 
The discrete version of the gradient and divergence mentioned in \secref{sec:graph_hypergraph} is also characterized.
In this section, we briefly introduce the chain and cochain complexes defined for a finite graph or a hypergraph and discrete calculus\cite{biggs1997Bull.Lond.Math.Soc.,grady2010a,lim2020SIAMRev.,sunada2012}.
We first introduce the complexes for a graph $\Graph$ and then extend them to a hypergraph $\HGraph$ algebraically\footnote{The complex used here should not be confused with complexes used in CRN theory\cite{feinberg2019}}.
It should be noted that the conventional discrete calculus (the discrete version of the theory for differential forms) presumes the Riemannian metric structure in the dual space of chains and cochains or that of cochains on primal and dual complexes \cite{desbrun2005,hirani2003}. 
However, we are going to introduce Legendre duality instead. 
For this purpose, our introduction of chain and cochain complexes depends only on the topological (algebraic) information of the underlying graph and hypergraph\cite{grady2010a} without specifying the metric information. 

\subsection{Chain and cochain complexes on graphs}\label{sec:chain_cochain}
The elements of a graph $\Graph$ are called cells in discrete calculus \footnote{We follow the terminology in \cite{grady2010a}. While we use ``cell'', we do not presume any $N$-dimensional topological manifold underlying the graph. The graph is just treated algebraically as in algebraic graph theory and homological algebra.}.
A vertex and an edge are, respectively, called $0$-cell and $1$-cell, and the graph $\Graph$ is denoted as a cell-complex\footnote{Depending on the choice of which elements of a graph are considered, the content of the complex changes. For example, vertices and edges are the major ingredients of the complex of a graph. The faces of a graph are often included in the complex. The definition of the higher-order elements than edges requires additional structural information to the incidence matrix of the graph, e.g., the edge-face incidence matrix. }. 
For each type of the cells, we consider vectors (chains and cochains) defined on the cells.
For $\Graph$, a $0$-chain with field $\Real$ is an $N_{\node}$-tuple of real scalars, each of which is assigned to a vertex, i.e., a $0$ cell.
Thus, a $0$-chain is a real vector defined on the vertices of $\Graph$ with the basis $\{\node_{i}\}$.
This basis is called the standard basis.
The vector space of real $0$-chains is called the vertex space here and denoted as $\chain_{0}(\Graph)=\Real^{N_{\node}}$\cite{biggs1997Bull.Lond.Math.Soc.}\footnote{In algebraic graph theory, the chain of a graph $\Graph$ is defined as an integer-valued vector space $\Integer^{N_{\node}}$ to represent the discrete and combinatorial nature of $\Graph$ and also to specify the domain of integration. Here, we use $\Real$ as the field of the vector space. }. 
The components of the vector $\Vc{x} \in \chain_{0}(\Graph)$ are given as $\Vc{x}(\node_{i})\defeq x_{i}$. 
Similarly, a real $1$-chain is a real vector defined on the edges of $\Graph$. 
The real vector space of $1$-chains is called the edge space and denoted as $\chain_{1}(\Graph)=\Real^{N_{\edge}}$.
The standard basis is introduced by using edges $\{\edge_{e}\}$, accordingly.
A flux $\Vc{\flux}$ is a $1$-chain: $\Vc{\flux}(\edge_{e})\defeq\flux_{e}$.
The graph incidence matrix $\IncMatrix$ induces the discrete differential $\boundary_{1}: \chain_{1}(\Graph) \to \chain_{0}(\Graph)$ as $\boundary_{1}\Vc{\flux}\defeq \IncMatrix \Vc{\flux}$\footnote{In algebraic graph theory, $\IncMatrix$ is also identical to the discrete boundary operator from $\chain_{1}(\Graph;\Integer)$ to $\chain_{0}(\Graph;\Integer)$.}.

To obtain an exact sequence, we algebraically define the $(-1)$ and $2$ chains and the corresponding differentials $\boundary_{0}$ and $\boundary_{2}$.
Let $\chain_{2}(\Graph)=\Real^{N_{\cycle}}$ where $N_{\cycle}=\Dim [\Ker \IncMatrix]$ and $\{\Vc{v}_{i}\}_{i\in[1,N_{\cycle}]}$ is a set of complete basis of $\Ker \IncMatrix$ where $\Vc{v}_{i} \in \{0,+1,-1\}^{N_{\edge}}$\footnote{$N_{\cycle}=0$ when $\Graph$ is a set of trees.}.
In algebraic graph theory, $\Ker \IncMatrix$ is called a cycle subspace\cite{biggs1997Bull.Lond.Math.Soc.,knauer2011,godsil2013}.
For a graph $\Graph$, we can construct $\{\Vc{v}_{i}\}_{i\in[1,N_{\cycle}]}$ by, for example, using the fundamental cycle basis of $\Graph$ obtained from a fixed spanning tree of $\Graph$\footnote{The spanning tree chosen specifies a fundamental cycle and cocycle bases.}\cite{godsil2013}. 
Thus, $\chain_{2}(\Graph)$ is the vector space defined on the cycles of $\Graph$ and isomorphic to the cycle subspace.
We define a matrix, $\cycMatrix \defeq (\Vc{v}_{1},\cdots, \Vc{v}_{N_{\cycle}})$\footnote{$\cycMatrix^{\Transpose}$ is called the fundamental tieset matrix in graph theory}, and the differential $\boundary_{2}: \chain_{2}(\Graph) \to \chain_{1}(\Graph)$ as $\boundary_{2}\defeq \cycMatrix$.
From the construction, $\IncMatrix \cycMatrix=\boundary_{1}\boundary_{2}=0$ and $\Img[\boundary_{2}]=\Ker[\boundary_{1}]$ hold.
Similarly, let $\chain_{-1}(\Graph)=\Real^{N_{\cons}}$ where $N_{\cons}=\Dim [\Ker \IncMatrix^{\Transpose}]$ and $\{\Vc{u}_{\ell}\}_{\ell\in[1,N_{\cons}]}$ is a set of complete basis of $\Ker \IncMatrix^{\Transpose}$ where $\Vc{u}_{\ell} \in \{0,+1,-1\}^{N_{\node}}$. 
The subspace $\Ker \IncMatrix^{\Transpose}$ is related to the connected components of $\Graph$ and $\Vc{u}_{i}$ can be chosen such that $u_{i,\ell}=+1$ if the $i$th vertex is included in the $\ell$th connected component and $u_{i,\ell}=0$, otherwise.
Thus, $\chain_{-1}(\Graph)$ is the vector space on the connected components.
From the matrix $\consMatrix \defeq (\Vc{u}_{1},\cdots, \Vc{u}_{N_{\cons}})^{\Transpose}$, the differential $\boundary_{0}: \chain_{0}(\Graph) \to \chain_{-1}(\Graph)$ is defined as $\boundary_{0}\defeq \consMatrix$.
From the construction, $\consMatrix\IncMatrix =\boundary_{0}\boundary_{1}=0$ and $\Img[\boundary_{1}]=\Ker[\boundary_{0}]$ hold. 
Then, we obtain the exact chain sequence\footnote{We should note that the sequence is not canonical because $\consMatrix$ and $\cycMatrix$ depend on the choice of bases.}\footnote{Upon necessity, we can consider the harmonic components by employing an under-complete basis for $\cycMatrix$.}:
\begin{align}
0 \xleftarrow{} \chain_{-1}(\Graph)\xleftarrow{\boundary_{0}=\consMatrix}\chain_{0}(\Graph)\xleftarrow{\boundary_{1}=\IncMatrix}\chain_{1}(\Graph)\xleftarrow{\boundary_{2}=\cycMatrix}\chain_{2}(\Graph)\xleftarrow{}0.
\end{align}

Because $\chain_{p}(\Graph)$ is a vector space for each $p\in\{-1,0,1,2\}$, we can consider its dual vector space $\chain^{p}(\Graph)\defeq \chain_{p}^{*}(\Graph)$ consisting of the linear functions on $\chain_{p}(\Graph)$.
An element of $\chain^{p}(\Graph)$ is called $p$-cochain.
Let $\langle\cdot , \cdot\rangle: \chain_{p}(\Graph) \times \chain^{p}(\Graph)\to \Real$ be the standard bilinear pairing of the $p$-chain and $p$-cochain defined with the standard basis.
The transposes of $\consMatrix$, $\IncMatrix$, and $\cycMatrix$ induce the differentials between cochains as $\boundary^{-1}\defeq \consMatrix^{\Transpose}: \chain^{-1}(\Graph)\to \chain^{0}(\Graph)$, 
$\boundary^{0}\defeq \IncMatrix^{\Transpose}: \chain^{0}(\Graph)\to \chain^{1}(\Graph)$, and $\boundary^{1}\defeq \cycMatrix^{\Transpose}: \chain^{1}(\Graph)\to \chain^{2}(\Graph)$. 
The differentials $\boundary^{p}$ on cochains are the adjoints of the differentials $\boundary_{p}$ on chains, which induce the exact cochain sequence:
\begin{align}
0 \xrightarrow{} \chain^{-1}(\Graph)\xrightarrow{\boundary^{-1}=\consMatrix^{\Transpose}}\chain^{0}(\Graph)\xrightarrow{\boundary^{0}=\IncMatrix^{\Transpose}}\chain^{1}(\Graph)\xrightarrow{\boundary^{1}=\cycMatrix^{\Transpose}}\chain^{2}(\Graph)\xrightarrow{}0.
\end{align}
Note that the definition of chains, cochains, and differential operators are topological in the sense that we do not include any metric information.

\subsection{Chain and cochain complexes on hypergraphs}\label{sec:chain_cochain_hyper}
The definitions of chain and cochain complexes introduced above are algebraically extended to hypergraphs $\HGraph$ simply by replacing the graph incidence matrix $\IncMatrix$ with the hypergraph incidence matrix $\HIncMatrix$.
\begin{dfn}[Exact chain and cochain sequences on a hypergraph]\label{dfn:exactchainhypergraph}
The chain and cochain complexes on a hypergraph are defined by the following diagram:
\begin{alignat}{10}
0 & \xrightarrow{} &\chain^{-1}(\HGraph) &\xrightarrow{\boundary^{-1}=\consMatrix^{\Transpose}}&\chain^{0}(\HGraph)&\xrightarrow{\boundary^{0}=\HIncMatrix^{\Transpose}}&\chain^{1}(\HGraph)&\xrightarrow{\boundary^{1}=\cycMatrix^{\Transpose}}&\chain^{2}(\HGraph)&\xrightarrow{}& 0&\notag\\
0 & \xleftarrow{} &\chain_{-1}(\HGraph)&\xleftarrow{\boundary_{0}=\consMatrix}&\chain_{0}(\HGraph)&\xleftarrow{\boundary_{1}=\HIncMatrix}&\chain_{1}(\HGraph)&\xleftarrow{\boundary_{2}=\cycMatrix}&\chain_{2}(\HGraph)&\xleftarrow{}&0&\notag.
\end{alignat}
where $\chain^{-1}(\HGraph)\simeq\chain_{-1}(\HGraph)\simeq\Real^{N_{\cons}}$, $\chain^{0}(\HGraph)\simeq\chain_{0}(\HGraph)\simeq\Real^{N_{\molX}}$, $\chain^{1}(\HGraph)\simeq\chain_{1}(\HGraph)\simeq\Real^{N_{\edge}}$, and $\chain^{2}(\HGraph)\simeq\chain_{2}(\HGraph)\simeq\Real^{N_{\cycle}}$.
\end{dfn}
The bases, $\cycMatrix$ and $\consMatrix$, are obtained as integral bases, i.e., the components of $\cycMatrix$ and $\consMatrix$ can be chosen from $\Integer$ because $\HIncMatrix$ is an integer-valued matrix\footnote{As far as we know, there is not a systematic and widely-appreciated way to define these bases because we have multiple ways to extend the notion of spanning tree of a graph to a hypergraph.}.
As we will explain in \secref{sec:orthogonality} and \secref{sec:InfoGeoNonEquilibrium}, the meaning of $\chain_{2}(\HGraph)$ can be retained as the space on generalized cycles.  The meaning of $\chain_{-1}(\HGraph)$ becomes the space of conserved quantities under the dynamics (\eqnref{CRN_rateEq}).

\subsection{Discrete calculus on graphs and hypergraphs}
The $p$-cochain and $p$-chain introduced above are an algebraic abstraction of the $p$-differential form and its Hodge dual on a differential manifold\cite{grady2010a}.
Accordingly, the discrete versions of gradient, divergence, and curl are associated with the differentials (exterior derivative). 
\begin{dfn}[Discrete gradients, divergences, and curls]\label{dfn:discreteOperators}
The discrete gradient is defined as $\Grad_{\IncMatrix}\defeq \boundary^{0}=\IncMatrix^{\Transpose}$ for $\Graph$ and also as $\Grad_{\HIncMatrix}\defeq \boundary^{0}=\HIncMatrix^{\Transpose}$ for $\HGraph$. 
The adjoints of the gradients are defined with the corresponding adjoint differentials: $\Grad^{*}_{\IncMatrix}\defeq \boundary_{1}=\IncMatrix$ and $\Grad^{*}_{\HIncMatrix}\defeq \boundary_{1}=\HIncMatrix$.
They are called discrete divergences and denoted also as  $\Div_{\IncMatrix}=\Grad^{*}_{\IncMatrix}$ and $\Div_{\HIncMatrix}=\Grad^{*}_{\HIncMatrix}$\footnote{These notations are consistent with those in \secref{sec:graph_hypergraph}}.
The discrete curl and its adjoint are defined as $\Curl_{\cycMatrix}\defeq \boundary^{1}=\cycMatrix^{\Transpose}$ and $\Curl^{*}_{\cycMatrix}\defeq \boundary_{2}=\cycMatrix$, respectively.
\end{dfn}

\subsection{Linear Graph Laplacian Dynamics and Metric structure in discrete calculus}
In the theory of graph Laplacian, a metric matrix $\metric_{p}$ and its associated inner product are typically endowed for each $p$. 
To contrast it with the Legendre duality introduced later, we briefly outline it here. 
For an edge-weighted graph $\Graph_{\Vc{\kcoef}^{\pm}}$ and for the case that $\Vc{\kcoef}^{+}=\Vc{\kcoef}^{-}=\Vc{\kcoef}\in \Real_{>0}^{N_{\edge}}$, $\metric_{0}=\identityM$ and $\metric_{1}=\diag[1/\Vc{\kcoef}]$ are conventionally employed. 
With these metric matrices, the graph Laplacian introduced in \eqnref{eq:graphLaplacian} can be described as 
\begin{align}
    \gLap_{\Vc{\kcoef}}=\Div_{\IncMatrix} \metric^{1}  \Grad_{\IncMatrix} \metric_{0} \label{eq:graphLaplacianLinear}
\end{align}
where $\metric^{p}\defeq \metric_{p}^{-1}$.
By including such metric information, the following pair of metric gradient and divergence is often used in graph theory and network theory: 
$\Grad_{\metric}\defeq  \sqrt{\metric^{1}}\IncMatrix^{\Transpose}$ and $\Div_{\metric} \defeq \IncMatrix \sqrt{\metric^{1}}$ where $\sqrt{\metric^{1}}\defeq \diag[\sqrt{\Vc{\kcoef}}]$.
This symmetric graph Laplacian $\gLap_{\Vc{\kcoef}}$ induces a linear dynamics of $\Vc{x}\in \Real^{N_{\node}}$ on graph via \eqnref{eq:LDMall}\footnote{Here $\Vc{x}$ is not density but a vector in $\Real^{N_{\node}}$.}:
\begin{align}
\dot{\Vc{x}}=-\gLap_{\Vc{\kcoef}}\Vc{x}. \label{eq:gLaplacian_dynamics}
\end{align}
The eigenvalues and eigenvectors of $\gLap_{\Vc{\kcoef}}$ enable us to obtain spectral information of the underlying graph\cite{chung1997}.  
Even for nonlinear dynamics on a hypergraph as in \eqnref{CRN_rateEqall}, the same symmetric Laplacian can provide some information when $\Vc{\kcoef}^{+}=\Vc{\kcoef}^{-}=\Vc{\kcoef}$. 
We can also include other information in the metric matrices such as the degree of vertices\cite{chen2012a}. Various normalizations of the graph Laplacian can be attributed to the choice of metrics.

However, such a choice of metric matrices ends up only with linear dynamics on $\Real^{N_{\node}}$ and is relevant only when the weighting is symmetric: $\Vc{\kcoef}^{+}=\Vc{\kcoef}^{-}=\Vc{\kcoef}$.
In addition, it may not always capture important aspects of the density dynamics such as gradient flow properties and information-theoretic properties, because nonlinear terms such as $\ln \Vc{p}$ appear in information-theoretic quantities.
To extend the class of dynamics being covered and to enable the information-geometric characterization of dynamics, we have to generalize the conventional inner product structure by replacing it with the Legendre dual structure induced by convex functions.


\section{Dually flat spaces on vertices and edges and generalized flow}\label{sec:doubly_dual_flat_spaces}
In this section, we introduce two pairs of dually flat spaces (\fgref{fig:Intro}): one is associated with the vertex spaces, i.e., the dual spaces of $0$-chains and $0$-cochains.
The other corresponds to the edge spaces, i.e., the dual spaces of $1$-chains and $1$-cochains. 
By combining them, the dynamics on graphs and hypergraphs are characterized as a generalized flow.

\subsection{Dually flat spaces on vertices and thermodynamic functions}
We work on the density $\Vc{x}$ and the vertex space for CRN because its reduction to rLDG is straightforward.
For a probability vector $\Vc{p}$, the introduction of dually flat spaces of $\Vc{p}$ and $\ln \Vc{p}$ is natural from the information-geometric viewpoint.
In CRN, $\Vc{x}$ is the vector of concentrations of molecular species. 
As we recently clarified\cite{sughiyama2021ArXiv211212403Cond-MatPhysicsphysics,kobayashi2021ArXiv211214910Phys.}, the dually flat spaces, in this case, result from the Legendre duality between extensive and intensive variables in thermodynamics, which is also natural from the physical viewpoint.

\begin{dfn}[Density space (primal vertex affine space)]\label{dfn:primalVspace}
The density space (also called primal affine vertex space) is the positive orthant $\X \defeq \Real_{>0}^{N_{\molX}}$ of a vector space $\Real^{N_{\molX}}$, which is isomorphic to $\chain_{0}(\HGraph)$: $\Real^{N_{\molX}} \simeq \chain_{0}(\HGraph)$ (\fgref{fig:Intro}, lower left). 
\end{dfn}

\begin{rmk}
The density space $\X$ is defined as the positive orthant rather than as $\Real_{\ge0}^{N_{\molX}}$.
This excludes the cases where some elements of $\Vc{x}$ become $0$. 
From the viewpoint of information geometry, this restriction is necessary to consider densities with the same support (all $\Vc{x}$ in $\X$ should be equivalent in terms of absolute continuity of measures). 
From the viewpoint of dynamical systems, depending on the specific functional form of the flux $\Vc{\flux}(\Vc{x})$, the trajectory $\Vc{x}(t)$ may not be restricted within $\X$. 
The property $\Vc{x}(t)$ in $\X$ for $t\in[0,\infty]$ is known as persistence\footnote{The persistence of a dynamical system is a hard problem, and the persistence for a subclass of CRN is an open problem\cite{craciun2013SIAMJ.Appl.Math.,craciun2016}, which goes by the name of Global Attractor Conjecture since 1974.}. 
Without going into this intricate problem, we simply assume that $\Vc{x}(t) \in \X$ for $t\in [0,\infty]$.
We call $\partial \X \defeq \Real^{N_{\molX}}_{\ge 0}\setminus\X$ the boundary of $\X$.
\end{rmk}
We define the dual of the density space by the Legendre transformation via the thermodynamic function:
\begin{dfn}[Primal thermodynamic function]\label{dfn:thermofunc}
A strictly convex differentiable function $\Pfunc: \X \to \Real$ is called the primal thermodynamic function\footnote{In information geometry, the convex function inducing duality is often called a potential function. We avoid using the word "potential" to discriminate it with an element of the dual vertex affine space, which is called a potential (field) or chemical potential in physics and chemistry. }\footnote{We may consider a convex function $\Pfunc(\Vc{x})$, which does not induce a bijection between $\X$ and $\Y$, e.g., the one which is not strictly convex. Such a situation can arise if a phase transition occurs. It would be an important direction to include this class of functions in this framework.} if the following two conditions are satisfied: 
(1) the associated Legendre transformation
\begin{align}
    \partial\Pfunc: \X & \to   \Real^{N_{\molX}}\\
     \Vc{x} & \longmapsto \Vc{y} \defeq \partial_{\Vc{x}}\Pfunc(\Vc{x}) = \left(\frac{\partial \Pfunc(\Vc{x})}{\partial x_{1}}, \cdots, \frac{\partial \Pfunc(\Vc{x})}{\partial x_{N_{\molX}}}\right)^{\Transpose}
\end{align}
has the image $\Y \defeq \left\{\Vc{y}|\Vc{y}=\partial\Pfunc(\Vc{x}),\Vc{x}\in \X\right\}$ being equal to $\Real^{N_{\molX}}$, i.e., $\Y=\Real^{N_{\molX}}$;
(2) for any $\Vc{x}_{in}\in \X$ and any point on the boundary $\Vc{x}_{bd}\in \partial\X$, 
\begin{align}
\lim_{\lambda \to +0}\frac{\dd \Pfunc(\Vc{x}_{\lambda})}{\dd \lambda} =- \infty \label{eq:thermofuncCond2}
\end{align}
holds where $\Vc{x}_{\lambda}\defeq\lambda \Vc{x}_{in} + (1-\lambda) \Vc{x}_{bd}$ for $\lambda \in [0,1]$.
\end{dfn}

\begin{dfn}[Potential space (dual affine vertex space) and dual thermodynamic function]\label{dfn:dualVspace}
The potential (field) space $\Y= \Real^{N_{\molX}}$ (also called the dual affine vertex space) is an affine space dual to $\X$ with the associated vector space $\chain^{0}(\HGraph)$((\fgref{fig:Intro}, upper left))\footnote{$\Y$ is not only associated with but also isomorphic to the $0$-cochain.  This condition is important when we consider information-geometric projections in the later sections. 
In the theory of differential forms, a $0$-form is often described as a potential field on a manifold. 
Our choice of the potential space is consistent with this convention.}.
The dual thermodynamic function $\Pfunc^{*}: \Y \to \Real$ is the Legendre-Fenchel conjugate of the primal thermodynamic function: 
\begin{align}
    \Pfunc^{*}: \Y \to \Real, \quad \Vc{y} \mapsto \Pfunc^{*}(\Vc{y}) \defeq \max_{\Vc{x}'\in \X}\left[\langle\Vc{x}',\Vc{y} \rangle - \Pfunc(\Vc{x}') \right],
\end{align}
where $\langle\cdot ,\cdot \rangle: \X \times \Y \to \Real$ is the bilinear pairing under the standard basis.
From the properties of the primal function, $\Pfunc^{*}(\Vc{y})$ is also a strictly convex differentiable function.
From $\Pfunc^{*}(\Vc{y})$, we have the inverse Legendre transformation $\partial\Pfunc^{*}: \Y\to\X,\,\Vc{y}\mapsto \Vc{x} =\partial_{\Vc{y}}\Pfunc^{*}(\Vc{y})$.
\end{dfn}
The Legendre transformations, $\partial \Pfunc$ and $\partial\Pfunc^{*}$, are continuous and establish a bijection between $\X$ and $\Y$, where $\partial\Pfunc^{*}=\partial \Pfunc^{-1}$.
In the following, we regard a pair $(\Vc{x},\Vc{y})$ with the same decoration as a Legendre dual pair satisfying $\Vc{y}=\partial \Pfunc(\Vc{x})$.
For a pair, the Legendre-Fenchel-Young identity holds:
\begin{align}
    \Pfunc(\Vc{x})+\Pfunc^{*}(\Vc{y})=\langle \Vc{x},\Vc{y}\rangle.\label{eq:Legendre_Fenchel_Identity}
\end{align}
Different pairs are discriminated with the difference of decorations as $(\Vc{x}', \Vc{y}')$ or $(\Vc{x}_{p}, \Vc{y}_{p})$.

Based on the thermodynamic function, the Bregman divergence can be defined:
\begin{dfn}[Bregman divergence\cite{bregman1967USSRComputationalMathematicsandMathematicalPhysics,amari2016}]\label{dfn:BregmanDiv}
The Bregman divergence on $\X$ with the generating thermodynamic function $\Pfunc(\Vc{x})$ is defined as
\begin{align}
    \BD^{\X}_{\Pfunc}[\Vc{x}\|\Vc{x}']\defeq \Pfunc(\Vc{x})-\Pfunc(\Vc{x}') - \langle\Vc{x}-\Vc{x}', \partial \Pfunc(\Vc{x}') \rangle \in \Real_{\ge 0}.
\end{align}
The non-negativity of the Bregman divergence follows from the Fenchel-Young inequality for products\cite{rockafellar1997,mitroi2011Abstr.Appl.Anal.}.
Furthermore, from the strict convexity of the thermodynamic function, $\BD^{\X}_{\Pfunc}[\Vc{x}\|\Vc{x}']$ is also strictly convex with respect to $\Vc{x}$ and $\BD^{\X}_{\Pfunc}[\Vc{x}\|\Vc{x}']=0$ if and only if $\Vc{x}=\Vc{x}'$.
Bregman divergences are defined for $(\Vc{y},\Vc{y}')$ and also for $(\Vc{x},\Vc{y}')$ as
\begin{align}
    \BD^{\Y}_{\Pfunc^{*}}[\Vc{y}'\|\Vc{y}] &\defeq \Pfunc^{*}(\Vc{y}')-\Pfunc^{*}(\Vc{y}) - \langle\partial \Pfunc^{*}(\Vc{y}), \Vc{y}'-\Vc{y}  \rangle,\\
    \BD^{\X,\Y}_{\Pfunc,\Pfunc^{*}}[\Vc{x}; \Vc{y}'] &\defeq \Pfunc(\Vc{x})+\Pfunc^{*}(\Vc{y}') - \langle\Vc{x}, \Vc{y}'  \rangle.
\end{align}
Because $(\Vc{x},\Vc{y})$ and $(\Vc{x}', \Vc{y}')$ are Legendre pairs, all the three representations are equivalent\footnote{$\BD^{\X,\Y}_{\Pfunc,\Pfunc^{*}}[\Vc{x}; \Vc{y}']$ is also called Fenchel-Young divergence\cite{nielsen2021ProgressinInformationGeometry:TheoryandApplications}.}: $\BD^{\X}_{\Pfunc}[\Vc{x}\|\Vc{x}']=\BD^{\Y}_{\Pfunc^{*}}[\Vc{y}'\|\Vc{y}]=\BD^{\X,\Y}_{\Pfunc,\Pfunc^{*}}[\Vc{x}; \Vc{y}']$\footnote{We here used  the Legendre-Fenchel-Young identity (\eqnref{eq:Legendre_Fenchel_Identity}).}.
$\BD^{\X}_{\Pfunc}[\Vc{x}\|\Vc{x}']$, $ \BD^{\Y}_{\Pfunc^{*}}[\Vc{y}'\|\Vc{y}]$, and $\BD^{\X,\Y}_{\Pfunc,\Pfunc^{*}}[\Vc{x}; \Vc{y}']$ are abbreviated as $\BD^{\X}[\Vc{x}\|\Vc{x}']$, $\BD^{\Y}[\Vc{y}'\|\Vc{y}]$, and $\BD^{\X,\Y}[\Vc{x}; \Vc{y}']$, respectively.
\end{dfn}

Finally, the Hessian matrices of the primal and dual thermodynamic functions
are defined when they are twice differentiable\footnote{When we work on Hessian matrices, we always suppose additionally that they are twice-differentiable.}:
\begin{dfn}[Hessian matrices]\label{dfn:HessianMatrices}
The primal and dual Hessian matrices, $\FM_{\Vc{x}}\in \Real^{N_{\node}\times N_{\node}}$ and $\FM_{\Vc{y}}^{*}\in \Real^{N_{\node}\times N_{\node}}$, of thermodynamic functions, $\Pfunc(\Vc{x})$ and $\Pfunc^{*}(\Vc{y})$, are defined as
\begin{align}
    (\FM_{\Vc{x}})_{i,j}&\defeq \frac{\partial^{2}\Pfunc(\Vc{x})}{\partial x_{i} \partial x_{j}}, & (\FM_{\Vc{y}}^{*})_{i,j}&\defeq \frac{\partial^{2}\Pfunc^{*}(\Vc{y})}{\partial y_{i} \partial y_{j}}. \label{eq:Hessian_X_Y}
\end{align}
In addition, they are positive definite and $\FM_{\Vc{x}}^{-1}=\FM_{\Vc{y}}^{*}$ holds for a Legendre dual pair $\Vc{x}$ and $\Vc{y}$.
\end{dfn}
The Hessian matrices induce a Riemannian metric over $\mathcal{X}$.
The tangent and cotangent spaces $\Tan_{\Vc{x}}\X$ and $\Tan^{*}_{\Vc{x}}\X$ are isomorphic to the corresponding tangent and cotangent spaces $\Tan^{*}_{\Vc{y}}\Y$ and  $\Tan_{\Vc{y}}\Y$ over $\Y$ and also to  $\chain_{0}(\HGraph)$ and $\chain^{0}(\HGraph)$: $\Tan_{\Vc{x}}\X \cong \Tan^{*}_{\Vc{y}}\Y \cong \chain_{0}(\HGraph)$ and $\Tan^{*}_{\Vc{x}}\X \cong\Tan_{\Vc{y}}\Y \cong  \chain^{0}(\HGraph)$.

The typical example of the duality between $\Vc{x}$ and $\Vc{y}$ in statistics is. 
Other than this typical one, depending on the purpose, we adopt different forms of thermodynamic functions $(\Pfunc(\Vc{x}), \Pfunc^{*}(\Vc{y}))$, associated dual variables, and Bregman divergence to endow different properties to inference or estimation methods that we are designing\cite{amari2016}.
In the case of CRN, the thermodynamic functions and Legendre duality are associated with the equilibrium thermodynamics\cite{callen1985}.
Specifically, as we recently demonstrated\cite{sughiyama2021ArXiv211212403Cond-MatPhysicsphysicsa}, $\X$ and $\Y$ are the conjugate spaces of the extensive and intensive thermodynamic variables (density of molecules and their chemical potential), $\Pfunc(\Vc{x})$ is the thermodynamic potential function of the system, and the Bregman divergence becomes the difference of the total entropy.
These correspondences are derived directly from the axiomatic formulation of thermodynamics\cite{callen1985,sughiyama2021ArXiv211212403Cond-MatPhysicsphysicsa}.
The explicit functional form of $\Pfunc(\Vc{x})$ is then determined by the physical details of the thermodynamic system that we work on.

Before closing this subsection, we introduce the notion of separability, which will be linked to the locality of the flux.
\begin{dfn}[Separability of a thermodynamic function]
A thermodynamic function $\Pfunc(\Vc{x})$ is separable if it can be represented as 
\begin{align}
\Pfunc(\Vc{x})=\sum_{i=1}^{N_{\node}}c_{i}\pfunc(x_{i}/x_{i}^{o}),
\end{align}
where $c_{i}>0$, $x_{i}^{o}>0$, and $\pfunc(x):\Real_{>0} \to \Real$ is a scalar primal thermodynamic function.  
\end{dfn}

If $\Pfunc(\Vc{x})$ is separable, then its conjugate $\Pfunc^{*}(\Vc{y})$ is also separable as $\Pfunc^{*}(\Vc{y})=\sum_{i=1}^{N_{\node}}c_{i}\pfunc^{*}(\frac{x_{i}^{o}}{c_{i}}y_{i})$ where $\pfunc^{*}(y): \Real \to \Real$ is the Legendre conjugate of $\pfunc(x)$\footnote{One may further generalize the separability so that $\phi(x)$ depends on $i$ as $\phi_{i}(x)$. }. If a thermodynamic function is separable, then the corresponding Bregman divergence is separable.
The Hessian matrices become diagonal for a separable thermodynamic function.
Most of our results can hold without the separability, but common thermodynamic functions and related quantities are typically separable. 
For example, the Kullback-Leibler divergence is an example of separable Bregman divergences.

\subsection{Dually flat spaces on edges and dissipation functions}
Next, we introduce another dually flat structure onto the edge space of graphs and hypergraphs based on the flux-force relation.

\begin{dfn}[Flux and force spaces (primal and dual edge spaces)]\label{dfn:PrimalDualEdgeSpaces}
The flux and force spaces on the edges, $\Jspace_{\Vc{x}}=\Real^{N_{\edge}}$ and $\Fspace_{\Vc{x}}=\Real^{N_{\edge}}$, are a pair of the primal and dual vector spaces defined for each $\Vc{x}\in \X$, which are isomorphic to $\chain_{1}(\HGraph)$ and $\chain^{1}(\HGraph)$, respectively (\fgref{fig:Intro}, right). 
The bilinear pairing under the standard basis $\langle \cdot, \cdot \rangle: \chain_{1}(\HGraph) \times \chain^{1}(\HGraph) \to \Real$ is inherited to $(\Jspace_{\Vc{x}}, \Fspace_{\Vc{x}})$.
\end{dfn}
To introduce Legendre duality on $(\Jspace_{\Vc{x}}, \Fspace_{\Vc{x}})$, we use the dissipation functions:
\begin{dfn}[Dissipation function\footnote{The definition of dissipation functions is more strict than those used in the previous works, e.g., \cite{mielke2014PotentialAnala}. This is because we define extended projections in this space as in \cite{kobayashi2022Phys.Rev.Researcha}. }]\label{dfn:DissipationFunc}
A dissipation function on $\Fspace_{\Vc{x}}$, $\Dissp^{*}_{\Vc{x}}:\Fspace_{\Vc{x}} \to \Real, \Vc{\tf} \mapsto \Dissp^{*}_{\Vc{x}}(\Vc{\tf})$, is a strictly convex and continuously differentiable function with respect to $\Vc{\tf}$ for all $\Vc{x} \in \X$ that also satisfies the following additional conditions:
\begin{align}
    \mbox{1-coercive:} & & \frac{\Dissp^{*}_{\Vc{x}}(\Vc{\tf})}{\|\Vc{\tf}\|}& \to \infty  \quad \mbox{as $\|\Vc{\tf}\| \to \infty$ }, \label{eq:DissipCond3}\\
    \mbox{Symmetric:} & &  \Dissp^{*}_{\Vc{x}}(\Vc{\tf})&=\Dissp^{*}_{\Vc{x}}(-\Vc{\tf})\label{eq:DissipCond1}\\
    \mbox{Bounded below by $0$:} & &  \Dissp^{*}_{\Vc{x}}(\Vc{0})&=0\label{eq:DissipCond2},
\end{align}
\end{dfn}

\begin{prop}[Duality of dissipation functions]
The Legendre-Fenchel conjugate of $\Dissp_{\Vc{x}}^{*}(\Vc{\tf})$, i.e., $\Dissp_{\Vc{x}}(\Vc{\flux}) \defeq \max_{\Vc{\tf}}\left[\langle\Vc{\flux},\Vc{\tf} \rangle - \Dissp^{*}_{\Vc{x}}(\Vc{\tf}) \right]$, is also the dissipation function on $\Jspace_{\Vc{x}}$.  
$\Dissp_{\Vc{x}}(\Vc{\flux})$ and $\Dissp^{*}_{\Vc{x}}(\Vc{\tf})$ are called primal and dual dissipation functions.
\end{prop}
\begin{proof}
For each $\Vc{x}\in\X$, the function $\Dissp_{\Vc{x}}(\Vc{\flux})$ is strictly convex, continuously differentiable, $1$-coercive, and $\Dissp_{\Vc{x}}(\Vc{\flux})<+\infty$ for all $\Vc{j}\in \Jspace_{\Vc{x}}$ because $\Dissp_{\Vc{x}}^{*}(\Vc{\tf})$ is ( see Corollary 4.1.4 in \cite{hiriart-urruty1996}).
For $\Vc{\flux}\in \Jspace_{\Vc{x}}$, the symmetry holds as $\Dissp_{\Vc{x}}(-\Vc{\flux}) = \max_{\Vc{\tf}}\left[\langle\Vc{-\flux},\Vc{\tf} \rangle - \Dissp^{*}_{\Vc{x}}(\Vc{\tf}) \right]= \max_{\Vc{\tf}}\left[\langle\Vc{\flux},-\Vc{\tf} \rangle - \Dissp^{*}_{\Vc{x}}(\Vc{\tf}) \right]= \max_{\Vc{\tf}}\left[\langle\Vc{\flux},\Vc{\tf} \rangle - \Dissp^{*}_{\Vc{x}}(-\Vc{\tf}) \right]= \max_{\Vc{\tf}}\left[\langle\Vc{\flux},\Vc{\tf} \rangle - \Dissp^{*}_{\Vc{x}}(\Vc{\tf}) \right]=\Dissp_{\Vc{x}}(\Vc{\flux})$.
From the convexity and symmetry, the minimum of $\Dissp_{\Vc{x}}(\Vc{\flux})$ is attained at $\Vc{\flux}=\Vc{0}$ and $\min_{\Vc{\flux}}\Dissp_{\Vc{x}}(\Vc{\flux}) = \Dissp_{\Vc{x}}(\Vc{0}) =\max_{\Vc{\tf}}\left[\langle\Vc{0},\Vc{\tf} \rangle - \Dissp^{*}_{\Vc{x}}(\Vc{\tf}) \right]=-\min_{\Vc{\tf}}\Dissp^{*}_{\Vc{x}}(\Vc{\tf})=0$.
\end{proof}

From these properties, for each $\Vc{x}\in X$, the one-to-one Legendre duality via Legendre transformations is established for all over $(\Jspace_{\Vc{x}}, \Fspace_{\Vc{x}})$:
\begin{align}
    \Vc{\flux}&=\partial_{\Vc{\tf}} \Dissp^{*}_{\Vc{x}}(\Vc{\tf}), & \Vc{\tf}&=\partial_{\Vc{\flux}} \Dissp_{\Vc{x}}(\Vc{\flux}). \label{eq:LegendreTF}
\end{align}
In the following, we abbreviate the Legendre transformations as $\partial_{\Vc{\tf}} \Dissp^{*}_{\Vc{x}}(\Vc{\tf})=\partial \Dissp^{*}_{\Vc{x}}(\Vc{\tf})$ and $\partial_{\Vc{\flux}} \Dissp_{\Vc{x}}(\Vc{\flux})=\partial \Dissp_{\Vc{x}}(\Vc{\flux})$ \footnote{We do not use differentiation of $\Dissp^{*}_{\Vc{x}}(\Vc{\tf})$ and $\Dissp_{\Vc{x}}(\Vc{\flux})$ with respect to $\Vc{x}$ in this work.}.
Similarly to the Legendre dual pair $(\Vc{x},\Vc{y})$ in $\X$ and $\Y$, a pair of flux and force with the same decoration, e.g., $(\Vc{\flux},\Vc{\tf})_{\Vc{x}}$ or $(\Vc{\flux}_{0},\Vc{\tf}_{0})_{\Vc{x}}$, represents a Legendre dual pair linked by \eqnref{eq:LegendreTF} at $\Vc{x}$.
We omit the $\Vc{x}$-dependency for simplicity.
The Legendre dual pair $(\Vc{\flux},\Vc{\tf})$ satisfies the Legendre-Fenchel-Young identity for each $\Vc{x}\in \X$:
\begin{align}
\Dissp^{*}_{\Vc{x}}(\Vc{\tf})+ \Dissp_{\Vc{x}}(\Vc{\flux})-\langle \Vc{\flux},\Vc{\tf}\rangle=0. \label{eq:LegendreIdentity}
\end{align}

Furthermore, the additional conditions of dissipation functions enable the Legendre duality to work as an extension of a Riemannian metric structure:
\begin{prop}[\cite{mielke2014PotentialAnala}]
The Legendre transformations satisfy the following properties: 
\begin{align}
 \mbox{Pairing of  $\Vc{0}\in \Jspace$ and $\Vc{0}\in \Fspace$:} & & \Vc{0}&=\partial\Dissp^{*}_{\Vc{x}}(\Vc{0}),& \Vc{0}&=\partial\Dissp_{\Vc{x}}(\Vc{0})\\
\mbox{Symmetry:} & &-\Vc{\tf}&=\partial \Dissp_{\Vc{x}}(-\Vc{\flux}),& -\Vc{\flux}&=\partial \Dissp^{*}_{\Vc{x}}(-\Vc{\tf}).\\
  \mbox{Nonnegativity of bilinear pairing:} & &   \langle \Vc{\flux},\Vc{\tf}\rangle &= \Dissp^{*}_{\Vc{x}}(\Vc{\tf})+ \Dissp_{\Vc{x}}(\Vc{\flux})\ge 0. \label{eq:j_f_positive}
\end{align}
\end{prop}
The first property means that zero force $\Vc{\tf}=\Vc{0}$ and zero flux $\Vc{\flux}=\Vc{0}$ are always Legendre dual regardless of $\Vc{x}$, and the second one indicates that if $(\Vc{\flux},\Vc{\tf})$ is a Legendre dual pair, then $(-\Vc{\flux},-\Vc{\tf})$ is as well\footnote{From the physical point of view, these conditions are consistent with the thermodynamic requirement that, if the force is zero, the corresponding flux becomes zero, and vice versa and that a sign-reversed force induced the sign-reversed flux.}.
The third property, as well as the nonnegativity of the dissipation functions, enables them to play the similar roles to the metric-induced  norm in Riemannian geometry\footnote{In the context of thermodynamics, the nonnegativity of $\langle \Vc{\flux},\Vc{\tf}\rangle$ is linked to the nonnegativity of the entropy production rate and thus the 2nd law of thermodynamics.}.

With the dissipation functions, $\Dissp_{\Vc{x}}(\Vc{\flux})$ and $\Dissp^{*}_{\Vc{x}}(\Vc{\tf})$, we now have the second dually flat structure on the edge spaces $(\Jspace_{\Vc{x}},\Fspace_{\Vc{x}})$. 
In these dually flat spaces, we define the Bregman divergence and Hessian matrices:
\begin{dfn}[Bregman divergence and Hessian matrices on the edge spaces]
 For each $\Vc{x}\in \X$, the Bregman divergence between $\Vc{\flux}\in \Jspace_{\Vc{x}}$ and $\Vc{\tf}'\in \Fspace_{\Vc{x}}$ is defined as 
\begin{align}
    \BD_{\Vc{x}}^{\Jspace,\Fspace}[\Vc{\flux};\Vc{\tf}']:=\Dissp_{\Vc{x}}(\Vc{\flux})+\Dissp^{*}_{\Vc{x}}(\Vc{\tf}')-\langle\Vc{\flux}, \Vc{\tf}'\rangle. \label{eq:BD}
\end{align}
$\BD_{\Vc{x}}^{\Jspace}[\Vc{\flux}\|\Vc{\flux}']$ and $\BD_{\Vc{x}}^{\Fspace}[\Vc{\tf}\|\Vc{\tf}']$ are also defined analogously to the Bregman divergence on the vertex space $(\X, \Y)$.
For a Legendre conjugate pair of twice differentiable dissipation functions, the Hessian matrices, $\FM_{\Vc{x},\Vc{\flux}}$ and $\FM_{\Vc{x},\Vc{\tf}}^{*}$, are defined as
\begin{align}
    (\FM_{\Vc{x},\Vc{\flux}})_{e,e'}&\defeq \frac{\partial^{2}\Dissp_{\Vc{x}}(\Vc{\flux})}{\partial \flux_{e} \partial \flux_{e'}}, & (\FM_{\Vc{x},\Vc{\tf}}^{*})_{e,e'}&\defeq \frac{\partial^{2}\Dissp_{\Vc{x}}^{*}(\Vc{\tf})}{\partial \tf_{e} \partial \tf_{e'}}. \label{eq:HessianMatrixEdge}
\end{align}
These matrices are positive-definite.
\end{dfn}

The Legendre dual structure via the dissipation functions provides an extension of a Riemannian metric structure in the following sense.
If the dissipation function is a quadratic function, i.e., a positive definite quadratic form as 
\begin{align}
\Dissp^{q,*}_{\Vc{x}}(\Vc{\tf})\defeq \frac{1}{2}\langle \Vc{\tf},\metric^{*}_{\Vc{x}} \Vc{\tf} \rangle, \label{eq:quadratic_dissp}
\end{align}
where $\metric^{*}_{\Vc{x}}$ is a positive definite $N^{\edge}\times N^{\edge}$ matrix, the Legendre transformation is reduced to the linear mapping $\Vc{\flux}=\partial \Dissp^{q,*}_{\Vc{x}}(\Vc{\tf})=\metric^{*}_{\Vc{x}}\Vc{\tf}$\footnote{This correspondence illustrates that the dependency of $\Dissp^{q,*}_{\Vc{x}}(\Vc{\tf})$ on $\Vc{x}$ is a formal generalization of the Riemannian metric. But for this case, the relevant state space for $\Vc{x}$ is not the positive orthant but the vector spece $\Real^{N_{\node}}$. }.
Then, the bilinear paring, $\langle\Vc{\flux},\Vc{\tf}'\rangle = \langle\Vc{\flux},\metric_{\Vc{x}}\Vc{\flux}'\rangle=\langle\metric^{*}_{\Vc{x}}\Vc{\tf},\Vc{\tf}'\rangle$ , becomes the inner product under the metric matrix $\metric_{\Vc{x}}$ where $\metric_{\Vc{x}}=(\metric_{\Vc{x}}^{*})^{-1}$.
The dissipation functions are associated with the induced norms: $\Dissp^{*}_{\Vc{x}}(\Vc{\tf})=\frac{1}{2}\|\Vc{\tf}\|_{\metric^{*}_{\Vc{x}}}^{2}$, $\Dissp_{\Vc{x}}(\Vc{\flux})=\frac{1}{2}\|\Vc{\flux}\|_{\metric_{\Vc{x}}}^{2}$.
The Bregman divergence is reduced to the norm-induced squared distance: $\BD_{\Vc{x}}^{\Jspace,\Fspace}[\Vc{\flux};\Vc{\tf}']=\frac{1}{2}\|\Vc{\flux}-\Vc{\flux}'\|_{\metric_{\Vc{x}}}^{2}=\frac{1}{2}\|\Vc{\tf}-\Vc{\tf}'\|_{\metric^{*}_{\Vc{x}}}^{2}$.

Finally, we also introduce the notion of separability to the dissipation functions:
\begin{dfn}[Separability and locality of dissipation functions]
A dissipation function $\Dissp^{*}_{\Vc{x}}(\Vc{\tf})$ is separable if it can be represented as 
\begin{align}
\Dissp^{*}_{\Vc{x}}(\Vc{\tf})=\sum_{e=1}^{N_{\edge}}\frenecy_{e}(\Vc{x})\dissip^{*}(\tf_{e}/\tf^{o}_{e}(\Vc{x})),\label{eq:separable_dissip}
\end{align}
where $\frenecy_{e}(\Vc{x})>0$ and $\tf^{o}_{e}(\Vc{x})>0$ for $\Vc{x}\in \X$ are positive weights and $\dissip^{*}(\tf):\Real \to \Real$ is a scalar dissipation function, i.e., a strictly convex differentiable scalar function satisfying \eqnref{eq:DissipCond1}, \eqnref{eq:DissipCond2}, and \eqnref{eq:DissipCond3}.
If $\frenecy_{e}(\Vc{x})$ and $\tf^{o}_{e}(\Vc{x})$ are additionally local, then the dissipation function is separable and local.
If $\Dissp^{*}_{\Vc{x}}(\Vc{\tf})$ is separable, then its dual $\Dissp_{\Vc{x}}(\Vc{\flux})$ is also separable.
The same is true for the locality.
\end{dfn}

\begin{rmk}[Young functions and N functions]
The scalar dissipation function is a N function, which appears in the theory of Orlicz spaces.
A function $\tilde{\dissip}(\flux): [0,\infty) \to [0,\infty]$ represented as $\tilde{\dissip}(\flux)=\int_{0}^{\flux}\varsigma(\flux')\dd \flux'$ is called Young function where $\varsigma(\flux):[0,\infty)\to [0,\infty]$ is a non-decreasing function satisfying $\varsigma(0)=0$ and being left-continuous on $(0,\infty)$. If $\varsigma(\flux)$ additionally satisfies $0<\varsigma(\flux)<+\infty (0<\flux<\infty)$, $\lim_{\flux \to +0} \varsigma(\flux)=0$, and $\lim_{\flux \to \infty} \varsigma(\flux)=+\infty$, then $\tilde{\dissip}(\flux)$ is called an N-function.
If we define a function $\dissip(\flux)$ with a N-function $\tilde{\dissip}(\flux)$ as $\dissip(\flux)=\tilde{\dissip}(|\flux|)$, this becomes a scalar dissipation function\cite{krasnoselskij1962}.
A separable dissipation function (\eqnref{eq:separable_dissip}) is often called a weighted N-function\cite{lods2015Entropy,pistone2018Inf.Geom.ItsAppl.}. The dissipation function and induced Legendre duality are, therefore, related to Birnbaum-Orlicz spaces, which are an extension of $L^{p}$ spaces.
\end{rmk}

\subsection{Generalized flow on graphs and hypergraphs and its steady state}
Because of the one-to-one Legendre duality between $(\Vc{\flux},\Vc{\tf})_{\Vc{x}}$, the continuity equation (\eqnref{CRN_rateEq}) can be represented as a generalized flow driven by the force $\Vc{\tf}(\Vc{x})$ dual to $\Vc{\flux}(\Vc{x})$\cite{renger2018Entropy}:
\begin{dfn}[Generalized flow]
A curve $\Vc{x}(t)$ is a generalized flow on $\HGraph$ driven by force $\Vc{\tf}(\Vc{x})$ under the dissipation function $\Dissp^{*}_{\Vc{x}}$ if it can be represented as
\begin{align}
    \dot{\Vc{x}}=-\Div_{\HIncMatrix} \Vc{\flux}(\Vc{x})=-\Div_{\HIncMatrix} \partial \Dissp^{*}_{\Vc{x}}[\Vc{\tf}(\Vc{x})].\label{eq:gFlow}
\end{align}
\end{dfn}
This representation is not dependent on the specific functional form of $\Vc{\tf}(\Vc{x})$ and $\Dissp^{*}_{\Vc{x}}(\Vc{\tf})$ and also on the definition of  $\Div_{\HIncMatrix}$ as long as the generated $\Vc{\flux}(\Vc{x})$ is consistent with $\HGraph$\footnote{The consistency is required because of our choice of $\Real_{\ge 0}^{N_{\molX}}$ as the density space. }\footnote{The consistency with $\HGraph$ is assumed to hold. }.
Thus, we can potentially apply this framework to various systems by choosing these functions appropriately depending on the system or the problem we work on. 

The generalized flow naturally encompasses three types of steady states:
\begin{dfn}[Steady state, complex-balanced state, and detailed-balanced state]\label{eq:st_cb_eq_state}
We define the manifolds of steady state $\Manifold^{\mathrm{ST}}$, complex-balanced (CB) state $\Manifold^{\mathrm{CB}}$, and detailed-balanced (DB) state $\Manifold^{\mathrm{DB}}$, respectively, as follows:
\begin{align}
\Manifold^{\mathrm{ST}} &\defeq \{\Vc{x} \in \X| \HIncMatrix \Vc{\flux}(\Vc{x})=0\},\\
\Manifold^{\mathrm{CB}} &\defeq \{\Vc{x} \in \X| \IncMatrix \Vc{\flux}(\Vc{x})=0\},\label{eq:DBC}\\
\Manifold^{\mathrm{DB}} &\defeq \{\Vc{x} \in \X| \Vc{\flux}(\Vc{x})=0\}=\{\Vc{x} \in \X| \Vc{\tf}(\Vc{x})=0\},
\end{align}
where we used $\Vc{\flux}(\Vc{x})=0$ iff $\Vc{\tf}(\Vc{x})=0$ from the properties of the dissipation functions.
The relations $\Vc{\flux}(\Vc{x})=0$ and $\IncMatrix \Vc{\flux}(\Vc{x})=0$ are called the detail-balanced (DB) condition and the complex-balanced (CB) condition, respectively.
From the decomposition $\HIncMatrix=\cmMatrix \IncMatrix$, an inclusion relation holds: $\Manifold^{\mathrm{DB}} \subseteq \Manifold^{\mathrm{CB}} \subseteq \Manifold^{\mathrm{ST}}$.
It should be noted that, depending on the details of $\Vc{\flux}(\Vc{x})$, these manifolds can be empty.
\end{dfn}
A steady state is a state at which $\dot{\Vc{x}}=0$ holds.
The DB condition $\Vc{\flux}(\Vc{x})=\Vc{0}$ means that all the fluxes are zero at $\Vc{x}$. 
In other words, all the forward and reverse fluxes are balanced at $\Vc{x}$, i.e., $\flux_{e}^{+}(\Vc{x})=\flux_{e}^{-}(\Vc{x})$.
The CB condition is equivalent to the balance of all influx and outflux at each hypervertex of $\HGraph$.
As we will see later, DB states are tightly linked to the equilibrium state and equilibrium flow.
The CB state is relevant as an extension of the equilibrium state to nonequilibrium flows.

\subsection{Generalized gradient flow and De Giorgi's formulation}
When $\Vc{\tf}(\Vc{x})$ can be represented as a gradient, i.e., $\Vc{\tf}(\Vc{x})=\Grad_{\HIncMatrix}\partial \FEnergy(\Vc{x})$ of a function $\FEnergy(\Vc{x})\in \Real$ on the density space, \eqnref{eq:gFlow} is reduced to the generalized gradient flow of $\FEnergy(\Vc{x})$.
\begin{dfn}[Generalized gradient flow]\label{dfn:gGradFlow}
$\Vc{x}(t)$ is a generalized gradient flow when it is a generalized flow driven by a gradient force of $\FEnergy(\Vc{x})$, i.e., $\Vc{\tf}(\Vc{x})=\Grad_{\HIncMatrix}\partial \FEnergy(\Vc{x})$ and  
\begin{align}
    \dot{\Vc{x}}=-\Div_{\HIncMatrix} \Vc{\flux}(\Vc{x})=-\Div_{\HIncMatrix} \partial \Dissp^{*}_{\Vc{x}}[\Grad_{\HIncMatrix}\partial \FEnergy(\Vc{x})].\label{eq:gGradFlow}
\end{align}
\end{dfn}
The following proposition ensures that the generalized gradient flow behaves like the conventional gradient flow:
\begin{prop}[$\FEnergy(\Vc{x})$ is non-increasing along the trajectory of generalized gradient flow ]\label{prop:gradflow}
For a trajectory $\{\Vc{x}_{t}\}_{t \in [0,\tau]}$ of the generalized gradient flow of $\FEnergy(\Vc{x})$, $\FEnergy(\Vc{x}_{t})$ is always decreasing except at the DB states $\Manifold^{\mathrm{DB}}$.
In addition, all the steady states of the generalized gradient flow are the DB states, i.e., $\Manifold^{\mathrm{ST}}=\Manifold^{DB}$\footnote{$\Manifold^{\mathrm{ST}}=\Manifold^{DB}=\emptyset$ can hold, e.g., when $\FEnergy(\Vc{x})$ is a strictly monotonous function.}. 
\end{prop}
\begin{proof}
$\FEnergy(\Vc{x}_{t})$ is non-increasing over time as follows:
\begin{align}
    \dot{\FEnergy}(\Vc{x}_{t})&=\langle\dot{\Vc{x}} , \partial_{\Vc{x}} \FEnergy(\Vc{x})\rangle = -\langle\Div_{\HIncMatrix} \partial \Dissp^{*}_{\Vc{x}}[\Vc{\tf}(\Vc{x})] , \partial_{\Vc{x}} \FEnergy(\Vc{x})\rangle = -\langle \partial \Dissp^{*}_{\Vc{x}}[\Vc{\tf}(\Vc{x})] , \Grad_{\HIncMatrix}\partial_{\Vc{x}} \FEnergy(\Vc{x})\rangle \notag \\
    &=-\langle \Vc{\flux}(\Vc{x}),\Vc{\tf}(\Vc{x})\rangle= -\left( \Dissp_{\Vc{x}}[\Vc{\flux}(\Vc{x})] + \Dissp^{*}_{\Vc{x}}[\Vc{\tf}(\Vc{x})]\right)\le 0, \label{eq:generalized_gradient_flow}
\end{align}
where \eqnref{eq:j_f_positive} is used. The equality holds iff $\Vc{\tf}(\Vc{x})=\Grad_{\HIncMatrix}\partial \FEnergy(\Vc{x})=0$ because $\Dissp^{*}_{\Vc{x}}[\Vc{\tf}(\Vc{x})]=\Dissp_{\Vc{x}}[\Vc{\flux}(\Vc{x})]=0$ iff $\Vc{\tf}(\Vc{x})=\Vc{\flux}(\Vc{x})=0$.
Thus, $\dot{\FEnergy}(\Vc{x}_{t})=0$ iff $\Vc{x}_{t}\in \Manifold^{\mathrm{DB}}$. Because $\dot{\Vc{x}}_{t}=0 \Rightarrow \dot{\FEnergy}(\Vc{x}_{t})=0$,  $\Manifold^{\mathrm{ST}}=\Manifold^{\mathrm{DB}}$.
\end{proof}
It should be noted that, even if $\FEnergy(\Vc{x})$ has a single minimum, the steady state $\Vc{x}_{st}\defeq \lim_{t \to \infty}\Vc{x}(t)$ may not be the minimum, because $\dot{\FEnergy}(\Vc{x}_{t})=0$ holds for any $\Vc{x}\in \Manifold^{\mathrm{DB}}$\footnote{In addition, there exists the possibility that $\Vc{x}(t)$ converges to the boundary of $\X$.}.

The generalized gradient flow of this form (\eqnref{eq:gGradFlow}) was devised in the process to extend the conventional gradient flow to metric spaces\cite{ambrosio2006a,liero2016SIAMJ.Math.Anal.}\footnote{The metric here means a general metric, which is not restricted to one associated with the inner product.}.
Furthermore, dissipation functions have been recognized since the seminal work of Onsager\cite{onsager1931Phys.Rev.,onsager1931Phys.Rev.a,machlup1953Phys.Rev.}.
However, only quadratic dissipation functions have been investigated until very recently\cite{mielke2014PotentialAnala,mielke2017SIAMJ.Appl.Math.,renger2018Entropy,kaiser2018JStatPhys,patterson2021ArXiv210314384Math-Ph,peletier2022Calc.Var.,renger2021DiscreteContin.Dyn.Syst.-S,peletier2022}. 
This may be partly because we lack an adequate geometric language to handle the non-quadratic cases, i.e., information geometry.
Actually, if the dissipation function is quadratic $\Dissp^{q,*}_{\Vc{x}}[\Vc{\tf}]$ as in \eqnref{eq:quadratic_dissp}, then the generalized flow (\eqnref{eq:gFlow}) formally reduces to the flow on a Riemannian manifold with the metric $(\HIncMatrix\metric^{*}_{\Vc{x}}\HIncMatrix^{\Transpose})^{-1}$\footnote{Here, $\Vc{x}$ should be regarded not as density but as $\Vc{x} \in \Real^{N_{\node}}$.}.

The non-negativity of $\dot{\FEnergy}(\Vc{x}_{t})$ is essentially attributed to the fact that $\dot{\FEnergy}(\Vc{x}_{t})=-\langle \Vc{\flux}(\Vc{x}),\Vc{\tf}(\Vc{x})\rangle$ holds in \eqnref{eq:generalized_gradient_flow} for the generalized gradient flow. 
The converse also holds.
\begin{prop}[De Giorgi's formulation of generalized gradient flow\cite{mielke2014PotentialAnala,patterson2021ArXiv210314384Math-Ph}]
Let $\Vc{x}_{t}$ be a generalized flow induced by a force $\Vc{\tf}(\Vc{x})$. $\Vc{x}_{t}$ is the generalized gradient flow of $\FEnergy(\Vc{x})$ iff
\begin{align}
    \dot{\FEnergy}(\Vc{x}_{t})&=-\langle \Vc{\flux}(\Vc{x}),\Vc{\tf}(\Vc{x})\rangle = -\left(\Dissp_{\Vc{x}}[\Vc{\flux}(\Vc{x})] + \Dissp^{*}_{\Vc{x}}[\Vc{\tf}(\Vc{x})]\right). \label{eq:Feq_j_f}
\end{align}
holds. The integral form of \eqnref{eq:Feq_j_f} 
\begin{align}
    \FEnergy(\Vc{x}_{0})-\FEnergy(\Vc{x}_{t})=\int_{0}^{t}\left[\Dissp^{*}_{\Vc{x}_{t'}}(\Vc{\tf}(\Vc{x}_{t'}))+ \Dissp_{\Vc{x}_{t'}}(\Vc{\flux}(\Vc{x}_{t'}))\right]\dt', \label{eq:DeGiorgi}
\end{align}
is called De Giorgi's $(\Dissp,\Dissp^{*})$-formulation of generalized gradient flow. 
\end{prop}
\begin{proof}
For a generalized flow $\Vc{x}_{t}$ driven by force $\Vc{\tf}(\Vc{x})$ as in \eqnref{eq:gFlow} and for any $\FEnergy(\Vc{x})$, the following inequality holds:
\begin{align}
    \dot{\FEnergy}(\Vc{x}_{t})=\left\langle\dot{\Vc{x}},\frac{\partial \FEnergy(\Vc{x}_{t})}{\partial \Vc{x}} \right\rangle
    &=\left\langle-\Vc{\flux}(\Vc{x}_{t}),\Grad_{\HIncMatrix}\frac{\partial \FEnergy(\Vc{x}_{t})}{\partial \Vc{x}} \right\rangle\\
    &=-\left[\Dissp_{\Vc{x}_{t}}(\Vc{\flux}(\Vc{x}_{t})) + \Dissp^{*}_{\Vc{x}_{t}}(\Vc{\tf}'(\Vc{x}_{t}))\right] + \BD^{\Jspace,\Fspace}_{\Vc{x}_{t}}[\Vc{\flux}(\Vc{x}_{t});\Vc{\tf}'(\Vc{x}_{t})]\\
    &\ge -\left[\Dissp_{\Vc{x}_{t}}(\Vc{\flux}(\Vc{x}_{t})) + \Dissp^{*}_{\Vc{x}_{t}}(\Vc{\tf}'(\Vc{x}_{t}))\right], 
\end{align}
where we define $\Vc{\tf}'(\Vc{x})\defeq \Grad_{\HIncMatrix}\partial \FEnergy(\Vc{x})$.
The last inequality becomes an equality if and only if $\Vc{\tf}'(\Vc{x}_{t})$ is the Legendre dual of $\Vc{\flux}(\Vc{x}_{t})$\footnote{These inequality and equality conditions are usually derived by using Cauchy-Schwarz inequality\cite{lisini2009ESAIM:COCV}. From the information-geometric framework, they are trivially attributed to the non-negativity of Bregman divergence. }, i.e., 
\begin{align}
    \BD^{\Jspace,\Fspace}_{\Vc{x}_{t}}[\Vc{\flux}(\Vc{x}_{t});\Vc{\tf}'(\Vc{x}_{t})]=0 \Longleftrightarrow \Vc{\flux}(\Vc{x}_{t})=\partial \Dissp^{*}_{\Vc{x}_{t}}[\Vc{\tf}'(\Vc{x}_{t})] \label{eq:DeGiorgi2}
\end{align}
Thus, \eqnref{eq:Feq_j_f} holds only when $\Vc{x}_{t}$ is the generalized gradient flow of $\FEnergy(\Vc{x})$.
\end{proof}
De Giorgi's formulation is a well-established approach for defining gradient flow in metric spaces\cite{ambrosio2006a}.
\subsection{Equilibrium and nonequilibrium flow}
In this work, we mainly focus on the case that $\FEnergy(\Vc{x})= \BD^{\X}_{\Pfunc}[\Vc{x}\|\tilde{\Vc{x}}]$ where $\BD^{\X}_{\Pfunc}$ is the Bregman divergence associated with a thermodynamic function $\Pfunc$.

\begin{figure}[h]
\includegraphics[bb=0 420 1024 768, width=\linewidth]{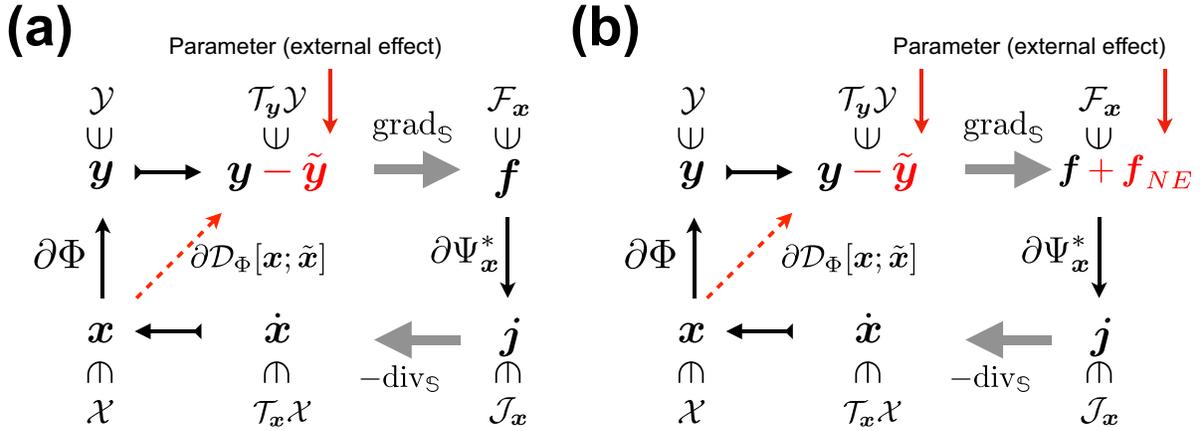}
\caption{Schematic representation of the equilibrium (a) and nonequilibrium flow (b).}
\label{fig:Equ_Noneq}
\end{figure}

\begin{dfn}[Equilibrium force, equilibrium flux, and equilibrium flow]
The force generated by the gradient of Bregman divergence associated with a thermodynamic function $\Pfunc$ is called the (thermodynamic) equilibrium force, and the following equation is denoted as the thermodynamic gradient equation:
\begin{align}
    \Vc{\tf}(\Vc{x})=\Grad_{\HIncMatrix}\partial \BD^{\X}_{\Pfunc}[\Vc{x}\|\tilde{\Vc{x}}],\label{eq:gradEq}
\end{align}
where $\tilde{\Vc{x}}\in \X$ is a parameter. The dual of $\Vc{\tf}(\Vc{x})$, i.e., $\Vc{\flux}(\Vc{x})=\partial \Dissp^{*}_{\Vc{x}}[\Vc{\tf}(\Vc{x})]$, is called the equilibrium flux: 
A generalized flow $\Vc{x}(t)$ is an equilibrium flow if it is driven by the equilibrium force :
\begin{align}
    \dot{\Vc{x}}=-\Div_{\HIncMatrix} \partial \Dissp^{*}_{\Vc{x}}[\Grad_{\HIncMatrix}\partial \BD^{\X}_{\Pfunc}[\Vc{x}\|\tilde{\Vc{x}}]].\label{eq:gGF}
\end{align}
Using the relation $\partial \BD^{\X}_{\Pfunc}[\Vc{x}\|\tilde{\Vc{x}}]=\partial \Pfunc(\Vc{x}) - \tilde{\Vc{y}}$ where $\tilde{\Vc{y}}=\partial \Pfunc(\tilde{\Vc{x}})$, 
\eqnref{eq:gGF} can be rewritten as
\begin{align}
    \dot{\Vc{x}}=-\Div_{\HIncMatrix}\left[ \partial \Dissp^{*}_{\Vc{x}}\left[\Grad_{\HIncMatrix}\left\{\partial \Pfunc(\Vc{x}) - \tilde{\Vc{y}}\right\}\right]\right], \label{eq:gGF2}
\end{align}
which explicitly shows the contribution of both the thermodynamic function and the dissipation function to the dynamics (\fgref{fig:Equ_Noneq} (a)). 
\end{dfn}
Various properties of the equilibrium flow (\eqnref{eq:gGF}) can be obtained from the doubly dual flat structure as we will see in the following sections.
In addition, the equilibrium flow captures the properties that the dynamics of thermodynamic equilibrium systems should hold. In this sense, the equilibrium flow is the mathematical representation of the dynamics of equilibrium systems.

Beyond the gradient equilibrium flow, we also consider the non-gradient nonequilibrium flow of the following type: 
\begin{dfn}[Nonequilibrium force and nonequilibrium flow]
The force generated by a shift of the equilibrium force 
\begin{align}
    \Vc{\tf}(\Vc{x})=\Grad_{\HIncMatrix}\partial \BD^{\X}_{\Pfunc}[\Vc{x}\|\tilde{\Vc{x}}] + \Vc{\tf}_{NE},\label{eq:nonequF}
\end{align}
is called nonequilibrium force if $\Vc{\tf}_{NE}\not\in \Img[\HIncMatrix^{\Transpose}]$\footnote{In physics, such $\Vc{\tf}_{NE}$ can be identified with a nonequilibrium force applied externally to the system.}.
If the shift $\Vc{\tf}_{NE}$ satisfies $\Vc{\tf}_{NE}\in \Img[\HIncMatrix^{\Transpose}]$, then $\Vc{\tf}(\Vc{x})$ is reduced to the equilibrium force $\Vc{\tf}_{NE}=\Vc{0}$ by appropriately changing $\tilde{\Vc{x}}$.
The nonequilibrium flow is the flow induced by the nonequilibrium force (\fgref{fig:Equ_Noneq} (b)):
\begin{align}
    \dot{\Vc{x}}=-\Div_{\HIncMatrix} \partial \Dissp^{*}_{\Vc{x}}\left[\left[\Grad_{\HIncMatrix}\partial_{\Vc{x}} \BD^{\X}_{\Pfunc}[\Vc{x}\|\tilde{\Vc{x}}]\right]+\Vc{\tf}_{NE}\right].\label{eq:gNGF}
\end{align}
\end{dfn}
In the next section, we show that this equation can cover a sufficiently wide class of models, e.g., all types of rLDG and CRN with extended LMA kinetics. 
Equation \ref{eq:gNGF} can also be associated with nonequilibrium dynamics with a constant environmental force. 
The techniques in information geometry, Hessian geometry, and convex analysis enable us to investigate such non-gradient dynamics.
\begin{rmk}[Variational modeling\cite{peletier2014}]
We introduced and characterized dynamics based on the thermodynamic functions and dissipation functions. While we employed a restricted definition in order to link dynamics to information geometry, we may further generalize this approach by appropriately choosing the state space, $\Vc{\tf}(\Vc{x})$,  $\Dissp^{*}_{\Vc{x}}(\Vc{\tf})$, and $\Div_{\HIncMatrix}$. 
For example, we may consider a $\Vc{x}$-dependent and noninteger-valued for the matrix $\HIncMatrix(\Vc{x})$. The equilibrium flow may not be restricted to $\FEnergy(\Vc{x})= \BD^{\X}_{\Pfunc}[\Vc{x}\|\tilde{\Vc{x}}]$, and the nonequilibrium flow may be defined for $\Vc{x}$-dependent $\Vc{\tf}_{NE}(\Vc{x})$.
This type of approach for modeling dissipative dynamics has been known as variational modeling.
\end{rmk}
Before closing this section, we mention that the existence of DB states, i.e., $\Manifold^{\mathrm{DB}}\neq \emptyset$, is necessary and sufficient for a nonequilibrium flow to be an equilibrium flow.
\begin{prop}[Detailed balance condition and  equilibrium flow]
Consider a flow given by \eqnref{eq:gNGF}. 
If $\Manifold^{\mathrm{DB}}\neq \emptyset$, then the flow is equilibrium, i.e.,  $\Vc{\tf}_{NE}\in \Img \HIncMatrix^{\Transpose}$.
\end{prop}
\begin{proof}
$\Manifold^{\mathrm{DB}}\neq \emptyset$ means that there exists $\Vc{x}_{DB}$ satisfying $\Vc{\flux}(\Vc{x}_{DB})=\Vc{0}$. 
Then we have $\Vc{\flux}(\Vc{x}_{DB})=\Vc{0}\Leftrightarrow \Vc{\tf}(\Vc{x}_{DB})=\Vc{0}$. 
If $\Vc{\tf}_{NE}\not\in \Img[\HIncMatrix^{\Transpose}]$, $\Vc{\tf}_{NE}\neq \Vc{0}$ and thus $\Vc{\tf}(\Vc{x}) \neq \Vc{0}$ for all $\Vc{x}\in \X$. Thus,  $\Vc{\tf}_{NE} \in \Img[\HIncMatrix^{\Transpose}]$ if $\Manifold^{\mathrm{DB}}\neq \emptyset$.
\end{proof}
The necessity follows basically from \propref{prop:gradflow}, but we have to show $\Manifold^{\mathrm{ST}} \neq \emptyset$. This will be shown in the following section (\lmmref{lmm:DualFoliationVertex}).

\section{Explicit form of thermodynamic and dissipation functions}\label{sec:explict_form}
Before investigating the dynamics of the equilibrium (\eqnref{eq:gGF}) and nonequilibrium (\eqnref{eq:gNGF}) flow, we show how the flow can be associated with the dynamics on graphs and hypergraphs via specific forms of the thermodynamic and dissipation functions.
The forms of functions depend on the functional form of the flux that we assume: \eqnref{eq:LDM2} for rLDG, \eqnref{eq:CRN_rate} for CRN with LMA kinetics, and \eqnref{eq:flux_FPE} for FPE. 
It should be noted that the choice of the thermodynamic function and the dissipation function is not unique for a given dynamics in general. 
Depending on the purpose, we should choose or find an appropriate set of functions.

\subsection{Explicit form of thermodynamic functions for rLDG and CRN}
For rLDG (\eqnref{eq:LDM2}) and CRN with LMA kinetics (\eqnref{eq:CRN_rate}), the following pair of thermodynamic functions is particularly relevant\footnote{For CRN, these forms of the thermodynamic functions are derived from the conventional thermodynamics of ideal gas or dilute solution with non-reactive solvent\cite{sughiyama2021ArXiv211212403Cond-MatPhysicsphysicsa}. Mathematically, we may employ other functions as we introduce different information geometric structures onto a family of probabilities depending on the purpose. Such exploitation is an interesting open problem.}:
\begin{align}
    \Pfunc(\Vc{x}) & \defeq \left[\ln \Vc{x} - \ln\Vc{x}^{o} - \Vc{1}\right]^{\Transpose}\Vc{x} = \sum_{i=1}^{N_{\molX}}\left[\ln \frac{x_{i}}{x_{i}^{o}}-1\right]x_{i}, & 
    \Pfunc^{*}(\Vc{y}) &\defeq  (\Vc{x}^{o})^{T}e^{\Vc{y}}=\sum_{i=1}^{N_{\molX}}x_{i}^{o}e^{y_{i}}, \label{eq:potentialfunction_X_MAK}
\end{align}
which induce the following Legendre transformation: 
\begin{align}
    \Vc{y}&=\partial\Pfunc(\Vc{x}) = \ln \Vc{x}-\ln\Vc{x}^{o}, &
    \Vc{x}&=\partial\Pfunc^{*}(\Vc{y}) = \Vc{x}^{o}\circ e^{\Vc{y}}. 
\end{align}
Here, $\Y=\Real^{N_{\molX}}$, and $\Vc{x}^{o}\in \X$ is a parameter determining the point in $\X$ that is associated with the origin of $\Y$ via the Legendre transformation.
For these thermodynamic functions, the Bregman divergence is reduced to the generalized Kullback-Leibler divergence.
\begin{align}
    \BD^{\X}[\Vc{x}\|\Vc{x}']=\left(\ln\frac{\Vc{x}}{\Vc{x}'}\right)^{\Transpose}\Vc{x}-\Vc{1}^{\Transpose}(\Vc{x}-\Vc{x}'). \label{eq:KL}
\end{align}
These thermodynamic functions and the Bregman divergence are separable.

If we choose $\Vc{x}^{o}=\Vc{1}$, then the conventional dual representation for the probability density $\Vc{p}$ on a discrete space is recovered: 
\begin{align}
    \Pfunc(\Vc{p}) & = \left[\ln \Vc{p} - \Vc{1}\right]^{\Transpose}\Vc{p}, & 
    \Pfunc^{*}(\Vc{y}) &=  \Vc{1}^{T}e^{\Vc{y}}, &
    \Vc{y}&=\partial_{\Vc{p}}\Pfunc(\Vc{p}) = \ln \Vc{p}, & \Vc{p}&=\partial_{\Vc{y}}\Pfunc^{*}(\Vc{y}) = e^{\Vc{y}}.\label{eq:potential_p}
\end{align}
In this case, $\Y$ is the space of the logarithm of $\Vc{p}$.
These representations hold even if $\Vc{p}$ is not a probability density.
If $\Vc{p}$ satisfies $\Vc{1}^{\Transpose}\Vc{p}=1$, the generalized KL divergence becomes the normal KL divergence $\BD^{\X}[\Vc{p}\|\Vc{p}']=\left(\ln\frac{\Vc{p}}{\Vc{p}'}\right)^{\Transpose}\Vc{p}$\footnote{As we will see later, the condition $\Vc{1}^{\Transpose}\Vc{p}(t)=1$ need not be assumed but is automatically satisfied due to the topological constraint of the graph and the initial condition $\Vc{1}^{\Transpose}\Vc{p}(0)=1$ when we work on rMJP.}.

\subsection{Explicit form of dissipation functions for rLDG and CRN}
To determine the dissipation functions, we need the definition of force, which may depend on the phenomena and purpose\footnote{This is parallel to the problem of how to define the dual of $\Vc{x}$. The choice of logarithm is contingent on the domain and knowledge of physics and statistics.}. 
In physics, the flux-force relations, which are also called constitutive equations\cite{truesdell2004}, are central because they determine what kind of change is induced by an incurred force \footnote{Actual forms of the relations depend on the respective phenomena. 
Some relations were obtained empirically through experiments and others were computed theoretically from microscopic models.}. 
For rMJP and CRNs, the flux and force are conventionally defined using the one-way fluxes, $\Vc{\flux}^{+}(\Vc{x})$ and $\Vc{\flux}^{-}(\Vc{x})$ as
\begin{align}
    \Vc{\flux}&=\Vc{\flux}^{+}-\Vc{\flux}^{-}, &    \Vc{\tf}&=\ln \Vc{\flux}^{+} - \ln \Vc{\flux}^{-}, \label{eq:flux_force_relation}
\end{align}
where the dependency of $\Vc{\flux}^{\pm}(\Vc{x})$ on $\Vc{x}$ is abbreviated for notational simplicity.
In physics, assuming this form of force-flux relation goes by the name of the local detailed balance (LDB) assumption\footnote{LDB assumption is different from the DB condition in \defref{eq:st_cb_eq_state}.}, or the generalized detailed balance assumption\footnote{The validity of LDB was shown for rMJP and CRN with LMA kinetics via large deviation theory for the corresponding microscopic Markovian models or via its consistency with the macroscopic chemical thermodynamics\cite{bergmann1955Phys.Rev.,maes2021SciPostPhys.Lect.Notes}. }.
By defining the frenetic activity \cite{maes2017}: 
\begin{align}
    \Vc{\frenecy}\defeq 2 \sqrt{\Vc{\flux}^{+}\circ\Vc{\flux}^{-}}\in \mathbb{R}_{\com{\ge} 0}^{N_{\edge}}, \label{eq:frenecy},
\end{align}
we have a relation $\Vc{\flux}=\Vc{\frenecy}\circ \left[\frac{\exp(\Vc{\tf}/\com{2})-\exp(-\Vc{\tf}/\com{2})}{2}\right]$. 
For a fixed $\Vc{\frenecy}$, this relation between the pair $(\Vc{\flux},\Vc{\tf})$ is a one-to-one Legendre duality induced by the following specific form of dissipation functions: 
\begin{align}
\begin{split}
    \Dissp^{*}_{\Vc{\frenecy}}(\Vc{\tf})&\defeq \com{2} \Vc{\frenecy}^{\Transpose} \left[\cosh(\Vc{\tf}/\com{2})-\Vc{1}\right], \\
    \Dissp_{\Vc{\frenecy}}(\Vc{\flux})&\defeq \com{2}\Vc{\frenecy}^{\Transpose}\left(\diag\left[\frac{\Vc{\flux}}{\Vc{\frenecy}}\right] \sinh^{-1}\left(\frac{\Vc{\flux}}{\Vc{\frenecy}}\right) -  \left[\sqrt{\Vc{1}+\left(\frac{\Vc{\flux}}{\Vc{\frenecy}}\right)^{2}}-\Vc{1}\right]\right),
    \end{split}
    \label{eq:CRN_dissip}
\end{align}
which lead to the Legendre transformation: 
\begin{align}
    \Vc{\flux}&=\partial \Dissp^{*}_{\Vc{\frenecy}}(\Vc{\tf})=\Vc{\frenecy}\circ \sinh(\Vc{\tf}/\com{2}),&     \Vc{\tf}&=\partial \Dissp_{\Vc{\frenecy}}(\Vc{\flux})=\com{2} \sinh^{-1}\left(\frac{\Vc{\flux}}{\Vc{\frenecy}}\right).\label{eq:LT_flux_force}
\end{align}
We can easily verify that these functions satisfy the conditions for dissipation functions, i.e., \eqnref{eq:DissipCond1}, \eqnref{eq:DissipCond2}, and \eqnref{eq:DissipCond3}.

For the flux $\Vc{\flux}_{\mathrm{MA}}(\Vc{x})$ of LMA kinetics (\eqnref{eq:CRN_rate})\footnote{The dissipation functions in \eqnref{eq:CRN_dissip} and the induced Legendre transformation in \eqnref{eq:LT_flux_force} are not necessarily restricted to these particular types of force and activity. Actually, the extended LMA kinetics (\eqnref{eq:CRN_rate_eLMA}) can also be represented by replacing $\Vc{\frenecy}_{\mathrm{MA}}(\Vc{x})$ with $\Vc{\frenecy}_{\mathrm{eMA}}(\Vc{x})=\Vc{g}(\Vc{x})\circ \Vc{\frenecy}_{\mathrm{MA}}(\Vc{x})$.
 Thus, \eqnref{eq:CRN_dissip} could be applied to a wider class of kinetics than \eqnref{eq:CRN_rate}.}, the force and activity become
\begin{align}
\Vc{\tf}_{\mathrm{MA}}(\Vc{x};\Vc{K})&=\left[\ln \Vc{K} + \HIncMatrix^{\Transpose}\ln\Vc{x}\right], &
\Vc{\frenecy}_{\mathrm{MA}}(\Vc{x};\Vc{\kappa})&=2\Vc{\kappa}\circ\Vc{x}^{[\cmMatrix(\B^{+}+\B^{-})]^{\Transpose}/2}.
\label{eq:jffMAK}
\end{align}
where we introduced a transformation of the kinetic parameters $(\Vc{\kcoef}^{+}, \Vc{\kcoef}^{-})$ into the force part $\Vc{K}$ and activity part $\Vc{\kappa}$ as $\Vc{\kappa}\defeq \sqrt{\Vc{\kcoef}^{+}\circ\Vc{\kcoef}^{-}}$ and $\Vc{K}\defeq  \Vc{\kcoef}^{+}/\Vc{\kcoef}^{-}$\footnote{For CRN, $\Vc{K}$ is referred as the equilibrium constant in chemistry.}. 
Because $\Vc{\kcoef}^{\pm}=\Vc{\kappa}\circ \Vc{K}^{\pm 1/2}$ holds, 
$(\Vc{\kappa},\Vc{K})$ has the same information as $(\Vc{\kcoef}^{+}, \Vc{\kcoef}^{-})$.
Moreover, we can verify that the force and activity are dependent only on $\Vc{K}$ and $\Vc{\kappa}$, respectively.
The dissipation functions of the forms above and their relations to rLDG and CRN were derived from the large deviation function of the corresponding microscopic stochastic models\cite{mielke2014PotentialAnala,patterson2019MathPhysAnalGeom}.
Actually, the Bregman divergence $\BD_{\Vc{x}}^{\Jspace}[\Vc{\flux};\Vc{\flux}_{\mathrm{MA}}(\Vc{x})]$ of the dissipation functions is identical to the rate function of the flux for rMJP and CRN.
Thus, these dissipation functions are keystones connecting macroscopic and microscopic dynamics.

If there exists $\tilde{\Vc{y}}$ satisfying $-\HIncMatrix^{\Transpose}\tilde{\Vc{y}}=\ln \Vc{K}$, i.e., $\ln \Vc{K}\in \Img \HIncMatrix^{\Transpose}$, the force in \eqnref{eq:jffMAK} is represented as
\begin{align}
    \Vc{\tf}_{\mathrm{MA}}(\Vc{x};\Vc{K})&=\Grad_{\HIncMatrix}\left(\ln\frac{\Vc{x}}{\tilde{\Vc{x}}}\right)=\Grad_{\HIncMatrix} \partial_{\Vc{x}} \BD^{\X}[\Vc{x}\|\tilde{\Vc{x}}] \in \Img \HIncMatrix^{\Transpose}, \label{eq:equilibrium_force_CRN}
\end{align}
where $\tilde{\Vc{x}}$ is the Legendre conjugate of $\tilde{\Vc{y}}$\footnote{It should be noted that, while $\tilde{\Vc{y}}$ is not uniquely determined by $\Vc{K}$ in general, it does not cause problems. This is clarified in the following section (\secref{sec:InfoGeoEquilibrium}) by introducing appropriate affine subspaces.}. 
Thus, CRN (and rMJP) is an equilibrium flow of the generalized KL divergence $\BD^{\X}[\Vc{x}\|\tilde{\Vc{x}}]$ when the parameter $\Vc{K}$ satisfies $\ln \Vc{K}\in \Img \HIncMatrix^{\Transpose}$. 
In chemistry, the condition $\ln \Vc{K}\in \Img \HIncMatrix^{\Transpose}$ is called Wegscheider's equilibrium condition\cite{wegscheider1902Z.FuerPhys.Chem.,rao2016Phys.Rev.X}, and the CRN satisfying this parametric condition is called equilibrium CRN
\footnote{Historically, the equilibrium chemical systems were characterized by macroscopic thermodynamics. The equilibrium condition was derived as the necessary and sufficient condition that the flux of the LMA kinetics (\eqnref{eq:CRN_rate}) should satisfy to have consistent properties with the thermodynamic equilibrium systems. It was found only recently that the equilibrium properties are mathematically attributed to the generalized gradient flow structure.}.
Even if $\ln \Vc{K}\in \Img \HIncMatrix^{\Transpose}$ is not satisfied, we can represent $\ln \Vc{K}= -\HIncMatrix^{\Transpose}\tilde{\Vc{y}} + \Vc{\tf}_{NE}$ with $\Vc{\tf}_{NE}\not\in \Img \HIncMatrix^{\Transpose}$. 
The force in \eqnref{eq:jffMAK} is always represented as
\begin{align}
    \Vc{\tf}_{\mathrm{MA}}(\Vc{x};\Vc{K})&=\Grad_{\HIncMatrix}\left(\ln\frac{\Vc{x}}{\tilde{\Vc{x}}}\right) + \Vc{\tf}_{NE}=\left[\Grad_{\HIncMatrix} \left[\partial_{\Vc{x}} \BD^{\X}[\Vc{x}\|\tilde{\Vc{x}}]\right] + \Vc{\tf}_{NE}\right] , \label{eq:nonequilibrium_force_CRN}
\end{align}
which leads to the nonequilibrium flow (\eqnref{eq:gNGF}).
Thus, CRN with LMA kinetics as well as rLDG are generally within the class of \eqnref{eq:gNGF}. 

\TJK{\begin{ex}[Simplified Brusselator CRN\cite{feinberg2019,yoshimura2022} (continued)]\label{ex:BrusselatorFlow}
For the Brusselator CRN introduced in \exref{ex:BrusselatorEq}, the force and activity defined in \eqnref{eq:jffMAK} can be explicitly represented as
\begin{align}
\Vc{\tf}(\Vc{x})&=\overbrace{\begin{pmatrix}
\ln (\kcoef_{1}^{+}/\kcoef_{1}^{-})\\
\ln (\kcoef_{2}^{+}/\kcoef_{2}^{-})\\
\ln (\kcoef_{3}^{+}/\kcoef_{3}^{-})\\
\end{pmatrix}}^{\ln \Vc{K}}+
\overbrace{\begin{pmatrix}
-\ln x_{1} \\
\ln x_{1} - \ln x_{2}\\
\ln x_{2}-\ln x_{1}
\end{pmatrix}}^{\HIncMatrix^{\Transpose}\ln\Vc{x}}
, &
\Vc{\frenecy}(\Vc{x})&=2\overbrace{\begin{pmatrix}
\sqrt{\kcoef_{1}^{+}\kcoef_{1}^{-}}\\
\sqrt{\kcoef_{2}^{+}\kcoef_{2}^{-}}\\
\sqrt{\kcoef_{3}^{+}\kcoef_{3}^{-}}\\
\end{pmatrix}}^{\Vc{\kappa}}
\circ
\overbrace{\begin{pmatrix}
\sqrt{x_{1}}\\
\sqrt{x_{1}x_{2}}\\
\sqrt{x_{1}^{5}x_{2}}
\end{pmatrix}}^{\Vc{x}^{[\cmMatrix(\B^{+}+\B^{-})]^{\Transpose}/2}}.
\end{align}
\end{ex}
}

\begin{rmk}[Wegscheider's equilibrium condition and Detailed balance condition]
While we defined equilibrium flow by the specific functional form of force and obtained Wegscheider's equilibrium condition as the necessary and sufficient condition to have the equilibrium force under LMA kinetics, the equilibrium dynamics is often defined by the existence of the steady state satisfying the DB condition (\eqnref{eq:DBC}) in CRN theory.    
In addition, the DB condition is also often assumed in statistics when we design or analyze a random walk in parameter spaces, e.g., in the Markov Chain Monte Carlo (MCMC) simulations or in other random-walk-based optimization schemes\footnote{The DB condition is conventionally adopted because, for example, it makes it easy to obtain an MCMC with a desirable stationary distribution. This nice property comes from the gradient-flow property of the equilibrium flow.}.
These two are equivalent for (extended) LMA kinetics. Actually, $\Manifold^{\mathrm{DB}} \neq \emptyset$ means that there exists $\Vc{x}_{DB} \in \X$ such that $\Vc{\flux}_{MA}(\Vc{x}_{DB})=\Vc{0} \Leftrightarrow - \HIncMatrix^{\Transpose}\ln \Vc{x}_{DB}=\ln\Vc{K}$.
From the Fredholm alternative, we obtain the Wegscheider's equilibrium condition $\ln \Vc{K}\in \Img \HIncMatrix^{\Transpose}$ for the existence of $\Vc{x}_{DB}$.
\end{rmk}

\begin{rmk}[Linear graph Laplacian dynamics]
The linear graph Laplacian dynamics defined by \eqnref{eq:gLaplacian_dynamics} can be formally regarded as a generalized flow.
From the form of the graph Laplacian (\eqnref{eq:graphLaplacianLinear})\footnote{It should be noted that this representation holds only when $\Vc{\kcoef}^{+}=\Vc{\kcoef}^{-}$ holds.}, it is easy to see that \eqnref{eq:gLaplacian_dynamics} coincides with \eqnref{eq:gGF} if
\begin{align}
    \Pfunc(\Vc{x})&=\frac{1}{2}\langle\Vc{x},M_{0}\Vc{x} \rangle, & \Dissp^{*,q}_{\Vc{x}}(\Vc{\tf})&=\frac{1}{2}\langle\Vc{\tf},M^{1}\Vc{\tf}\rangle,
\end{align}
where $M_{0}=\identityM$, $M^{1}=\diag[\Vc{\kcoef}]$, and $\tilde{\Vc{x}}=\Vc{0}$. 
In contrast to rLDG, the natural state space and the corresponding dual is $\X=\Y=\Real^{N_{\node}}$\footnote{Even if we restrict the dynamics to $\X=\Y=\Real^{N_{\node}}_{>0}$, no problem arises for defining the generalized flow as long as we do not consider projections that we are going to introduce.}.
In \cite{ohara2021Geom.Sci.Inf.}, non-quadratic general $\Pfunc(\Vc{x})$ is considered as a class of nonlinear diffusion on a network from information geometric viewpoint.
\end{rmk}

\subsection{Some remarks on the dissipation functions for rLDG and CRN}\label{sec:dissp_CRN_LDG}
The dissipation functions in \eqnref{eq:CRN_dissip} have several notable properties.
First, they are separable:
\begin{align}
    \Dissp^{*}_{\Vc{\frenecy}(\Vc{x})}(\Vc{\tf})&= \sum_{e=1}^{N_{\edge}} \frenecy_{e}(\Vc{x}) \psi^{*}(\tf_{e}), &    \Dissp_{\Vc{\frenecy}(\Vc{x})}(\Vc{\flux})&= \sum_{e=1}^{N_{\edge}} \frenecy_{e}(\Vc{x}) \psi(\flux_{e}/\frenecy_{e}(\Vc{x})),
\end{align}
where 
\begin{align}
    \psi^{*}(\tf)&\defeq \com{2} \left[\cosh(\tf/\com{2})-1\right]\in [0,\infty), \\
    \psi(\flux)&\defeq \com{2}\left(\flux \sinh^{-1}\left(\flux\right) -  \left[\sqrt{\Vc{1}+\bar{\flux}^{2}}-1\right]\right)\in [0,\infty).
\end{align}
and $\Vc{\frenecy}(\Vc{x})$ is local: $\frenecy_{e}(\Vc{x})=2 \kappa_{e}\prod_{i=1}^{N_{\molX}}x_{i}^{(\gamma^{+}_{i,e}+\gamma^{-}_{i,e})/2}$.
The thermodynamic functions in \eqnref{eq:potentialfunction_X_MAK} are also separable\footnote{The locality and separability may sound natural. However, from the physical viewpoint, the Onsager matrix can have nondiagonal components, which implies nonseparable dissipation functions. In addition, equilibrium thermodynamics does not preclude thermodynamic functions from being nonseparable.}.  

Second, the scalar function $\psi^{*}(\tf)$ is the N-function.
The N-function of the $(\cosh(\tf)-1)$-type and the associated Orlicz space have been employed for establishing the infinite-dimensional information geometry by Pistone\cite{pistone1995Ann.Stat.,pistone2013Entropy,pistone2021ProgressinInformationGeometry:TheoryandApplications}.
In functional analysis, the Orlicz space is a generalization of the $L^{p}$ spaces, which arise naturally when we work on the $L \log^{+} L$ space for the divergences and large deviation functions.
Hence, the dissipation functions in \eqnref{eq:CRN_dissip} are tightly related to such topics.

Third, various information geometric measures and quantities are related to the dissipation functions in \eqnref{eq:CRN_dissip} and also to the associated quantities as follows:
\begin{align}
\frac{1}{4} \Dissp^{*}_{\Vc{\frenecy}}(\Vc{\tf})&=\frac{1}{2} \sum_{e=1}^{N_{\edge}}\left[\sqrt{\flux_{e}^{+}}- \sqrt{\flux_{e}^{-}}\right]^{2} =: \BD_{Hel}[\Vc{\flux}^{+};\Vc{\flux}^{-}]^{2}\\
\frac{1}{2}\Vc{1}^{\Transpose}\Vc{\frenecy} & = \sum_{e=1}^{N_{\edge}}\sqrt{\flux_{e}^{+}\flux_{e}^{+}}=:\mathrm{BC}[\Vc{\flux}^{+};\Vc{\flux}^{-}]\\
\langle \Vc{\flux},\Vc{\tf}\rangle&= \sum_{e=1}^{N_{\edge}}(\flux^{+}_{e}-\flux^{-}_{e})\ln \frac{\flux^{+}_{e}}{\flux^{-}_{e}}=: \BD_{Jef}[\Vc{\flux}^{+};\Vc{\flux}^{-}],
\end{align}
where $\BD_{Hel}[\Vc{\flux}^{+};\Vc{\flux}^{-}]$, $\mathrm{BC}[\Vc{\flux}^{+};\Vc{\flux}^{-}]$, and $\BD_{Jef}[\Vc{\flux}^{+};\Vc{\flux}^{-}]$ are the Hellinger--Kakutani distance, the Bhattacharyya coefficient, and the Jeffreys divergence (symmetrized KL divergence) for $\Vc{\flux}^{+}$ and $\Vc{\flux}^{-}$, respectively.
In addition, in physics, the bilinear pairing $\langle\Vc{\flux},\Vc{\tf}\rangle$ of a Legendre dual pair and its approximation using the Hessian matrix are often referred to as the entropy production rate (EPR) $\EPR$ and pseudo-entropy production rate (pEPR) $\pEPR$, respectively\cite{shiraishi2021JStatPhys,kobayashi2022Phys.Rev.Researcha}:
\begin{align}
\EPR&\defeq \langle \Vc{\flux},\Vc{\tf}\rangle=\sum_{e=1}^{N_{\edge}}(\flux^{+}_{e}-\flux^{-}_{e})\ln \frac{\flux^{+}_{e}}{\flux^{-}_{e}}.\label{eq:EPR}\\
\pEPR&\defeq \langle\Vc{\flux}, \FM_{\Vc{\frenecy},\Vc{\flux}}\Vc{\flux}\rangle = 2\sum_{e=1}^{N_{\edge}}\frac{(\flux^{+}_{e}-\flux^{-}_{e})^{2}}{\flux^{+}_{e}+\flux^{-}_{e}}\label{eq:pEPR},
\end{align}
where we treat $\Vc{\flux}\in \Jspace_{\Vc{x}}$ as a member of $\Tan_{\Vc{\flux}}\Jspace_{\Vc{x}}$ by the isomorphism: $\Jspace_{\Vc{x}} \cong \Tan_{\Vc{\flux}}\Jspace_{\Vc{x}} \cong \chain_{1}(\HGraph)$.
The pEPR $\pEPR$ is an approximation of EPR by replacing $\Vc{\tf}=\partial \Dissp_{\Vc{\frenecy}}(\Vc{\flux})$ with $\FM_{\Vc{\frenecy},\Vc{\flux}}\Vc{\flux}$ and works as a lower bound of $\EPR$: $\EPR \ge \pEPR$\cite{shiraishi2021JStatPhys}\footnote{This inequality is obtained directly from the inequality $2(a-b)^{2}/(a+b)\le (a-b)\ln a/b$ for $a,b>0$.}.

Finally, the dissipation functions in \eqnref{eq:CRN_dissip} are not the unique choice to reproduce the force-flux relation in \eqnref{eq:flux_force_relation}.
The quadratic dissipation functions $\Dissp^{q,*}_{\Vc{x}}(\Vc{\tf})\defeq \frac{1}{2}\langle \Vc{\tf},\metric^{*}_{\Vc{x}} \Vc{\tf} \rangle$ in \eqnref{eq:quadratic_dissp} with the following diagonal metric tensor can reproduce the relation in \eqnref{eq:flux_force_relation}:
\begin{align}
    \metric^{*}_{\Vc{x}}=\diag\left[\frac{\Vc{\flux}^{+}(\Vc{x})-\Vc{\flux}^{-}(\Vc{x})}{\ln \Vc{\flux}^{+}(\Vc{x})- \ln \Vc{\flux}^{-}(\Vc{x})} \right]=\diag\left[\left(\frac{\flux^{+}_{e}(\Vc{x})-\flux^{-}_{e}(\Vc{x})}{\ln \flux^{+}_{e}(\Vc{x})- \ln \flux^{-}_{e}(\Vc{x})}\right)_{e} \right].
\end{align}
This type of quadratic dissipation function was proposed even earlier than the non-quadratic ones\cite{chow2012ArchRationalMechAnal,maas2011JournalofFunctionalAnalysis,mielke2013Calc.Var.} and has been investigated\cite{liero2013Philos.Trans.R.Soc.Math.Phys.Eng.Sci.,yoshimura2022,vanvu2022}. 
Its advantage is that the induced geometry is Riemannian, and thus the information geometric argument is not necessarily required.
In addition, this Riemannian geometric structure is analogous to the formal Riemannian geometric structure of FPE and other diffusion processes on continuous manifolds induced via the $L^{2}$-Wasserstein geometry\cite{otto2001Commun.PartialDiffer.Equ.,villani2003} (\fgref{fig:IG_Wasser}).
Thus, this quadratic dissipation function provides a consistent extension of these results for FPE and diffusion processes to graphs and hypergraphs.
Nevertheless, the doubly dual flat structure with the non-quadratic dissipation functions that we introduce is also another sound generalization of the formal Riemannian geometry of FPE, as we see in the next subsection. 

As long as we focus only on the trajectory of the generalized flow (\eqnref{eq:gFlow}), the difference does not matter because both induce the same dynamics. 
However, the Bregman divergence of the quadratic dissipation functions is not directly related to the rate function of the microscopic stochastic models, while that of nonquadratic ones in \eqnref{eq:CRN_dissip} is\cite{patterson2019MathPhysAnalGeom}.
Thus, if we consider projections of fluxes and forces in the edge spaces, different choices of dissipation functions lead to different projections. \TJK{In addition, for non-quadratic dissipation functions, the contributions of the kinetic parameters $\Vc{\kcoef}^{\pm}$ can be clearly separated into the force part $\Vc{K}$ and the activity part $\Vc{\kappa}$ in the case of CRN with the LMA kinetics (\eqnref{eq:jffMAK}). This separation enables a physical realization of the projected flux as we derive in the following section. }

\begin{figure}[h]
\includegraphics[bb=0 500 1024 768, width=\linewidth]{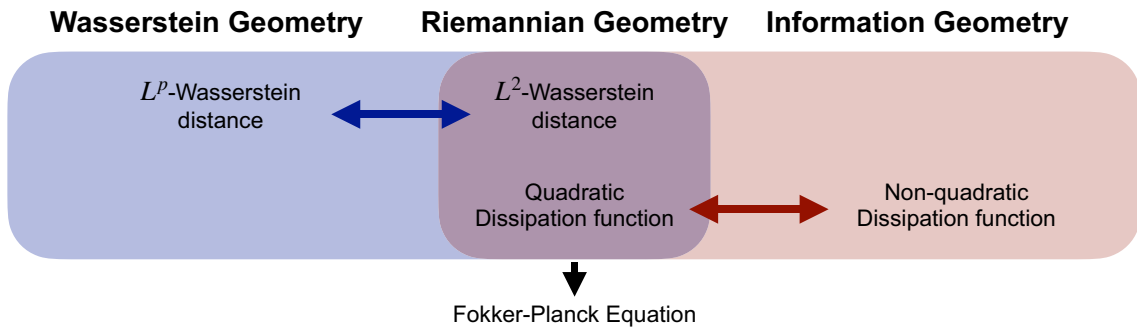}
\caption{A relationship between Wasserstein geometry and information geometry. The formal Riemannian geometric structure appears at their intersection. It should be noted that, while the regions of $L^{p}$-Wasserstein distance for $p\neq 2$ and nonquadratic dissipation function are not overlapping in this figure, this does not mean that they are unrelated. There may be undiscovered relations between these two regions.}
\label{fig:IG_Wasser}
\end{figure}

\subsection{Explicit forms of thermodynamic and dissipation functions for FPE}\label{sec:FPE_functions}
For FPE, the dualistic representation of the density $p(\Vc{r})$ and its logarithm $y(\Vc{r})=\ln p(\Vc{r})$ is also relevant. 
This duality is induced formally by the following thermodynamic functions\footnote{The base measure is omitted because this is just a formal one.}:
\begin{align}
    \Pfunc[p] & = \int [\ln p(\Vc{r}) - \Vc{1}]p(\Vc{r})\dd \Vc{r}, & 
    \Pfunc^{*}[y] &=  \int e^{y(\Vc{r})}\dd \Vc{r}, \label{eq:potential_FPE}
\end{align}
the Legendre transformations of which are
\begin{align}
    y(\Vc{r})&=\frac{\delta \Pfunc[p]}{\delta p} = \ln p(\Vc{r}), & p(\Vc{r})&=\frac{\delta \Pfunc^{*}[y]}{\delta y} = e^{y(\Vc{r})}.
\end{align}
The Bregman divergence becomes the KL divergence $\BD_{\X}[p\|p']=\int \dd \Vc{r}p(\Vc{r})\ln\frac{p(\Vc{r})}{p'(\Vc{r})}$.
In physics, the flux and force for FPE are defined conventionally as
\begin{align}
    \Vc{\flux}_{\mathrm{FP}}[p(\Vc{r})]&=D_{0}p(\Vc{r})\left\{\Vc{F}(\Vc{r})/D_{0}-\nabla \ln p(\Vc{r}) \right\}, \\
    \Vc{\tf}_{\mathrm{FP}}[p(\Vc{r})]&=D^{-1}_{0}\Vc{F}(\Vc{r})- \nabla \ln p(\Vc{r}) .
\end{align}
The dissipation functions associated with the force-flux relation above are 
\begin{align}
    \Dissp^{\mathrm{FP},*}_{\Vc{\frenecy}[p]}[\Vc{\tf}]&= \frac{1}{2}\int \Vc{\tf}(p(\Vc{r}))^{\Transpose}\metric^{*}_{p(\Vc{r})}\Vc{\tf}(p(\Vc{r}))\dd \Vc{r}, & 
    \Dissp^{\mathrm{FP}}_{\Vc{\frenecy}[p]}[\Vc{\flux}]&= \frac{1}{2}\int \Vc{\flux}(p(\Vc{r}))^{\Transpose}\metric_{p(\Vc{r})} \Vc{\flux}(p(\Vc{r}))\dd \Vc{r},
\end{align}
where $\metric^{*}_{p(\Vc{r})}\defeq \diag[ \Vc{\frenecy}[p(\Vc{r})] ] $, $\metric_{p(\Vc{r})}\defeq (\metric^{*}_{p(\Vc{r})})^{-1}$, and $\frenecy_{i}[p(\Vc{r})]=D_{0} p(\Vc{r})$.
Thus, the dissipation functions are formally quadratic and positive definite. 
If $\Vc{F}(\Vc{r})$ is a gradient of $U(\Vc{r})$ as $\Vc{F}(\Vc{r})=D_{0}\nabla U(\Vc{r})$, $\Vc{\tf}_{\mathrm{FP}}[p(\Vc{r})]=- \nabla \ln \frac{p(\Vc{r})}{\tilde{p}(\Vc{r})}$ holds where $\tilde{p}(\Vc{r})\defeq \exp[U(\Vc{r})]$.
Then, the dissipation functions, the bilinear pairing $\langle \Vc{\flux}_{\mathrm{FP}},\Vc{\tf}_{\mathrm{FP}}\rangle$, the EPR $\EPR_{\mathrm{FP}}$ in \eqnref{eq:EPR}, and the pEPR $\pEPR_{\mathrm{FP}}$ in \eqnref{eq:pEPR} formally consolidate into the same quantity:  
\begin{align}
     2\Dissp^{\mathrm{FP},*}_{\Vc{\frenecy}[p]}[\Vc{\tf}_{\mathrm{FP}}]=2\Dissp^{\mathrm{FP}}_{\Vc{\frenecy}[p]}[\Vc{\flux}_{\mathrm{FP}}]=\langle \Vc{\flux}_{\mathrm{FP}},\Vc{\tf}_{\mathrm{FP}}\rangle=\EPR_{\mathrm{FP}}=\pEPR_{\mathrm{FP}}= D_{0}\int p(\Vc{r})\left(\nabla_{\Vc{r}} \ln \frac{p(\Vc{r})}{\tilde{p}(\Vc{r})}\right)^{2}\mathrm{d}\Vc{r}.
\end{align}
The last quantity without $D_{0}$ is known as relative Fisher information\cite{villani2003,yamano2013J.Math.Phys.} and Hyv\"arinen divergence\cite{hyvarinen2005J.Mach.Learn.Res.,lods2015Entropy} between $p$ and $\tilde{p}$.
For $U(\Vc{f})=0$, it reduces to the Fisher information number $\FIN[p]$ in \eqnref{eq:FisherInfoNum}.
This consolidation is a source of confusion, because the same quantity for FPE or linear diffusion processes has different names in different contexts and in different disciplines.
However, they actually have different definitions, roles, and meanings, which becomes explicit in the information-geometric formulation.

\section{Orthogonal subspaces, dual foliation, and Pythagorean relation}\label{sec:orthogonality}
To investigate the behaviors and properties of the equilibrium (\eqnref{eq:gGF}) and nonequilibrium (\eqnref{eq:gNGF}) flow, especially its topological and algebraic constraints from the graph or hypergraph structure,  information geometry provides the ideal tools.
In particular, the four affine subspaces associated with the cycle and cocycle subspaces of the chain and cochain complexes (\fgref{fig:Subspaces}) form dual foliations via the Legendre transformation, whose geometric properties are captured by information geometry\cite{shima2007,amari2016}. 
It should be noted that the results of this section do not assume the specific forms of the thermodynamic and dissipation functions introduced in \secref{sec:explict_form}.

\begin{figure}[h]
\includegraphics[bb=0 200 1024 768, width=\linewidth]{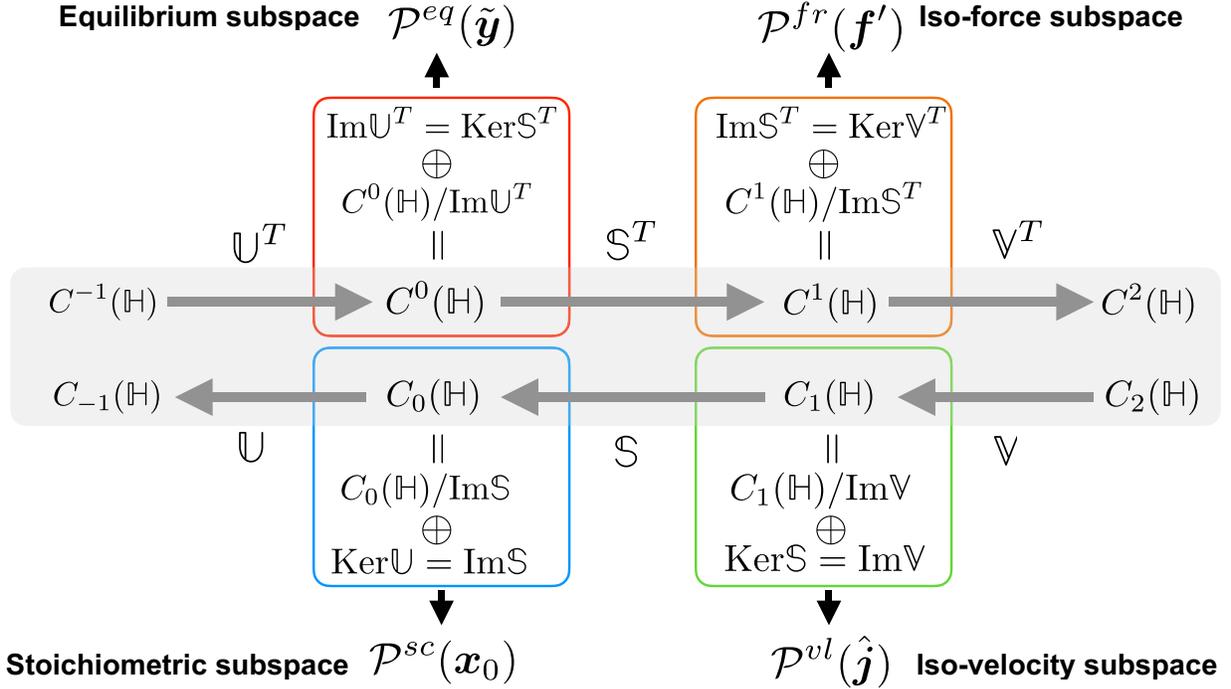}
\caption{Diagrammatic representation of the four subspaces and their relationship with the chain and cochain complexes of $\HGraph$.}
\label{fig:Subspaces}
\end{figure}

\subsection{Four affine subspaces}
Two families of orthogonally complement affine subspaces are naturally introduced on $\X$ and $\Y$, respectively, from the topological structure of graph and hypergraph, i.e., $\IncMatrix$ and $\HIncMatrix$. 

\begin{dfn}[Stoichiometric subspaces in $\X$]\label{dfn:stoiSubspace}
The stoichiometric subspaces are defined as\footnote{In CRN theory, a stoichiometric subspace is called stoichiometric compatibility class\cite{feinberg2019}.} 
\begin{align}
    \Polytope^{sc}(\Vc{x}_{0})&\defeq \{\Vc{x}\in \X | \Vc{x}-\Vc{x}_{0} \in \Img\HIncMatrix\}, \quad \Vc{x}_{0}\in \X \label{eq:stoichiometric_subspace}
\end{align}
where $\Vc{x}_{0}$ is a parameter to specify the position of the subspace (\fgref{fig:Subspaces}, lower left)\footnote{Because $\Polytope^{sc}(\Vc{x}_{0})$ is restricted within the positive orthant $\X$, $\Polytope^{sc}(\Vc{x}_{0})$ is a polyhedron. If bounded, it is called a polytope in discrete geometry and also in combinatorial optimization\cite{wolsey1999}. 
However, we abuse the word (affine) subspace for $\Polytope^{sc}(\Vc{x}_{0})$, and use polyhedron or polytope when we care about the boundary.}.
\end{dfn}

\begin{dfn}[Equilibrium subspaces in $\Y$]\label{dfn:EquiSubspace}
The equilibrium subspaces (\fgref{fig:Subspaces}, upper left) are defined as 
\begin{align}
    \Polytope^{eq}(\tilde{\Vc{y}})&\defeq\left\{\Vc{y}\in \Y| \Vc{y}-\tilde{\Vc{y}} \in \Ker \HIncMatrix^{\Transpose}\right\}, \quad \tilde{\Vc{y}}\in \Y.\label{eq:riumstoichiometric_subspace}
\end{align}
\end{dfn}
$\Polytope^{sc}(\Vc{x}_{0})$ and $\Polytope^{eq}(\tilde{\Vc{y}})$ are of orthogonal complement to each other: $\langle\Vc{x}-\Vc{x}_{0},\Vc{y}'-\tilde{\Vc{y}}\rangle=0$ for $\Vc{x}\in \Polytope^{sc}(\Vc{x}_{0})$ and $\Vc{y}' \in \Polytope^{eq}(\tilde{\Vc{y}})$\footnote{In this work, orthogonality always means the orthogonal complement in dual vector spaces except otherwise stated.}.
Because $\HIncMatrix$ and $\HIncMatrix^{\Transpose}$ are the discrete differentials, $\delta_{1}$ and $\delta^{0}$,  $\Polytope^{sc}(\Vc{x}_{0})$ and $\Polytope^{eq}(\tilde{\Vc{y}})$ are associated with the $0$-cycle and $0$-cocycle spaces, respectively.

Two other families of orthogonal-complement subspaces are introduced on $\Jspace_{\Vc{x}}$ and $\Fspace_{\Vc{x}}$. 
\begin{dfn}[Iso-velocity subspaces in $\Jspace_{\Vc{x}}$]\label{dfn:IsovelSubspace}
The iso-velocity subspaces (\fgref{fig:Subspaces}, lower right) are defined as
\begin{align}
\Polytope^{vl}(\hat{\Vc{\flux}})=\left\{\Vc{\flux}\in \Jspace_{\Vc{x}}|\Vc{\flux}-\hat{\Vc{\flux}} \in \Ker\stoiMatrix\right\}, \qquad \hat{\Vc{\flux}}\in \Jspace_{\Vc{x}}.
\end{align}
\end{dfn}

\begin{dfn}[Iso-force subspaces in $\Fspace_{\Vc{x}}$]\label{dfn:IsoforceSubspace}
The iso-external-force subspaces, iso-force subspaces in short, (\fgref{fig:Subspaces}, upper right)  are defined as  
\begin{align}
    \Polytope^{fr}(\Vc{\tf}')\defeq \left\{\Vc{\tf}\in \Fspace_{\Vc{x}}|\Vc{\tf}-\Vc{\tf}'\in\Img\stoiMatrix^{\Transpose} \right\}, \qquad \Vc{\tf}'\in \Fspace_{\Vc{x}}.
\end{align}
\end{dfn}

Again, from the correspondence of $\delta_{1}=\HIncMatrix$ and $\delta^{0}=\HIncMatrix^{\Transpose}$,  $\Polytope^{vl}(\hat{\Vc{\flux}})$ and $\Polytope^{fr}(\Vc{\tf}')$ are associated with the $1$-cycle and $1$-cocycle spaces, respectively.
We specifically call $\Polytope^{vl}(\Vc{0})$ and $\Polytope^{fr}(\Vc{0})$ zero-velocity subspace and equilibrium force subspace, respectively.

\subsection{Meaning of the subspaces}
All four subspaces are natural constituents in the theory of algebraic graph theory and homological algebra.
Here, we provide their meaning in terms of the dynamics on graphs and hypergraphs.

The stoichiometric and iso-velocity subspaces, $\Polytope^{sc}(\Vc{x}_{0})$ and $\Polytope^{vl}(\hat{\Vc{\flux}})$, are related by the continuity equation (\eqnref{CRN_rateEq}).
From the continuity equation, $\Polytope^{vl}(\hat{\Vc{\flux}})$ is the set of fluxes that induce the same velocity as a reference $\hat{\Vc{\flux}}$ does: $\Vc{\flux}\in \Polytope^{vl}(\hat{\Vc{\flux}}) \Longleftrightarrow \dot{\Vc{x}}=-\HIncMatrix\hat{\Vc{\flux}}=-\HIncMatrix\Vc{\flux} $.
Thereby, $\Polytope^{vl}(\hat{\Vc{\flux}})$ is parametrized as follows:
\begin{align}
\Polytope^{vl}(\dot{\Vc{x}})&= \{\Vc{\flux}\in \Jspace_{\Vc{x}}| - \HIncMatrix\Vc{\flux} =\dot{\Vc{x}}\},\quad \dot{\Vc{x}}\in \Img[\HIncMatrix]=\Ker[\consMatrix],
\end{align}
This subspace is crucial to characterize fluxes that can realize the same dynamics as the reference one.

The stoichiometric subspace $\Polytope^{sc}(\Vc{x}_{0})$ determines the subspace in which the dynamics are algebraically constrained via the topology of the underlying graph or hypergraph.
Because $\dot{\Vc{x}}=-\HIncMatrix\Vc{\flux}(\Vc{x}(t))$, for an initial state $\Vc{x}(0)=\Vc{x}_{0}$, $\Vc{x}(t)-\Vc{x}_{0} \in \Img[\HIncMatrix]$ should hold, meaning that $\Vc{x}(t)\in \Polytope^{sc}(\Vc{x}_{0})$.
Thus, $\Polytope^{sc}(\Vc{x}_{0})$ is the subspace in which the dynamics are restricted by the initial condition $\Vc{x}_{0}$.
$\Polytope^{sc}(\Vc{x}_{0})$ can also be represented parametrically by the quantities which are conserved by the dynamics. 
For any vector $\Vc{u} \in \Ker \HIncMatrix^{\Transpose}$, $\eta(t)\defeq \Vc{u}^{\Transpose}\Vc{x}(t)$ is constant over time:
\begin{align}
    \dot{\eta}(t)=\frac{\dd \Vc{u}^{\Transpose}\Vc{x}(t)}{\dt}=\Vc{u}^{\Transpose}\frac{\dd \Vc{x}(t)}{\dt}=-\Vc{u}^{\Transpose}\HIncMatrix\Vc{\flux}(\Vc{x})=0.
\end{align}
In \secref{sec:chain_cochain}, we defined a matrix $\consMatrix$ by a complete basis of $\Ker \HIncMatrix^{\Transpose}$ so that $\Img \consMatrix^{\Transpose}=\Ker \HIncMatrix^{\Transpose}$.
Using $\consMatrix$, the conserved quantities for a given initial condition $\Vc{x}_{0}$ are obtained as $\Vc{\eta}=\consMatrix\Vc{x}_{0}=\consMatrix\Vc{x}(t)$.
Because $\Img \consMatrix$ is isomorphic to $\chain_{-1}(\HGraph)$, the stoichiometric subspace is explicitly parametrized by the conserved quantities (an element of $\chain_{-1}(\HGraph)$):
\begin{align}
    \Polytope^{sc}(\Vc{\eta})&= \{\Vc{x}\in \X| \consMatrix\Vc{x} =\Vc{\eta}\},\quad \Vc{\eta}\in \chain_{-1}(\HGraph).
\end{align}
For rMJP, the conserved quantity is reduced to the conservation of probability $\Vc{1}\Vc{p}(t)=1$ and $\Polytope^{sc}(\Vc{p}_{0})$ becomes the probability simplex. 
Because $\Ker\IncMatrix^{\Transpose}$ determines the connected components of the graph $\Graph$ and we conventionally assume that the underlying graph is connected in rMJP, we only have the one-dimensional cokernel space and one conserved quantity, which is $\eta =1$.
Thus, the conservation of probability or, equivalently, the restriction of $\Vc{p}$ in the probability simplex is automatically guaranteed from the topological constraint of the dynamics if we start from the initial state satisfying $\Vc{1}^{\Transpose}\Vc{p}_{0}=1$.

The iso-force subspace $\Polytope^{fr}(\Vc{\tf}')$ and the equilibrium subspace $\Polytope^{eq}(\tilde{\Vc{y}})$ are related to the equilibrium and nonequilibrium force equations, \eqnref{eq:gradEq} and \eqnref{eq:nonequF}.
The equilibrium force defined in \eqnref{eq:gradEq} satisfies $\Vc{\tf}(\Vc{x})\in \Img \HIncMatrix^{\Transpose}=\Polytope^{fr}(\Vc{0})$.
Thus, the equilibrium-force subspace $\Polytope^{fr}(\Vc{0})$ is literally the set of equilibrium forces. 
$\Polytope^{fr}(\Vc{\tf}')$ is its shift by $\Vc{\tf}'\in \Fspace_{\Vc{x}}$.
Using $\cycMatrix$ defined in \secref{sec:chain_cochain}, we can represent $\Polytope^{fr}$ parametrically as
\begin{align}
    \Polytope^{fr}(\Vc{\zeta})&= \{\Vc{\tf}\in \Fspace_{\Vc{x}}| \cycMatrix^{\Transpose}\Vc{\tf} =\Vc{\zeta}\},\quad \Vc{\zeta}\in \chain^{2}(\HGraph)
\end{align}
because $\Fspace_{\Vc{x}}/\Img \HIncMatrix^{\Transpose}\cong \Fspace_{\Vc{x}}/\Ker \cycMatrix^{\Transpose}\cong\Img \cycMatrix^{\Transpose}\cong\chain^{2}(\HGraph)$.
Thus, $\Vc{\zeta}$ characterizes the type of nonequilibrium forces quotient by the equilibrium forces.

Finally, the equilibrium subspace $\Polytope^{eq}(\tilde{\Vc{y}})$ can also be regarded as the set of potentials $\Vc{y}$ that generate the same equilibrium force because any $\Vc{y}\in \Polytope^{eq}(\tilde{\Vc{y}})$ satisfies $\Vc{\tf}'=\HIncMatrix^{\Transpose}\Vc{y}=\HIncMatrix^{\Transpose}\tilde{\Vc{y}} \in \Polytope^{fr}(\Vc{0})$.
Due to this, the equilibrium subspace $\Polytope^{eq}$ is parameterized as
\begin{align}
    \Polytope^{eq}(\Vc{\tf}')&= \{\Vc{y}\in \Y| \HIncMatrix^{\Transpose}\Vc{y} =\Vc{\tf}'\},\quad \Vc{\tf}'\in \Img[\HIncMatrix^{\Transpose}]. \label{eq:Peq_param}
\end{align}
The parametric forms of the subspaces are summarized as follows:
\begin{align}
    \Polytope^{vl}(\dot{\Vc{x}})&= \{\Vc{\flux}\in \Jspace_{\Vc{x}}| - \HIncMatrix\Vc{\flux} =\dot{\Vc{x}}\}, & \dot{\Vc{x}} &\in \Img[\HIncMatrix]=\Ker[\consMatrix],\\
    \Polytope^{sc}(\Vc{\eta})&= \{\Vc{x}\in \X| \consMatrix\Vc{x} =\Vc{\eta}\},& \Vc{\eta} &\in \Img[\consMatrix]=\chain_{-1}(\HGraph),\\
    \Polytope^{fr}(\Vc{\zeta})&= \{\Vc{\tf}\in \Fspace_{\Vc{x}}| \cycMatrix^{\Transpose}\Vc{\tf} =\Vc{\zeta}\}, &  \Vc{\zeta}&\in \Img[\cycMatrix^{\Transpose}]=\chain^{2}(\HGraph),\\
    \Polytope^{eq}(\Vc{\tf}')&= \{\Vc{y}\in \Y| \HIncMatrix^{\Transpose}\Vc{y} =\Vc{\tf}'\}, &  \Vc{\tf}' &\in \Img[\HIncMatrix^{\Transpose}]=\Ker[\cycMatrix^{\Transpose}].
\end{align}
From these subspaces, we can obtain dual foliations on the vertex and edge spaces.

\subsection{Dual Manifold, Dual Foliation, and Pythagorean relation in vertex spaces}
For the subspaces $\Polytope^{sc}$ and $\Polytope^{eq}$ in the density and potential spaces, we introduce their Legendre transformation via the thermodynamic functions, $\Pfunc(\Vc{x})$ and $\Pfunc^{*}(\Vc{y})$, which form the dual foliation with the subspaces of orthogonal complement (\fgref{fig:Foliations}, left).
\begin{figure}[h]
\includegraphics[bb=0 30 1024 768, width=\linewidth]{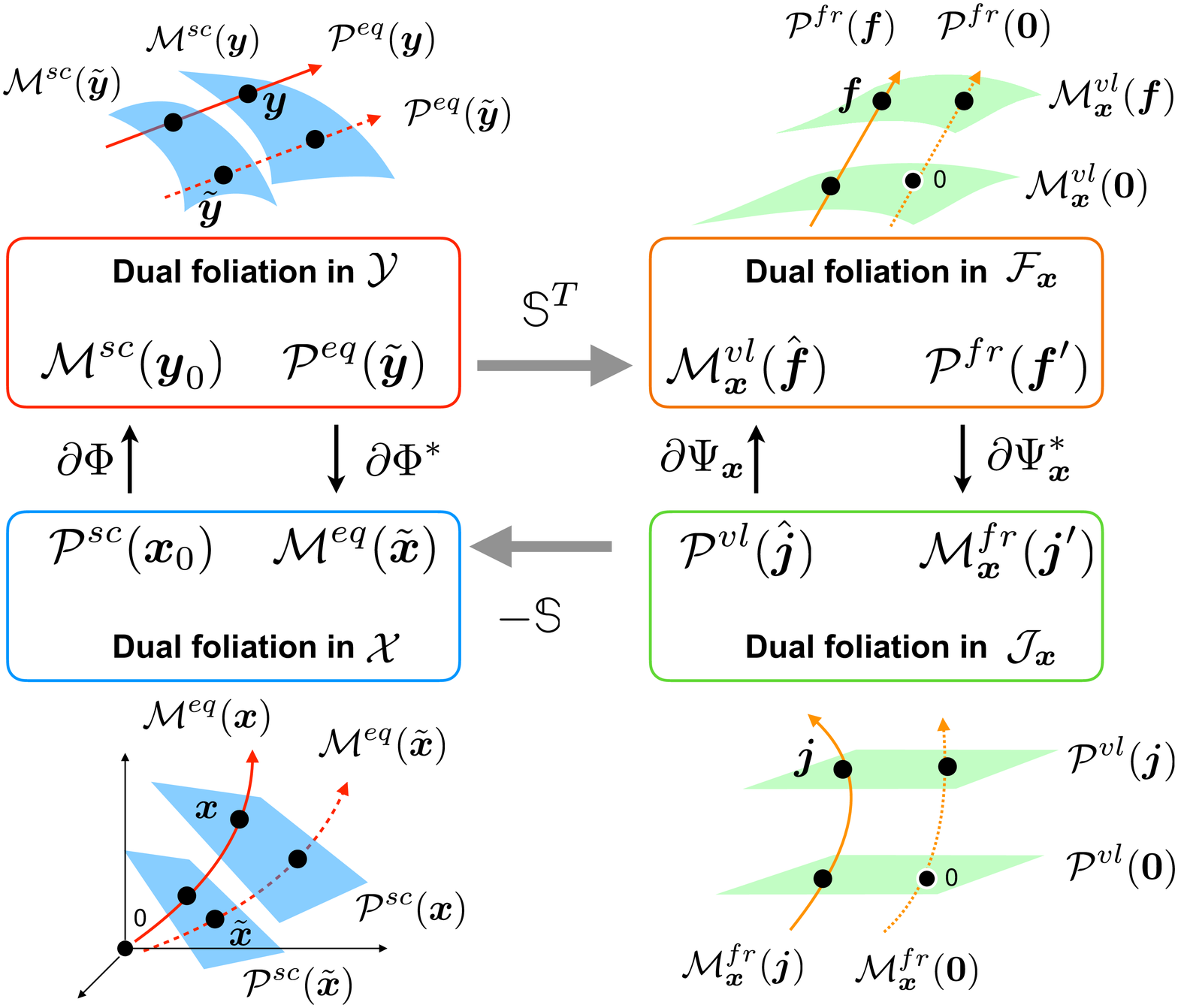}
\caption{Diagrammatic representation of the dual foliations in $\X$, $\Y$, $\Jspace_{\Vc{x}}$, and $\Fspace_{\Vc{x}}$ spaces.}
\label{fig:Foliations}
\end{figure}

\begin{dfn}[Stoichiometric manifold in $\Y$ and equilibrium manifold in $\X$]\label{dfn:StoiEquManifolds}
The stoichiometric and equilibrium manifolds (\fgref{fig:Foliations}, left) are defined respectively as 
\begin{align}
    \Variety^{sc}(\Vc{y}_{0})&\defeq \partial \Pfunc[\Polytope^{sc}(\Vc{x}_{0})]\subset \Y, \quad \Vc{y}_{0}=\partial \Pfunc(\Vc{x}_{0}),\\
    \Variety^{eq}(\tilde{\Vc{x}}) & \defeq \partial \Pfunc^{*}[\Polytope^{eq}(\tilde{\Vc{y}})] \subset \X, \quad \tilde{\Vc{x}}=\partial \Pfunc^{*}(\tilde{\Vc{y}}). \label{eq:Manifold_eq}
\end{align}
\end{dfn}
\begin{lmm}[Dual foliations in density and potential spaces\cite{sughiyama2022Phys.Rev.Research}]\label{lmm:DualFoliationVertex}
$\Polytope^{sc}$ and $\Variety^{eq}$ are foliations of $\X$, and $\Variety^{sc}$ and $\Polytope^{eq}$ are foliations of $\Y$.
For each pair of $(x_{0},\tilde{x})$, the intersection of $\Polytope^{sc}(\Vc{x}_{0})$ and $\Variety^{eq}(\tilde{\Vc{x}})$ is unique and transversal. The same applies to $\Variety^{sc}(\Vc{y}_{0})$ and $\Polytope^{eq}(\tilde{\Vc{y}})$.
Then, $(\Polytope^{sc}, \Variety^{eq})$ and $(\Variety^{sc}, \Polytope^{eq})$ form dual foliations (nonlinear coordinate systems) in $\X$ and $\Y$ spaces, respectively.
\end{lmm}
\begin{proof}
The polyhedron $\Polytope^{sc}(\Vc{x}_{0})$ and the affine subspace $\Polytope^{eq}(\tilde{\Vc{y}})$ can cover the whole $\X$ and $\Y$ by changing $\Vc{x}_{0}$ and $\tilde{\Vc{y}}$, respectively.
Similarly,  $\Variety^{eq}(\tilde{\Vc{x}})$ and $\Variety^{sc}(\Vc{y}_{0})$ can cover the whole $\X$ and $\Y$ because Legendre transformations by the thermodynamic functions are one-to-one between $\X$ and $\Y$.
Consider the intersection of $\Polytope^{sc}(\Vc{x}_{0})$ and  $\Variety^{eq}(\tilde{\Vc{x}})$ in $\X$ space. 
The condition that $\Polytope^{sc}(\Vc{x}_{0}) \cap \Variety^{eq}(\tilde{\Vc{x}}) \neq \emptyset$ is related to the existence of $\Vc{x}^{\dagger}$ defined by the following convex optimization problem:
\begin{align}
\Vc{x}^{\dagger}\defeq \arg\min_{\Vc{x}\in \overline{\Polytope^{sc}(\Vc{x}_{0})}} \BD^{\X}_{\Pfunc}[\Vc{x}\|\tilde{\Vc{x}}]. \label{eq:opt_prob_x}
\end{align}
Because of the properties of $\Pfunc(\Vc{x})$, $\BD^{\X}_{\Pfunc}[\Vc{x}\|\tilde{\Vc{x}}]$ and its restriction to $\overline{\Polytope^{sc}(\Vc{x}_{0})}$ are strictly convex with respect to $\Vc{x}$.
Thus, $\Vc{x}^{\dagger}$ is unique and either satisfies the stationarity condition $\Vc{x}^{\dagger} \in \Polytope^{sc}(\Vc{x}_{0}) \cap \Variety^{eq}(\tilde{\Vc{x}})$ if $\Vc{x}^{\dagger}\in \Polytope^{sc}(\Vc{x}_{0})$  or locates on the boundary $\partial \X$ if $\Vc{x}^{\dagger}\not\in \Polytope^{sc}(\Vc{x}_{0})$, where we used $\HIncMatrix^{\Transpose}\frac{\partial \BD^{\X}[\Vc{x}\|\tilde{\Vc{x}}]}{\partial \Vc{x}}  =\Vc{0} \Leftrightarrow \HIncMatrix^{\Transpose}(\Vc{y}-\tilde{\Vc{y}}) =\Vc{0} \Leftrightarrow \Vc{y}\in \Polytope^{eq}(\tilde{\Vc{y}})\Leftrightarrow \Vc{x}\in \Manifold^{eq}(\tilde{\Vc{x}})$.
Let $\Vc{x}_{bd}\in \partial \X$ and $\Vc{x}_{in}\in \X$ be arbitrary points on the boundary and interior of $\X$.
From the condition \eqnref{eq:thermofuncCond2} of the thermodynamic function, 
for $\Vc{x}_{\lambda}\defeq \lambda \Vc{x}_{in} + (1-\lambda)\Vc{x}_{bd}$ where $\lambda \in [0,1]$, 
\begin{align}
\lim_{\lambda \to +0}\frac{\dd \BD^{\X}[\Vc{x}_{\lambda}\|\tilde{\Vc{x}}]}{\dd \lambda}=\lim_{\lambda \to +0}\left[\frac{\dd \Pfunc(\Vc{x}_{\lambda})}{\dd \lambda} - \left\langle\tilde{\Vc{y}}, \frac{\dd \Vc{x}_{\lambda}}{\dd \lambda} \right\rangle \right] = -\infty.
\end{align}
Thus,  $\Vc{x}^{\dagger}\not\in \X$ is excluded, and the intersection exists, i.e., $\Vc{x}^{\dagger}\in \Polytope^{sc}(\Vc{x}_{0}) \cap \Variety^{eq}(\tilde{\Vc{x}})$.
The intersection $\Vc{x}^{\dagger}$ is unique and transversal because $\langle\Vc{x}_{sc}-\Vc{x}^{\dagger}, \Vc{y}_{eq}-\Vc{y}^{\dagger}\rangle=0$ holds for any $\Vc{x}_{sc}\in \Polytope^{sc}(\Vc{x}_{0})$ and $\Vc{x}_{eq} \in \Variety^{eq}(\tilde{\Vc{x}})$ and the dimensions of $\Polytope^{sc}(\Vc{x}_{0})$ and $\Variety^{eq}(\tilde{\Vc{x}})$ are complementary because $\Polytope^{sc}(\Vc{x}_{0})$ and $\Polytope^{eq}(\tilde{\Vc{y}})$ are of orthogonal complement (see also the proof in \cite{kobayashi2022Phys.Rev.Researcha}).
As a result, $\Vc{x}^{\dagger} \in \Polytope^{sc}(\Vc{x}_{0}) \cap \Variety^{eq}(\tilde{\Vc{x}})$ always exists, and $(\Polytope^{sc}, \Variety^{eq})$ forms a dual foliation in $\X$.
Also $(\Variety^{sc}, \Polytope^{eq})$ does in $\Y$ because they are bijective Legendre duals of $(\Polytope^{sc}, \Variety^{eq})$.
\end{proof}
This result is reduced to Birch's theorem\cite{craciun2009JournalofSymbolicComputation,pachter2005} and the seminal result by Horn and Jackson\cite{horn1972Arch.RationalMech.Anal.} when the thermodynamic function is the generalized KL divergence.

With the dual foliation, we can consider the generalized Pythagorean relations and orthogonal decomposition.
For any three points satisfying $\Vc{x} \in \Polytope^{sc}(\Vc{x}_{0})$, $\Vc{x}_{q} \in \Variety^{eq}(\tilde{\Vc{x}})$, and $\Vc{x}^{\dagger} = \Polytope^{sc}(\Vc{x}_{0}) \cap \Variety^{eq}(\tilde{\Vc{x}})$\footnote{We abuse the notation $\Vc{x}^{\dagger} = \Polytope^{sc}(\Vc{x}_{0}) \cap \Variety^{eq}(\tilde{\Vc{x}})$ because the intersection $\Polytope^{sc}(\Vc{x}_{0}) \cap \Variety^{eq}(\tilde{\Vc{x}})$ is a unique point.}, we have the generalized Pythagorean relation:
\begin{align}
\BD^{\X}[\Vc{x}\|\Vc{x}_{q}]=\BD^{\X}[\Vc{x}\|\Vc{x}^{\dagger}] + \BD^{\X}[\Vc{x}^{\dagger}\|\Vc{x}_{q}]. \label{eq:gPR_X}
\end{align}

In $\Y$ space, we also have the dual version of the relations as 
\begin{align}
\BD^{\Y}[\Vc{y}_{q}\|\Vc{y}]=\BD^{\Y}[\Vc{y}_{q}\|\Vc{y}^{\dagger}] + \BD^{\Y}[\Vc{y}^{\dagger}\|\Vc{y}].
\end{align}
These relations are used to characterize the steady state of equilibrium and nonequilibrium flow geometrically and also variationally.

\begin{rmk}[Interpretation in terms of statistical inference]
The meaning of the equilibrium manifold in statistics can be clarified more explicitly by considering the specific form of thermodynamic function (\eqnref{eq:potential_p}).
For this thermodynamic function, the equilibrium manifold $\Variety^{eq}(\tilde{\Vc{p}})$ is represented as 
\begin{align}
    \Variety^{eq}(\tilde{\Vc{p}})&=\left\{\Vc{p}\in \X| \ln\Vc{p}-\ln\tilde{\Vc{p}} \in \Ker \HIncMatrix^{\Transpose}\right\}
    =\left\{\Vc{p}\in \X| \Vc{p}=\tilde{\Vc{p}}\circ \exp\left[\consMatrix^{\Transpose}\Vc{\eta}^{*}\right], \Vc{\eta}^{*}\in \chain^{-1}(\HGraph) \right\} \label{eq:exp_family}
\end{align}
where we use the fact $\HIncMatrix^{\Transpose}\consMatrix^{\Transpose}=0$. 
Thus, $\Variety^{eq}(\tilde{\Vc{p}})$ is an exponential family with algebraic constraints via $\consMatrix^{\Transpose}$.
In contrast, $\Polytope^{sc}(\Vc{\eta})$ can be regarded as the data manifold, which constrains $\Vc{p}$ by $\Vc{\eta}=\consMatrix\Vc{p}$, because $\consMatrix\Vc{p}$ can be interpreted as expectation of observables $\{\Vc{u}_{\ell}\}_{\ell\in [1,N_{\cons}]}$. 
Thus, the intersection $\Vc{p}^{\dagger}=\Polytope^{sc}(\Vc{\eta})\cap \Variety^{eq}(\tilde{\Vc{p}})$ is the maximum likelihood estimator. 
The exponential family with linear algebraic constraints as in \eqnref{eq:exp_family} appears in algebraic statistics where $\consMatrix$ is sometimes called the design matrix\cite{pachter2005}\footnote{In algebraic statistics, $\consMatrix$ is explicitly given as constraints of a statistical model. In the dynamics on graphs and hypergraphs, $\HIncMatrix$ is explicitly given, and $\consMatrix$ is implicitly defined as a complete basis of $\Ker \HIncMatrix^{\Transpose}$. As a result, their connection is not apparently obvious.}.
\end{rmk}

\subsection{Dual manifold, dual foliation in edge spaces and generalized Helmholtz-Hodge-Kodaira decomposition}
For the edge spaces, we similarly introduce the iso-velocity and iso-force manifolds, which are the duals of $\Polytope^{fr}(\Vc{\tf}')$ and $\Polytope^{vl}(\hat{\Vc{\flux}})$, respectively, via the Legendre transformations by the dissipation functions, $\Dissp_{\Vc{x}}(\Vc{\flux})$ and $\Dissp^{*}_{\Vc{x}}(\Vc{\tf})$ (\fgref{fig:Foliations}, right):
\begin{dfn}[Iso-velocity manifold in $\Fspace_{\Vc{x}}$ and iso-force manifold in $\Jspace_{\Vc{x}}$]\label{dfn:IsoVelForceManifolds}
The iso-velocity and iso-force manifolds (\fgref{fig:Foliations}, right) are defined as follows:
\begin{align}
    \Variety^{vl}_{\Vc{x}}(\hat{\Vc{\tf}})&\defeq \partial \Dissp_{\Vc{x}}[\Polytope^{vl}(\hat{\Vc{\flux}})]\subset \Fspace_{\Vc{x}}, & \hat{\Vc{\tf}}&=\partial \Dissp_{\Vc{x}}(\hat{\Vc{\flux}}),\\
    \Variety^{fr}_{\Vc{x}}(\Vc{\flux}') & \defeq \partial \Dissp^{*}_{\Vc{x}}[\Polytope^{fr}(\Vc{\tf}')] \subset \Jspace_{\Vc{x}}, & \Vc{\flux}' &=\partial \Dissp^{*}_{\Vc{x}}(\Vc{\tf}').
\end{align}
\end{dfn}
It should be noted that $\Variety^{vl}_{\Vc{x}}(\hat{\Vc{\tf}})$ and $\Variety^{fr}_{\Vc{x}}(\Vc{\flux}')$ are dependent on $\Vc{x}$ via the $\Vc{x}$ dependence of the dissipation functions.
We obtain the intersections in $\Jspace_{\Vc{x}}$ and $\Fspace_{\Vc{x}}$:
\begin{align}
    \Vc{\flux}^{\dagger} &\defeq \Polytope^{vl}(\hat{\Vc{\flux}}) \cap \Variety^{fr}_{\Vc{x}}(\Vc{\flux}'), & 
    \Vc{\tf}^{\dagger} &\defeq \Variety^{vl}_{\Vc{x}}(\hat{\Vc{\tf}}) \cap \Polytope^{fr}(\Vc{\tf}'), \label{eq:intersections_edge}
\end{align}
which are also unique and transversal for each $\Vc{x}\in \X$ because $\Jspace_{\Vc{x}}$ and $\Fspace_{\Vc{x}}$ are whole vector spaces and the Legendre transformations are one-to-one.
Thus, similarly to the case of vertex space, we have the dual foliation:
\begin{lmm}[Dual foliations in edge space\cite{kobayashi2022Phys.Rev.Researcha}]
For each $\Vc{x}\in X$, $(\Polytope^{vl}, \Variety^{fr}_{\Vc{x}})$ and $(\Variety^{vl}_{\Vc{x}}, \Polytope^{fr})$ form dual foliations in $\Jspace_{\Vc{x}}$ and $\Fspace_{\Vc{x}}$ spaces, respectively.
\end{lmm}

For $\hat{\Vc{\flux}}$ and $\Vc{\tf}'$, and their intersections $\Vc{\flux}^{\dagger}$ and $\Vc{\tf}^{\dagger}$ defined in \eqnref{eq:intersections_edge}, $\langle\hat{\Vc{\flux}}-\Vc{\flux}^{\dagger},\Vc{\tf}^{\dagger}-\Vc{\tf}' \rangle=0$ holds.
Thus, we have the generalized Pythagorean relations:
\begin{align}
    \BD^{\Jspace}_{\Vc{x}}[\hat{\Vc{\flux}}\|\Vc{\flux}'] &=\BD^{\Jspace}_{\Vc{x}}[\hat{\Vc{\flux}}\|\Vc{\flux}^{\dagger}]+\BD^{\Jspace}_{\Vc{x}}[\Vc{\flux}^{\dagger}\|\Vc{\flux}'], &
    \BD^{\Fspace}_{\Vc{x}}[\Vc{\tf}'\|\hat{\Vc{\tf}}] &=\BD^{\Fspace}_{\Vc{x}}[\Vc{\tf}'\|\Vc{\tf}^{\dagger}]+\BD^{\Fspace}_{\Vc{x}}[\Vc{\tf}^{\dagger}\|\hat{\Vc{\tf}}].
\end{align}

In contrast to the thermodynamic functions $(\Pfunc, \Pfunc^{*})$, the dissipation functions have symmetry, which makes the origins $\Vc{0}$ in $\Jspace_{\Vc{x}}$ and $\Fspace_{\Vc{x}}$ special and leads to an extension of Helmholtz-Hodge-Kodaira decomposition.
\begin{thm}[Information-geometric extension of Helmholtz-Hodge-Kodaira  (HHK) decomposition\cite{kobayashi2022Phys.Rev.Researcha}]\label{thm:HHK}
For a given flux-force Legendre pair $(\Vc{\flux},\Vc{\tf})\in (\Jspace_{\Vc{x}},\Fspace_{\Vc{x}})$, we have their unique $\Vc{x}$-dependent decompositions:
\begin{align}
\Vc{\flux}&=\Vc{\flux}_{eq}(\Vc{x})+(\Vc{\flux}-\Vc{\flux}_{eq}(\Vc{x})), &\Vc{\tf}&=\Vc{\tf}_{st}(\Vc{x})+(\Vc{\tf}-\Vc{\tf}_{st}(\Vc{x})),
\end{align}
such that $\Vc{\tf}_{eq}(\Vc{x}) \in \Polytope^{fr}(\Vc{0})$, $\Vc{\flux}-\Vc{\flux}_{eq}(\Vc{x}) \in \Polytope^{vl}(\Vc{0})$, $\Vc{\tf}-\Vc{\tf}_{st}(\Vc{x}), \in \Polytope^{fr}(\Vc{0})$, and $\Vc{\flux}_{st}(\Vc{x}) \in \Polytope^{vl}(\Vc{0})$ hold.
In addition, $\Vc{\flux}_{eq}(\Vc{x})$ and $\Vc{\tf}_{st}(\Vc{x})$ are characterized geometrically as
\begin{align}
    \Vc{\flux}_{eq}(\Vc{x}) &\defeq\Polytope^{vl}(\Vc{\flux})\cap \Variety^{fr}_{\Vc{x}}(\Vc{0}), &
    \Vc{\tf}_{st}(\Vc{x})   &\defeq\Variety^{vl}_{\Vc{x}}(\Vc{0})\cap \Polytope^{fr}(\Vc{\tf}). \label{eq:j_eq_f_st_intersection}
\end{align}
Furthermore, $\Vc{\flux}_{eq}$ and $\Vc{\tf}_{st}$ are also characterized variationally as the minimizers of dissipation functions:  
\begin{align}
    \Vc{\flux}_{eq}(\Vc{x}) &= \arg \min_{\Vc{\flux}' \in \Polytope^{vl}(\Vc{\flux})} \Dissp_{\Vc{x}}(\Vc{\flux}'),&
    \Vc{\tf}_{st}(\Vc{x}) &= \arg \min_{\Vc{\tf}'' \in \Polytope^{fr}(\Vc{\tf})} \Dissp_{\Vc{x}}^{*}(\Vc{\tf}''). \label{eq:j_eq_f_st_variational}
\end{align}
\end{thm}
\begin{proof}
The uniqueness of $\Vc{\flux}_{eq}(\Vc{x})$ and $\Vc{\tf}_{st}(\Vc{x})$ as intersections in \eqnref{eq:j_eq_f_st_intersection} follows immediately from the property of the dual foliations. Because, for any $\Vc{\flux}' \in \Polytope^{vl}(\Vc{\flux})$ and $\Vc{\tf}'' \in \Polytope^{fr}(\Vc{\tf})$, $\langle\Vc{\flux}'-\Vc{\flux}_{eq},\Vc{\tf}_{eq} \rangle=0$ and $\langle\Vc{\flux}_{st}, \Vc{\tf}''-\Vc{\tf}_{st} \rangle=0$ hold, the generalized Pythagorean relations lead to
\begin{align}
    \begin{aligned}
    \BD^{\Jspace}_{\Vc{x}}[\Vc{\flux}'\| \Vc{0}] &=\BD^{\Jspace}_{\Vc{x}}[\Vc{\flux}'\|\Vc{\flux}_{eq}]+\BD^{\Jspace}_{\Vc{x}}[\Vc{\flux}_{eq}\|\Vc{0}],&
    \BD^{\Fspace}_{\Vc{x}}[\Vc{\tf}''\| \Vc{0}] &=\BD^{\Fspace}_{\Vc{x}}[\Vc{\tf}''\|\Vc{\tf}_{st}]+\BD^{\Fspace}_{\Vc{x}}[\Vc{\tf}_{st}\|\Vc{0}].
    \end{aligned}\notag
\end{align}
Because $\BD^{\Jspace}_{\Vc{x}}[\Vc{\flux}'\|\Vc{0}]=\Dissp_{\Vc{x}}(\Vc{\flux}')$ and $\BD^{\Fspace}_{\Vc{x}}[\Vc{\tf}''\|\Vc{0}]=\Dissp^{*}_{\Vc{x}}(\Vc{\tf}'')$ hold,  the relations are reduced to
\begin{align}
    \Dissp_{\Vc{x}}(\Vc{\flux}')&=\BD_{\Vc{x}}^{\Jspace}[\Vc{\flux}'\|\Vc{\flux}_{eq}]+\Dissp_{\Vc{x}}(\Vc{\flux}_{eq}),&
    \Dissp^{*}_{\Vc{x}}(\Vc{\tf}'')&=\BD_{\Vc{x}}^{\Fspace}[\Vc{\tf}''\|\Vc{\tf}_{st}]+\Dissp^{*}_{\Vc{x}}(\Vc{\tf}_{st}).
 \label{eq:PGR}
\end{align}
Then \eqnref{eq:j_eq_f_st_variational} follows.
\end{proof}
The decomposed flux $\Vc{\flux}_{eq}$ and force $\Vc{\tf}_{st}$ play a particularly important role in dynamics. 
From the definition, $\Vc{\flux}_{eq}$ is the equilibrium flux, which induces the same instantaneous velocity $\dot{\Vc{x}}$ as $\Vc{\flux}$ does, i.e., $\dot{\Vc{x}}=-\Div_{\HIncMatrix}\Vc{\flux}=-\Div_{\HIncMatrix}\Vc{\flux}_{eq}$. 
Thus, $\Vc{\flux}_{eq}$ is the equilibrium flux mimicking the instantaneous dynamics induced by the nonequilibrium flux $\Vc{\flux}$.
This equilibrium flux is uniquely determined owing to the information-geometric orthogonality of $\Polytope^{vl}(\Vc{\flux})$ and $\Variety^{fr}_{\Vc{x}}(\Vc{0})$.
Moreover, the decomposition $\Vc{\flux}=\Vc{\flux}_{eq}+(\Vc{\flux}-\Vc{\flux}_{eq})$ can be regarded as an information-geometric extension of the Helmholtz-Hodge-Kodaira decomposition in vector calculus and differential form, because $(\Vc{\flux}-\Vc{\flux}_{eq})$ is divergence free, i.e., $\Div_{\HIncMatrix}(\Vc{\flux}-\Vc{\flux}_{eq})=0$, and $\Vc{\tf}_{eq}$ is a curl-free equilibrium force, i.e., $\Vc{\tf}_{eq}\in \Polytope^{fr}(\Vc{0})=\Img[\HIncMatrix^{\Transpose}]=\Ker[\cycMatrix^{\Transpose}]$. 

On the contrary, by definition, $\Vc{\flux}_{st}\in \Polytope^{v}(\Vc{0})$ is the flux that makes the state $\Vc{x}$ a steady state, i.e., $\dot{\Vc{x}}=0$, and is also induced by the force in the same quotient set of force $\Polytope^{fr}(\Vc{\tf})$ as $\Vc{\tf}$.
The decomposition $\Vc{\tf}=\Vc{\tf}_{st}+(\Vc{\tf}-\Vc{\tf}_{st})$ is also a HHK decomposition because $\Vc{\flux}_{st}\in \Polytope^{vl}(\Vc{0})$ is divergence free, i.e., $\Div_{\HIncMatrix}\Vc{\flux}_{st}=0$, and $\Vc{\tf}-\Vc{\tf}_{st}$ is a curl-free equilibrium force, i.e., $\Vc{\tf}-\Vc{\tf}_{st}\in \Polytope^{fr}(\Vc{0})=\Img[\HIncMatrix^{\Transpose}]=\Ker[\cycMatrix^{\Transpose}]$.
\footnote{It should be noted that, in general,  $\Vc{\flux}_{st}\neq \Vc{\flux}-\Vc{\flux}_{eq}$ and $\Vc{\tf}_{eq}\neq\Vc{\tf}-\Vc{\tf}_{st}$ and as a result $\Vc{\flux} \neq \Vc{\flux}_{eq}+\Vc{\flux}_{st}$ and $\Vc{\tf} \neq \Vc{\tf}_{eq}+\Vc{\tf}_{st}$ due to the nonlinearity of Legendre transformation, except when the dissipation functions are quadratic under which the Legendre dual relation is reduced to the linear inner-product relation.}
These decompositions are used in the subsequent sections (\secref{sec:InfoGeoEquilibrium} and \secref{sec:InfoGeoEquilibrium}).

\section{Central affine manifold and Hilbert orthogonality}\label{sec:central_affine}
The dual foliation is an essential geometric object in information geometry.
While less common than the dual foliation, the central affine manifold defined by a convex function also plays an integral role in information geometry\cite{shima2007}.
\begin{dfn}[Central affine manifolds in $\Jspace_{\Vc{x}}$ and $\Fspace_{\Vc{x}}$]\label{dfn:cAffineManifolds}
The central affine manifolds in $\Jspace_{\Vc{x}}$ and $\Fspace_{\Vc{x}}$ are defined as the level sets of $\Dissp_{\Vc{x}}(\Vc{\flux})$ and $\Dissp^{*}_{\Vc{x}}(\Vc{\tf})$, respectively\footnote{While we introduce the central affine manifolds only on the edge space, they can be defined on the vertex space as well. The central affine manifolds on the vertex space of CRN become fundamental when we work on the isobaric processes in which the volume changes in conjunction with the reactions\cite{sughiyama2022Phys.Rev.Researchb}. In this case, the volume is a global variable affecting all the reactions simultaneously.}:
\begin{align}
    \cManifold_{\Vc{x}}^{\Dissp}(c)&\defeq \{\Vc{\flux}|\Dissp_{\Vc{x}}(\Vc{\flux})=c\} \subset \Jspace_{\Vc{x}}, & 
    \cManifold_{\Vc{x}}^{\Dissp^{*}}(c)&\defeq \{\Vc{\tf}|\Dissp^{*}_{\Vc{x}}(\Vc{\tf})=c\}\subset \Fspace_{\Vc{x}}, 
\end{align}
where $c\in \Real_{\ge 0}$.
For a given $\Vc{\flux}'\in \Jspace_{\Vc{x}}$ or $\Vc{\tf}'\in \Fspace_{\Vc{x}}$, the manifolds are also denoted as
\begin{align}
    \cManifold_{\Vc{x}}^{\Dissp}(\Vc{\flux}')&\defeq \{\Vc{\flux}|\Dissp_{\Vc{x}}(\Vc{\flux})=\Dissp_{\Vc{x}}(\Vc{\flux}')\}, & 
    \cManifold_{\Vc{x}}^{\Dissp^{*}}(\Vc{\tf}')&\defeq \{\Vc{\tf}|\Dissp^{*}_{\Vc{x}}(\Vc{\tf})=\Dissp^{*}_{\Vc{x}}(\Vc{\tf}')\}.
\end{align}
Their Legendre transformations are also called (dual) central affine manifolds:
\begin{align}
    \Manifold_{\Vc{x}}^{\Dissp}(c)&\defeq \partial \Dissp_{\Vc{x}}[\cManifold_{\Vc{x}}^{\Dissp}(c)]\subset \Fspace_{\Vc{x}}, & 
    \Manifold_{\Vc{x}}^{\Dissp^{*}}(c)&\defeq \partial \Dissp^{*}_{\Vc{x}}[\cManifold_{\Vc{x}}^{\Dissp^{*}}(c)]\subset \Jspace_{\Vc{x}}. \label{eq:dual_central_affine_manifolds}
\end{align}
\end{dfn}
\subsection{Pseudo-Hilbert-isosceles orthogonality and decomposition}
By employing the central affine manifold, we can introduce another type of generalized orthogonality:
\begin{dfn}[Pseudo-Hilbert-isosceles orthogonality\cite{kaiser2018JStatPhys,renger2021DiscreteContin.Dyn.Syst.-S,patterson2021ArXiv210314384Math-Ph}]\label{dfn:PHorthogonality}
Pseudo-Hilbert-isosceles orthogonality between $\Vc{\flux}_{S}, \Vc{\flux}_{A} \in \Jspace_{\Vc{x}}$ and between $\Vc{\tf}_{S}, \Vc{\tf}_{A} \in \Fspace_{\Vc{x}}$ are defined as follows:
\begin{align}
    \Vc{\flux}_{S}\perp_{H} \Vc{\flux}_{A} &\Longleftrightarrow \Dissp_{\Vc{x}}(\Vc{\flux}_{S}+\Vc{\flux}_{A})=\Dissp_{\Vc{x}}(\Vc{\flux}_{S}-\Vc{\flux}_{A})\\
    \Vc{\tf}_{S}\perp_{H} \Vc{\tf}_{A} &\Longleftrightarrow \Dissp_{\Vc{x}}^{*}(\Vc{\tf}_{S}+\Vc{\tf}_{A})=\Dissp_{\Vc{x}}^{*}(\Vc{\tf}_{S}-\Vc{\tf}_{A}).
\end{align}

\end{dfn}

This orthogonality is motivated by the relation $\|\Vc{\flux}_{S}+\Vc{\flux}_{A}\|^{2}=\|\Vc{\flux}_{S}-\Vc{\flux}_{A}\|^{2}$ satisfied by an orthogonal pair $\Vc{\flux}_{S}\perp \Vc{\flux}_{A}$ under a usual inner product structure and its induced norm $\|\cdot\|^{2}$\footnote{If $\|\cdot\|^{2}$ is a squared norm that is not necessarily induced from an inner product, the orthogonality is called isosceles or James orthogonality\cite{davis1959Math.Mag.}. Here, $\|\cdot\|^{2}$ is further replaced with the dissipation function, which does not satisfy some conditions required to be a norm. }.
By employing this orthogonality, we obtain pseudo-Hilbert isosceles decompositions of  $\Vc{\flux}$ and $\Vc{\tf}$ as follows:
\begin{lmm}[Positive decompositions of the bilinear pairing via pseudo-Hilbert-isosceles orthogonality\cite{kaiser2018JStatPhys,renger2021DiscreteContin.Dyn.Syst.-S,patterson2021ArXiv210314384Math-Ph,kobayashi2022Phys.Rev.Researcha}]
For a given $\Vc{\flux}\in \Jspace_{\Vc{x}}$ and any $\Vc{\flux}'$ on the same central affine manifold as $\Vc{\flux}$, i.e., $\Vc{\flux}' \in \cManifold_{\Vc{x}}^{\Dissp}(\Vc{\flux})$, we obtain the pseudo-Hilbert-isosceles orthogonal decomposition $\Vc{\flux}=\Vc{\flux}_{S}+\Vc{\flux}_{A}$:
\begin{align}
    \Vc{\flux}_{S}&\defeq \frac{1}{2}(\Vc{\flux}+\Vc{\flux}'), & \Vc{\flux}_{A}&\defeq \frac{1}{2}(\Vc{\flux}-\Vc{\flux}'),
\end{align}
where $\Vc{\flux}_{S} \perp_{H}\Vc{\flux}_{A}$ and $\Vc{\flux}'=\Vc{\flux}_{S}-\Vc{\flux}_{A}$ hold.
In addition, this decomposition induces a positive decomposition of the bilinear product $\langle\Vc{\flux},\Vc{\tf}\rangle=\langle\Vc{\flux}_{S},\Vc{\tf}\rangle+\langle\Vc{\flux}_{A},\Vc{\tf}\rangle$ where 
\begin{align}
    \langle\Vc{\flux}_{S}, \Vc{\tf}\rangle&=\frac{1}{2}\BD^{\Jspace,\Fspace}_{\Vc{x}}[\Vc{\flux}'; -\Vc{\tf}]\ge 0, &
    \langle\Vc{\flux}_{A}, \Vc{\tf}\rangle&=\frac{1}{2}\BD^{\Jspace,\Fspace}_{\Vc{x}}[\Vc{\flux}'; \Vc{\tf}]\ge 0,  
\end{align}
hold. Similarly, for $\Vc{\tf} \in \Fspace_{\Vc{x}}$ and $\Vc{\tf}''\in \cManifold_{\Vc{x}}^{\Dissp^{*}}(\Vc{\tf})$, a positive orthogonal decomposition $\Vc{\tf}=\Vc{\tf}_{S}+\Vc{\tf}_{A}$ is obtained by $\Vc{\tf}_{S}\defeq \frac{1}{2}(\Vc{\tf}+\Vc{\tf}'')$ and $\Vc{\tf}_{A}\defeq \frac{1}{2}(\Vc{\tf}-\Vc{\tf}'')$, which satisfy the associated relations:
\begin{align}
    \langle\Vc{\flux}, \Vc{\tf}_{A}\rangle&=\frac{1}{2}\BD^{\Jspace,\Fspace}_{\Vc{x}}[\Vc{\flux}; \Vc{\tf}'']\ge 0, &
    \langle\Vc{\flux}, \Vc{\tf}_{S}\rangle &=\frac{1}{2}\BD^{\Jspace,\Fspace}_{\Vc{x}}[\Vc{\flux}; -\Vc{\tf}'']\ge 0. \label{eq:pHilbert2}
\end{align}
\end{lmm}
These decompositions were introduced in \cite{kaiser2018JStatPhys} for rMJP and extended to CRN in \cite{renger2021DiscreteContin.Dyn.Syst.-S,patterson2021ArXiv210314384Math-Ph}, whereas we pointed out its information-geometric aspect in \cite{kobayashi2022Phys.Rev.Researcha}.
The decomposition plays a role of characterizing the gradient-flow-like property of non-gradient flows.


\section{Information-geometric Properties of Equilibrium Flow}\label{sec:InfoGeoEquilibrium}
In this section, we describe several properties of the equilibrium flow (\eqnref{eq:gGF}) from the viewpoint of information geometry by employing the objects introduced in the previous sections. 
Such properties include the existence and uniqueness of the steady state (static property), convergence to the state (kinetic property), and the balance between information-geometric quantities associated with the steady state and convergence along the trajectory (the connection between static and kinetic properties).
These properties are consistent with those that thermodynamic equilibrium systems should have.
In addition, several results are extensions of the results obtained for FPE in the context of functional analysis, partial differential equations, and optimal transport.

\subsection{Properties of equilibrium flow}
The following property of the equilibrium state characterizes the static aspect of the equilibrium flow and is fundamentally ascribed to the dually flat structure of density and potential spaces\footnote{Actually, this result is independent of the detail of the dissipation functions.}:
\begin{prop}[Equilibrium state and its geometric and variational characterizations]\label{prop:Eqstate}
The steady state of the equilibrium flow $\Vc{x}_{t}$ (\eqnref{eq:gGF}) starting from $\Vc{x}(0)=\Vc{x}_{0}$ is called the equilibrium state $\Vc{x}_{eq}$.  For each $\Vc{x}_{0}$, the equilibrium state is identical to the intersection $\Vc{x}^{\dagger}=\Polytope^{sc}(\Vc{x}_{0}) \cap \Variety^{eq}(\tilde{\Vc{x}})$, i.e., $\Vc{x}_{eq}=\Vc{x}^{\dagger}$, and thus uniquely exists for a given pair of the initial state $\Vc{x}_{0}$ and the parameter of equilibrium flow $\tilde{\Vc{x}}$. The equilibrium state $\Vc{x}_{eq}$ is also characterized variationally as 
\begin{align}
    \Vc{x}_{eq}&= \arg \min_{\Vc{x}\in \Polytope^{sc}(\Vc{x}_{0})}\BD^{\X}_{\Pfunc}[\Vc{x}\|\tilde{\Vc{x}}]
    =\arg \min_{\Vc{x}_{q}\in \Variety^{eq}(\tilde{\Vc{x}})}\BD^{\X}_{\Pfunc}[\Vc{x}_{0}\|\Vc{x}_{q}]\label{eq:PGTx}. 
\end{align}
Moreover, $\Variety^{eq}(\tilde{\Vc{x}})=\Variety^{\mathrm{DB}}$ holds.
\end{prop}
\begin{proof}
From \propref{prop:gradflow}, $\Vc{x}_{eq} \in \Manifold^{\mathrm{DB}}$, from which $\Vc{\tf}(\Vc{x}_{eq})=0$ follows.
For the equilibrium force (\eqnref{eq:gradEq}), $ \Manifold^{\mathrm{DB}}=\{\Vc{x}|\Vc{\tf}(\Vc{x})=0\}=\{\Vc{x}|\HIncMatrix^{\Transpose}(\partial\Pfunc[\Vc{x}]-\partial\Pfunc[\tilde{\Vc{x}}])=0\}=\Manifold^{eq}(\tilde{\Vc{x}})$ holds. Thus, $\Vc{x}_{eq} \in \Manifold^{eq}(\tilde{\Vc{x}})$. Because the initial state is $\Vc{x}_{0}$, $\Vc{x}_{eq}\in \Polytope^{sc}(\Vc{x}_{0})$. 
Thus, $\Vc{x}_{eq} = \Vc{x}^{\dagger} \in \Polytope^{sc}(\Vc{x}_{0}) \cap \Manifold^{eq}(\tilde{\Vc{x}})$.
The first equality of \eqnref{eq:PGTx} is obvious from the proof of the dual foliation (\lmmref{lmm:DualFoliationVertex}). The second equality is from the generalized Pythagorean relation (\eqnref{eq:gPR_X}).
\end{proof}

The second property of the equilibrium flow is kinetic in nature and characterizes the Bregman divergence as the generalized driving potential, which ensures the convergence of $\Vc{x}_{t}$ to the equilibrium state.
This property is attributed to the dually flat structure on the edge spaces.
\begin{prop}[Bregman divergence and Gibbs' H-Theorem]
For the trajectory of the equilibrium flow $\Vc{x}_{t}$ (\eqnref{eq:gGF}) starting from $\Vc{x}(0)=\Vc{x}_{0}$, the thermodynamic function $\BD^{\X}_{\Pfunc}[\Vc{x}_{t}\|\tilde{\Vc{x}}]$ decreases, that is, $\dd \BD^{\X}_{\Pfunc}[\Vc{x}_{t}\|\tilde{\Vc{x}}]/\dt < 0$ except at $\Vc{x}_{t} \in \Manifold^{eq}(\tilde{\Vc{x}})$ where $\dd \BD^{\X}_{\Pfunc}[\Vc{x}_{t}\|\tilde{\Vc{x}}]/\dt=0$ holds. 
Thus, the equilibrium state $\Vc{x}_{eq} \in  \Polytope^{sc}(\Vc{x}_{0}) \cap \Manifold^{eq}(\tilde{\Vc{x}})$ is locally and asymptotically stable.
\end{prop}
\begin{proof}
By replacing $\FEnergy(\Vc{x})$ in \propref{prop:gradflow} with $\BD^{\X}_{\Pfunc}[\Vc{x}\|\tilde{\Vc{x}}]$, 
we obtain $\dd \BD^{\X}_{\Pfunc}[\Vc{x}_{t}\|\tilde{\Vc{x}}]/\dt=-\left[\Dissp^{*}_{\Vc{x}_{t}}(\Vc{\tf}(\Vc{x}_{t}))+ \Dissp_{\Vc{x}_{t}}(\Vc{\flux}(\Vc{x}_{t}))\right] \le 0$ and the equality holds if and only if $\Vc{\tf}(\Vc{x}_{t})=0 (\Leftrightarrow \Vc{x}_{t} \in \Manifold^{eq}(\tilde{\Vc{x}})$. 
\end{proof}
Because $\BD^{\X}_{\Pfunc}[\Vc{x}\|\tilde{\Vc{x}}]$ can be identified with the difference of total entropy between $\Vc{x}$ and $\tilde{\Vc{x}}$ for thermodynamic systems such as CRN\cite{sughiyama2021ArXiv211212403Cond-MatPhysicsphysics}, $\dd \BD^{\X}_{\Pfunc}[\Vc{x}_{t}\|\tilde{\Vc{x}}_{eq}]/\dt \le 0$ corresponds to the nondecreasing property of thermodynamic entropy, which is also referred as Gibbs' H-theorem\footnote{We have multiple types of H-theorems. The most famous one is Boltzmann's H theorem, in which the H function is derived from the microscopic dynamics of a system. In Gibb's H-theorem, H function is obtained by coarse-graining the microscopic system\cite{gibbs2010}}.

The third property provides a connection between the thermodynamic function and the dissipation function, which is immediately obtained from the De Giorgi's formulation of the generalized gradient flow (\eqnref{eq:DeGiorgi}):
\begin{prop}[Balancing of thermodynamic function and dissipation function]
For the trajectory of the equilibrium flow $\Vc{x}_{t}$ (\eqnref{eq:gGF}) starting from $\Vc{x}(0)=\Vc{x}_{0}$, the following relation holds for the thermodynamic function and the dissipation function:
\begin{align}
    \BD^{\X}_{\Pfunc}[\Vc{x}_{0}\|\tilde{\Vc{x}}]-\BD^{\X}_{\Pfunc}[\Vc{x}_{t}\|\tilde{\Vc{x}}]= \int_{t'=0}^{t}\left[\Dissp^{*}_{\Vc{x}_{t'}}(\Vc{\tf}(\Vc{x}_{t'}))+ \Dissp_{\Vc{x}_{t'}}(\Vc{\flux}(\Vc{x}_{t'}))\right]\dt' = \int_{t'=0}^{t}\EPR_{t'}\dt', \label{eq:H-theorem}
\end{align}
\end{prop}
In physics and chemistry, this relation means that the difference in the thermodynamic (potential) function between $\Vc{x}_{t}$ and $\Vc{x}_{0}$ (the left-hand side), i.e., the change in total entropy, is equal to the integral of dissipation along $\Vc{x}_{t}$(the right-hand side), i.e., the entropy production, for equilibrium systems.

All these results indicate that the equilibrium flow and its properties mathematically abstract the properties of physical equilibrium systems.
The equilibrium state $\Vc{x}_{eq}$ is characterized algebraically by the unique intersection of $\Polytope^{sc}(\Vc{x}_{0})$ and $\Variety^{eq}(\tilde{\Vc{x}})$ and also variationally by \eqnref{eq:PGTx}.
The convergence to $\Vc{x}_{eq}$ is guaranteed by $\dd \BD^{\X}_{\Pfunc}[\Vc{x}_{t}\|\tilde{\Vc{x}}]/\dt \le 0$.  
Furthermore, the entropy-dissipation balance relation (\eqnref{eq:H-theorem}) itself defines the equilibrium system abstractly as the De Giorgi's formulation (\eqnref{eq:DeGiorgi}) does.

\subsection{Induced dually flat structure on tangent-cotangent spaces}
The equilibrium state is characterized geometrically and variationally via the information-geometric structure on the vertex spaces ($\X$, $\Y$)  as in \propref{prop:Eqstate}. 
Similarly, the flux (kinetic law) of equilibrium systems (gradient systems) can be obtained variationally as the flux minimizing the dissipation function under the restriction of the continuity equation. 
\begin{figure}[h]
\includegraphics[bb=0 0 1024 468, width=\linewidth]{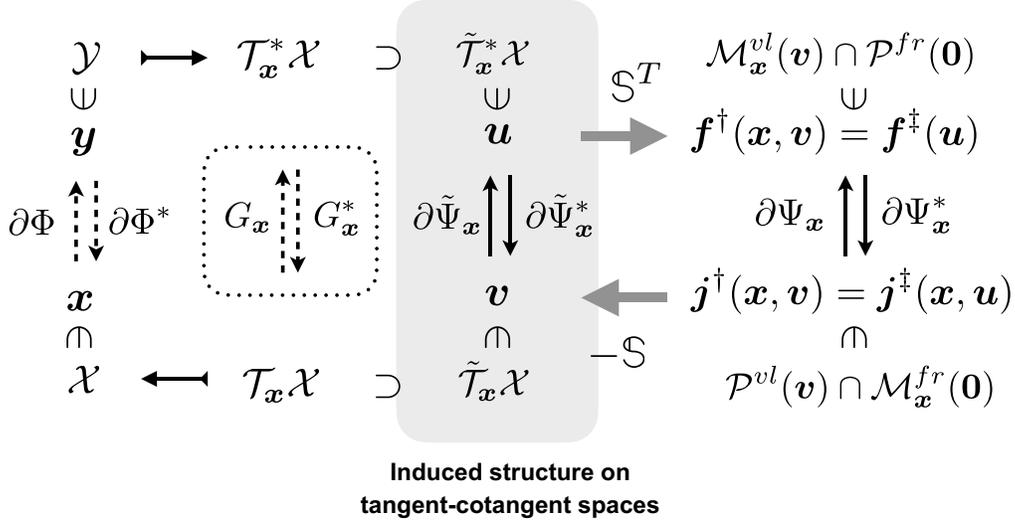}
\caption{\TJK{The induced dually flat structure on the restricted tangent and cotangent spaces from the dissipation functions on the edge spaces (gray region). The relationship is compared with the Riemannian metric on tangent and cotangent spaces via Fisher information matrices induced by the thermodynamic functions (dotted box). }}
\label{fig:Induced_Duality}
\end{figure}

\begin{lmm}[Equilibrium force as the minimizer of primal dissipation function]
For a given trajectory $\{\Vc{x}_{t}\}$, we define the trajectory of the flux $\{\Vc{\flux}^{\dagger}_{t}\}$ minimizing the primal dissipation:
\begin{align}
    \{\Vc{\flux}^{\dagger}_{t}\}\defeq \arg \min_{\{\Vc{\flux}_{t}\}}\int_{0}^{t}\Dissp_{\Vc{x}_{t'}}[\Vc{\flux}_{t'}]\dd t',\quad \mbox{s.t.   $\dot{\Vc{x}}_{t'}+\Div_{\HIncMatrix} \Vc{\flux}_{t'}=0$ for all $t'\in[0,t]$}. \label{eq:min_Dissip_J}
\end{align}
Then, $\Vc{\flux}^{\dagger}_{t}$ is generated by the equilibrium force, $\Vc{\tf}_{t}^{\dagger}=\partial \Dissp_{\Vc{x}}[\Vc{\flux}^{\dagger}_{t}] \in \Polytope^{fr}(\Vc{0})$. Thus, the minimum primal dissipation flux that generates the given $\{\Vc{x}_{t}\}$ is the equilibrium flux.
\end{lmm}
\begin{proof}
Because the minimization of \eqnref{eq:min_Dissip_J} can be conducted pointwise-manner for each $t'\in[0,t]$ and $\dot{\Vc{x}}_{t'}+\Div_{\HIncMatrix} \Vc{\flux}_{t'}=0\Longleftrightarrow \Vc{\flux}_{t'}\in\Polytope^{vl}(\dot{\Vc{x}}_{t'})$, we have 
\begin{align}
    \Vc{\flux}_{t'}^{\dagger}=\Vc{\flux}^{\dagger}(\Vc{x}_{t'},\dot{\Vc{x}}_{t'})= \arg \min_{\Vc{\flux}\in\Polytope^{vl}(\dot{\Vc{x}}_{t'}) }\Dissp_{\Vc{x}_{t'}}[\Vc{\flux}]=\Polytope^{vl}(\dot{\Vc{x}}_{t'}) \cap \Variety^{fr}_{\Vc{x}_{t'}}(\Vc{0}), \label{eq:jdagger}
\end{align}
where we used \eqnref{eq:j_eq_f_st_intersection} and \eqnref{eq:j_eq_f_st_variational}.
Thus, from $\Vc{\flux}_{t'}^{\dagger} \in \Variety^{fr}_{\Vc{x}_{t'}}(\Vc{0}) \Longleftrightarrow \Vc{\tf}_{t'}^{\dagger} \in \Polytope^{fr}(\Vc{0})$, the minimum dissipation flux $\{\Vc{\flux}^{\dagger}_{t}\}$ is generated by the equilibrium force, $\{\Vc{\tf}_{t}^{\dagger}\}=\{\Vc{\tf}^{\dagger}(\Vc{x}_{t'},\dot{\Vc{x}}_{t'})\}\in \Polytope^{fr}(\Vc{0})$ where $\Vc{\tf}^{\dagger}(\Vc{x},\dot{\Vc{x}})\defeq \partial \Dissp_{\Vc{x}}[\Vc{\flux}^{\dagger}(\Vc{x},\dot{\Vc{x}})]$. 
\end{proof}

By exploiting this unique pairing between $\dot{\Vc{x}}_{t'}$ and $\Vc{\flux}_{t'}^{\dagger}$ or $\dot{\Vc{x}}_{t'}$ and $\Vc{\tf}_{t'}^{\dagger}$, we can obtain an induced dually flat structure on the restricted tangent and cotangent spaces of $\X$ and $\Y$ (\fgref{fig:Induced_Duality}), which can be regarded as an information-geometric extension of the Otto structure.
\begin{thm}[Induced dually flat structure on tangent and cotangent spaces]\label{thm:inducedHessiangeometry1}
Let $\tilde{\Tan}_{\Vc{x}} \X\defeq \Img \HIncMatrix\cong \Polytope^{sc}(\Vc{0})\subset \Tan_{\Vc{x}}\X$ and $\tilde{\Tan}_{\Vc{x}}^{*} \X\defeq \Tan_{\Vc{x}}^{*}\X/\Ker \HIncMatrix^{\Transpose}$ be tangent and cotangent spaces on $\X$ restricted by $\HIncMatrix$. 
On $\tilde{\Tan}_{\Vc{x}} \X$ and $\tilde{\Tan}_{\Vc{x}}^{*} \X$, we have the Legendre conjugate dissipation functions $\tilde{\Dissp}_{\Vc{x}}: \tilde{\Tan}_{\Vc{x}} \X \to \Real$ and $\tilde{\Dissp}_{\Vc{x}}^{*}: \tilde{\Tan}_{\Vc{x}}^{*} \X \to \Real$ induced by the dissipation functions on the edge spaces (\fgref{fig:Induced_Duality}).
\end{thm}
\begin{proof}
By employing \eqnref{eq:jdagger}, for each $\Vc{v}\in  \tilde{\Tan}_{\Vc{x}}\X$, we can uniquely determine $\Vc{\flux}^{\dagger}(\Vc{x},\Vc{v})$, $\Vc{\tf}^{\dagger}(\Vc{x},\Vc{v})\in \Polytope^{fr}(\Vc{0})$, and $\Vc{u}^{\dagger}(\Vc{x},\Vc{v}) \in \tilde{\Tan}_{\Vc{x}}^{*}\X$\footnote{Because $\Vc{\tf}^{\dagger}(\Vc{x},\Vc{v}) \in \Polytope^{fr}(\Vc{0})=\Img[\HIncMatrix^{\Transpose}]$ and $\tilde{\Tan}_{\Vc{x}}^{*} \X\defeq \Y/\Ker \HIncMatrix^{\Transpose}$, $\Vc{u}^{\dagger}(\Vc{x},\Vc{v})$ is uniquely determined.}. 
They satisfy
\begin{align}
    \Vc{v} & = - \HIncMatrix \Vc{\flux}^{\dagger}(\Vc{x},\Vc{v}), & \Vc{\flux}^{\dagger}(\Vc{x},\Vc{v}) & = \partial \Dissp^{*}_{\Vc{x}}[\Vc{\tf}^{\dagger}(\Vc{x},\Vc{v})], & \Vc{\tf}^{\dagger}(\Vc{x},\Vc{v}) & = - \HIncMatrix^{\Transpose} \Vc{u}^{\dagger}(\Vc{x},\Vc{v}). \label{eq:v_pairing}
\end{align}
Conversely, for a given $\Vc{u}\in \tilde{\Tan}_{\Vc{x}}^{*}\X$, we have $\Vc{\tf}^{\ddagger}(\Vc{u})$, $\Vc{\flux}^{\ddagger}(\Vc{x},\Vc{u})$, and $\Vc{v}^{\ddagger}(\Vc{x},\Vc{u})$ as follows:
\begin{align}
    \Vc{v}^{\ddagger}(\Vc{x},\Vc{u}) & = - \HIncMatrix \Vc{\flux}^{\ddagger}(\Vc{x},\Vc{u}), & \Vc{\flux}^{\ddagger}(\Vc{x},\Vc{u}) & = \partial \Dissp^{*}_{\Vc{x}}[\Vc{\tf}^{\ddagger}(\Vc{u})], & \Vc{\tf}^{\ddagger}(\Vc{u}) & = - \HIncMatrix^{\Transpose} \Vc{u}.\label{eq:u_pairing}
\end{align}
Thus, for a pair of $(\Vc{v},\Vc{u})_{\Vc{x}}$ satisfying $\Vc{u}=\Vc{u}^{\dagger}(\Vc{x},\Vc{v})$, we have $\Vc{v}=\Vc{v}^{\ddagger}(\Vc{x},\Vc{u})$, $\Vc{\flux}^{\dagger}(\Vc{x},\Vc{v})=\Vc{\flux}^{\ddagger}(\Vc{x},\Vc{u})$, and $\Vc{\tf}^{\dagger}(\Vc{x},\Vc{v})=\Vc{\tf}^{\ddagger}(\Vc{x},\Vc{u})$.
This pairing establishes a bijection between $\tilde{\Tan}_{\Vc{x}}\X$ and $\tilde{\Tan}_{\Vc{x}}^{*}\X$.
Moreover, this bijection is realized by the Legendre transformations of the following induced dissipation functions on $ \tilde{\Tan}_{\Vc{x}}\X$ and $\tilde{\Tan}_{\Vc{x}}^{*}\X$:
\begin{align}
\tilde{\Dissp}_{\Vc{x}}(\Vc{v}) & \defeq \Dissp_{\Vc{x}}(\Vc{\flux}^{\dagger}(\Vc{x},\Vc{v})), & \tilde{\Dissp}_{\Vc{x}}^{*}(\Vc{u}) & \defeq \Dissp_{\Vc{x}}^{*}(\Vc{\tf}^{\ddagger}(\Vc{u})). \label{eq:induced_dissp_func_tan}
\end{align}
These functions are Legendre conjugate as follows:
\begin{align*}
 \max_{\Vc{u}'\in\tilde{\Tan}^{*}_{\Vc{x}}\X}&\left[\langle\Vc{v},\Vc{u}'\rangle-\tilde{\Dissp}_{\Vc{x}}^{*}(\Vc{u}') \right]=\max_{\substack{\Vc{u}'\in\tilde{\Tan}^{*}_{\Vc{x}}\X\\ \Vc{\flux}\in \Polytope^{vl}(\Vc{v})}}\left[\langle-\HIncMatrix\Vc{\flux},\Vc{u}'\rangle-\tilde{\Dissp}_{\Vc{x}}^{*}(\Vc{u}') \right]=\max_{\substack{\Vc{u}'\in\tilde{\Tan}^{*}_{\Vc{x}}\X\\ \Vc{\flux}\in \Polytope^{vl}(\Vc{v})}}\left[\langle\Vc{\flux},-\HIncMatrix^{\Transpose}\Vc{u}'\rangle-\Dissp_{\Vc{x}}^{*}(-\HIncMatrix^{\Transpose}\Vc{u}') \right]\\
&=\max_{\substack{\Vc{\tf}'\in \Polytope^{fr}(\Vc{0})\\ \Vc{\flux}\in \Polytope^{vl}(\Vc{v})}}\left[\langle\Vc{\flux},\Vc{\tf}'\rangle-\Dissp_{\Vc{x}}^{*}(\Vc{\tf}') \right]
=\max_{\substack{\Vc{\tf}'\in \Polytope^{fr}(\Vc{0})\\ \Vc{\flux}\in \Polytope^{vl}(\Vc{v})}}\left[\langle\Vc{\flux}^{\dagger}(\Vc{x},\Vc{v}),\Vc{\tf}'\rangle-\Dissp_{\Vc{x}}^{*}(\Vc{\tf}')  + \langle(\Vc{\flux}-\Vc{\flux}^{\dagger}(\Vc{x},\Vc{v})),\Vc{\tf}'\rangle\right]\\
&=\max_{\substack{\Vc{\tf}'\in \Polytope^{fr}(\Vc{0})}}\left[\langle\Vc{\flux}^{\dagger}(\Vc{x},\Vc{v}),\Vc{\tf}'\rangle-\Dissp_{\Vc{x}}^{*}(\Vc{\tf}')  \right]=\Dissp_{\Vc{x}}(\Vc{\flux}^{\dagger}(\Vc{x},\Vc{v}))=\tilde{\Dissp}_{\Vc{x}}(\Vc{v}),
\end{align*}
where we used $\langle\Vc{\flux}-\Vc{\flux}^{\dagger}(\Vc{x},\Vc{v}),\Vc{\tf}'\rangle=0$ because $\Vc{\tf}' \in \Polytope^{fr}(\Vc{0})=\Img \HIncMatrix^{\Transpose}$ and $(\Vc{\flux}-\Vc{\flux}^{\dagger}(\Vc{x},\Vc{v})) \in \Ker \HIncMatrix$.
The inverse is also shown:
\begin{align*}
 \max_{\Vc{v}'\in\tilde{\Tan}_{\Vc{x}}\X}&\left[\langle\Vc{v}',\Vc{u}\rangle-\tilde{\Dissp}_{\Vc{x}}(\Vc{v}') \right]
 =\max_{\substack{\Vc{v}'\in\tilde{\Tan}_{\Vc{x}}\X}}\left[\langle-\HIncMatrix\Vc{\flux}^{\dagger}(\Vc{x},\Vc{v}'),\Vc{u}\rangle-\Dissp_{\Vc{x}}(\Vc{\flux}^{\dagger}(\Vc{x},\Vc{v}')) \right]\\
 &=\max_{\substack{\Vc{\flux}^{\dagger}\in  \Variety^{fr}_{\Vc{x}}(\Vc{0})}}\left[\langle\Vc{\flux}^{\dagger},-\HIncMatrix^{\Transpose}\Vc{u}\rangle-\Dissp_{\Vc{x}}(\Vc{\flux}^{\dagger}) \right]
 =\max_{ \Vc{\flux}^{\dagger}\in  \Variety^{fr}_{\Vc{x}}(\Vc{0})}\left[\langle\Vc{\flux}^{\dagger},\Vc{\tf}^{\ddagger}(\Vc{u})\rangle-\Dissp_{\Vc{x}}(\Vc{\flux}^{\dagger}) \right]\\
 &=\Dissp_{\Vc{x}}^{*}(\Vc{\tf}^{\ddagger}(\Vc{u}))=\tilde{\Dissp}_{\Vc{x}}^{*}(\Vc{u}),
\end{align*}
where we used the fact that $\{\Vc{\flux}^{\dagger}(\Vc{x},\Vc{v}')\}_{\Vc{v}'\in\tilde{\Tan}_{\Vc{x}}\X}=\Variety^{fr}_{\Vc{x}}(\Vc{0})$ in the second line.
The pair $(\Vc{v}$, $\Vc{u})_{\Vc{x}}$ are Legendre dual of these functions:
\begin{align}
    \partial_{\Vc{v}} \tilde{\Dissp}_{\Vc{x}}(\Vc{v})&= \left[\frac{\partial \Vc{\flux}^{\dagger}(\Vc{x},\Vc{v})}{\partial \Vc{v}}\right]^{\Transpose}\left.\frac{\partial \Dissp_{\Vc{x}}(\Vc{\flux})}{\partial \Vc{\flux}}\right|_{\Vc{\flux}=\Vc{\flux}^{\dagger}(\Vc{x},\Vc{v})}
    =\left[\frac{\partial \Vc{\flux}^{\dagger}(\Vc{x},\Vc{v})}{\partial \Vc{v}}\right]^{\Transpose}\Vc{\tf}^{\dagger}(\Vc{x},\Vc{v})\\
    &=-\left[\frac{\partial \Vc{\flux}^{\dagger}(\Vc{x},\Vc{v})}{\partial \Vc{v}}\right]^{\Transpose}\HIncMatrix^{\Transpose}\Vc{u}^{\dagger}(\Vc{x},\Vc{v})=\Vc{u},\\
    \partial_{\Vc{u}} \tilde{\Dissp}^{*}_{\Vc{x}}(\Vc{u}) &= \left[\frac{\partial \Vc{\tf}^{\ddagger}(\Vc{u})}{\partial \Vc{u}}\right]^{\Transpose}\left.\frac{\partial \Dissp^{*}_{\Vc{x}}(\Vc{\tf})}{\partial \Vc{\tf}}\right|_{\Vc{\tf}=\Vc{\tf}^{\ddagger}(\Vc{u})}=\left[\frac{\partial \Vc{\tf}^{\ddagger}(\Vc{u})}{\partial \Vc{u}}\right]^{\Transpose}\Vc{\flux}^{\ddagger}(\Vc{x},\Vc{u})\\
    &=-\HIncMatrix\Vc{\flux}^{\ddagger}(\Vc{x},\Vc{u})=\Vc{v},
\end{align}
where we used $\left[\frac{\partial \Vc{\flux}^{\dagger}(\Vc{x},\Vc{v})}{\partial \Vc{v}}\right]^{\Transpose}\HIncMatrix^{\Transpose}=-\identityM$ from $\frac{\partial}{\partial \Vc{v}}[\Vc{v}+\HIncMatrix \Vc{\flux}^{\dagger}(\Vc{x},\Vc{v})]=\identityM + \HIncMatrix \frac{\partial \Vc{\flux}^{\dagger}(\Vc{x},\Vc{v})}{\partial \Vc{v}}=0$ and $\frac{\partial \Vc{\tf}^{\ddagger}(\Vc{u})}{\partial \Vc{u}} =-\HIncMatrix^{\Transpose}$.
They are dissipation functions; strict convexity and $1$-coercivity follow from those of the original dissipation functions.
Also, we have
\begin{align}
\mbox{Symmetry}& & \tilde{\Dissp}_{\Vc{x}}(-\Vc{v}) & = \Dissp_{\Vc{x}}(\Vc{\flux}^{\dagger}(\Vc{x},-\Vc{v}))=\Dissp_{\Vc{x}}(-\Vc{\flux}^{\dagger}(\Vc{x},\Vc{v}))=\tilde{\Dissp}_{\Vc{x}}(\Vc{v})\\
\mbox{Bounded by $0$ at $\Vc{0}$}& &     \tilde{\Dissp}_{\Vc{x}}(\Vc{v}=\Vc{0})&=\Dissp_{\Vc{x}}(\Vc{\flux}^{\dagger}(\Vc{x},\Vc{0}))=\Dissp_{\Vc{x}}(\Vc{0})=0.
\end{align}
\end{proof}
Using the induced dissipation functions, we define the Bregman divergence on $(\tilde{\Tan}_{\Vc{x}}\X, \tilde{\Tan}^{*}_{\Vc{x}}\X)$, which is associated with the Bregman divergence on $(\Jspace_{\Vc{x}},\Fspace_{\Vc{x}})$:
\begin{align}
    \BD^{\X,\Y}_{\tilde{\Dissp}_{\Vc{x}}}[\Vc{v}\|\Vc{u}']&\defeq \tilde{\Dissp}_{\Vc{x}}(\Vc{v}) + \tilde{\Dissp}_{\Vc{x}}^{*}(\Vc{u}')-\langle \Vc{v}, \Vc{u}'\rangle\\
    &=\Dissp_{\Vc{x}}(\Vc{\flux}^{\dagger}) + \Dissp_{\Vc{x}}^{*}(\Vc{\tf}'^{\ddagger})-\langle \Vc{\flux}^{\dagger}, \Vc{\tf}'^{\ddagger}\rangle=\BD^{\Jspace,\Fspace}_{\Vc{x}}[\Vc{\flux}^{\dagger}\|\Vc{\tf}'^{\ddagger}],
\end{align}
where $\Vc{\flux}^{\dagger}=\Vc{\flux}^{\dagger}(\Vc{x},\Vc{v})$ and $\Vc{\tf}'^{\ddagger}=\Vc{\tf}^{\ddagger}(\Vc{x},\Vc{u}')$.
Therefore, we have the induced dually flat structure on $(\tilde{\Tan}_{\Vc{x}}\X, \tilde{\Tan}_{\Vc{x}}^{*}\X)$.
This induced structure can be regarded as an extension to discrete manifolds of the Otto structure \cite{otto2001Commun.PartialDiffer.Equ.,villani2003}: the formal Riemannian structure induced by the $L^{2}$ -Wasserstein distance.
This is also related to Pistone's infinite-dimensional information geometry\cite{pistone2018Inf.Geom.ItsAppl.,pistone2021ProgressinInformationGeometry:TheoryandApplications}.

\subsection{Fisher information, natural gradient, mirror descent, evolutionary computation, and optimal transport}
In information geometry, it is conventional to use the Fisher information matrices, i.e., the Hessian matrices $\FM_{\Vc{x}}$ and $\FM^{*}_{\Vc{y}}$ (\eqnref{eq:Hessian_X_Y}) as the metric tensor (Fisher--Rao metric) on  $(\Tan_{\Vc{x}}\X, \Tan_{\Vc{x}}^{*}\X)$ or equivalently on $(\Tan_{\Vc{y}}\Y, \Tan_{\Vc{y}}^{*}\Y)$  (\fgref{fig:Induced_Duality}). 
Gradient systems have been defined information-geometrically\cite{fujiwara1995PhysicaD:NonlinearPhenomena} as a Riemannian gradient flow using the Bregman divergence and the Fisher information matrix of $\Pfunc(\Vc{x})$ as the gradient function and the metric tensor, respectively: $\dot{\Vc{x}}=-\FM_{\Vc{x}}^{-1} \partial \BD^{\X}_{\Pfunc}[\Vc{x}\|\tilde{\Vc{x}}]$.
Because both $\FM_{\Vc{x}}$ and $\BD^{\X}_{\Pfunc}$ are derived from $\Pfunc(\Vc{x})$, this gradient flow becomes a geodesic in $\Y$ space: $\dot{\Vc{y}}=- (\Vc{y}-\tilde{\Vc{y}})$.
In natural gradient descent \cite{amari1996Proc.9thInt.Conf.NeuralInf.Process.Syst.,amari1998NeuralComputation,raskutti2015IEEETrans.Inf.Theory}, the Fisher information matrix is used to find the steepest descent gradient of a function $\FEnergy(\Vc{\theta})$ on a parameter space $\Theta$ as $\dot{\Vc{\theta}}=-\FM_{\Vc{\theta}}^{-1} \partial \FEnergy(\Vc{\theta})$, where $\FM_{\Vc{\theta}}$ is determined independently of $\FEnergy(\Vc{\theta})$ by considering the underlying model parameter space.
In optimization, the natural gradient is fundamental in information-geometric optimization algorithms, which contain various evolutionary optimization schemes\cite{ollivier2017J.Mach.Learn.Res.}.
In relation to machine learning, the mirror descent is identified with the natural gradient descent by a naive continuous limit\cite{raskutti2015IEEETrans.Inf.Theory,gunasekar2021}. 
Furthermore, optimal transport has recently been employed to replace or integrate the Fisher-Rao metric with the Wasserstein metrics\cite{li2018Info.Geo.,amari2018Info.Geo.}.
Because the Wasserstein metric can take the information of the base manifold into account, their integration may provide more amenable ways to accommodate various prior and structural information.

The doubly dual flat structure introduced in this work actually provides a solution to generalize those results and the associated problems. 
The base space $\X$ with the dually flat structure and the associated Fisher information matrix accommodates the conventional natural gradient.
The graph or hypergraph structure endows the additional topological relation to the base space of $\X$. 
The dissipation functions on the edge spaces or their induced versions bestow a more flexible way than the Fisher-Rao metric to represent the loss of the potential function, i.e., the dissipation, at each point in the state space. 
Upon necessity, we may combine both of them (\fgref{fig:Induced_Duality}), for example, as $\dot{\Vc{x}}=-\FM_{\Vc{x}}^{-1} \partial \FEnergy^{(1)}(\Vc{x})-\Div_{\HIncMatrix} \partial \Dissp^{*}_{\Vc{x}}[\Grad_{\HIncMatrix}\partial\FEnergy^{(2)}(\Vc{x})]$ where $\FEnergy^{(1)}(\Vc{x})$ and $\FEnergy^{(2)}(\Vc{x})$ could be different.
This flexibility may contribute to the design of new algorithms for machine learning.
Actually, this integrated representation is quite relevant to the filtering equations\cite{bain2008} in sequential inference where the first term, i.e., $-\FM_{\Vc{x}}^{-1} \partial \FEnergy^{(1)}(\Vc{x})$, can usually be associated with the update of posterior probability by observation and the second term, $-\Div_{\HIncMatrix} \partial \Dissp^{*}_{\Vc{x}}[\Grad_{\HIncMatrix}\partial\FEnergy^{(2)}(\Vc{x})]$ can represent the prediction by the prior information on the dynamics.
Our framework may provide a unified information-geometric perspective to various information-geometric analyses and extensions of filtering, e.g., projection-filters\cite{brigo1998IEEETrans.Autom.Control}, information-geometric nonlinear filtering\cite{li2017201736thChin.ControlConf.CCC}, and information geometric optimization\cite{ollivier2017J.Mach.Learn.Res.}
Furthermore, the generalized gradient flow can be regarded as a continuous time limit of the mirror descent where the nonlinear Legendre duality between primal and dual spaces is preserved at the limit.
This fact may be employed to design new gradient-based algorithms based on the doubly dual flat structure.

\section{Information-geometric Properties of Generalized Nonequilibrium Flow}\label{sec:InfoGeoNonEquilibrium}
In this section, we consider the nonequilibrium flow defined by \eqnref{eq:gNGF}, i.e.,
\begin{align}
    \dot{\Vc{x}}=-\Div_{\HIncMatrix} \partial \Dissp^{*}_{\Vc{x}}\left[\left[\Grad_{\HIncMatrix}\partial \BD^{\X}_{\Pfunc}[\Vc{x}\|\tilde{\Vc{x}}]\right]+\Vc{\tf}_{NE}\right], \label{eq:gNGF2}
\end{align}
with $\Vc{\tf}_{NE}\not\in \Img \HIncMatrix^{\Transpose}$, and show how information geometry can be employed to analyze such dynamics.
While we can obtain several properties of equilibrium flow independently of the detail of the thermodynamic function and the dissipation function, these functions should be related so as to obtain nice properties for the nonequilibrium flow. We will observe that the thermodynamic function and the dissipation function of LMA kinetics actually have such a relation.

\subsection{Gradient-flow-like property and Lyapunov function of nonequilibrium flow}
For the equilibrium flow (\eqnref{eq:gGF}), the Bregman divergence $\BD^{\X}_{\Pfunc}[\Vc{x}\|\tilde{\Vc{x}}]$ is a Lyapunov function.
The Bregman divergence can still be a Lyapunov function even for the nonequilibrium flow (\eqnref{eq:gNGF2}) under the following conditions:
\begin{lmm}\label{lmm:orth_decomp}
Suppose that, for all $\Vc{x}\in \X$, the force $\Vc{\tf}(\Vc{x})=\Grad_{\HIncMatrix}\partial \BD^{\X}_{\Pfunc}[\Vc{x}\|\tilde{\Vc{x}}]+\Vc{\tf}_{NE}$ is orthogonally decomposed as $\Vc{\tf}(\Vc{x})=\Vc{\tf}_{S}(\Vc{x}) + \Vc{\tf}_{A}(\Vc{x})$ where $\Vc{\tf}_{S}(\Vc{x})\defeq \Grad_{\HIncMatrix}\partial_{\Vc{x}}\BD^{\X}_{\Pfunc}[\Vc{x}\|\tilde{\Vc{x}}_{CB}]$ and $\Vc{\tf}_{A}(\Vc{x})\in \Fspace_{\Vc{x}}$ satisfy the pseudo-Hilbert-isosceles orthogonality $\Vc{\tf}_{S}(\Vc{x}) \perp_{H} \Vc{\tf}_{A}(\Vc{x})$.
Then $\frac{\dd}{\dt}\BD^{\X}_{\Pfunc}[\Vc{x}_{t}\|\tilde{\Vc{x}}_{CB}]\le 0$ holds. 
In addition, $\Vc{x}_{CB}=\Polytope^{sc}(\Vc{x}_{0}) \cap \Variety^{eq}(\tilde{\Vc{x}}_{CB})$ is the unique steady state of \eqnref{eq:gNGF2} with the initial state $\Vc{x}_{0}$ that attains $\frac{\dd}{\dt}\BD^{\X}_{\Pfunc}[\Vc{x}_{t}\|\tilde{\Vc{x}}_{CB}]=0$ .
Thus, $\Vc{x}_{CB}$ is locally and asymptotically stable\footnote{To have global stability, we have to consider the boundary of $\X$.}.
\end{lmm}
\begin{proof}
We can directly verify $\frac{\dd}{\dt}\BD^{\X}_{\Pfunc}[\Vc{x}_{t}\|\tilde{\Vc{x}}_{CB}]\le 0$ as follows:
\begin{align}
    \frac{\dd}{\dt}\BD^{\X}_{\Pfunc}[\Vc{x}_{t}\|\tilde{\Vc{x}}_{CB}]&=\langle\dot{\Vc{x}}, \partial_{\Vc{x}}\BD^{\X}_{\Pfunc}[\Vc{x}_{t}\|\tilde{\Vc{x}}_{CB}] \rangle=-\langle \Div_{\HIncMatrix}\Vc{\flux}(\Vc{x}), \partial_{\Vc{x}}\BD^{\X}_{\Pfunc}[\Vc{x}_{t}\|\tilde{\Vc{x}}_{CB}] \rangle\\
    &=-\langle \Vc{\flux}(\Vc{x}), \Grad_{\HIncMatrix}\partial_{\Vc{x}}\BD^{\X}_{\Pfunc}[\Vc{x}_{t}\|\tilde{\Vc{x}}_{CB}] \rangle\notag\\
    &=-\langle\Vc{\flux}(\Vc{x}), \Vc{\tf}_{S}(\Vc{x}) \rangle=-\frac{1}{2}\BD^{\Jspace,\Fspace}_{\Vc{x}}[\Vc{\flux}(\Vc{x})\|-\Vc{\tf}''(\Vc{x})] \le 0
\end{align}
where we used \eqnref{eq:pHilbert2} and $\Vc{\tf}''(\Vc{x})\defeq\Vc{\tf}_{S}(\Vc{x})-\Vc{\tf}_{A}(\Vc{x})$. 
The equality holds if and only if $\Vc{\tf}(\Vc{x})=-\Vc{\tf}''(\Vc{x})$, which means that
\begin{align}
    \Vc{\tf}(\Vc{x})=-\Vc{\tf}''(\Vc{x}) \Longleftrightarrow \Vc{\tf}_{S}(\Vc{x})=0 \Longleftrightarrow \Grad_{\HIncMatrix}\partial_{\Vc{x}}\BD^{\X}_{\Pfunc}[\Vc{x}\|\tilde{\Vc{x}}_{CB}]=0 \Longleftrightarrow \Vc{x} \in \Variety^{eq}(\tilde{\Vc{x}}_{CB}).
\end{align}
Because $\Vc{x}_{t}\in \Polytope^{sc}(\Vc{x}_{0})$,  $\Vc{x}_{CB}=\Polytope^{sc}(\Vc{x}_{0}) \cap \Variety^{eq}(\tilde{\Vc{x}}_{CB})$ holds.
\end{proof}
Thus, if the pseudo-Hilbert-isosceles orthogonal decomposition exists, then the nonequilibrium flow behaves like the equilibrium flow.

\subsection{Complex-balanced state and pseudo-Hilbert-isosceles orthogonality}
More specific conditions or situations under which the orthogonal decomposition in \lmmref{lmm:orth_decomp} exists is still an open problem. 
However, for CRN with LMA kinetics, the decomposition holds if a complex-balanced steady state exists.
\begin{prop}[Complex-balanced steady state and orthogonal decomposition for CRN with LMA kinetics\cite{kaiser2018JStatPhys,renger2021DiscreteContin.Dyn.Syst.-S,patterson2021ArXiv210314384Math-Ph,kobayashi2022Phys.Rev.Researcha}]
Suppose that a complex balanced steady state exists, i.e., $\Manifold^{\mathrm{CB}}\neq \emptyset$ for CRN with LMA kinetics (\eqnref{eq:CRN_rate}).
Using any $\tilde{\Vc{x}}_{CB} \in \Manifold^{\mathrm{CB}}$, consider a decomposition of the force $\Vc{\tf}_{MA}(\Vc{x})$ as $\Vc{\tf}_{MA}(\Vc{x})=\Vc{\tf}_{S}(\Vc{x})+\Vc{\tf}_{A}$ where $\Vc{\tf}_{S}(\Vc{x})=\Grad_{\HIncMatrix}\partial_{\Vc{x}}\BD^{\X}_{\Pfunc}[\Vc{x}\|\tilde{\Vc{x}}_{CB}]$, $\Vc{\tf}_{A}=\ln \Vc{K}+\HIncMatrix^{\Transpose}\ln \tilde{\Vc{x}}_{CB}$, and $\Pfunc(\Vc{x})$ is as in \eqnref{eq:potentialfunction_X_MAK}.
Then, for the dissipation functions in \eqnref{eq:CRN_dissip}, the pseudo-Hilbert isosceles orthogonality $\Vc{\tf}_{S}(\Vc{x})\perp_{H} \Vc{\tf}_{A}$ holds for all $\Vc{x} \in \X$.
In addition, $\Manifold^{eq}(\tilde{\Vc{x}}_{CB})=\Manifold^{\mathrm{CB}}$ holds.
\end{prop}
\begin{proof}
We can prove the orthogonality by direct computation. The orthogonality condition is
\begin{align}
\Dissp^{*}_{\Vc{x}}[\Vc{\tf}_{S}(\Vc{x})+\Vc{\tf}_{A}]&=\Dissp^{*}_{\Vc{x}}[\Vc{\tf}_{S}(\Vc{x})-\Vc{\tf}_{A}] \label{eq:orth_cb_cond1}\\
&\Leftrightarrow   \left\langle\Vc{\flux}^{+}_{MA}(\Vc{x}),\left(\frac{\Vc{x}}{\tilde{\Vc{x}}_{CB}}\right)^{-\HIncMatrix^{\Transpose}}-\Vc{1}\right\rangle+\left\langle\Vc{\flux}^{-}_{MA}(\Vc{x}),\left(\frac{\Vc{x}}{\tilde{\Vc{x}}_{CB}}\right)^{\HIncMatrix^{\Transpose}}-\Vc{1}\right\rangle=0.\label{eq:orth_cb_cond2}
\end{align}
Consider the following equality:
\begin{align}
    \left\langle\Vc{\flux}^{\pm}_{MA}(\Vc{x}),\left(\frac{\Vc{x}}{\tilde{\Vc{x}}_{CB}}\right)^{\mp\HIncMatrix^{\Transpose}}-\Vc{1}\right\rangle &=\sum_{e=1}^{N_{\edge}}\kcoef_{e}^{\pm}\tilde{\Vc{x}}_{CB}^{\Vc{\gamma}_{e}^{\pm}}\left[\left(\frac{\Vc{x}}{\tilde{\Vc{x}}_{CB}}\right)^{\Vc{\gamma}_{e}^{\mp}}-\left(\frac{\Vc{x}}{\tilde{\Vc{x}}_{CB}}\right)^{\Vc{\gamma}_{e}^{\pm}}\right],
\end{align}
where $\Vc{\gamma}_{e}^{\pm}\defeq \cmMatrix \Vc{b}_{e}^{\pm}$.
By using this, we have the following:
\begin{align}
    \mbox{\eqnref{eq:orth_cb_cond2}}&=\sum_{e=1}^{N_{\edge}}\left(\kcoef_{e}^{+}\tilde{\Vc{x}}_{CB}^{\gamma_{e}^{+}}-\kcoef_{e}^{-}\tilde{\Vc{x}}_{CB}^{\gamma_{e}^{-}}\right)\left[\left(\frac{\Vc{x}}{\tilde{\Vc{x}}_{CB}}\right)^{\gamma_{e}^{-}}-\left(\frac{\Vc{x}}{\tilde{\Vc{x}}_{CB}}\right)^{\gamma_{e}^{+}}\right]\\
    &=\left\langle (\Vc{\flux}^{+}_{MA}(\tilde{\Vc{x}}_{CB})-\Vc{\flux}^{-}_{MA}(\tilde{\Vc{x}}_{CB})), \left[\left(\frac{\Vc{x}}{\tilde{\Vc{x}}_{CB}}\right)^{(\cmMatrix \IncMatrix^{-})^{\Transpose}}-\left(\frac{\Vc{x}}{\tilde{\Vc{x}}_{CB}}\right)^{(\cmMatrix \IncMatrix^{+})^{\Transpose}}\right] \right\rangle \\
    &=\left\langle \Vc{\flux}_{MA}(\tilde{\Vc{x}}_{CB}), \IncMatrix^{\Transpose}_{-}\left(\frac{\Vc{x}}{\tilde{\Vc{x}}_{CB}}\right)^{\cmMatrix ^{\Transpose}}-\IncMatrix^{\Transpose}_{+}\left(\frac{\Vc{x}}{\tilde{\Vc{x}}_{CB}}\right)^{\cmMatrix^{\Transpose}}\right\rangle=\left\langle \IncMatrix\Vc{\flux}_{MA}(\tilde{\Vc{x}}_{CB}), \left(\frac{\Vc{x}}{\tilde{\Vc{x}}_{CB}}\right)^{\cmMatrix^{\Transpose}}\right\rangle.
\end{align}
Thus, \eqnref{eq:orth_cb_cond1} holds for $\Vc{x}\in \X$ if $\IncMatrix\Vc{\flux}_{MA}(\tilde{\Vc{x}}_{CB})=\Vc{0}$ holds\footnote{The transformation of \eqnref{eq:orth_cb_cond2} here is strongly dependent on the specific functional form of $\Vc{\flux}_{MA}(\Vc{x})$. }. 
$\Manifold^{eq}(\tilde{\Vc{x}}_{CB})=\Manifold^{\mathrm{CB}}$ can be proved by obtaining the parametric representation of $\Manifold^{\mathrm{CB}}$ as $\Manifold^{\mathrm{CB}}=\{\Vc{x}\in \X| \ln \Vc{x}-\ln \tilde{\Vc{x}}_{CB} \in \Ker \HIncMatrix^{\Transpose} \}$ via solving $\Vc{\flux}_{MA}(\Vc{x})=\Vc{0}$\footnote{We skip the derivation because it is involved. See the original derivation\cite{craciun2009JournalofSymbolicComputation} or our rephrased version\cite{kobayashi2022Phys.Rev.Research}}.  This representation is identical to that of $\Manifold^{eq}(\tilde{\Vc{x}}_{CB})$ (\eqnref{eq:Manifold_eq}).
\end{proof}
\begin{rmk}[Algebraic structure of detailed balanced and complex balanced manifolds]
We here mention about the underlying algebraic source of why $\Manifold^{eq}(\tilde{\Vc{x}}_{CB})=\Manifold^{\mathrm{CB}}$ holds.
First, we already showed that $\Variety^{DB}=\Manifold^{eq}(\tilde{\Vc{x}})$ holds generally if $\Manifold^{\mathrm{DB}} \neq \emptyset$. 
Under LMA kinetics (\eqnref{eq:CRN_rate}), the DB condition $\Vc{\flux}_{MA}(\Vc{x})=\Vc{0}$ is nothing but the binomial equations because $\Vc{\flux}^{\pm}_{MA}(\Vc{x})$ are vectors of monominals of $\Vc{x}$.
Owing to this, $\Variety^{DB}$ becomes a toric variety\footnote{The real variety generated by a toric ideal, i.e., a binomial and prime ideal\cite{craciun2009JournalofSymbolicComputation}.}. 
In contrast, the CB condition $\IncMatrix\Vc{\flux}_{MA}(\Vc{x}_{CB})=\Vc{0}$ is a set of polynomial equations for LMA kinetics. 
Nonetheless, it was shown that $\Variety^{\mathrm{CB}}$ is binomially generated and has the same structural matrix $\HIncMatrix^{\Transpose}$ as the equilibrium manifold\cite{craciun2009JournalofSymbolicComputation} .
Because of that, they become equivalent as manifolds.
\end{rmk}

Because rLDG (\eqnref{eq:LDM2}) is a subclass of CRN where $\cmMatrix=\identityM$ and thus $\HIncMatrix=\IncMatrix$ holds, the complex-balanced condition is always satisfied for rLDG.
\begin{cor}[rLDG is unconditionally complex-balanced\cite{feinberg2019}]
All the steady states of rLDG are complex-balanced states, i.e., $\Variety^{\mathrm{ST}} = \Variety^{\mathrm{CB}}$ independently of the parameter values $\Vc{\kcoef}^{\pm}$ of the flux (\eqnref{eq:LDM2}) \footnote{Such a situation is called unconditionally complex balanced.}. 
Thus, KL divergence (\eqnref{eq:KL}) always works as a Lyapunov function of rLDG\footnote{More generally, if $\mathrm{Rank}[\HIncMatrix]=\mathrm{Rank}[\IncMatrix]$, CRN is unconditionally complex-balanced. This condition is called the deficiency zero condition\cite{feinberg2019}.}.
\end{cor}
The properties described in this corollary are well-known for rMJP and are usually obtained by using the Perron-Frobenius theorem for linear operators. 
The framework of the generalized flow enables us to extend them to the nonlinear regime.

\subsection{Effective fluw of the nonequilibrium flow by the primal information-geometric projection}
In general, the nonequilibrium force or flux has redundant degrees of freedom in terms of generating a specific vector field or trajectory $\{\Vc{x}_{t}\}$ on the density space.
By using the extended HHK projective decomposition (\thmref{thm:HHK}), we can carve out the effective part of the flux for the trajectory  $\{\Vc{x}_{t}\}$. 
In addition, we can obtain an effective time-dependent equilibrium flux that mimics the trajectory $\{\Vc{x}_{t}\}$:
\begin{lmm}[\TJK{Effective time-dependent equilibrium flux}]\label{lmn:effectiveflux}
Suppose that $\Vc{\flux}(\Vc{x})$ is the flux of a generalized flow (\eqnref{eq:gFlow}), and define the corresponding effective equilibrium flux by $\Vc{\flux}_{eq}(\Vc{x})=\partial \Dissp^{*}_{\Vc{x}}[\HIncMatrix^{\Transpose}\Vc{u}_{eq}(\Vc{x})]$ where $\Vc{u}_{eq}(\Vc{x})\defeq \partial \tilde{\Dissp}_{\Vc{x}}[-\HIncMatrix \Vc{\flux}(\Vc{x})]$.
Then, $\Vc{\flux}_{eq}(\Vc{x})$ induces the same velocity as $\Vc{\flux}(\Vc{x})$ does\footnote{$\Vc{\flux}_{eq}(\Vc{x})$ may not be equilibrium flux because $\Vc{u}_{eq}(\Vc{x})$ is not necessarily represented by the gradient of a certain function $\FEnergy(\Vc{x})$ as $\Vc{u}_{eq}(\Vc{x})=\partial_{\Vc{x}} \FEnergy(\Vc{x})$.}.
Furthermore, for a given trajectory $\{\Vc{x}_{t}\}$ of $\Vc{\flux}(\Vc{x})$, we can construct a time-dependent equilibrium flux $\Vc{\flux}_{eq}(t,\Vc{x})$ that can generate the same $\{\Vc{x}_{t}\}$ as follows
\begin{align}
    \Vc{\flux}_{eq}(t,\Vc{x})&\defeq \partial \Dissp^{*}_{\Vc{x}}\left[\Grad_{\HIncMatrix}\partial \BD^{\X}_{\Pfunc}[\Vc{x}\|\tilde{\Vc{x}_{t}}]\right], &\tilde{\Vc{x}}_{t} \defeq  \partial \Pfunc^{*}[\partial \Pfunc(\Vc{x}_{t})+\Vc{u}_{eq}(\Vc{x}_{t})]. \label{eq:ef_time_dep_flux}
\end{align}
\end{lmm}
\begin{proof}
From \thmref{thm:HHK}, $\Vc{\flux}(\Vc{x})$ can be decomposed as $\Vc{\flux}(\Vc{x})=\Vc{\flux}_{eq}(\Vc{x})+(\Vc{\flux}(\Vc{x})-\Vc{\flux}_{eq}(\Vc{x}))$. Because $\Vc{\tf}_{eq}(\Vc{x})=\partial \Dissp_{\Vc{x}}^{*}[\Vc{\flux}_{eq}(\Vc{x})] \in \Polytope^{fr}(\Vc{0})$, there exists $\Vc{u}_{eq}(\Vc{x})$ satisfying $-\HIncMatrix^{\Transpose}\Vc{u}_{eq}(\Vc{x})=\Vc{\tf}_{eq}(\Vc{x})$. By employing the duality introduced in \thmref{thm:inducedHessiangeometry1}, $\Vc{u}_{eq}(\Vc{x})$ can be represented as $\Vc{u}_{eq}(\Vc{x})=\partial \tilde{\Dissp}_{\Vc{x}}[\Vc{v}(\Vc{x})]$ where $\Vc{v}(\Vc{x})=-\HIncMatrix \Vc{\flux}(\Vc{x})$. Because $-\HIncMatrix \Vc{\flux}_{eq}(\Vc{x})=\Vc{v}(\Vc{x})$, $\Vc{\flux}_{eq}(\Vc{x})$ generates the same dynamics or vector field as $\Vc{\flux}(\Vc{x})$ does. By solving $-\Vc{u}_{eq}(\Vc{x}_{t})=\partial \Pfunc(\Vc{x}_{t})-\partial \Pfunc(\tilde{\Vc{x}}_{t})$, we have \eqnref{eq:ef_time_dep_flux}.
\end{proof}
The effective time-dependent equilibrium flux is obtained more explicitly for CRN with LMA kinetics:
\begin{cor}[\TJK{Effective equilibrium force and flux of LMA kinetics}]\label{cor:effective_flux_LMA}
Consider the following quantities of CRN with LMA kinetics:
\begin{align}
\mbox{Flux defined in \eqnref{eq:CRN_rate}} & & \Vc{\flux}_{MA}(\Vc{x};\Vc{\kcoef}^{\pm})\\
\mbox{Force defined in \eqnref{eq:jffMAK}} & & \Vc{\tf}_{\mathrm{MA}}(\Vc{x};\Vc{K})\\
\mbox{Thermodynamic function defined in \eqnref{eq:potentialfunction_X_MAK}} & & \Pfunc(\Vc{x})\\
\mbox{Dissipation functions defined in \eqnref{eq:CRN_dissip} with \eqnref{eq:jffMAK}} & & \Dissp^{*}_{\Vc{x},\Vc{\kappa}}[\Vc{\tf}] = \Dissp^{*}_{\Vc{\frenecy}_{\mathrm{MA}}(\Vc{x};\Vc{\kappa})}[\Vc{\tf}],
\end{align}
where $\Vc{\kcoef}^{\pm}=\Vc{\kappa}\circ \Vc{K}^{\pm 1/2}$ holds. For a trajectory $\{\Vc{x}_{t}\}$ generated by $\Vc{\flux}_{MA}(\Vc{x};\Vc{\kcoef}^{\pm})$, the effective time-dependent equilibrium force $\Vc{\tf}_{eq}(t,\Vc{x})$ can be described as $\Vc{\tf}_{MA}(\Vc{x}; \Vc{K}_{eq}(t))$ where $\Vc{K}_{eq}(t)$ is determined by 
\begin{align}
\Vc{K}_{eq}(t) &\defeq \exp\left[-\HIncMatrix^{\Transpose}\left(\Vc{u}_{eq}(\Vc{x}_{t};\Vc{\kappa})+\partial \Pfunc(\Vc{x}_{t}) \right)\right], & \Vc{u}_{eq}(\Vc{x}_{t};\Vc{\kappa}) &\defeq \partial \tilde{\Dissp}_{\Vc{x}_{t},\Vc{\kappa}}[-\HIncMatrix \Vc{\flux}_{MA}(\Vc{x}_{t})] \label{eq:time_dependent_Keq}
\end{align}
Thus, the effective time-dependent equilibrium flux $\Vc{\flux}_{eq}(t,\Vc{x})$ is represented as $\Vc{\flux}_{eq}(t,\Vc{x})=\Vc{\flux}_{MA}(\Vc{x}; \Vc{\kcoef}^{\pm}_{eq}(t))$ where $\Vc{\kcoef}^{\pm}_{eq}(t)=\Vc{\kappa}\circ \Vc{K}^{\pm 1/2}_{eq}(t)$.
\end{cor}
\TJK{This corollary means that the effective time-dependent flux of LMA kinetics is always obtained by a time-dependent modulation of the kinetic parameters $\Vc{\kcoef}^{\pm}$. More specifically, the modulation of force part  $\Vc{K}_{eq}(t)$ is sufficient while the activity part $\Vc{\kappa}$ is kept constant\footnote{While this result may sound not so significant mathematically, for physics and chemistry, the result means that the effective flux is physically realizable and tested.}.}
\begin{figure}[t]
\includegraphics[bb=0 0 1424 768, width= \linewidth]{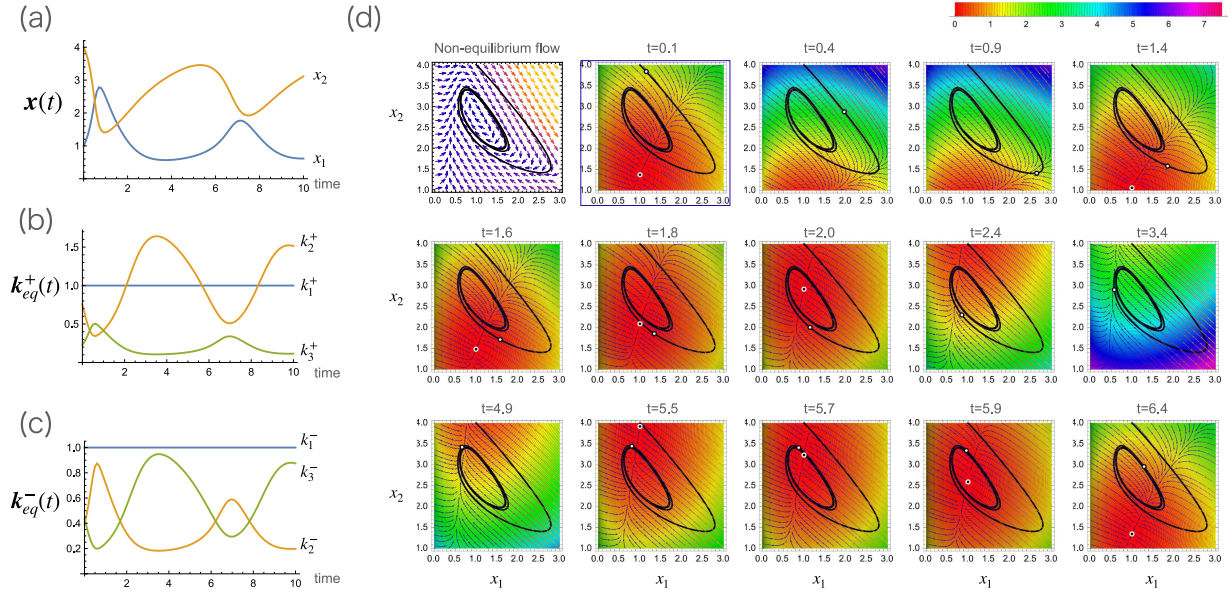}
\caption{\TJK{A nonequilibrium trajectory of the Brusselator CRN (\exref{ex:BrusselatorEq}) and the associated time-dependent effective equilibrium flux $\Vc{\flux}_{eq}(t,\Vc{x})=\Vc{\flux}_{MA}(\Vc{x}; \Vc{\kcoef}^{\pm}_{eq}(t))$. 
(a) A nonequilibrium trajectory $\{\Vc{x}_{t}\}$ of the Brusselator CRN for a set of the kinetic parameters, $\kcoef_{1}^{+}=1.0$, $\kcoef_{1}^{-}=1.0$, $\kcoef_{2}^{+}=3.0$, $\kcoef_{2}^{-}=0.1$, $\kcoef_{3}^{+}=1.0$, $\kcoef_{3}^{-}=0.1$, and the initial state $(\Vc{x}_{1}(0),\Vc{x}_{2}(0))=(1.0, 4.0)$.
(b,c) The time-dependent kinetic parameter set $\Vc{\kcoef}^{\pm}_{eq}(t)$, which generates the time-dependent equilibrium flux $\Vc{\flux}_{eq}(t,\Vc{x})$. 
(d) The top left panel shows the nonequilibrium trajectory $\{\Vc{x}_{t}\}$ (black curve) and the vector field $\Vc{v}(\Vc{x})= - \HIncMatrix \Vc{\flux}(\Vc{x})$  (arrows). The other panels show the time-dependent vector field $\Vc{v}_{eq}(t, \Vc{x})$ induced by $\Vc{\flux}_{eq}(t,\Vc{x})$: $\Vc{v}_{eq}(t, \Vc{x})= - \HIncMatrix \Vc{\flux}_{eq}(t,\Vc{x})$. The nonequilibrium trajectory $\{\Vc{x}_{t}\}$ is also depicted as a reference (the black curve). The color indicates the value of $\BD^{\X}_{\Pfunc}[\Vc{x}\|\tilde{\Vc{x}_{t}}]$. 
In each panel, the white circle on the trajectory is $\Vc{x}_{t}$ at which $\Vc{\flux}_{eq}(t,\Vc{x})$ is computed, and the black circle with a white border corresponds to $\tilde{\Vc{x}}_{t}$ at which $\BD^{\X}_{\Pfunc}[\Vc{x}\|\tilde{\Vc{x}_{t}}]$ is $0$.}
}
\label{fig:Projection}
\end{figure}

\TJK{\begin{ex}[Simplified Brusselator CRN\cite{feinberg2019,yoshimura2022} (continued)]\label{ex:BrusselatorProjection}
By using the Brusselator CRN (\exref{ex:BrusselatorEq}), we numerically obtained a nonequilibrium trajectory $\{\Vc{x}_{t}\}$ (\fgref{fig:Projection} (a)  and \fgref{fig:Projection} (d) top left panel).
By using \corref{cor:effective_flux_LMA}, we also computed the corresponding time-dependent kinetic parameter set $\Vc{\kcoef}^{\pm}_{eq}(t)$ that generates the time-dependent equilibrium flux $\Vc{\flux}_{eq}(t,\Vc{x})=\Vc{\flux}_{MA}(t,\Vc{x}; \Vc{\kcoef}^{\pm}_{eq}(t))$ (\fgref{fig:Projection} (b,c)). 
Figure \ref{fig:Projection} (d) shows the vector field $\Vc{v}_{eq}(t, \Vc{x})=-\HIncMatrix \Vc{\flux}_{eq}(t,\Vc{x})$ induced by the time-dependent equilibrium flux $\Vc{\flux}_{eq}(t,\Vc{x})$ and the contours of $\BD^{\X}_{\Pfunc}[\Vc{x}\|\tilde{\Vc{x}_{t}}]$ where $\tilde{\Vc{x}}_{t}$ follows from \eqnref{eq:ef_time_dep_flux}.
From \fgref{fig:Projection}, we can see that the trajectory $\{\Vc{x}_{t}\}$ originally generated by the nonequilibrium flux $\Vc{\flux}(\Vc{x})$ can be traced by the time-dependent equilibrium flux $\Vc{\flux}_{eq}(t,\Vc{x})$ and also that $\Vc{\flux}_{eq}(t,\Vc{x})$ can be physically realized by the modulation of the kinetic parameters $\Vc{\kcoef}^{\pm}_{eq}(t)$.
\end{ex}
}

\subsection{Characterization of the nonequilibrium flow by the dual information geometric projection}
The nonequilibrium flow is redundant in terms of generating a specific vector field or trajectory $\{\Vc{x}_{t}\}$. 
Such redundancy is crucial to characterize the extent of nonequilibrium.
One approach for the characterization is to investigate the cycle force or flux, which has been employed in the linear theory of dynamics on graphs and also in graph-theoretic approaches to nonequilibrium phenomena\cite{hill2005,schnakenberg1976Rev.Mod.Phys.,altaner2012Phys.Rev.Ea,polettini2014J.Chem.Phys.,strang2020}.
To extract such cyclic components, we can use $\cycMatrix^{\Transpose}=\Curl_{\cycMatrix}$, its adjoint $\cycMatrix=\Curl^{*}_{\cycMatrix}$, the associated cycle subspaces $C^{2}(\HGraph)$ and $C_{2}(\HGraph)$, and also the generalized HKK decomposition (\thmref{thm:HHK}).

\begin{dfn}[Cycle spaces]\label{dfn:cycle_spaces}
The cycle spaces at $\Vc{x}\in \X$ are defined as $\Z_{\Vc{x}} = \chain_{2}(\HGraph)=\Ker[\cycMatrix]$ and $\Zeta_{\Vc{x}} = \chain^{2}(\HGraph)=\Img[\cycMatrix^{\Transpose}]$.
\end{dfn}
For a given nonequilibrium force $\Vc{\tf}(\Vc{x})$, we can obtain its cycle component $\Vc{\zeta} = \Curl_{\cycMatrix}\Vc{\tf}(\Vc{x})=\cycMatrix^{\Transpose}\Vc{\tf}(\Vc{x}) \in \Zeta_{\Vc{x}}$.
$\Vc{\zeta}$ contains the information to categorize the force because $\Polytope^{fr}(\Vc{\tf})=\Polytope^{fr}(\Vc{\zeta})$ is the quotient space of force by the equilibrium force.
For each $\zeta \in \Zeta_{\Vc{x}}$, we obtain the representative force $\Vc{\tf}^{\lozenge}(\Vc{x},\Vc{\zeta})$ via the following variational problem:
\begin{lmm}[Steady (zero-velocity) force as the minimizer of the dual dissipation function]
For a given $\Vc{\zeta}\in \Zeta_{\Vc{x}}$, we define the force $\Vc{\tf}^{\lozenge}(\Vc{x},\Vc{\zeta})$ minimizing the dual dissipation function:
\begin{align}
    \Vc{\tf}^{\lozenge}(\Vc{x},\Vc{\zeta})\defeq \arg \min_{\Vc{\tf}}\Dissp_{\Vc{x}}^{*}[\Vc{\tf}],\quad \mbox{s.t.   $\Curl_{\cycMatrix} \Vc{\tf}=\Vc{\zeta}$}. \label{eq:min_Dissip_F}
\end{align}
Then, $\Vc{\tf}^{\lozenge}(\Vc{x},\Vc{\zeta})$ is the steady (zero-velocity) force, i.e., $\Vc{\flux}^{\lozenge}=\partial \Dissp_{\Vc{x}}^{*}[\Vc{\tf}^{\lozenge}] \in \Polytope^{vl}(\Vc{0})$. 
\end{lmm}
\begin{proof}
From \eqnref{eq:j_eq_f_st_intersection} in \thmref{thm:HHK}, $\Vc{\tf}^{\lozenge}(\Vc{x},\Vc{\zeta}) \in \Polytope^{fr}(\Vc{\zeta})\cap \Manifold^{vl}(\Vc{0})$. Thus, $\Vc{\flux}^{\lozenge}\in \Polytope^{vl}(\Vc{0})$. 
\end{proof}
Among various $\Vc{\tf}$ that has the same cyclic component $\Vc{\zeta}$, the force $\Vc{\tf}^{\lozenge}(\Vc{x},\Vc{\zeta})$ is the one that induces no dynamics of $\Vc{x}$. Because any dynamics of $\Vc{x}$ can be represented by the effective equilibrium flux as in \lmmref{lmn:effectiveflux}, $\Vc{\tf}^{\lozenge}(\Vc{x},\Vc{\zeta})$ can be regarded as the force being purely relevant to the cycle.
\begin{figure}[h]
\includegraphics[bb=0 0 1024 468, width=\linewidth]{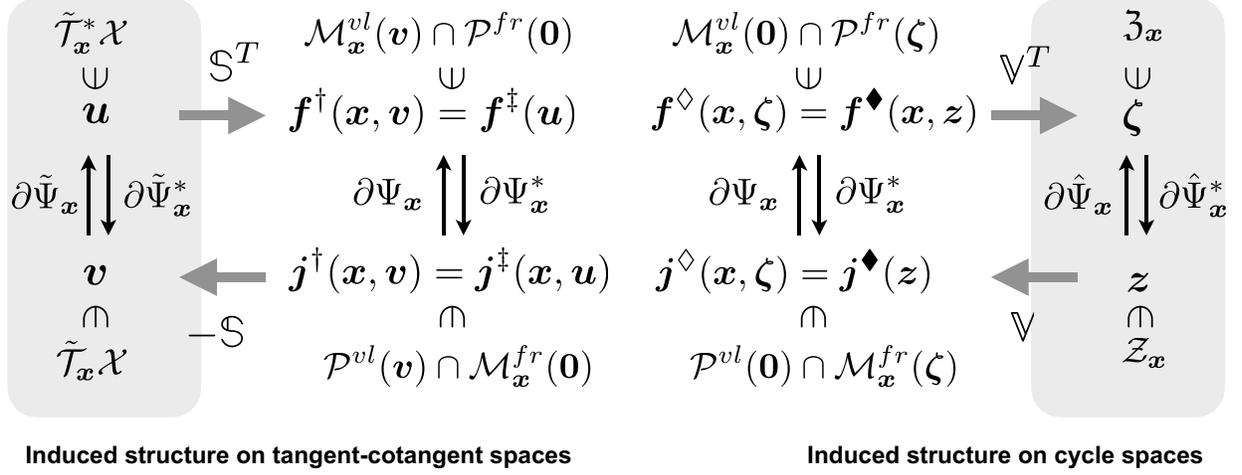}
\caption{\TJK{The induced dually flat structure on the cycle spaces (right gray region) in comparison with the induced structure on the restricted tangent and cotangent spaces(left gray region). }}
\label{fig:Induced_Duality2}
\end{figure}

Using $\Vc{\tf}^{\lozenge}(\Vc{x},\Vc{\zeta})$, we can establish the induced duality between $\Z_{\Vc{x}}$ and $\Zeta_{\Vc{x}}$ spaces:
\begin{thm}[Induced dually flat structure on cycle spaces]\label{thm:inducedHessiangeometry2}
On the cycle spaces, $\Z_{\Vc{x}}$ and $\Zeta_{\Vc{x}}$, we have the Legendre conjugate dissipation functions $\hat{\Dissp}_{\Vc{x}}: \Z_{\Vc{x}} \to \Real$ and $\hat{\Dissp}_{\Vc{x}}^{*}: \Zeta_{\Vc{x}} \to \Real$ induced by the dissipation functions on the edge spaces  (\fgref{fig:Induced_Duality2}).
\end{thm}
\begin{proof}
For each $\Vc{\zeta}\in \Zeta_{\Vc{x}}$, we can uniquely determine $\Vc{\tf}^{\lozenge}(\Vc{x},\Vc{\zeta}) \in \Fspace_{\Vc{x}}$ and $\Vc{\flux}^{\lozenge}(\Vc{x},\Vc{\zeta})\in \Jspace_{\Vc{x}}$ as
\begin{align}
\Vc{\tf}^{\lozenge}(\Vc{x},\Vc{\zeta}) &\defeq \Polytope^{fr}(\Vc{\zeta})\cap\Manifold^{vl}_{\Vc{x}}(\Vc{0})=\arg\min_{\Vc{\tf}\in \Polytope^{fr}(\Vc{\zeta})} \Dissp_{\Vc{x}}^{*}[\Vc{\tf}],\\
\Vc{\flux}^{\lozenge}(\Vc{x},\Vc{\zeta}) &\defeq \partial \Dissp^{*}_{\Vc{x}}[\Vc{\tf}^{\lozenge}(\Vc{x},\Vc{\zeta})] \in \Polytope^{vl}(\Vc{0}).
\end{align}
In addition, $\Vc{z}^{\lozenge}(\Vc{x},\Vc{\zeta}) \in \Z_{\Vc{x}}$ satisfying $\cycMatrix\Vc{z}^{\lozenge}(\Vc{x},\Vc{\zeta})  = \Vc{\flux}^{\lozenge}(\Vc{x},\Vc{\zeta})$ is also uniquely determined because $\cycMatrix: \Z_{\Vc{x}} \to \Jspace_{\Vc{x}}$, $\Vc{\flux}^{\lozenge}(\Vc{x},\Vc{\zeta}) \in \Polytope^{vl}(\Vc{0})=\Img[\cycMatrix]$ and $\Ker \cycMatrix =\{\Vc{0}\}$. 
For these quantities,
\begin{align}
    \Vc{\zeta} & = \cycMatrix^{\Transpose} \Vc{\tf}^{\lozenge}(\Vc{x},\Vc{\zeta}), & \Vc{\tf}^{\lozenge}(\Vc{x},\Vc{\zeta}) & = \partial \Dissp_{\Vc{x}}[\Vc{\flux}^{\lozenge}(\Vc{x},\Vc{\zeta})], & \Vc{\flux}^{\lozenge}(\Vc{x},\Vc{\zeta}) & =  \cycMatrix \Vc{z}^{\lozenge}(\Vc{x},\Vc{\zeta}), \label{eq:zeta_pairing}
\end{align}
hold.
Conversely, for a given $\Vc{z}\in \Z_{\Vc{x}}$, we have $\Vc{\flux}^{\blacklozenge}(\Vc{x},\Vc{u})$, $\Vc{\tf}^{\blacklozenge}(\Vc{u})$, and $\Vc{\zeta}^{\blacklozenge}(\Vc{x},\Vc{u})$ as follows:
\begin{align}
    \Vc{\zeta}^{\blacklozenge}(\Vc{x},\Vc{z}) & =  \cycMatrix^{\Transpose} \Vc{\tf}^{\blacklozenge}(\Vc{x},\Vc{z}), & \Vc{\tf}^{\blacklozenge}(\Vc{x},\Vc{z}) & = \partial \Dissp_{\Vc{x}}[\Vc{\flux}^{\blacklozenge}(\Vc{z})], & \Vc{\flux}^{\blacklozenge}(\Vc{z}) & = \cycMatrix \Vc{z}.\label{eq:z_pairing}
\end{align}
For a pair of $(\Vc{z},\Vc{\zeta})_{\Vc{x}}$ that satisfies $\Vc{z}=\Vc{z}^{\lozenge}(\Vc{x},\Vc{\zeta})$, then $\Vc{\zeta}=\Vc{\zeta}^{\blacklozenge}(\Vc{x},\Vc{z})$, $\Vc{\tf}^{\lozenge}(\Vc{x},\Vc{\zeta})=\Vc{\tf}^{\blacklozenge}(\Vc{x},\Vc{z})$, and $\Vc{\flux}^{\lozenge}(\Vc{x},\Vc{\zeta})=\Vc{\flux}^{\blacklozenge}(\Vc{z})$ hold.
This pairing establishes a bijection between $\Z_{\Vc{x}}$ and $\Zeta_{\Vc{x}}$.
This bijection is realized by the Legendre transformations of the following induced dissipation functions on $\Z_{\Vc{x}}$ and $\Zeta_{\Vc{x}}$:
\begin{align}
\hat{\Dissp}_{\Vc{x}}(\Vc{z}) & \defeq \Dissp_{\Vc{x}}(\Vc{\flux}^{\blacklozenge}(\Vc{z})), & 
\hat{\Dissp}_{\Vc{x}}^{*}(\Vc{\zeta}) & \defeq \Dissp_{\Vc{x}}^{*}(\Vc{\tf}^{\lozenge}(\Vc{x},\Vc{\zeta})). \label{eq:induced_dissp_func_cycle}
\end{align}
These functions are Legendre conjugate as follows:
\begin{align*}
 \max_{\Vc{z}'\in\Z_{\Vc{x}}}&\left[\langle\Vc{z}', \Vc{\zeta}\rangle-\hat{\Dissp}_{\Vc{x}}(\Vc{z}') \right]
 =\max_{\substack{\Vc{z}'\in\Z_{\Vc{x}}\\ \Vc{\tf}\in \Polytope^{fr}(\Vc{\zeta})}}\left[\langle\Vc{z}', \cycMatrix^{\Transpose}\Vc{\tf}\rangle-\hat{\Dissp}_{\Vc{x}}(\Vc{z}') \right]
 =\max_{\substack{\Vc{z}'\in\Z_{\Vc{x}}\\ \Vc{\tf}\in \Polytope^{fr}(\Vc{\zeta})}}\left[\langle\cycMatrix\Vc{z}', \Vc{\tf}\rangle-\Dissp_{\Vc{x}}(\cycMatrix\Vc{z}') \right]\\
&=\max_{\substack{\Vc{\flux}'\in \Polytope^{vl}(\Vc{0})\\ \Vc{\tf}\in \Polytope^{fr}(\Vc{\zeta})}}\left[\langle\Vc{\flux}', \Vc{\tf}\rangle-\Dissp_{\Vc{x}}(\Vc{\flux}') \right]
=\max_{\substack{\Vc{\flux}'\in \Polytope^{vl}(\Vc{0})\\ \Vc{\tf}\in \Polytope^{fr}(\Vc{\zeta})}}\left[\langle\Vc{\flux}', \Vc{\tf}^{\lozenge}(\Vc{x},\Vc{\zeta})\rangle-\Dissp_{\Vc{x}}(\Vc{\flux}')  + \langle\Vc{\flux}', (\Vc{\tf}-\Vc{\tf}^{\lozenge}(\Vc{x},\Vc{\zeta}))\rangle\right]\\
&=\max_{\substack{\Vc{\flux}'\in \Polytope^{vl}(\Vc{0})}}\left[\langle\Vc{\flux}', \Vc{\tf}^{\lozenge}(\Vc{x},\Vc{\zeta})\rangle-\Dissp_{\Vc{x}}(\Vc{\flux}')  \right]=\Dissp_{\Vc{x}}^{*}(\Vc{\tf}^{\lozenge}(\Vc{x},\Vc{\zeta}))=\hat{\Dissp}_{\Vc{x}}(\Vc{\zeta}),
\end{align*}
where we used $\langle\Vc{\flux}', (\Vc{\tf}-\Vc{\tf}^{\lozenge}(\Vc{x},\Vc{\zeta}))\rangle=0$ because $\Vc{\flux}' \in \Polytope^{vl}(\Vc{0})=\Img \cycMatrix$ and $\Vc{\tf}-\Vc{\tf}^{\lozenge}(\Vc{x},\Vc{\zeta}) \in \Polytope^{fr}(\Vc{0})=\Ker \cycMatrix^{\Transpose}$.
The inverse is also shown:
\begin{align*}
 \max_{\Vc{\zeta}'\in\Zeta_{\Vc{x}}}&\left[\langle\Vc{z}, \Vc{\zeta}'\rangle-\hat{\Dissp}_{\Vc{x}}^{*}(\Vc{\zeta}') \right]
 =\max_{\substack{\Vc{\zeta}'\in\Zeta_{\Vc{x}}}}\left[\langle\Vc{z}, \cycMatrix^{\Transpose}\Vc{\tf}^{\lozenge}(\Vc{x},\Vc{\zeta}')\rangle-\Dissp_{\Vc{x}}^{*}(\Vc{\tf}^{\lozenge}(\Vc{x},\Vc{\zeta}')) \right]\\
 &=\max_{\substack{\Vc{\tf}^{\lozenge}\in  \Variety^{vl}_{\Vc{x}}(\Vc{0})}}\left[\langle\cycMatrix\Vc{z}, \Vc{\tf}^{\lozenge}\rangle-\Dissp_{\Vc{x}}^{*}(\Vc{\tf}^{\lozenge}) \right]
 =\max_{ \Vc{\tf}^{\lozenge}\in  \Variety^{vl}_{\Vc{x}}(\Vc{0})}\left[\langle\Vc{\flux}^{\blacklozenge}(\Vc{z}), \Vc{\tf}^{\lozenge}\rangle-\Dissp_{\Vc{x}}^{*}(\Vc{\tf}^{\lozenge}) \right]\\
 &=\Dissp_{\Vc{x}}(\Vc{\flux}^{\blacklozenge}(\Vc{z}))=\hat{\Dissp}_{\Vc{x}}(\Vc{z})
\end{align*}
The pair $(\Vc{z}$, $\Vc{\zeta})_{\Vc{x}}$ is Legendre dual with respect to these functions:
\begin{align}
    \partial_{\Vc{\zeta}} \hat{\Dissp}_{\Vc{x}}^{*}(\Vc{\zeta})&= \left[\frac{\partial \Vc{\tf}^{\lozenge}(\Vc{x},\Vc{\zeta})}{\partial \Vc{\zeta}}\right]^{\Transpose}\left.\frac{\partial \Dissp_{\Vc{x}}^{*}(\Vc{\tf})}{\partial \Vc{\tf}}\right|_{\Vc{\tf}=\Vc{\tf}^{\lozenge}(\Vc{x},\Vc{\zeta})}=\left[\frac{\partial \Vc{\tf}^{\lozenge}(\Vc{x},\Vc{\zeta})}{\partial \Vc{\zeta}}\right]^{\Transpose}\Vc{\flux}^{\lozenge}(\Vc{x},\Vc{\zeta})\\
    &=\left[\frac{\partial \Vc{\tf}^{\lozenge}(\Vc{x},\Vc{\zeta})}{\partial \Vc{\zeta}}\right]^{\Transpose}\cycMatrix\Vc{z}^{\lozenge}(\Vc{x},\Vc{\zeta})=\Vc{z},\\
    \partial_{\Vc{z}} \hat{\Dissp}_{\Vc{x}}(\Vc{z}) &= \left[\frac{\partial \Vc{\flux}^{\blacklozenge}(\Vc{z})}{\partial \Vc{z}}\right]^{\Transpose}\left.\frac{\partial \Dissp_{\Vc{x}}(\Vc{\flux})}{\partial \Vc{\flux}}\right|_{\Vc{\flux}=\Vc{\flux}^{\blacklozenge}(\Vc{z})}
    =\left[\frac{\partial \Vc{\flux}^{\blacklozenge}(\Vc{z})}{\partial \Vc{z}}\right]^{\Transpose}\Vc{\tf}^{\blacklozenge}(\Vc{x},\Vc{z})=\cycMatrix^{\Transpose}\Vc{\tf}^{\blacklozenge}(\Vc{x},\Vc{z})=\Vc{\zeta},
\end{align}
where we used $\left[\frac{\partial \Vc{\tf}^{\lozenge}(\Vc{x},\Vc{\zeta})}{\partial \Vc{\zeta}}\right]^{\Transpose}\cycMatrix=\identityM$ from $\frac{\partial}{\partial \Vc{\zeta}}[\Vc{\zeta}-\cycMatrix^{\Transpose} \Vc{\tf}^{\lozenge}(\Vc{x},\Vc{\zeta})]=\identityM - \cycMatrix^{\Transpose} \frac{\partial \Vc{\tf}^{\lozenge}(\Vc{x},\Vc{\zeta})}{\partial \Vc{\zeta}}=0$ and $\frac{\partial \Vc{\flux}^{\blacklozenge}(\Vc{x},\Vc{z})}{\partial \Vc{z}} =\cycMatrix$.
They are dissipation functions: strict convexity and $1$-coercivity follow from those of the original dissipation functions.
Also, we have
\begin{align}
\mbox{Symmetry:}& & \hat{\Dissp}_{\Vc{x}}(-\Vc{z}) & = \Dissp_{\Vc{x}}(\Vc{\flux}^{\blacklozenge}(-\Vc{z}))=\Dissp_{\Vc{x}}(-\Vc{\flux}^{\blacklozenge}(\Vc{z}))=\hat{\Dissp}_{\Vc{x}}(\Vc{z})\\
\mbox{Bounded by $0$ at $\Vc{0}$:}& &     \hat{\Dissp}_{\Vc{x}}(\Vc{z}=\Vc{0})&=\Dissp_{\Vc{x}}(\Vc{\flux}^{\blacklozenge}(\Vc{0}))=\Dissp_{\Vc{x}}(\Vc{0})=0.
\end{align}
\end{proof}
For a given force $\Vc{\tf}(\Vc{x})$ of the generalized flow (\eqnref{eq:gFlow}), $\Vc{\tf}_{st}(\Vc{x})=\Vc{\tf}^{\lozenge}(\Vc{x}, \cycMatrix^{\Transpose}\Vc{\tf}(\Vc{x}))$ works as the effective cycle force for each $\Vc{x} \in \X$.
Similarly to \corref{cor:effective_flux_LMA}, we obtain the effective cycle force and the flux for LMA kinetics by parametric modulation:
\begin{cor}[\TJK{Effective cycle force and flux for LMA kinetics}]\label{cor:effective_cycle_flux_LMA}
Consider CRN with LMA kinetics as in \corref{cor:effective_flux_LMA}.
For each $\Vc{x}\in \X$, the effective cycle force $\Vc{\tf}_{st}(\Vc{x})$ associated with $\Vc{\flux}_{MA}(\Vc{x};\Vc{\kcoef}^{\pm})$ can be described as $\Vc{\tf}_{st}(\Vc{x})=\Vc{\tf}_{MA}(\Vc{x}; \Vc{K}_{st}(\Vc{x}))$ where $\Vc{K}_{st}(\Vc{x})$ is determined by 
\begin{align}
\Vc{K}_{st}(\Vc{x}) &\defeq \exp\left[\Vc{\tf}^{\lozenge}(\Vc{x};\cycMatrix^{\Transpose}\Vc{\tf}_{MA}(\Vc{x};\Vc{K}))-\HIncMatrix^{\Transpose}\partial \Pfunc(\Vc{x}) \right]. \label{eq:Kst}
\end{align}
Thus, the effective cycle flux $\Vc{\flux}_{st}(\Vc{x})$ is represented as $\Vc{\flux}_{st}(\Vc{x})=\Vc{\flux}_{MA}(\Vc{x}; \Vc{\kcoef}^{\pm}_{st}(\Vc{x}))$ where $\Vc{\kcoef}^{\pm}_{st}(\Vc{x})=\Vc{\kappa}\circ \Vc{K}^{\pm 1/2}_{st}(\Vc{x})$.
For a given trajectory $\{\Vc{x}_{t}\}$, which is generated by a generalized flow, we have the effective time-dependent cycle flux $\Vc{\flux}_{st}(t, \Vc{x})$ as $\Vc{\flux}_{st}(t, \Vc{x})=\Vc{\flux}_{MA}(\Vc{x}; \Vc{\kcoef}^{\pm}_{st}(\Vc{x}_{t}))$. From the construction, this time-dependent cycle flux makes $\Vc{x}_{t}$ a steady state for each $t$, i.e., $\HIncMatrix \Vc{\flux}_{st}(t, \Vc{x}_{t})=\Vc{0}$ holds for all $t$.
\end{cor}

\begin{figure}[t]
\includegraphics[bb=0 0 1424 768, width= \linewidth]{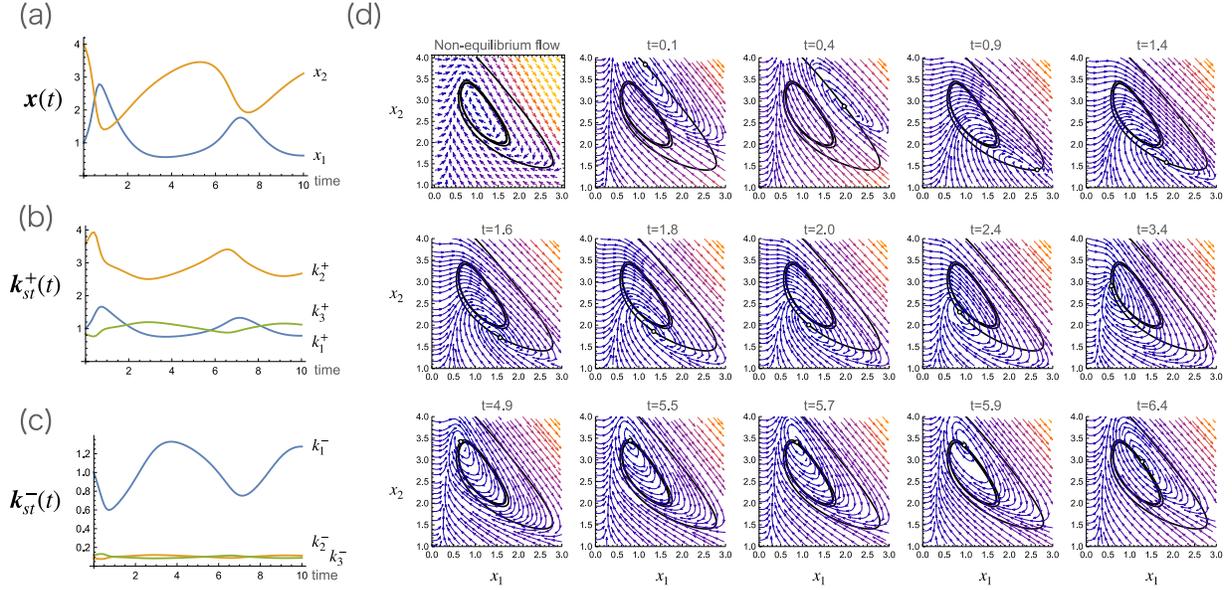}
\caption{\TJK{A nonequilibrium trajectory of the Brusselator CRN (\exref{ex:BrusselatorEq}) and the associated time-dependent cycle flux $\Vc{\flux}_{st}(t,\Vc{x})=\Vc{\flux}_{MA}(\Vc{x}; \Vc{\kcoef}^{\pm}_{st}(t))$. 
(a) A nonequilibrium trajectory $\{\Vc{x}_{t}\}$ of the Brusselator CRN obtained with the same parameter values in \fgref{fig:Projection} (a) .
(b,c) The time-dependent kinetic parameter set $\Vc{\kcoef}^{\pm}_{st}(t)$, which generates the time-dependent cycle flux $\Vc{\flux}_{st}(t,\Vc{x})$. 
(d) The top left panel shows the nonequilibrium trajectory $\{\Vc{x}_{t}\}$ (the black curve) and vector field $\Vc{v}(\Vc{x})= - \HIncMatrix \Vc{\flux}(\Vc{x})$  (arrows). The other panels show the time-dependent vector field $\Vc{v}_{st}(t, \Vc{x})$ induced by $\Vc{\flux}_{st}(t,\Vc{x})$: $\Vc{v}_{st}(t, \Vc{x})= - \HIncMatrix \Vc{\flux}_{st}(t,\Vc{x})$. The nonequilibrium trajectory $\{\Vc{x}_{t}\}$ is also depicted as for a reference (the black curve). In each panel, the white circle on the trajectory is $\Vc{x}_{t}$ at which $\Vc{\flux}_{st}(t,\Vc{x})$ is computed.}
}
\label{fig:Projection_ST}
\end{figure}

\TJK{\begin{ex}[Simplified Brusselator CRN\cite{feinberg2019,yoshimura2022} (continued)]\label{ex:BrusselatorProjection}
For the nonequilibrium trajectory of the Brusselator CRN in \fgref{fig:Projection_ST} (a), we numerically obtained the effective cycle flux $\Vc{\flux}_{st}(t,\Vc{x})$ and the corresponding time-dependent kinetic parameter set $\Vc{\kcoef}_{st}^{\pm}(t)$ (\fgref{fig:Projection_ST} (b,c)).
Figure \ref{fig:Projection_ST} (d) shows the vector field $\Vc{v}_{st}(t,\Vc{x})=-\HIncMatrix \Vc{\flux}_{st}(t,\Vc{x})$ induced by the effective cycle flux $\Vc{\flux}_{st}(t,\Vc{x})$.
From \fgref{fig:Projection_ST}, we can see that any point on the trajectory $\{\Vc{x}_{t}\}$ originally generated by the nonequilibrium flux $\Vc{\flux}(\Vc{x}_{t})$ can be kept steady with the time-dependent cycle flux $\Vc{\flux}_{st}(t,\Vc{x})$ realized by the modulation of the kinetic parameter $\Vc{\kcoef}^{\pm}_{st}(t)$.
\end{ex}
}

\TJK{In modern nonequilibrium thermodynamics, it has been a great challenge to establish thermodynamic characterizations for nonequilibrium phenomena. To this end, the dissection of dynamics and the corresponding flux and force has been attempted\cite{oono1998ProgressofTheoreticalPhysicsSupplement,komatsu2008Phys.Rev.Lett.,maes2014JStatPhys,dechant2022ArXiv210912817Cond-Mat,yoshimura2022}. For a given trajectory $\{\Vc{x}_{t}\}$, 
the effective time-dependent equilibrium flux $\Vc{\flux}_{eq}(t, \Vc{x})$ generates exactly the same trajectory, which dissects and mimics the dynamic aspect of the trajectory. On the other hand, the effective time-dependent cycle flux $\Vc{\flux}_{st}(t, \Vc{x})$ makes each point on the trajectory steady, which can be recognized as the nonequilibrium aspect of the trajectory.
Moreover, these two types of fluxes can be realized by appropriately modulating the kinetic parameter set $\Vc{\kcoef}^{\pm}$, which makes the dissected fluxes physically meaningful and accessible. 
More specifically, the modulation of the force part $\Vc{K}$ of the kinetic parameter is sufficient for realization (\eqnref{eq:time_dependent_Keq} and \eqnref{eq:Kst}) while the activity part $\Vc{\kappa}$ is kept constant. 
In the case of CRN, the former is linked to the free energy difference between reactants and products of each reaction, and the other is associated with the height of the energy barrier between them. This clear separation of different physical parameters in our framework is advantageous to further investigate physical aspects of dynamics on graphs and hypergraphs. Thus, the dually flat structure on the edge space and the HHK decomposition provides a new and promising way to characterize the nonequilibrium flow\footnote{It should be noted that various definitions of housekeeping and excess EPRs have been proposed in the research field of nonequilibrium thermodynamics. The definition here is just one of them. It is an open problem how to define the notion of effective entropy for nonequilibrium dynamics.}. }

\section{Summary and Discussion}\label{sec:summary}
In this work, we have shown that the doubly dual flat structure of the vertex and edge spaces on graphs and hypergraphs provides the information-geometric basis for the dynamics on graphs and hypergraphs.
Two notions of orthogonality, pseudo-Hilbert isosceles orthogonality and information-geometric orthogonality, have been introduced and shown to dissect the equilibrium and nonequilibrium aspects of the dynamics into the induced structures on the tangent and cotangent spaces and the cycle spaces.
The doubly dual flat structure naturally connects the topological information of underlying discrete manifolds, i.e., graphs and hypergraphs, with the dynamics on them and thus endows more flexibility and representation power to the information-geometric modeling of dynamics.
Furthermore, the generalized equilibrium and nonequilibrium flow, as well as the generalized flow, accommodate a sufficiently wide range of models, which include the reversible Markov jump processes on finite graphs and CRN with LMA kinetics (a class of PDS).
These results could substantially extend the applicability of information geometry to dynamical problems.

\subsection{Extension of other relations involving information measures}
While we demonstrated that the generalized flow and the doubly dual flat structure can extend several results known for FPE and diffusion processes, we still have potentially relevant results and problems that could be explained and extended in our framework.
For example, for FPE and diffusion processes, the Fisher information number $\FIN$ was extended to the relative Fisher information (also known as Hyv\"arinen divergence\cite{hyvarinen2005J.Mach.Learn.Res.,lods2015Entropy}).
The relative Fisher information of two trajectories $p^{(1)}_{t}(\boldsymbol{r})$ and  $p^{(2)}_{t}(\boldsymbol{r})$ are known to satisfy information-theoretic relations such as the De Brujin identity\cite{stam1959InformationandControl} and its extensions\cite{yamano2013Eur.Phys.J.B,wibisono20172017IEEEInt.Symp.Inf.TheoryISIT}.
In addition, the logarithmic Sobolev inequality also constitutes a relationship between the Fisher information number and the KL divergence (or Shannon information)\cite{gross1975Am.J.Math.,gross1975DukeMath.J.}.
It would be an important future problem to associate these results with the doubly dual flat structure.

Moreover, several relations potentially being related to De Giorgi's formulation (\eqnref{eq:DeGiorgi}) have been known for mutual information in filtering and control theories. 
For example, Guo, Shamai, and Verdu found a relation between mutual information and the minimum mean square error (MMSE) in Gaussian channels \cite{guo2005IEEETrans.Inf.Theory}.
Relations similar to these have also been reported by Mayer-Wolf and Zakai\cite{mayer-wolf1984Filter.ControlRandomProcess.,mayer-wolf2007Ann.Appl.Probab.}.
Our framework may offer a unified perspective behind these different types of relation involving information measures.

\subsection{Extensions of the doubly dual flat structure}
There is also room for extensions of the doubly dual flat structure.
While we consider only strictly convex thermodynamic functions and dissipation functions, the strict convexity is not necessarily required, at least for defining the generalized flow and the equilibrium and nonequilibrium flows\footnote{The strict convexity of the dissipation functions is assumed to work on the projections in the edge spaces where the bijection between $\Jspace_{\Vc{x}}$ and $\Fspace_{\Vc{x}}$ are important. In addition, $\Y=\Real^{N_{\node}}$ is assumed to make the induced dually flat structure well-defined for all given $\Vc{v}$.
Only to define the generalized flow and equilibrium and nonequilibrium flow, injectivity of $\partial \Dissp^{*}_{\Vc{x}}$ and $\partial\Pfunc$ is sufficient and  $\Y=\Real^{N_{\node}}$ is not required.}.
Actually, in terms of thermodynamics, the thermodynamic function can be non-strictly convex when a phase transition of the system occurs \cite{callen1985}. 
The loss of bijectivity via the loss of the strict convexity can happen in complicated and degenerate statistical models\cite{amari2006NeuralComputation}. 
Techniques from algebraic geometry could be employed to address such situations\cite{watanabe2009}.

Moreover, the structure introduced for rLDG and CRN may be extended to irreversible cases, where some edges have only either forward or reverse jumps or reactions.
For this purpose, we may take advantage of several results about the CB states obtained in CRN theory \cite{feinberg2019} where the reversibility is not necessarily assumed and those in stochastic thermodynamics for absolute irreversible processes \cite{murashita2014Phys.Rev.E}.

While the nonequilibrium flow is general enough to cover at least all reversible CRN with LMA kinetics, the classes of nonlinear dynamics other than CRN are much wider in general.
To further extend the range of models that can be covered, GENERIC (General Equation for Non-Equilibrium Reversible-Irreversible Coupling) would be a good candidate \cite{ottinger2005}.
GENERIC is a theoretical framework to integrate dissipative dynamics (gradient flow dynamics) and conservative dynamics (Hamiltonian dynamics).  
The extension of the generalized flow to GENERIC has already been attempted but is still ongoing \cite{renger2018Entropy,patterson2021ArXiv210314384Math-Ph}.
\TJK{One might also consider Hamiltonian-type dynamics, which differs from the GENERIC structure mentioned above. In the doubly dual flat structure, dual spaces are statically coupled by the Legendre duality. However, we could consider the coupling of two dynamics, each of which is defined on the primal and the dual spaces. Such coupling has been investigated in relation to accelerations of gradient flows\cite{wang2021JSciComput}, optimal control problems\cite{li2023J.Comput.Phys.}, and also mean field game problems\cite{gao2022}. It would be quite an interesting problem to formulate this dynamic coupling in relation to our results and also the results of GENERIC.}
The information geometry could offer new insights and techniques to achieve these missions.

\subsection{Homological algebra and differential geometric formulations}
From the viewpoint of the standard homological algebra, the doubly dual flat structure that we introduced is an extension of chain and cochain complexes with inner product structure.
Because the homological algebra used here is an abstraction of the differential form, the doubly dual flat structure can also be viewed as an extension of the differential form and might be called \textit{dually flat form}. 
It would be an interesting mission to characterize this stricture under a more rigorous mathematical formulation and to investigate if the Legendre duality can be consistently introduced for chains and cochains higher than those of the edge space.
From the viewpoint of differential geometry, the dual flat structure can be defined independently of the specific coordinate by the Hessian geometry\cite{shima2007}. 
While we stick to the standard basis of the graph or hypergraph on which the convex thermodynamic function is defined, we can formulate it more generally.
It would be an important future work to clarify how the doubly dual flat structure can be formulated from a differential geometric perspective.

\subsection{Consistency and Persistence}
Finally, we would like to mention the problem of consistency and persistence.
In this work, we presume that the flux $\Vc{\flux}(\Vc{x})$ is consistent with $\HGraph$ and that the trajectory is persistent.
The explicit conditions when these are satisfied are still elusive.
Actually, the condition for consistency is intricate, even for the separable cases.
For an illustrative example, suppose that the $i$th molecule is involved in the $e$th reaction of a CRN with LMA kinetics. 
For $x_{i} \to 0$, $j_{e}(\Vc{x}) \to 0$ holds. However, their Legendre duals diverge as $y_{i} \to -\infty$ and $f_{e}(\Vc{x}) \to -\infty$. 
The flux $j_{e}(\Vc{x})$ goes to $0$ because $\frenecy_{e}(\Vc{x})\to 0$ holds.
This example suggests that we have to consider a certain limit of relevant quantities to appropriately address the consistency condition.

The persistence of the nonequilibrium flow would be a much harder problem.
While persistence has been approached for CB CRN with LMA kinetics using techniques from algebraic geometry\cite{craciun2016}, its connection to information geometry has yet to be clarified.
Moreover, from an information-geometric viewpoint, the loss of persistence means a change in the support of probability or positive density, which is effectively accompanied by a change in the topology of the underlying graph or hypergraph.
To resolve the problem, we may need a deeper understanding of the interrelationship among dynamics, information-geometric structure, and the underlying topology.

\section{Acknowledgement}
This research is supported by JST (JPMJCR2011,JPMJCR1927) and JSPS (19H05799). The authors thank Hideyuki Miyahara and Tomonari Sei for their critical comments. We also express our sincere thanks to the anonymous reviewer who carefully read our manuscript and provided valuable comments. 
\section{Conflict of interest}
The authors have no conflict of interest, financial or otherwise.


%

\appendix
\section{Symbols, Notations, and Abbreviations}
\begin{table}[h]
  \centering
  \begin{tabular}{|c|l|l|}
\hline
$\Graph$, $\Graph_{\Vc{\kcoef}^{\pm}}$ & finite graph and edge-weighted finite graph. & \defref{dfn:graph} \\
$\node$, $\edge$ & vertex and edges of graph$\Graph$. & \defref{dfn:graph} \\
$\IncMatrix$ & incidence matrix& \defref{dfn:graph} \\
$\HGraph$, $\HGraph_{\Vc{\kcoef}^{\pm}}$ & Reversible CRN hypergraph and edge-weighted hypergraph.& \defref{dfn:rCRNhyperG}\\
$\molX$, $\hat{\node}$, $\edge$ & vertex,  hypervertex, and hyperedges of CRN hypergraph $\HGraph$. & \defref{dfn:rCRNhyperG}\\
$\HIncMatrix$ & Hypergraph incidence matrix. & \defref{dfn:rCRNhyperG}\\
$\cmMatrix$ & Hypervertex matrix& \defref{dfn:rCRNhyperG}\\
\hline\hline
$\chain_{p}(\Graph)$, $\chain^{p}(\Graph)$& $p$-chain and $p$-cochain of graph with field $\Real$ & \secref{sec:chain_cochain}\\
$\chain_{p}(\HGraph)$, $\chain^{p}(\HGraph)$ & $p$-chain and $p$-cochain of hypergraph with field $\Real$ & \defref{dfn:exactchainhypergraph}\\
$\boundary_{p}$, $\boundary^{p}$ & $p$-discrete differentials & \secref{sec:chain_cochain}, \defref{dfn:exactchainhypergraph}\\
$\Grad_{\HIncMatrix}\defeq \boundary^{0}=\HIncMatrix^{\Transpose}$ & Discrete gradient & \defref{dfn:discreteOperators}\\
$\Div_{\HIncMatrix}=\Grad^{*}_{\HIncMatrix}\defeq \boundary_{1}=\HIncMatrix$& Discrete divergence (adjoint of discrete gradient) & \defref{dfn:discreteOperators}\\
$\Curl_{\cycMatrix}\defeq \boundary^{1}=\cycMatrix^{\Transpose}$ & Discrete curl& \defref{dfn:discreteOperators}\\
$\Curl^{*}_{\cycMatrix}\defeq \boundary_{2}=\cycMatrix$ & Adjoint of discrete curl& \defref{dfn:discreteOperators}\\
\hline
  \end{tabular}
  \caption{List of symbols and notations related to graph, hypergraph, and homological algebra.}
  \label{tb:symbols1}
\end{table}

\begin{table}[h]
  \centering
  \begin{tabular}{|c|l|l|}
\hline
$\X$, $\Y$ & Primal and dual vertex spaces (density and potential spaces)& \defref{dfn:primalVspace}, \defref{dfn:dualVspace}\\
$\Vc{x}$ & The density defined on vertices of graph or hypergraph&   \defref{dfn:rLDG}, \defref{dfn:rCRNhyperG}\\
$\Vc{y}$ & The potential field defined on vertices of graph or hypergraph&   \defref{dfn:dualVspace}\\
\hline \hline
$\Jspace_{\Vc{x}}$, $\Fspace_{\Vc{x}}$& Primal and dual edge spaces & \defref{dfn:PrimalDualEdgeSpaces}\\
$\Vc{\flux}^{+}$, $\Vc{\flux}^{-}$& The forward and reverse oneway fluxes of graph or hypergraph &   \defref{dfn:rLDG}\\
$\Vc{\flux}$ & The total flux on edges of graph or hypergraph &   \defref{dfn:rLDG}\\
$\Vc{\tf}$ & The force on edges of graph or phyergraph & \eqnref{eq:LegendreTF}\\
\hline \hline
$\flux^{\pm}_{\mathrm{MA}}$, $ \Vc{\flux}^{\pm}_{\mathrm{eMA}}$ & Oneway fluxes of the normal and extended LMA kinetics. & \eqnref{eq:CRN_rate}, \eqnref{eq:CRN_rate_eLMA}\\
$\Vc{\flux}_{\mathrm{FP}}$ & Flux of the FPE & \eqnref{eq:flux_FPE}\\ 
$\Vc{\tf}_{NE}$ & Nonequilibrium force& \eqnref{eq:nonequF}\\
$\Vc{\tf}_{\mathrm{MA}}(\Vc{x};\Vc{K})$, $\Vc{\frenecy}_{\mathrm{MA}}(\Vc{x};\Vc{\kappa})$ & Force and frenetic activity of LMA kinetics & \eqnref{eq:jffMAK}\\
\hline \hline
$\Manifold^{\mathrm{ST}}$& Steady state manifold& \eqnref{eq:DBC}\\
$\Manifold^{\mathrm{CB}}$& Complex-balanced manifold& \eqnref{eq:DBC}\\
$\Manifold^{\mathrm{DB}}$& Detailed-balanced manifold& \eqnref{eq:DBC}\\
$\gLap_{\Vc{\theta}}$ & Weighted asymmetric graph Laplacian& \defref{dfn:wagLaplacian}\\
\hline
  \end{tabular}
  \caption{List of symbols and notations related to the dynamics on graph and hypergraph.}
  \label{tb:symbols2}
\end{table}

\begin{table}[h]
  \centering
  \begin{tabular}{|c|l|l|}
\hline
$\Pfunc$, $\Pfunc^{*}$ & Primal and dual thermodynamic functions & \defref{dfn:thermofunc}, \defref{dfn:dualVspace}\\
$ \BD^{\X}_{\Pfunc}[\Vc{x}\|\Vc{x}']$, $\BD^{\Y}_{\Pfunc^{*}}[\Vc{y}'\|\Vc{y}]$ & Bregman divergences on vertex spaces& \defref{dfn:BregmanDiv}\\
 $\Tan_{\Vc{x}}\X$,  $\Tan^{*}_{\Vc{x}}\X$ & Tangent and cotangent spaces of $\X$  & \defref{dfn:HessianMatrices}\\
 $\Tan_{\Vc{y}}\Y$, $\Tan^{*}_{\Vc{y}}\Y$ & Tangent and cotangent spaces of  $\Y$ & \defref{dfn:HessianMatrices}\\
$\FM_{\Vc{x}}$, $\FM_{\Vc{y}}^{*}$& Hessian matrices on vertex spaces& \defref{dfn:HessianMatrices}\\
\hline\hline
$\Dissp_{\Vc{x}}$, $\Dissp^{*}_{\Vc{x}}$& Primal and dual dissipation functions & \defref{dfn:DissipationFunc}\\
$\BD_{\Vc{x}}^{\Jspace,\Fspace}[\Vc{\flux};\Vc{\tf}']$ & Bregman divergence on the edge space & \eqnref{eq:BD}\\
$\FM_{\Vc{x},\Vc{\flux}}$, $\FM_{\Vc{x},\Vc{\tf}}^{*}$& Hessian matrices on the edge space& \eqnref{eq:HessianMatrixEdge}\\
$\Dissp^{q,*}_{\Vc{x}}(\Vc{\tf})$& Quadratic dissipation function& \eqnref{eq:quadratic_dissp}\\
\hline\hline
$\Polytope^{sc}(\Vc{x}_{0})$ & Stoichiometric subspace& \defref{dfn:stoiSubspace}\\
$\Polytope^{eq}(\tilde{\Vc{y}})$ & Equilibrium subspace & \defref{dfn:EquiSubspace}\\
$\Polytope^{vl}(\hat{\Vc{\flux}})$ & Iso-velocity subspace & \defref{dfn:IsovelSubspace}\\
$\Polytope^{fr}(\Vc{\tf}')$ & Iso-force subspace & \defref{dfn:IsoforceSubspace}\\
$\Variety^{sc}(\Vc{y}_{0})$, $\Variety^{eq}(\tilde{\Vc{x}})$ & Stoichiometric and equilibrium manifolds & \defref{dfn:StoiEquManifolds}\\
$ \Variety^{vl}_{\Vc{x}}(\hat{\Vc{\tf}})$, $\Variety^{fr}_{\Vc{x}}(\Vc{\flux}')$ & Iso-velocity and Iso-force manifolds & \defref{dfn:IsoVelForceManifolds}\\
\hline\hline
$\cManifold_{\Vc{x}}^{\Dissp}(c)$, $\cManifold_{\Vc{x}}^{\Dissp^{*}}(c)$ & Primal central affine manifolds & \defref{dfn:cAffineManifolds}\\
$ \Manifold_{\Vc{x}}^{\Dissp}(c)$, $\Manifold_{\Vc{x}}^{\Dissp^{*}}(c)$ & Dual central affine manifolds & \defref{eq:dual_central_affine_manifolds}\\
$\perp_{H}$ & Pseudo-Hilbert-isosceles orthogonality & \defref{dfn:PHorthogonality}\\
\hline\hline
$\tilde{\Tan}_{\Vc{x}} \X$, $\tilde{\Tan}_{\Vc{x}}^{*} \X$ & Induced tangent and cotangent spaces of $\X$ & \thmref{thm:inducedHessiangeometry1}\\
$\tilde{\Dissp}_{\Vc{x}}(\Vc{v})$, $\tilde{\Dissp}_{\Vc{x}}^{*}(\Vc{u})$ & Induced dissipation functions on tangent and cotangent spaces & \eqnref{eq:induced_dissp_func_tan}\\
$\Z_{\Vc{x}}$, $\Zeta_{\Vc{x}}$ & Cycle spaces & \defref{dfn:cycle_spaces}\\
$\hat{\Dissp}_{\Vc{x}}(\Vc{z})$, $\hat{\Dissp}_{\Vc{x}}^{*}(\Vc{\zeta})$ & Induced dissipation functions on cycle spaces & \eqnref{eq:induced_dissp_func_cycle}\\
\hline
  \end{tabular}
  \caption{List of symbols and notations related to the information-geometric structure.}
  \label{tb:symbols3}
\end{table}

\begin{table}[h]
  \centering
  \begin{tabular}{|c|l|}
\hline
CRN & Chemical reaction network \\
LMA & Law of mass action \\
(r)MJP & (Reversible) Markov jump process\\
(r)LDG & (Reversible) linear dynamics on graph\\
PDS  & Polynomial dynamical systems\\
FPE & Fokker Planck Equation\\
CB  & complex-balanced\\
DB  & detail-balanced\\
LDB  & local detailed balance\\
EPR  & entropy production rate\\
HHK  & Helmholtz-Hodge-Kodaira\\
MCMC  & Markov chain Monte Carlo\\
\hline
  \end{tabular}
  \caption{List of abbreviations.}
  \label{tb:abbreviations}
\end{table}


\begin{thebibliography}{176}%
\makeatletter
\providecommand \@ifxundefined [1]{%
 \@ifx{#1\undefined}
}%
\providecommand \@ifnum [1]{%
 \ifnum #1\expandafter \@firstoftwo
 \else \expandafter \@secondoftwo
 \fi
}%
\providecommand \@ifx [1]{%
 \ifx #1\expandafter \@firstoftwo
 \else \expandafter \@secondoftwo
 \fi
}%
\providecommand \natexlab [1]{#1}%
\providecommand \enquote  [1]{``#1''}%
\providecommand \bibnamefont  [1]{#1}%
\providecommand \bibfnamefont [1]{#1}%
\providecommand \citenamefont [1]{#1}%
\providecommand \href@noop [0]{\@secondoftwo}%
\providecommand \href [0]{\begingroup \@sanitize@url \@href}%
\providecommand \@href[1]{\@@startlink{#1}\@@href}%
\providecommand \@@href[1]{\endgroup#1\@@endlink}%
\providecommand \@sanitize@url [0]{\catcode `\\12\catcode `\$12\catcode
  `\&12\catcode `\#12\catcode `\^12\catcode `\_12\catcode `\%12\relax}%
\providecommand \@@startlink[1]{}%
\providecommand \@@endlink[0]{}%
\providecommand \url  [0]{\begingroup\@sanitize@url \@url }%
\providecommand \@url [1]{\endgroup\@href {#1}{\urlprefix }}%
\providecommand \urlprefix  [0]{URL }%
\providecommand \Eprint [0]{\href }%
\providecommand \doibase [0]{https://doi.org/}%
\providecommand \selectlanguage [0]{\@gobble}%
\providecommand \bibinfo  [0]{\@secondoftwo}%
\providecommand \bibfield  [0]{\@secondoftwo}%
\providecommand \translation [1]{[#1]}%
\providecommand \BibitemOpen [0]{}%
\providecommand \bibitemStop [0]{}%
\providecommand \bibitemNoStop [0]{.\EOS\space}%
\providecommand \EOS [0]{\spacefactor3000\relax}%
\providecommand \BibitemShut  [1]{\csname bibitem#1\endcsname}%
\let\auto@bib@innerbib\@empty
\bibitem [{\citenamefont {Amari}(2016)}]{amari2016}%
  \BibitemOpen
  \bibfield  {author} {\bibinfo {author} {\bibfnamefont {S.-i.}\ \bibnamefont
  {Amari}},\ }\href@noop {} {\emph {\bibinfo {title} {Information {{Geometry}}
  and {{Its Applications}}}}}\ (\bibinfo  {publisher} {{Springer}},\ \bibinfo
  {year} {2016})\BibitemShut {NoStop}%
\bibitem [{\citenamefont {Ay}\ \emph {et~al.}(2018)\citenamefont {Ay},
  \citenamefont {Gibilisco},\ and\ \citenamefont {Mat{\'u}{\v s}}}]{ay2018}%
  \BibitemOpen
  \bibfield  {author} {\bibinfo {author} {\bibfnamefont {N.}~\bibnamefont
  {Ay}}, \bibinfo {author} {\bibfnamefont {P.}~\bibnamefont {Gibilisco}},\ and\
  \bibinfo {author} {\bibfnamefont {F.}~\bibnamefont {Mat{\'u}{\v s}}},\
  }\href@noop {} {\emph {\bibinfo {title} {Information {{Geometry}} and {{Its
  Applications}}: {{On}} the {{Occasion}} of {{Shun-ichi Amari}}'s 80th
  {{Birthday}}, {{IGAIA IV Liblice}}, {{Czech Republic}}, {{June}} 2016}}}\
  (\bibinfo  {publisher} {{Springer}},\ \bibinfo {year} {2018})\BibitemShut
  {NoStop}%
\bibitem [{\citenamefont {Risken}\ and\ \citenamefont
  {Frank}(1996)}]{risken1996}%
  \BibitemOpen
  \bibfield  {author} {\bibinfo {author} {\bibfnamefont {H.}~\bibnamefont
  {Risken}}\ and\ \bibinfo {author} {\bibfnamefont {T.}~\bibnamefont {Frank}},\
  }\href@noop {} {\emph {\bibinfo {title} {The {{Fokker-Planck Equation}}:
  {{Methods}} of {{Solution}} and {{Applications}}}}}\ (\bibinfo  {publisher}
  {{Springer Science \& Business Media}},\ \bibinfo {year} {1996})\BibitemShut
  {NoStop}%
\bibitem [{\citenamefont {Horsthemke}\ and\ \citenamefont
  {Lefever}(2006)}]{horsthemke2006}%
  \BibitemOpen
  \bibfield  {author} {\bibinfo {author} {\bibfnamefont {W.}~\bibnamefont
  {Horsthemke}}\ and\ \bibinfo {author} {\bibfnamefont {R.}~\bibnamefont
  {Lefever}},\ }\href@noop {} {\emph {\bibinfo {title} {Noise-{{Induced
  Transitions}}: {{Theory}} and {{Applications}} in {{Physics}}, {{Chemistry}},
  and {{Biology}}}}}\ (\bibinfo  {publisher} {{Springer Science \& Business
  Media}},\ \bibinfo {year} {2006})\BibitemShut {NoStop}%
\bibitem [{\citenamefont {Gardiner}(2010)}]{gardiner2010}%
  \BibitemOpen
  \bibfield  {author} {\bibinfo {author} {\bibfnamefont {C.}~\bibnamefont
  {Gardiner}},\ }\href@noop {} {\emph {\bibinfo {title} {Stochastic
  {{Methods}}: {{A Handbook}} for the {{Natural}} and {{Social Sciences}}}}}\
  (\bibinfo  {publisher} {{Springer Berlin Heidelberg}},\ \bibinfo {year}
  {2010})\BibitemShut {NoStop}%
\bibitem [{\citenamefont {Murray}(2007)}]{murray2007}%
  \BibitemOpen
  \bibfield  {author} {\bibinfo {author} {\bibfnamefont {J.~D.}\ \bibnamefont
  {Murray}},\ }\href@noop {} {\emph {\bibinfo {title} {Mathematical
  {{Biology}}: {{I}}. {{An Introduction}}}}}\ (\bibinfo  {publisher} {{Springer
  Science \& Business Media}},\ \bibinfo {year} {2007})\BibitemShut {NoStop}%
\bibitem [{\citenamefont {Beard}\ and\ \citenamefont {Qian}(2008)}]{beard2008}%
  \BibitemOpen
  \bibfield  {author} {\bibinfo {author} {\bibfnamefont {D.~A.}\ \bibnamefont
  {Beard}}\ and\ \bibinfo {author} {\bibfnamefont {H.}~\bibnamefont {Qian}},\
  }\href {https://doi.org/10.1017/CBO9780511803345} {\emph {\bibinfo {title}
  {Chemical {{Biophysics}}: {{Quantitative Analysis}} of {{Cellular
  Systems}}}}},\ Cambridge {{Texts}} in {{Biomedical Engineering}}\ (\bibinfo
  {publisher} {{Cambridge University Press}},\ \bibinfo {address}
  {{Cambridge}},\ \bibinfo {year} {2008})\BibitemShut {NoStop}%
\bibitem [{\citenamefont {Feinberg}(2019)}]{feinberg2019}%
  \BibitemOpen
  \bibfield  {author} {\bibinfo {author} {\bibfnamefont {M.}~\bibnamefont
  {Feinberg}},\ }\href@noop {} {\emph {\bibinfo {title} {Foundations of
  {{Chemical Reaction Network Theory}}}}}\ (\bibinfo  {publisher}
  {{Springer}},\ \bibinfo {year} {2019})\BibitemShut {NoStop}%
\bibitem [{\citenamefont {Amari}(1987{\natexlab{a}})}]{amari1987}%
  \BibitemOpen
  \bibfield  {author} {\bibinfo {author} {\bibfnamefont {S.}~\bibnamefont
  {Amari}},\ }\href@noop {} {\emph {\bibinfo {title} {Differential {{Geometry}}
  in {{Statistical Inference}}}}}\ (\bibinfo  {publisher} {{Institute of
  mathematical Statistics}},\ \bibinfo {address} {{Hayward, Calif}},\ \bibinfo
  {year} {1987})\BibitemShut {NoStop}%
\bibitem [{\citenamefont {Ravishanker}\ \emph {et~al.}(1990)\citenamefont
  {Ravishanker}, \citenamefont {Melnick},\ and\ \citenamefont
  {Tsai}}]{ravishanker1990J.TimeSer.Anal.}%
  \BibitemOpen
  \bibfield  {author} {\bibinfo {author} {\bibfnamefont {N.}~\bibnamefont
  {Ravishanker}}, \bibinfo {author} {\bibfnamefont {E.~L.}\ \bibnamefont
  {Melnick}},\ and\ \bibinfo {author} {\bibfnamefont {C.-L.}\ \bibnamefont
  {Tsai}},\ }\bibfield  {title} {\bibinfo {title} {Differential {{Geometry}} of
  {{Arma Models}}},\ }\href
  {https://doi.org/10.1111/j.1467-9892.1990.tb00057.x} {\bibfield  {journal}
  {\bibinfo  {journal} {J. Time Ser. Anal.}\ }\textbf {\bibinfo {volume}
  {11}},\ \bibinfo {pages} {259} (\bibinfo {year} {1990})}\BibitemShut
  {NoStop}%
\bibitem [{\citenamefont {Tanaka}\ and\ \citenamefont
  {Komaki}(2011)}]{tanaka2011SankhyaA}%
  \BibitemOpen
  \bibfield  {author} {\bibinfo {author} {\bibfnamefont {F.}~\bibnamefont
  {Tanaka}}\ and\ \bibinfo {author} {\bibfnamefont {F.}~\bibnamefont
  {Komaki}},\ }\bibfield  {title} {\bibinfo {title} {Asymptotic expansion of
  the risk difference of the {{Bayesian}} spectral density in the
  autoregressive moving average model},\ }\href
  {https://doi.org/10.1007/s13171-011-0005-1} {\bibfield  {journal} {\bibinfo
  {journal} {Sankhya A}\ }\textbf {\bibinfo {volume} {73}},\ \bibinfo {pages}
  {162} (\bibinfo {year} {2011})}\BibitemShut {NoStop}%
\bibitem [{\citenamefont
  {Amari}(1987{\natexlab{b}})}]{amari1987Math.SystemsTheory}%
  \BibitemOpen
  \bibfield  {author} {\bibinfo {author} {\bibfnamefont {S.-i.}\ \bibnamefont
  {Amari}},\ }\bibfield  {title} {\bibinfo {title} {Differential geometry of a
  parametric family of invertible linear systems--{{Riemannian}} metric, dual
  affine connections, and divergence},\ }\href
  {https://doi.org/10.1007/BF01692059} {\bibfield  {journal} {\bibinfo
  {journal} {Math. Systems Theory}\ }\textbf {\bibinfo {volume} {20}},\
  \bibinfo {pages} {53} (\bibinfo {year} {1987}{\natexlab{b}})}\BibitemShut
  {NoStop}%
\bibitem [{\citenamefont {Amari}(2001)}]{amari2001IEEETrans.Inf.Theory}%
  \BibitemOpen
  \bibfield  {author} {\bibinfo {author} {\bibfnamefont {S.-I.}\ \bibnamefont
  {Amari}},\ }\bibfield  {title} {\bibinfo {title} {Information geometry on
  hierarchy of probability distributions},\ }\href
  {https://doi.org/10.1109/18.930911} {\bibfield  {journal} {\bibinfo
  {journal} {IEEE Trans. Inf. Theory}\ }\textbf {\bibinfo {volume} {47}},\
  \bibinfo {pages} {1701} (\bibinfo {year} {2001})}\BibitemShut {NoStop}%
\bibitem [{\citenamefont {Nakagawa}\ and\ \citenamefont
  {Kanaya}(1993)}]{nakagawa1993IEEETrans.Inf.Theory}%
  \BibitemOpen
  \bibfield  {author} {\bibinfo {author} {\bibfnamefont {K.}~\bibnamefont
  {Nakagawa}}\ and\ \bibinfo {author} {\bibfnamefont {F.}~\bibnamefont
  {Kanaya}},\ }\bibfield  {title} {\bibinfo {title} {On the converse theorem in
  statistical hypothesis testing for {{Markov}} chains},\ }\href
  {https://doi.org/10.1109/18.212294} {\bibfield  {journal} {\bibinfo
  {journal} {IEEE Trans. Inf. Theory}\ }\textbf {\bibinfo {volume} {39}},\
  \bibinfo {pages} {629} (\bibinfo {year} {1993})}\BibitemShut {NoStop}%
\bibitem [{\citenamefont {Takeuchi}\ and\ \citenamefont
  {Barron}(1998)}]{takeuchi1998Proc.1998IEEEInt.Symp.Inf.TheoryCatNo98CH36252}%
  \BibitemOpen
  \bibfield  {author} {\bibinfo {author} {\bibfnamefont {J.}~\bibnamefont
  {Takeuchi}}\ and\ \bibinfo {author} {\bibfnamefont {A.}~\bibnamefont
  {Barron}},\ }\bibfield  {title} {\bibinfo {title} {Asymptotically minimax
  regret by {{Bayes}} mixtures},\ }in\ \href
  {https://doi.org/10.1109/ISIT.1998.708923} {\emph {\bibinfo {booktitle}
  {Proc. 1998 {{IEEE Int}}. {{Symp}}. {{Inf}}. {{Theory Cat No98CH36252}}}}}\
  (\bibinfo {year} {1998})\ pp.\ \bibinfo {pages} {318--}\BibitemShut {NoStop}%
\bibitem [{\citenamefont {Nagaoka}(2017)}]{nagaoka2017}%
  \BibitemOpen
  \bibfield  {author} {\bibinfo {author} {\bibfnamefont {H.}~\bibnamefont
  {Nagaoka}},\ }\href {https://doi.org/10.48550/arXiv.1701.06119} {\bibinfo
  {title} {The exponential family of {{Markov}} chains and its information
  geometry}} (\bibinfo {year} {2017}),\ \Eprint
  {https://arxiv.org/abs/1701.06119} {arxiv:1701.06119 [cs, math]} \BibitemShut
  {NoStop}%
\bibitem [{\citenamefont {Takeuchi}\ and\ \citenamefont
  {Kawabata}(2007)}]{takeuchi20072007IEEEInt.Symp.Inf.Theory}%
  \BibitemOpen
  \bibfield  {author} {\bibinfo {author} {\bibfnamefont {J.}~\bibnamefont
  {Takeuchi}}\ and\ \bibinfo {author} {\bibfnamefont {T.}~\bibnamefont
  {Kawabata}},\ }\bibfield  {title} {\bibinfo {title} {Exponential
  {{Curvature}} of {{Markov Models}}},\ }in\ \href
  {https://doi.org/10.1109/ISIT.2007.4557657} {\emph {\bibinfo {booktitle}
  {2007 {{IEEE Int}}. {{Symp}}. {{Inf}}. {{Theory}}}}}\ (\bibinfo {year}
  {2007})\ pp.\ \bibinfo {pages} {2891--2895}\BibitemShut {NoStop}%
\bibitem [{\citenamefont {Hayashi}\ and\ \citenamefont
  {Watanabe}(2016)}]{hayashi2016Ann.Stat.}%
  \BibitemOpen
  \bibfield  {author} {\bibinfo {author} {\bibfnamefont {M.}~\bibnamefont
  {Hayashi}}\ and\ \bibinfo {author} {\bibfnamefont {S.}~\bibnamefont
  {Watanabe}},\ }\bibfield  {title} {\bibinfo {title} {Information geometry
  approach to parameter estimation in {{Markov}} chains},\ }\href
  {https://doi.org/10.1214/15-AOS1420} {\bibfield  {journal} {\bibinfo
  {journal} {Ann. Stat.}\ }\textbf {\bibinfo {volume} {44}},\ \bibinfo {pages}
  {1495} (\bibinfo {year} {2016})}\BibitemShut {NoStop}%
\bibitem [{\citenamefont {Wolfer}\ and\ \citenamefont
  {Watanabe}(2021)}]{wolfer2021Info.Geo.}%
  \BibitemOpen
  \bibfield  {author} {\bibinfo {author} {\bibfnamefont {G.}~\bibnamefont
  {Wolfer}}\ and\ \bibinfo {author} {\bibfnamefont {S.}~\bibnamefont
  {Watanabe}},\ }\bibfield  {title} {\bibinfo {title} {Information {{Geometry}}
  of {{Reversible Markov Chains}}},\ }\href
  {https://doi.org/10.1007/s41884-021-00061-7} {\bibfield  {journal} {\bibinfo
  {journal} {Info. Geo.}\ }\textbf {\bibinfo {volume} {4}},\ \bibinfo {pages}
  {393} (\bibinfo {year} {2021})}\BibitemShut {NoStop}%
\bibitem [{\citenamefont {Pistone}\ and\ \citenamefont
  {Rogantin}(2013)}]{pistone2013AnnInstStatMath}%
  \BibitemOpen
  \bibfield  {author} {\bibinfo {author} {\bibfnamefont {G.}~\bibnamefont
  {Pistone}}\ and\ \bibinfo {author} {\bibfnamefont {M.~P.}\ \bibnamefont
  {Rogantin}},\ }\bibfield  {title} {\bibinfo {title} {The algebra of
  reversible {{Markov}} chains},\ }\href
  {https://doi.org/10.1007/s10463-012-0368-7} {\bibfield  {journal} {\bibinfo
  {journal} {Ann Inst Stat Math}\ }\textbf {\bibinfo {volume} {65}},\ \bibinfo
  {pages} {269} (\bibinfo {year} {2013})}\BibitemShut {NoStop}%
\bibitem [{\citenamefont {Obata}\ \emph {et~al.}(1992)\citenamefont {Obata},
  \citenamefont {Hara},\ and\ \citenamefont {Endo}}]{obata1992Phys.Rev.A}%
  \BibitemOpen
  \bibfield  {author} {\bibinfo {author} {\bibfnamefont {T.}~\bibnamefont
  {Obata}}, \bibinfo {author} {\bibfnamefont {H.}~\bibnamefont {Hara}},\ and\
  \bibinfo {author} {\bibfnamefont {K.}~\bibnamefont {Endo}},\ }\bibfield
  {title} {\bibinfo {title} {Differential geometry of nonequilibrium
  processes},\ }\href {https://doi.org/10.1103/PhysRevA.45.6997} {\bibfield
  {journal} {\bibinfo  {journal} {Phys. Rev. A}\ }\textbf {\bibinfo {volume}
  {45}},\ \bibinfo {pages} {6997} (\bibinfo {year} {1992})}\BibitemShut
  {NoStop}%
\bibitem [{\citenamefont {Ohara}(2009)}]{ohara2009Eur.Phys.J.B}%
  \BibitemOpen
  \bibfield  {author} {\bibinfo {author} {\bibfnamefont {A.}~\bibnamefont
  {Ohara}},\ }\bibfield  {title} {\bibinfo {title} {Geometric study for the
  {{Legendre}} duality of generalized entropies and its application to the
  porous medium equation},\ }\href {https://doi.org/10.1140/epjb/e2009-00170-y}
  {\bibfield  {journal} {\bibinfo  {journal} {Eur. Phys. J. B}\ }\textbf
  {\bibinfo {volume} {70}},\ \bibinfo {pages} {15} (\bibinfo {year}
  {2009})}\BibitemShut {NoStop}%
\bibitem [{\citenamefont {Ohara}\ and\ \citenamefont
  {Zhang}(2021)}]{ohara2021Geom.Sci.Inf.}%
  \BibitemOpen
  \bibfield  {author} {\bibinfo {author} {\bibfnamefont {A.}~\bibnamefont
  {Ohara}}\ and\ \bibinfo {author} {\bibfnamefont {X.}~\bibnamefont {Zhang}},\
  }\bibfield  {title} {\bibinfo {title} {Properties of {{Nonlinear Diffusion
  Equations}} on {{Networks}} and {{Their Geometric Aspects}}},\ }in\ \href
  {https://doi.org/10.1007/978-3-030-80209-7_79} {\emph {\bibinfo {booktitle}
  {Geom. {{Sci}}. {{Inf}}.}}},\ \bibinfo {series and number} {Lecture {{Notes}}
  in {{Computer Science}}},\ \bibinfo {editor} {edited by\ \bibinfo {editor}
  {\bibfnamefont {F.}~\bibnamefont {Nielsen}}\ and\ \bibinfo {editor}
  {\bibfnamefont {F.}~\bibnamefont {Barbaresco}}}\ (\bibinfo  {publisher}
  {{Springer International Publishing}},\ \bibinfo {address} {{Cham}},\
  \bibinfo {year} {2021})\ pp.\ \bibinfo {pages} {736--743}\BibitemShut
  {NoStop}%
\bibitem [{\citenamefont
  {Nakamura}(1993)}]{nakamura1993JapanJ.Indust.Appl.Math.}%
  \BibitemOpen
  \bibfield  {author} {\bibinfo {author} {\bibfnamefont {Y.}~\bibnamefont
  {Nakamura}},\ }\bibfield  {title} {\bibinfo {title} {Completely integrable
  gradient systems on the manifolds of {{Gaussian}} and multinomial
  distributions},\ }\href {https://doi.org/10.1007/BF03167571} {\bibfield
  {journal} {\bibinfo  {journal} {Japan J. Indust. Appl. Math.}\ }\textbf
  {\bibinfo {volume} {10}},\ \bibinfo {pages} {179} (\bibinfo {year}
  {1993})}\BibitemShut {NoStop}%
\bibitem [{\citenamefont {Fujiwara}\ and\ \citenamefont
  {Amari}(1995)}]{fujiwara1995PhysicaD:NonlinearPhenomena}%
  \BibitemOpen
  \bibfield  {author} {\bibinfo {author} {\bibfnamefont {A.}~\bibnamefont
  {Fujiwara}}\ and\ \bibinfo {author} {\bibfnamefont {S.-i.}\ \bibnamefont
  {Amari}},\ }\bibfield  {title} {\bibinfo {title} {Gradient systems in view of
  information geometry},\ }\href {https://doi.org/10.1016/0167-2789(94)00175-P}
  {\bibfield  {journal} {\bibinfo  {journal} {Physica D: Nonlinear Phenomena}\
  }\textbf {\bibinfo {volume} {80}},\ \bibinfo {pages} {317} (\bibinfo {year}
  {1995})}\BibitemShut {NoStop}%
\bibitem [{\citenamefont {Felice}\ and\ \citenamefont {Ay}(2018)}]{felice2018}%
  \BibitemOpen
  \bibfield  {author} {\bibinfo {author} {\bibfnamefont {D.}~\bibnamefont
  {Felice}}\ and\ \bibinfo {author} {\bibfnamefont {N.}~\bibnamefont {Ay}},\
  }\href {https://doi.org/10.48550/arXiv.1812.04461} {\bibinfo {title}
  {Dynamical {{Systems}} induced by {{Canonical Divergence}} in dually flat
  manifolds}} (\bibinfo {year} {2018}),\ \Eprint
  {https://arxiv.org/abs/1812.04461} {arxiv:1812.04461 [math-ph]} \BibitemShut
  {NoStop}%
\bibitem [{\citenamefont {Goto}\ and\ \citenamefont
  {Wada}(2018)}]{goto2018J.Phys.A:Math.Theor.}%
  \BibitemOpen
  \bibfield  {author} {\bibinfo {author} {\bibfnamefont {S.-i.}\ \bibnamefont
  {Goto}}\ and\ \bibinfo {author} {\bibfnamefont {T.}~\bibnamefont {Wada}},\
  }\bibfield  {title} {\bibinfo {title} {Hessian\textendash information
  geometric formulation of {{Hamiltonian}} systems and generalized {{Toda}}'s
  dual transform},\ }\href {https://doi.org/10.1088/1751-8121/aacbdf}
  {\bibfield  {journal} {\bibinfo  {journal} {J. Phys. A: Math. Theor.}\
  }\textbf {\bibinfo {volume} {51}},\ \bibinfo {pages} {324001} (\bibinfo
  {year} {2018})}\BibitemShut {NoStop}%
\bibitem [{\citenamefont {Chirco}\ \emph {et~al.}(2022)\citenamefont {Chirco},
  \citenamefont {Malag{\`o}},\ and\ \citenamefont
  {Pistone}}]{chirco2022Int.J.Geom.MethodsMod.Phys.}%
  \BibitemOpen
  \bibfield  {author} {\bibinfo {author} {\bibfnamefont {G.}~\bibnamefont
  {Chirco}}, \bibinfo {author} {\bibfnamefont {L.}~\bibnamefont {Malag{\`o}}},\
  and\ \bibinfo {author} {\bibfnamefont {G.}~\bibnamefont {Pistone}},\
  }\bibfield  {title} {\bibinfo {title} {Lagrangian and {{Hamiltonian
  Mechanics}} for {{Probabilities}} on the {{Statistical Manifold}}},\ }\href
  {https://doi.org/10.1142/S0219887822502140} {\bibfield  {journal} {\bibinfo
  {journal} {Int. J. Geom. Methods Mod. Phys.}\ ,\ \bibinfo {pages} {2250214}}
  (\bibinfo {year} {2022})},\ \Eprint {https://arxiv.org/abs/2009.09431}
  {arxiv:2009.09431 [hep-th, stat]} \BibitemShut {NoStop}%
\bibitem [{\citenamefont {Ihara}(1993)}]{ihara1993}%
  \BibitemOpen
  \bibfield  {author} {\bibinfo {author} {\bibfnamefont {S.}~\bibnamefont
  {Ihara}},\ }\href@noop {} {\emph {\bibinfo {title} {Information {{Theory}}
  for {{Continuous Systems}}}}}\ (\bibinfo  {publisher} {{World Scientific}},\
  \bibinfo {year} {1993})\BibitemShut {NoStop}%
\bibitem [{\citenamefont {Brigo}\ \emph {et~al.}(1999)\citenamefont {Brigo},
  \citenamefont {Hanzon},\ and\ \citenamefont {Gland}}]{brigo1999Bernoulli}%
  \BibitemOpen
  \bibfield  {author} {\bibinfo {author} {\bibfnamefont {D.}~\bibnamefont
  {Brigo}}, \bibinfo {author} {\bibfnamefont {B.}~\bibnamefont {Hanzon}},\ and\
  \bibinfo {author} {\bibfnamefont {F.~L.}\ \bibnamefont {Gland}},\ }\bibfield
  {title} {\bibinfo {title} {Approximate nonlinear filtering by projection on
  exponential manifolds of densities},\ }\href@noop {} {\bibfield  {journal}
  {\bibinfo  {journal} {Bernoulli}\ }\textbf {\bibinfo {volume} {5}},\ \bibinfo
  {pages} {495} (\bibinfo {year} {1999})}\BibitemShut {NoStop}%
\bibitem [{\citenamefont {Newton}(2018)}]{newton2018Inf.Geom.ItsAppl.}%
  \BibitemOpen
  \bibfield  {author} {\bibinfo {author} {\bibfnamefont {N.~J.}\ \bibnamefont
  {Newton}},\ }\bibfield  {title} {\bibinfo {title} {Nonlinear {{Filtering}}
  and {{Information Geometry}}: {{A Hilbert Manifold Approach}}},\ }in\ \href
  {https://doi.org/10.1007/978-3-319-97798-0_7} {\emph {\bibinfo {booktitle}
  {Inf. {{Geom}}. {{Its Appl}}.}}},\ \bibinfo {series and number} {Springer
  {{Proceedings}} in {{Mathematics}} \& {{Statistics}}},\ \bibinfo {editor}
  {edited by\ \bibinfo {editor} {\bibfnamefont {N.}~\bibnamefont {Ay}},
  \bibinfo {editor} {\bibfnamefont {P.}~\bibnamefont {Gibilisco}},\ and\
  \bibinfo {editor} {\bibfnamefont {F.}~\bibnamefont {Mat{\'u}{\v s}}}}\
  (\bibinfo  {publisher} {{Springer International Publishing}},\ \bibinfo
  {address} {{Cham}},\ \bibinfo {year} {2018})\ pp.\ \bibinfo {pages}
  {189--208}\BibitemShut {NoStop}%
\bibitem [{\citenamefont {Fleming}\ and\ \citenamefont
  {Mitter}(1982)}]{fleming1982Stochastics}%
  \BibitemOpen
  \bibfield  {author} {\bibinfo {author} {\bibfnamefont {W.~H.}\ \bibnamefont
  {Fleming}}\ and\ \bibinfo {author} {\bibfnamefont {S.~K.}\ \bibnamefont
  {Mitter}},\ }\bibfield  {title} {\bibinfo {title} {Optimal {{Control}} and
  {{Nonlinear Filtering}} for {{Nondegenerate Diffusion Processes}}},\ }\href
  {https://doi.org/10.1080/17442508208833228} {\bibfield  {journal} {\bibinfo
  {journal} {Stochastics}\ }\textbf {\bibinfo {volume} {8}},\ \bibinfo {pages}
  {63} (\bibinfo {year} {1982})}\BibitemShut {NoStop}%
\bibitem [{\citenamefont {Todorov}(2009)}]{todorov2009Proc.Natl.Acad.Sci.}%
  \BibitemOpen
  \bibfield  {author} {\bibinfo {author} {\bibfnamefont {E.}~\bibnamefont
  {Todorov}},\ }\bibfield  {title} {\bibinfo {title} {Efficient computation of
  optimal actions},\ }\href {https://doi.org/10.1073/pnas.0710743106}
  {\bibfield  {journal} {\bibinfo  {journal} {Proc. Natl. Acad. Sci.}\ }\textbf
  {\bibinfo {volume} {106}},\ \bibinfo {pages} {11478} (\bibinfo {year}
  {2009})}\BibitemShut {NoStop}%
\bibitem [{\citenamefont {Theodorou}\ and\ \citenamefont
  {Todorov}(2012)}]{theodorou20122012IEEE51stIEEEConf.Decis.ControlCDC}%
  \BibitemOpen
  \bibfield  {author} {\bibinfo {author} {\bibfnamefont {E.~A.}\ \bibnamefont
  {Theodorou}}\ and\ \bibinfo {author} {\bibfnamefont {E.}~\bibnamefont
  {Todorov}},\ }\bibfield  {title} {\bibinfo {title} {Relative entropy and free
  energy dualities: {{Connections}} to {{Path Integral}} and {{KL}} control},\
  }in\ \href {https://doi.org/10.1109/CDC.2012.6426381} {\emph {\bibinfo
  {booktitle} {2012 {{IEEE}} 51st {{IEEE Conf}}. {{Decis}}. {{Control CDC}}}}}\
  (\bibinfo {year} {2012})\ pp.\ \bibinfo {pages} {1466--1473}\BibitemShut
  {NoStop}%
\bibitem [{\citenamefont {Jaynes}(1957)}]{jaynes1957Phys.Rev.}%
  \BibitemOpen
  \bibfield  {author} {\bibinfo {author} {\bibfnamefont {E.~T.}\ \bibnamefont
  {Jaynes}},\ }\bibfield  {title} {\bibinfo {title} {Information {{Theory}} and
  {{Statistical Mechanics}}},\ }\href {https://doi.org/10.1103/PhysRev.106.620}
  {\bibfield  {journal} {\bibinfo  {journal} {Phys. Rev.}\ }\textbf {\bibinfo
  {volume} {106}},\ \bibinfo {pages} {620} (\bibinfo {year}
  {1957})}\BibitemShut {NoStop}%
\bibitem [{\citenamefont {Frieden}(2004)}]{frieden2004}%
  \BibitemOpen
  \bibfield  {author} {\bibinfo {author} {\bibfnamefont {B.~R.}\ \bibnamefont
  {Frieden}},\ }\href {https://doi.org/10.1017/CBO9780511616907} {\emph
  {\bibinfo {title} {Science from {{Fisher Information}}: {{A Unification}}}}}\
  (\bibinfo  {publisher} {{Cambridge University Press}},\ \bibinfo {address}
  {{Cambridge}},\ \bibinfo {year} {2004})\BibitemShut {NoStop}%
\bibitem [{\citenamefont {Sagawa}(2012)}]{sagawa2012}%
  \BibitemOpen
  \bibfield  {author} {\bibinfo {author} {\bibfnamefont {T.}~\bibnamefont
  {Sagawa}},\ }\href@noop {} {\emph {\bibinfo {title} {Thermodynamics of
  {{Information Processing}} in {{Small Systems}}}}}\ (\bibinfo  {publisher}
  {{Springer Science \& Business Media}},\ \bibinfo {year} {2012})\BibitemShut
  {NoStop}%
\bibitem [{\citenamefont {Kullback}\ and\ \citenamefont
  {Leibler}(1951)}]{kullback1951Ann.Math.Stat.}%
  \BibitemOpen
  \bibfield  {author} {\bibinfo {author} {\bibfnamefont {S.}~\bibnamefont
  {Kullback}}\ and\ \bibinfo {author} {\bibfnamefont {R.~A.}\ \bibnamefont
  {Leibler}},\ }\bibfield  {title} {\bibinfo {title} {On {{Information}} and
  {{Sufficiency}}},\ }\href {https://doi.org/10.1214/aoms/1177729694}
  {\bibfield  {journal} {\bibinfo  {journal} {Ann. Math. Stat.}\ }\textbf
  {\bibinfo {volume} {22}},\ \bibinfo {pages} {79} (\bibinfo {year}
  {1951})}\BibitemShut {NoStop}%
\bibitem [{\citenamefont {Lebowitz}\ and\ \citenamefont
  {Bergmann}(1957)}]{lebowitz1957AnnalsofPhysics}%
  \BibitemOpen
  \bibfield  {author} {\bibinfo {author} {\bibfnamefont {J.~L.}\ \bibnamefont
  {Lebowitz}}\ and\ \bibinfo {author} {\bibfnamefont {P.~G.}\ \bibnamefont
  {Bergmann}},\ }\bibfield  {title} {\bibinfo {title} {Irreversible gibbsian
  ensembles},\ }\href {https://doi.org/10.1016/0003-4916(57)90002-7} {\bibfield
   {journal} {\bibinfo  {journal} {Annals of Physics}\ }\textbf {\bibinfo
  {volume} {1}},\ \bibinfo {pages} {1} (\bibinfo {year} {1957})}\BibitemShut
  {NoStop}%
\bibitem [{\citenamefont {Shear}(1967)}]{shear1967JournalofTheoreticalBiology}%
  \BibitemOpen
  \bibfield  {author} {\bibinfo {author} {\bibfnamefont {D.}~\bibnamefont
  {Shear}},\ }\bibfield  {title} {\bibinfo {title} {An analog of the
  {{Boltzmann H-theorem}} (a {{Liapunov}} function) for systems of coupled
  chemical reactions},\ }\href {https://doi.org/10.1016/0022-5193(67)90005-7}
  {\bibfield  {journal} {\bibinfo  {journal} {Journal of Theoretical Biology}\
  }\textbf {\bibinfo {volume} {16}},\ \bibinfo {pages} {212} (\bibinfo {year}
  {1967})}\BibitemShut {NoStop}%
\bibitem [{\citenamefont {Horn}\ and\ \citenamefont
  {Jackson}(1972)}]{horn1972Arch.RationalMech.Anal.}%
  \BibitemOpen
  \bibfield  {author} {\bibinfo {author} {\bibfnamefont {F.}~\bibnamefont
  {Horn}}\ and\ \bibinfo {author} {\bibfnamefont {R.}~\bibnamefont {Jackson}},\
  }\bibfield  {title} {\bibinfo {title} {General mass action kinetics},\ }\href
  {https://doi.org/10.1007/BF00251225} {\bibfield  {journal} {\bibinfo
  {journal} {Arch. Rational Mech. Anal.}\ }\textbf {\bibinfo {volume} {47}},\
  \bibinfo {pages} {81} (\bibinfo {year} {1972})}\BibitemShut {NoStop}%
\bibitem [{\citenamefont {Goh}(1977)}]{goh1977Am.Nat.}%
  \BibitemOpen
  \bibfield  {author} {\bibinfo {author} {\bibfnamefont {B.~S.}\ \bibnamefont
  {Goh}},\ }\bibfield  {title} {\bibinfo {title} {Global {{Stability}} in
  {{Many-Species Systems}}},\ }\href {https://doi.org/10.1086/283144}
  {\bibfield  {journal} {\bibinfo  {journal} {Am. Nat.}\ }\textbf {\bibinfo
  {volume} {111}},\ \bibinfo {pages} {135} (\bibinfo {year}
  {1977})}\BibitemShut {NoStop}%
\bibitem [{\citenamefont {Figueiredo}\ \emph {et~al.}(2000)\citenamefont
  {Figueiredo}, \citenamefont {Gl{\'e}ria},\ and\ \citenamefont
  {Rocha~Filho}}]{figueiredo2000PhysicsLettersA}%
  \BibitemOpen
  \bibfield  {author} {\bibinfo {author} {\bibfnamefont {A.}~\bibnamefont
  {Figueiredo}}, \bibinfo {author} {\bibfnamefont {I.~M.}\ \bibnamefont
  {Gl{\'e}ria}},\ and\ \bibinfo {author} {\bibfnamefont {T.~M.}\ \bibnamefont
  {Rocha~Filho}},\ }\bibfield  {title} {\bibinfo {title} {Boundedness of
  solutions and {{Lyapunov}} functions in quasi-polynomial systems},\ }\href
  {https://doi.org/10.1016/S0375-9601(00)00175-4} {\bibfield  {journal}
  {\bibinfo  {journal} {Physics Letters A}\ }\textbf {\bibinfo {volume}
  {268}},\ \bibinfo {pages} {335} (\bibinfo {year} {2000})}\BibitemShut
  {NoStop}%
\bibitem [{\citenamefont {Gibbs}(2010)}]{gibbs2010}%
  \BibitemOpen
  \bibfield  {author} {\bibinfo {author} {\bibfnamefont {J.~W.}\ \bibnamefont
  {Gibbs}},\ }\href {https://doi.org/10.1017/CBO9780511686948} {\emph {\bibinfo
  {title} {Elementary {{Principles}} in {{Statistical Mechanics}}:
  {{Developed}} with {{Especial Reference}} to the {{Rational Foundation}} of
  {{Thermodynamics}}}}},\ Cambridge {{Library Collection}} - {{Mathematics}}\
  (\bibinfo  {publisher} {{Cambridge University Press}},\ \bibinfo {address}
  {{Cambridge}},\ \bibinfo {year} {2010})\BibitemShut {NoStop}%
\bibitem [{\citenamefont {Waage}\ and\ \citenamefont
  {Gulberg}(1986)}]{waage1986J.Chem.Educ.a}%
  \BibitemOpen
  \bibfield  {author} {\bibinfo {author} {\bibfnamefont {P.}~\bibnamefont
  {Waage}}\ and\ \bibinfo {author} {\bibfnamefont {C.~M.}\ \bibnamefont
  {Gulberg}},\ }\bibfield  {title} {\bibinfo {title} {Studies concerning
  affinity},\ }\href {https://doi.org/10.1021/ed063p1044} {\bibfield  {journal}
  {\bibinfo  {journal} {J. Chem. Educ.}\ }\textbf {\bibinfo {volume} {63}},\
  \bibinfo {pages} {1044} (\bibinfo {year} {1986})}\BibitemShut {NoStop}%
\bibitem [{\citenamefont {Ge}\ and\ \citenamefont
  {Qian}(2016)}]{ge2016ChemicalPhysics}%
  \BibitemOpen
  \bibfield  {author} {\bibinfo {author} {\bibfnamefont {H.}~\bibnamefont
  {Ge}}\ and\ \bibinfo {author} {\bibfnamefont {H.}~\bibnamefont {Qian}},\
  }\bibfield  {title} {\bibinfo {title} {Nonequilibrium thermodynamic formalism
  of nonlinear chemical reaction systems with
  {{Waage}}\textendash{{Guldberg}}'s law of mass action},\ }\href
  {https://doi.org/10.1016/j.chemphys.2016.03.026} {\bibfield  {journal}
  {\bibinfo  {journal} {Chemical Physics}\ }\textbf {\bibinfo {volume} {472}},\
  \bibinfo {pages} {241} (\bibinfo {year} {2016})}\BibitemShut {NoStop}%
\bibitem [{\citenamefont {Rao}\ and\ \citenamefont
  {Esposito}(2016)}]{rao2016Phys.Rev.X}%
  \BibitemOpen
  \bibfield  {author} {\bibinfo {author} {\bibfnamefont {R.}~\bibnamefont
  {Rao}}\ and\ \bibinfo {author} {\bibfnamefont {M.}~\bibnamefont {Esposito}},\
  }\bibfield  {title} {\bibinfo {title} {Nonequilibrium {{Thermodynamics}} of
  {{Chemical Reaction Networks}}: {{Wisdom}} from {{Stochastic
  Thermodynamics}}},\ }\href {https://doi.org/10.1103/PhysRevX.6.041064}
  {\bibfield  {journal} {\bibinfo  {journal} {Phys. Rev. X}\ }\textbf {\bibinfo
  {volume} {6}},\ \bibinfo {pages} {041064} (\bibinfo {year}
  {2016})}\BibitemShut {NoStop}%
\bibitem [{\citenamefont {Sughiyama}\ \emph
  {et~al.}(2021{\natexlab{a}})\citenamefont {Sughiyama}, \citenamefont
  {Loutchko}, \citenamefont {Kamimura},\ and\ \citenamefont
  {Kobayashi}}]{sughiyama2021ArXiv211212403Cond-MatPhysicsphysics}%
  \BibitemOpen
  \bibfield  {author} {\bibinfo {author} {\bibfnamefont {Y.}~\bibnamefont
  {Sughiyama}}, \bibinfo {author} {\bibfnamefont {D.}~\bibnamefont {Loutchko}},
  \bibinfo {author} {\bibfnamefont {A.}~\bibnamefont {Kamimura}},\ and\
  \bibinfo {author} {\bibfnamefont {T.~J.}\ \bibnamefont {Kobayashi}},\
  }\bibfield  {title} {\bibinfo {title} {A {{Hessian Geometric Structure}} of
  {{Chemical Thermodynamic Systems}} with {{Stoichiometric Constraints}}},\
  }\href@noop {} {\bibfield  {journal} {\bibinfo  {journal} {ArXiv211212403
  Cond-Mat Physicsphysics}\ } (\bibinfo {year} {2021}{\natexlab{a}})},\ \Eprint
  {https://arxiv.org/abs/2112.12403} {arxiv:2112.12403 [cond-mat,
  physics:physics]} \BibitemShut {NoStop}%
\bibitem [{\citenamefont {Seifert}(2012)}]{seifert2012Rep.Prog.Phys.}%
  \BibitemOpen
  \bibfield  {author} {\bibinfo {author} {\bibfnamefont {U.}~\bibnamefont
  {Seifert}},\ }\bibfield  {title} {\bibinfo {title} {Stochastic
  thermodynamics, fluctuation theorems and molecular machines},\ }\href
  {https://doi.org/10.1088/0034-4885/75/12/126001} {\bibfield  {journal}
  {\bibinfo  {journal} {Rep. Prog. Phys.}\ }\textbf {\bibinfo {volume} {75}},\
  \bibinfo {pages} {126001} (\bibinfo {year} {2012})}\BibitemShut {NoStop}%
\bibitem [{\citenamefont {Ito}(2018)}]{ito2018Phys.Rev.Lett.}%
  \BibitemOpen
  \bibfield  {author} {\bibinfo {author} {\bibfnamefont {S.}~\bibnamefont
  {Ito}},\ }\bibfield  {title} {\bibinfo {title} {Stochastic {{Thermodynamic
  Interpretation}} of {{Information Geometry}}},\ }\href
  {https://doi.org/10.1103/PhysRevLett.121.030605} {\bibfield  {journal}
  {\bibinfo  {journal} {Phys. Rev. Lett.}\ }\textbf {\bibinfo {volume} {121}},\
  \bibinfo {pages} {030605} (\bibinfo {year} {2018})}\BibitemShut {NoStop}%
\bibitem [{\citenamefont {Kolchinsky}\ and\ \citenamefont
  {Wolpert}(2021)}]{kolchinsky2021Phys.Rev.X}%
  \BibitemOpen
  \bibfield  {author} {\bibinfo {author} {\bibfnamefont {A.}~\bibnamefont
  {Kolchinsky}}\ and\ \bibinfo {author} {\bibfnamefont {D.~H.}\ \bibnamefont
  {Wolpert}},\ }\bibfield  {title} {\bibinfo {title} {Work, {{Entropy
  Production}}, and {{Thermodynamics}} of {{Information}} under {{Protocol
  Constraints}}},\ }\href {https://doi.org/10.1103/PhysRevX.11.041024}
  {\bibfield  {journal} {\bibinfo  {journal} {Phys. Rev. X}\ }\textbf {\bibinfo
  {volume} {11}},\ \bibinfo {pages} {041024} (\bibinfo {year}
  {2021})}\BibitemShut {NoStop}%
\bibitem [{\citenamefont {Yoshimura}\ and\ \citenamefont
  {Ito}(2021)}]{yoshimura2021Phys.Rev.Research}%
  \BibitemOpen
  \bibfield  {author} {\bibinfo {author} {\bibfnamefont {K.}~\bibnamefont
  {Yoshimura}}\ and\ \bibinfo {author} {\bibfnamefont {S.}~\bibnamefont
  {Ito}},\ }\bibfield  {title} {\bibinfo {title} {Information geometric
  inequalities of chemical thermodynamics},\ }\href
  {https://doi.org/10.1103/PhysRevResearch.3.013175} {\bibfield  {journal}
  {\bibinfo  {journal} {Phys. Rev. Research}\ }\textbf {\bibinfo {volume}
  {3}},\ \bibinfo {pages} {013175} (\bibinfo {year} {2021})}\BibitemShut
  {NoStop}%
\bibitem [{\citenamefont {Ohga}\ and\ \citenamefont
  {Ito}(2021)}]{ohga2021ArXiv211211008Cond-Mat}%
  \BibitemOpen
  \bibfield  {author} {\bibinfo {author} {\bibfnamefont {N.}~\bibnamefont
  {Ohga}}\ and\ \bibinfo {author} {\bibfnamefont {S.}~\bibnamefont {Ito}},\
  }\bibfield  {title} {\bibinfo {title} {Information-geometric {{Legendre}}
  duality in stochastic thermodynamics},\ }\href@noop {} {\bibfield  {journal}
  {\bibinfo  {journal} {ArXiv211211008 Cond-Mat}\ } (\bibinfo {year} {2021})},\
  \Eprint {https://arxiv.org/abs/2112.11008} {arxiv:2112.11008 [cond-mat]}
  \BibitemShut {NoStop}%
\bibitem [{\citenamefont {Stam}(1959)}]{stam1959InformationandControl}%
  \BibitemOpen
  \bibfield  {author} {\bibinfo {author} {\bibfnamefont {A.~J.}\ \bibnamefont
  {Stam}},\ }\bibfield  {title} {\bibinfo {title} {Some inequalities satisfied
  by the quantities of information of {{Fisher}} and {{Shannon}}},\ }\href
  {https://doi.org/10.1016/S0019-9958(59)90348-1} {\bibfield  {journal}
  {\bibinfo  {journal} {Information and Control}\ }\textbf {\bibinfo {volume}
  {2}},\ \bibinfo {pages} {101} (\bibinfo {year} {1959})}\BibitemShut {NoStop}%
\bibitem [{\citenamefont {Plastino}\ \emph {et~al.}(1998)\citenamefont
  {Plastino}, \citenamefont {Casas},\ and\ \citenamefont
  {Plastino}}]{plastino1998PhysicsLettersA}%
  \BibitemOpen
  \bibfield  {author} {\bibinfo {author} {\bibfnamefont {A.~R.}\ \bibnamefont
  {Plastino}}, \bibinfo {author} {\bibfnamefont {M.}~\bibnamefont {Casas}},\
  and\ \bibinfo {author} {\bibfnamefont {A.}~\bibnamefont {Plastino}},\
  }\bibfield  {title} {\bibinfo {title} {Fisher's information, {{Kullback}}'s
  measure, and {{H-theorems}}},\ }\href
  {https://doi.org/10.1016/S0375-9601(98)00567-2} {\bibfield  {journal}
  {\bibinfo  {journal} {Physics Letters A}\ }\textbf {\bibinfo {volume}
  {246}},\ \bibinfo {pages} {498} (\bibinfo {year} {1998})}\BibitemShut
  {NoStop}%
\bibitem [{\citenamefont {Wibisono}\ \emph {et~al.}(2017)\citenamefont
  {Wibisono}, \citenamefont {Jog},\ and\ \citenamefont
  {Loh}}]{wibisono20172017IEEEInt.Symp.Inf.TheoryISIT}%
  \BibitemOpen
  \bibfield  {author} {\bibinfo {author} {\bibfnamefont {A.}~\bibnamefont
  {Wibisono}}, \bibinfo {author} {\bibfnamefont {V.}~\bibnamefont {Jog}},\ and\
  \bibinfo {author} {\bibfnamefont {P.-L.}\ \bibnamefont {Loh}},\ }\bibfield
  {title} {\bibinfo {title} {Information and estimation in {{Fokker-Planck}}
  channels},\ }in\ \href {https://doi.org/10.1109/ISIT.2017.8007014} {\emph
  {\bibinfo {booktitle} {2017 {{IEEE Int}}. {{Symp}}. {{Inf}}. {{Theory
  ISIT}}}}}\ (\bibinfo {year} {2017})\ pp.\ \bibinfo {pages}
  {2673--2677}\BibitemShut {NoStop}%
\bibitem [{\citenamefont {Fisher}\ and\ \citenamefont
  {Russell}(1922)}]{fisher1922Philos.Trans.R.Soc.Lond.Ser.Contain.Pap.Math.Phys.Character}%
  \BibitemOpen
  \bibfield  {author} {\bibinfo {author} {\bibfnamefont {R.~A.}\ \bibnamefont
  {Fisher}}\ and\ \bibinfo {author} {\bibfnamefont {E.~J.}\ \bibnamefont
  {Russell}},\ }\bibfield  {title} {\bibinfo {title} {On the mathematical
  foundations of theoretical statistics},\ }\href
  {https://doi.org/10.1098/rsta.1922.0009} {\bibfield  {journal} {\bibinfo
  {journal} {Philos. Trans. R. Soc. Lond. Ser. Contain. Pap. Math. Phys.
  Character}\ }\textbf {\bibinfo {volume} {222}},\ \bibinfo {pages} {309}
  (\bibinfo {year} {1922})}\BibitemShut {NoStop}%
\bibitem [{\citenamefont {Rao}(1958)}]{rao1958Scand.Actuar.J.}%
  \BibitemOpen
  \bibfield  {author} {\bibinfo {author} {\bibfnamefont {B.~R.}\ \bibnamefont
  {Rao}},\ }\bibfield  {title} {\bibinfo {title} {On an analogue of
  {{Cram\'er-Rao}}'s inequality},\ }\href
  {https://doi.org/10.1080/03461238.1958.10405982} {\bibfield  {journal}
  {\bibinfo  {journal} {Scand. Actuar. J.}\ }\textbf {\bibinfo {volume}
  {1958}},\ \bibinfo {pages} {57} (\bibinfo {year} {1958})}\BibitemShut
  {NoStop}%
\bibitem [{\citenamefont {Papaioannou}\ and\ \citenamefont
  {Ferentinos}(2005)}]{papaioannou2005Commun.Stat.-TheoryMethodsa}%
  \BibitemOpen
  \bibfield  {author} {\bibinfo {author} {\bibfnamefont {T.}~\bibnamefont
  {Papaioannou}}\ and\ \bibinfo {author} {\bibfnamefont {K.}~\bibnamefont
  {Ferentinos}},\ }\bibfield  {title} {\bibinfo {title} {On {{Two Forms}} of
  {{Fisher}}'s {{Measure}} of {{Information}}},\ }\href
  {https://doi.org/10.1081/STA-200063386} {\bibfield  {journal} {\bibinfo
  {journal} {Commun. Stat. - Theory Methods}\ }\textbf {\bibinfo {volume}
  {34}},\ \bibinfo {pages} {1461} (\bibinfo {year} {2005})}\BibitemShut
  {NoStop}%
\bibitem [{\citenamefont {Kharazmi}\ and\ \citenamefont
  {Asadi}(2018)}]{kharazmi2018Braz.J.Probab.Stat.}%
  \BibitemOpen
  \bibfield  {author} {\bibinfo {author} {\bibfnamefont {O.}~\bibnamefont
  {Kharazmi}}\ and\ \bibinfo {author} {\bibfnamefont {M.}~\bibnamefont
  {Asadi}},\ }\bibfield  {title} {\bibinfo {title} {On the time-dependent
  {{Fisher}} information of a density function},\ }\href
  {https://doi.org/10.1214/17-BJPS366} {\bibfield  {journal} {\bibinfo
  {journal} {Braz. J. Probab. Stat.}\ }\textbf {\bibinfo {volume} {32}},\
  \bibinfo {pages} {795} (\bibinfo {year} {2018})}\BibitemShut {NoStop}%
\bibitem [{\citenamefont {Johnson}(2004)}]{johnson2004}%
  \BibitemOpen
  \bibfield  {author} {\bibinfo {author} {\bibfnamefont {O.}~\bibnamefont
  {Johnson}},\ }\href@noop {} {\emph {\bibinfo {title} {Information {{Theory}}
  and the {{Central Limit Theorem}}}}}\ (\bibinfo  {publisher} {{World
  Scientific}},\ \bibinfo {year} {2004})\BibitemShut {NoStop}%
\bibitem [{\citenamefont
  {Yamano}(2013{\natexlab{a}})}]{yamano2013Eur.Phys.J.B}%
  \BibitemOpen
  \bibfield  {author} {\bibinfo {author} {\bibfnamefont {T.}~\bibnamefont
  {Yamano}},\ }\bibfield  {title} {\bibinfo {title} {De {{Bruijn-type}}
  identity for systems with flux},\ }\href
  {https://doi.org/10.1140/epjb/e2013-40634-9} {\bibfield  {journal} {\bibinfo
  {journal} {Eur. Phys. J. B}\ }\textbf {\bibinfo {volume} {86}},\ \bibinfo
  {pages} {363} (\bibinfo {year} {2013}{\natexlab{a}})}\BibitemShut {NoStop}%
\bibitem [{\citenamefont {Gross}(1975{\natexlab{a}})}]{gross1975Am.J.Math.}%
  \BibitemOpen
  \bibfield  {author} {\bibinfo {author} {\bibfnamefont {L.}~\bibnamefont
  {Gross}},\ }\bibfield  {title} {\bibinfo {title} {Logarithmic {{Sobolev
  Inequalities}}},\ }\href {https://doi.org/10.2307/2373688} {\bibfield
  {journal} {\bibinfo  {journal} {Am. J. Math.}\ }\textbf {\bibinfo {volume}
  {97}},\ \bibinfo {pages} {1061} (\bibinfo {year} {1975}{\natexlab{a}})},\
  \Eprint {https://arxiv.org/abs/2373688} {2373688} \BibitemShut {NoStop}%
\bibitem [{\citenamefont {Gross}(1975{\natexlab{b}})}]{gross1975DukeMath.J.}%
  \BibitemOpen
  \bibfield  {author} {\bibinfo {author} {\bibfnamefont {L.}~\bibnamefont
  {Gross}},\ }\bibfield  {title} {\bibinfo {title} {Hypercontractivity and
  logarithmic {{Sobolev}} inequalities for the {{Clifford-Dirichlet}} form},\
  }\href {https://doi.org/10.1215/S0012-7094-75-04237-4} {\bibfield  {journal}
  {\bibinfo  {journal} {Duke Math. J.}\ }\textbf {\bibinfo {volume} {42}},\
  \bibinfo {pages} {383} (\bibinfo {year} {1975}{\natexlab{b}})}\BibitemShut
  {NoStop}%
\bibitem [{\citenamefont {Otto}(2001)}]{otto2001Commun.PartialDiffer.Equ.}%
  \BibitemOpen
  \bibfield  {author} {\bibinfo {author} {\bibfnamefont {F.}~\bibnamefont
  {Otto}},\ }\bibfield  {title} {\bibinfo {title} {The {{Geometry}} of
  {{Dissipative Evolution Equations}}: {{The Porous Medium Equation}}},\ }\href
  {https://doi.org/10.1081/PDE-100002243} {\bibfield  {journal} {\bibinfo
  {journal} {Commun. Partial Differ. Equ.}\ }\textbf {\bibinfo {volume} {26}},\
  \bibinfo {pages} {101} (\bibinfo {year} {2001})}\BibitemShut {NoStop}%
\bibitem [{\citenamefont {Villani}(2003)}]{villani2003}%
  \BibitemOpen
  \bibfield  {author} {\bibinfo {author} {\bibfnamefont {C.}~\bibnamefont
  {Villani}},\ }\href@noop {} {\emph {\bibinfo {title} {Topics in {{Optimal
  Transportation}}}}}\ (\bibinfo  {publisher} {{American Mathematical Soc.}},\
  \bibinfo {year} {2003})\BibitemShut {NoStop}%
\bibitem [{\citenamefont {Amari}(1998)}]{amari1998NeuralComputation}%
  \BibitemOpen
  \bibfield  {author} {\bibinfo {author} {\bibfnamefont {S.-i.}\ \bibnamefont
  {Amari}},\ }\bibfield  {title} {\bibinfo {title} {Natural {{Gradient Works
  Efficiently}} in {{Learning}}},\ }\href
  {https://doi.org/10.1162/089976698300017746} {\bibfield  {journal} {\bibinfo
  {journal} {Neural Computation}\ }\textbf {\bibinfo {volume} {10}},\ \bibinfo
  {pages} {251} (\bibinfo {year} {1998})}\BibitemShut {NoStop}%
\bibitem [{\citenamefont {Beck}\ and\ \citenamefont
  {Teboulle}(2003)}]{beck2003OperationsResearchLetters}%
  \BibitemOpen
  \bibfield  {author} {\bibinfo {author} {\bibfnamefont {A.}~\bibnamefont
  {Beck}}\ and\ \bibinfo {author} {\bibfnamefont {M.}~\bibnamefont
  {Teboulle}},\ }\bibfield  {title} {\bibinfo {title} {Mirror descent and
  nonlinear projected subgradient methods for convex optimization},\ }\href
  {https://doi.org/10.1016/S0167-6377(02)00231-6} {\bibfield  {journal}
  {\bibinfo  {journal} {Operations Research Letters}\ }\textbf {\bibinfo
  {volume} {31}},\ \bibinfo {pages} {167} (\bibinfo {year} {2003})}\BibitemShut
  {NoStop}%
\bibitem [{\citenamefont {Raskutti}\ and\ \citenamefont
  {Mukherjee}(2015)}]{raskutti2015IEEETrans.Inf.Theory}%
  \BibitemOpen
  \bibfield  {author} {\bibinfo {author} {\bibfnamefont {G.}~\bibnamefont
  {Raskutti}}\ and\ \bibinfo {author} {\bibfnamefont {S.}~\bibnamefont
  {Mukherjee}},\ }\bibfield  {title} {\bibinfo {title} {The {{Information
  Geometry}} of {{Mirror Descent}}},\ }\href
  {https://doi.org/10.1109/TIT.2015.2388583} {\bibfield  {journal} {\bibinfo
  {journal} {IEEE Trans. Inf. Theory}\ }\textbf {\bibinfo {volume} {61}},\
  \bibinfo {pages} {1451} (\bibinfo {year} {2015})}\BibitemShut {NoStop}%
\bibitem [{\citenamefont {Ollivier}\ \emph {et~al.}(2017)\citenamefont
  {Ollivier}, \citenamefont {Arnold}, \citenamefont {Auger},\ and\
  \citenamefont {Hansen}}]{ollivier2017J.Mach.Learn.Res.}%
  \BibitemOpen
  \bibfield  {author} {\bibinfo {author} {\bibfnamefont {Y.}~\bibnamefont
  {Ollivier}}, \bibinfo {author} {\bibfnamefont {L.}~\bibnamefont {Arnold}},
  \bibinfo {author} {\bibfnamefont {A.}~\bibnamefont {Auger}},\ and\ \bibinfo
  {author} {\bibfnamefont {N.}~\bibnamefont {Hansen}},\ }\bibfield  {title}
  {\bibinfo {title} {Information-{{Geometric Optimization Algorithms}}: {{A
  Unifying Picture}} via {{Invariance Principles}}},\ }\href@noop {} {\bibfield
   {journal} {\bibinfo  {journal} {J. Mach. Learn. Res.}\ }\textbf {\bibinfo
  {volume} {18}},\ \bibinfo {pages} {1} (\bibinfo {year} {2017})}\BibitemShut
  {NoStop}%
\bibitem [{\citenamefont {Hino}\ \emph {et~al.}(2022)\citenamefont {Hino},
  \citenamefont {Akaho},\ and\ \citenamefont {Murata}}]{hino2022}%
  \BibitemOpen
  \bibfield  {author} {\bibinfo {author} {\bibfnamefont {H.}~\bibnamefont
  {Hino}}, \bibinfo {author} {\bibfnamefont {S.}~\bibnamefont {Akaho}},\ and\
  \bibinfo {author} {\bibfnamefont {N.}~\bibnamefont {Murata}},\ }\bibfield
  {title} {\bibinfo {title} {Geometry of {{EM}} and related iterative
  algorithms},\ }\bibfield  {journal} {\bibinfo  {journal} {Info. Geo.}\ }\href
  {https://doi.org/10.1007/s41884-022-00080-y} {10.1007/s41884-022-00080-y}
  (\bibinfo {year} {2022})\BibitemShut {NoStop}%
\bibitem [{\citenamefont
  {Pistone}(2021)}]{pistone2021ProgressinInformationGeometry:TheoryandApplications}%
  \BibitemOpen
  \bibfield  {author} {\bibinfo {author} {\bibfnamefont {G.}~\bibnamefont
  {Pistone}},\ }\bibfield  {title} {\bibinfo {title} {Information {{Geometry}}
  of {{Smooth Densities}} on the {{Gaussian Space}}: {{Poincar\'e
  Inequalities}}},\ }in\ \href {https://doi.org/10.1007/978-3-030-65459-7_1}
  {\emph {\bibinfo {booktitle} {Progress in {{Information Geometry}}:
  {{Theory}} and {{Applications}}}}},\ \bibinfo {series and number} {Signals
  and {{Communication Technology}}},\ \bibinfo {editor} {edited by\ \bibinfo
  {editor} {\bibfnamefont {F.}~\bibnamefont {Nielsen}}}\ (\bibinfo  {publisher}
  {{Springer International Publishing}},\ \bibinfo {address} {{Cham}},\
  \bibinfo {year} {2021})\ pp.\ \bibinfo {pages} {1--17}\BibitemShut {NoStop}%
\bibitem [{\citenamefont {Grady}\ and\ \citenamefont
  {Polimeni}(2010)}]{grady2010a}%
  \BibitemOpen
  \bibfield  {author} {\bibinfo {author} {\bibfnamefont {L.~J.}\ \bibnamefont
  {Grady}}\ and\ \bibinfo {author} {\bibfnamefont {J.~R.}\ \bibnamefont
  {Polimeni}},\ }\href@noop {} {\emph {\bibinfo {title} {Discrete {{Calculus}}:
  {{Applied Analysis}} on {{Graphs}} for {{Computational Science}}}}}\
  (\bibinfo  {publisher} {{Springer Science \& Business Media}},\ \bibinfo
  {year} {2010})\BibitemShut {NoStop}%
\bibitem [{\citenamefont {Musielak}(1983)}]{musielak1983}%
  \BibitemOpen
  \bibfield  {author} {\bibinfo {author} {\bibfnamefont {J.}~\bibnamefont
  {Musielak}},\ }\href@noop {} {\emph {\bibinfo {title} {Orlicz {{Spaces}} and
  {{Modular Spaces}}}}}\ (\bibinfo  {publisher} {{Springer}},\ \bibinfo {year}
  {1983})\BibitemShut {NoStop}%
\bibitem [{\citenamefont {Mielke}\ \emph {et~al.}(2014)\citenamefont {Mielke},
  \citenamefont {Peletier},\ and\ \citenamefont
  {Renger}}]{mielke2014PotentialAnala}%
  \BibitemOpen
  \bibfield  {author} {\bibinfo {author} {\bibfnamefont {A.}~\bibnamefont
  {Mielke}}, \bibinfo {author} {\bibfnamefont {M.~A.}\ \bibnamefont
  {Peletier}},\ and\ \bibinfo {author} {\bibfnamefont {D.~R.~M.}\ \bibnamefont
  {Renger}},\ }\bibfield  {title} {\bibinfo {title} {On the {{Relation}}
  between {{Gradient Flows}} and the {{Large-Deviation Principle}}, with
  {{Applications}} to {{Markov Chains}} and {{Diffusion}}},\ }\href
  {https://doi.org/10.1007/s11118-014-9418-5} {\bibfield  {journal} {\bibinfo
  {journal} {Potential Anal}\ }\textbf {\bibinfo {volume} {41}},\ \bibinfo
  {pages} {1293} (\bibinfo {year} {2014})}\BibitemShut {NoStop}%
\bibitem [{\citenamefont {Mielke}\ \emph {et~al.}(2017)\citenamefont {Mielke},
  \citenamefont {Patterson}, \citenamefont {Peletier},\ and\ \citenamefont
  {Michiel~Renger}}]{mielke2017SIAMJ.Appl.Math.}%
  \BibitemOpen
  \bibfield  {author} {\bibinfo {author} {\bibfnamefont {A.}~\bibnamefont
  {Mielke}}, \bibinfo {author} {\bibfnamefont {R.~I.~A.}\ \bibnamefont
  {Patterson}}, \bibinfo {author} {\bibfnamefont {M.~A.}\ \bibnamefont
  {Peletier}},\ and\ \bibinfo {author} {\bibfnamefont {D.~R.}\ \bibnamefont
  {Michiel~Renger}},\ }\bibfield  {title} {\bibinfo {title} {Non-equilibrium
  {{Thermodynamical Principles}} for {{Chemical Reactions}} with {{Mass-Action
  Kinetics}}},\ }\href {https://doi.org/10.1137/16M1102240} {\bibfield
  {journal} {\bibinfo  {journal} {SIAM J. Appl. Math.}\ }\textbf {\bibinfo
  {volume} {77}},\ \bibinfo {pages} {1562} (\bibinfo {year}
  {2017})}\BibitemShut {NoStop}%
\bibitem [{\citenamefont {Renger}(2018)}]{renger2018Entropy}%
  \BibitemOpen
  \bibfield  {author} {\bibinfo {author} {\bibfnamefont {D.~R.~M.}\
  \bibnamefont {Renger}},\ }\bibfield  {title} {\bibinfo {title} {Gradient and
  {{GENERIC Systems}} in the {{Space}} of {{Fluxes}}, {{Applied}} to {{Reacting
  Particle Systems}}},\ }\href {https://doi.org/10.3390/e20080596} {\bibfield
  {journal} {\bibinfo  {journal} {Entropy}\ }\textbf {\bibinfo {volume} {20}},\
  \bibinfo {pages} {596} (\bibinfo {year} {2018})}\BibitemShut {NoStop}%
\bibitem [{\citenamefont {Kaiser}\ \emph {et~al.}(2018)\citenamefont {Kaiser},
  \citenamefont {Jack},\ and\ \citenamefont {Zimmer}}]{kaiser2018JStatPhys}%
  \BibitemOpen
  \bibfield  {author} {\bibinfo {author} {\bibfnamefont {M.}~\bibnamefont
  {Kaiser}}, \bibinfo {author} {\bibfnamefont {R.~L.}\ \bibnamefont {Jack}},\
  and\ \bibinfo {author} {\bibfnamefont {J.}~\bibnamefont {Zimmer}},\
  }\bibfield  {title} {\bibinfo {title} {Canonical {{Structure}} and
  {{Orthogonality}} of {{Forces}} and {{Currents}} in {{Irreversible Markov
  Chains}}},\ }\href {https://doi.org/10.1007/s10955-018-1986-0} {\bibfield
  {journal} {\bibinfo  {journal} {J Stat Phys}\ }\textbf {\bibinfo {volume}
  {170}},\ \bibinfo {pages} {1019} (\bibinfo {year} {2018})}\BibitemShut
  {NoStop}%
\bibitem [{\citenamefont {Patterson}\ \emph {et~al.}(2021)\citenamefont
  {Patterson}, \citenamefont {Renger},\ and\ \citenamefont
  {Sharma}}]{patterson2021ArXiv210314384Math-Ph}%
  \BibitemOpen
  \bibfield  {author} {\bibinfo {author} {\bibfnamefont {R.~I.~A.}\
  \bibnamefont {Patterson}}, \bibinfo {author} {\bibfnamefont {D.~R.~M.}\
  \bibnamefont {Renger}},\ and\ \bibinfo {author} {\bibfnamefont
  {U.}~\bibnamefont {Sharma}},\ }\bibfield  {title} {\bibinfo {title}
  {Variational structures beyond gradient flows: A macroscopic
  fluctuation-theory perspective},\ }\href@noop {} {\bibfield  {journal}
  {\bibinfo  {journal} {ArXiv210314384 Math-Ph}\ } (\bibinfo {year} {2021})},\
  \Eprint {https://arxiv.org/abs/2103.14384} {arxiv:2103.14384 [math-ph]}
  \BibitemShut {NoStop}%
\bibitem [{\citenamefont {Peletier}\ \emph {et~al.}(2022)\citenamefont
  {Peletier}, \citenamefont {Rossi}, \citenamefont {Savar{\'e}},\ and\
  \citenamefont {Tse}}]{peletier2022Calc.Var.}%
  \BibitemOpen
  \bibfield  {author} {\bibinfo {author} {\bibfnamefont {M.~A.}\ \bibnamefont
  {Peletier}}, \bibinfo {author} {\bibfnamefont {R.}~\bibnamefont {Rossi}},
  \bibinfo {author} {\bibfnamefont {G.}~\bibnamefont {Savar{\'e}}},\ and\
  \bibinfo {author} {\bibfnamefont {O.}~\bibnamefont {Tse}},\ }\bibfield
  {title} {\bibinfo {title} {Jump processes as generalized gradient flows},\
  }\href {https://doi.org/10.1007/s00526-021-02130-2} {\bibfield  {journal}
  {\bibinfo  {journal} {Calc. Var.}\ }\textbf {\bibinfo {volume} {61}},\
  \bibinfo {pages} {33} (\bibinfo {year} {2022})}\BibitemShut {NoStop}%
\bibitem [{\citenamefont {Renger}\ and\ \citenamefont
  {Zimmer}(2021)}]{renger2021DiscreteContin.Dyn.Syst.-S}%
  \BibitemOpen
  \bibfield  {author} {\bibinfo {author} {\bibfnamefont {D.~R.~M.}\
  \bibnamefont {Renger}}\ and\ \bibinfo {author} {\bibfnamefont
  {J.}~\bibnamefont {Zimmer}},\ }\bibfield  {title} {\bibinfo {title}
  {Orthogonality of fluxes in general nonlinear reaction networks},\ }\href
  {https://doi.org/10.3934/dcdss.2020346} {\bibfield  {journal} {\bibinfo
  {journal} {Discrete Contin. Dyn. Syst. - S}\ }\textbf {\bibinfo {volume}
  {14}},\ \bibinfo {pages} {205} (\bibinfo {year} {2021})}\BibitemShut
  {NoStop}%
\bibitem [{\citenamefont {Peletier}\ and\ \citenamefont
  {Schlichting}(2023)}]{peletier2022}%
  \BibitemOpen
  \bibfield  {author} {\bibinfo {author} {\bibfnamefont {M.~A.}\ \bibnamefont
  {Peletier}}\ and\ \bibinfo {author} {\bibfnamefont {A.}~\bibnamefont
  {Schlichting}},\ }\bibfield  {title} {\bibinfo {title} {Cosh gradient systems
  and tilting},\ }\href {https://doi.org/10.1016/j.na.2022.113094} {\bibfield
  {journal} {\bibinfo  {journal} {Nonlinear Analysis}\ }\bibinfo {series}
  {Variational {{Models}} for {{Discrete Systems}}},\ \textbf {\bibinfo
  {volume} {231}},\ \bibinfo {pages} {113094} (\bibinfo {year}
  {2023})}\BibitemShut {NoStop}%
\bibitem [{\citenamefont {Kobayashi}\ \emph
  {et~al.}(2022{\natexlab{a}})\citenamefont {Kobayashi}, \citenamefont
  {Loutchko}, \citenamefont {Kamimura},\ and\ \citenamefont
  {Sughiyama}}]{kobayashi2022Phys.Rev.Researcha}%
  \BibitemOpen
  \bibfield  {author} {\bibinfo {author} {\bibfnamefont {T.~J.}\ \bibnamefont
  {Kobayashi}}, \bibinfo {author} {\bibfnamefont {D.}~\bibnamefont {Loutchko}},
  \bibinfo {author} {\bibfnamefont {A.}~\bibnamefont {Kamimura}},\ and\
  \bibinfo {author} {\bibfnamefont {Y.}~\bibnamefont {Sughiyama}},\ }\bibfield
  {title} {\bibinfo {title} {Hessian geometry of nonequilibrium chemical
  reaction networks and entropy production decompositions},\ }\href
  {https://doi.org/10.1103/PhysRevResearch.4.033208} {\bibfield  {journal}
  {\bibinfo  {journal} {Phys. Rev. Research}\ }\textbf {\bibinfo {volume}
  {4}},\ \bibinfo {pages} {033208} (\bibinfo {year}
  {2022}{\natexlab{a}})}\BibitemShut {NoStop}%
\bibitem [{\citenamefont {Kobayashi}\ \emph
  {et~al.}(2022{\natexlab{b}})\citenamefont {Kobayashi}, \citenamefont
  {Loutchko}, \citenamefont {Kamimura},\ and\ \citenamefont
  {Sughiyama}}]{kobayashi2022Phys.Rev.Research}%
  \BibitemOpen
  \bibfield  {author} {\bibinfo {author} {\bibfnamefont {T.~J.}\ \bibnamefont
  {Kobayashi}}, \bibinfo {author} {\bibfnamefont {D.}~\bibnamefont {Loutchko}},
  \bibinfo {author} {\bibfnamefont {A.}~\bibnamefont {Kamimura}},\ and\
  \bibinfo {author} {\bibfnamefont {Y.}~\bibnamefont {Sughiyama}},\ }\bibfield
  {title} {\bibinfo {title} {Kinetic derivation of the {{Hessian}} geometric
  structure in chemical reaction networks},\ }\href
  {https://doi.org/10.1103/PhysRevResearch.4.033066} {\bibfield  {journal}
  {\bibinfo  {journal} {Phys. Rev. Research}\ }\textbf {\bibinfo {volume}
  {4}},\ \bibinfo {pages} {033066} (\bibinfo {year}
  {2022}{\natexlab{b}})}\BibitemShut {NoStop}%
\bibitem [{\citenamefont {Sughiyama}\ \emph
  {et~al.}(2022{\natexlab{a}})\citenamefont {Sughiyama}, \citenamefont
  {Loutchko}, \citenamefont {Kamimura},\ and\ \citenamefont
  {Kobayashi}}]{sughiyama2022Phys.Rev.Research}%
  \BibitemOpen
  \bibfield  {author} {\bibinfo {author} {\bibfnamefont {Y.}~\bibnamefont
  {Sughiyama}}, \bibinfo {author} {\bibfnamefont {D.}~\bibnamefont {Loutchko}},
  \bibinfo {author} {\bibfnamefont {A.}~\bibnamefont {Kamimura}},\ and\
  \bibinfo {author} {\bibfnamefont {T.~J.}\ \bibnamefont {Kobayashi}},\
  }\bibfield  {title} {\bibinfo {title} {Hessian geometric structure of
  chemical thermodynamic systems with stoichiometric constraints},\ }\href
  {https://doi.org/10.1103/PhysRevResearch.4.033065} {\bibfield  {journal}
  {\bibinfo  {journal} {Phys. Rev. Research}\ }\textbf {\bibinfo {volume}
  {4}},\ \bibinfo {pages} {033065} (\bibinfo {year}
  {2022}{\natexlab{a}})}\BibitemShut {NoStop}%
\bibitem [{\citenamefont {Sughiyama}\ \emph
  {et~al.}(2022{\natexlab{b}})\citenamefont {Sughiyama}, \citenamefont
  {Kamimura}, \citenamefont {Loutchko},\ and\ \citenamefont
  {Kobayashi}}]{sughiyama2022Phys.Rev.Researchb}%
  \BibitemOpen
  \bibfield  {author} {\bibinfo {author} {\bibfnamefont {Y.}~\bibnamefont
  {Sughiyama}}, \bibinfo {author} {\bibfnamefont {A.}~\bibnamefont {Kamimura}},
  \bibinfo {author} {\bibfnamefont {D.}~\bibnamefont {Loutchko}},\ and\
  \bibinfo {author} {\bibfnamefont {T.~J.}\ \bibnamefont {Kobayashi}},\
  }\bibfield  {title} {\bibinfo {title} {Chemical thermodynamics for growing
  systems},\ }\href {https://doi.org/10.1103/PhysRevResearch.4.033191}
  {\bibfield  {journal} {\bibinfo  {journal} {Phys. Rev. Research}\ }\textbf
  {\bibinfo {volume} {4}},\ \bibinfo {pages} {033191} (\bibinfo {year}
  {2022}{\natexlab{b}})}\BibitemShut {NoStop}%
\bibitem [{\citenamefont {Godsil}\ and\ \citenamefont
  {Royle}(2013)}]{godsil2013}%
  \BibitemOpen
  \bibfield  {author} {\bibinfo {author} {\bibfnamefont {C.}~\bibnamefont
  {Godsil}}\ and\ \bibinfo {author} {\bibfnamefont {G.~F.}\ \bibnamefont
  {Royle}},\ }\href@noop {} {\emph {\bibinfo {title} {Algebraic {{Graph
  Theory}}}}}\ (\bibinfo  {publisher} {{Springer Science \& Business Media}},\
  \bibinfo {year} {2013})\BibitemShut {NoStop}%
\bibitem [{\citenamefont {Bretto}(2013)}]{bretto2013}%
  \BibitemOpen
  \bibfield  {author} {\bibinfo {author} {\bibfnamefont {A.}~\bibnamefont
  {Bretto}},\ }\href@noop {} {\emph {\bibinfo {title} {Hypergraph {{Theory}}:
  {{An Introduction}}}}}\ (\bibinfo  {publisher} {{Springer Science \& Business
  Media}},\ \bibinfo {year} {2013})\BibitemShut {NoStop}%
\bibitem [{\citenamefont {Meyn}\ and\ \citenamefont
  {Tweedie}(2009)}]{meyn2009}%
  \BibitemOpen
  \bibfield  {author} {\bibinfo {author} {\bibfnamefont {S.}~\bibnamefont
  {Meyn}}\ and\ \bibinfo {author} {\bibfnamefont {R.~L.}\ \bibnamefont
  {Tweedie}},\ }\href@noop {} {\emph {\bibinfo {title} {Markov {{Chains}} and
  {{Stochastic Stability}}}}}\ (\bibinfo  {publisher} {{Cambridge University
  Press}},\ \bibinfo {year} {2009})\BibitemShut {NoStop}%
\bibitem [{\citenamefont {Schnakenberg}(1976)}]{schnakenberg1976Rev.Mod.Phys.}%
  \BibitemOpen
  \bibfield  {author} {\bibinfo {author} {\bibfnamefont {J.}~\bibnamefont
  {Schnakenberg}},\ }\bibfield  {title} {\bibinfo {title} {Network theory of
  microscopic and macroscopic behavior of master equation systems},\ }\href
  {https://doi.org/10.1103/RevModPhys.48.571} {\bibfield  {journal} {\bibinfo
  {journal} {Rev. Mod. Phys.}\ }\textbf {\bibinfo {volume} {48}},\ \bibinfo
  {pages} {571} (\bibinfo {year} {1976})}\BibitemShut {NoStop}%
\bibitem [{\citenamefont {Craciun}(2019)}]{craciun2019}%
  \BibitemOpen
  \bibfield  {author} {\bibinfo {author} {\bibfnamefont {G.}~\bibnamefont
  {Craciun}},\ }\href {https://doi.org/10.48550/arXiv.1901.02544} {\bibinfo
  {title} {Polynomial {{Dynamical Systems}}, {{Reaction Networks}}, and {{Toric
  Differential Inclusions}}}} (\bibinfo {year} {2019}),\ \Eprint
  {https://arxiv.org/abs/1901.02544} {arxiv:1901.02544 [math]} \BibitemShut
  {NoStop}%
\bibitem [{\citenamefont {Biggs}(1997)}]{biggs1997Bull.Lond.Math.Soc.}%
  \BibitemOpen
  \bibfield  {author} {\bibinfo {author} {\bibfnamefont {N.}~\bibnamefont
  {Biggs}},\ }\bibfield  {title} {\bibinfo {title} {Algebraic {{Potential
  Theory}} on {{Graphs}}},\ }\href {https://doi.org/10.1112/S0024609397003305}
  {\bibfield  {journal} {\bibinfo  {journal} {Bull. Lond. Math. Soc.}\ }\textbf
  {\bibinfo {volume} {29}},\ \bibinfo {pages} {641} (\bibinfo {year}
  {1997})}\BibitemShut {NoStop}%
\bibitem [{\citenamefont {Chung}\ and\ \citenamefont
  {Graham}(1997)}]{chung1997}%
  \BibitemOpen
  \bibfield  {author} {\bibinfo {author} {\bibfnamefont {F.~R.~K.}\
  \bibnamefont {Chung}}\ and\ \bibinfo {author} {\bibfnamefont {F.~C.}\
  \bibnamefont {Graham}},\ }\href@noop {} {\emph {\bibinfo {title} {Spectral
  {{Graph Theory}}}}}\ (\bibinfo  {publisher} {{American Mathematical Soc.}},\
  \bibinfo {year} {1997})\BibitemShut {NoStop}%
\bibitem [{\citenamefont {D{\"o}rfler}\ \emph {et~al.}(2018)\citenamefont
  {D{\"o}rfler}, \citenamefont {{Simpson-Porco}},\ and\ \citenamefont
  {Bullo}}]{dorfler2018Proc.IEEE}%
  \BibitemOpen
  \bibfield  {author} {\bibinfo {author} {\bibfnamefont {F.}~\bibnamefont
  {D{\"o}rfler}}, \bibinfo {author} {\bibfnamefont {J.~W.}\ \bibnamefont
  {{Simpson-Porco}}},\ and\ \bibinfo {author} {\bibfnamefont {F.}~\bibnamefont
  {Bullo}},\ }\bibfield  {title} {\bibinfo {title} {Electrical {{Networks}} and
  {{Algebraic Graph Theory}}: {{Models}}, {{Properties}}, and
  {{Applications}}},\ }\href {https://doi.org/10.1109/JPROC.2018.2821924}
  {\bibfield  {journal} {\bibinfo  {journal} {Proc. IEEE}\ }\textbf {\bibinfo
  {volume} {106}},\ \bibinfo {pages} {977} (\bibinfo {year}
  {2018})}\BibitemShut {NoStop}%
\bibitem [{\citenamefont {Saber}\ and\ \citenamefont
  {Murray}(2003)}]{saber2003Proc.2003Am.ControlConf.2003}%
  \BibitemOpen
  \bibfield  {author} {\bibinfo {author} {\bibfnamefont {R.}~\bibnamefont
  {Saber}}\ and\ \bibinfo {author} {\bibfnamefont {R.}~\bibnamefont {Murray}},\
  }\bibfield  {title} {\bibinfo {title} {Consensus protocols for networks of
  dynamic agents},\ }in\ \href {https://doi.org/10.1109/ACC.2003.1239709}
  {\emph {\bibinfo {booktitle} {Proc. 2003 {{Am}}. {{Control Conf}}. 2003}}},\
  Vol.~\bibinfo {volume} {2}\ (\bibinfo {year} {2003})\ pp.\ \bibinfo {pages}
  {951--956}\BibitemShut {NoStop}%
\bibitem [{\citenamefont {Veerman}\ and\ \citenamefont
  {Lyons}(2020)}]{veerman2020}%
  \BibitemOpen
  \bibfield  {author} {\bibinfo {author} {\bibfnamefont {J.~J.~P.}\
  \bibnamefont {Veerman}}\ and\ \bibinfo {author} {\bibfnamefont
  {R.}~\bibnamefont {Lyons}},\ }\href
  {https://doi.org/10.48550/arXiv.2002.02605} {\bibinfo {title} {A {{Primer}}
  on {{Laplacian Dynamics}} in {{Directed Graphs}}}} (\bibinfo {year} {2020}),\
  \Eprint {https://arxiv.org/abs/2002.02605} {arxiv:2002.02605 [math]}
  \BibitemShut {NoStop}%
\bibitem [{\citenamefont {Qian}\ and\ \citenamefont {Ge}(2021)}]{qian2021a}%
  \BibitemOpen
  \bibfield  {author} {\bibinfo {author} {\bibfnamefont {H.}~\bibnamefont
  {Qian}}\ and\ \bibinfo {author} {\bibfnamefont {H.}~\bibnamefont {Ge}},\
  }\href@noop {} {\emph {\bibinfo {title} {Stochastic {{Chemical Reaction
  Systems}} in {{Biology}}}}}\ (\bibinfo  {publisher} {{Springer Nature}},\
  \bibinfo {year} {2021})\BibitemShut {NoStop}%
\bibitem [{\citenamefont {Keener}\ and\ \citenamefont
  {Sneyd}(2008)}]{keener2008}%
  \BibitemOpen
  \bibfield  {author} {\bibinfo {author} {\bibfnamefont {J.}~\bibnamefont
  {Keener}}\ and\ \bibinfo {author} {\bibfnamefont {J.}~\bibnamefont {Sneyd}},\
  }\href@noop {} {\emph {\bibinfo {title} {Mathematical {{Physiology}}: {{I}}:
  {{Cellular Physiology}}}}}\ (\bibinfo  {publisher} {{Springer Science \&
  Business Media}},\ \bibinfo {year} {2008})\BibitemShut {NoStop}%
\bibitem [{\citenamefont {Sottile}(2008)}]{sottile2008arXiv:math/0212044}%
  \BibitemOpen
  \bibfield  {author} {\bibinfo {author} {\bibfnamefont {F.}~\bibnamefont
  {Sottile}},\ }\bibfield  {title} {\bibinfo {title} {Toric ideals, real toric
  varieties, and the algebraic moment map},\ }\href@noop {} {\bibfield
  {journal} {\bibinfo  {journal} {arXiv:math/0212044}\ } (\bibinfo {year}
  {2008})},\ \Eprint {https://arxiv.org/abs/math/0212044} {arxiv:math/0212044}
  \BibitemShut {NoStop}%
\bibitem [{\citenamefont {Cox}\ \emph {et~al.}(2011)\citenamefont {Cox},
  \citenamefont {Little},\ and\ \citenamefont {Schenck}}]{cox2011}%
  \BibitemOpen
  \bibfield  {author} {\bibinfo {author} {\bibfnamefont {D.~A.}\ \bibnamefont
  {Cox}}, \bibinfo {author} {\bibfnamefont {J.~B.}\ \bibnamefont {Little}},\
  and\ \bibinfo {author} {\bibfnamefont {H.~K.}\ \bibnamefont {Schenck}},\
  }\href@noop {} {\emph {\bibinfo {title} {Toric {{Varieties}}}}}\ (\bibinfo
  {publisher} {{American Mathematical Soc.}},\ \bibinfo {year}
  {2011})\BibitemShut {NoStop}%
\bibitem [{\citenamefont {Craciun}\ \emph {et~al.}(2009)\citenamefont
  {Craciun}, \citenamefont {Dickenstein}, \citenamefont {Shiu},\ and\
  \citenamefont {Sturmfels}}]{craciun2009JournalofSymbolicComputation}%
  \BibitemOpen
  \bibfield  {author} {\bibinfo {author} {\bibfnamefont {G.}~\bibnamefont
  {Craciun}}, \bibinfo {author} {\bibfnamefont {A.}~\bibnamefont
  {Dickenstein}}, \bibinfo {author} {\bibfnamefont {A.}~\bibnamefont {Shiu}},\
  and\ \bibinfo {author} {\bibfnamefont {B.}~\bibnamefont {Sturmfels}},\
  }\bibfield  {title} {\bibinfo {title} {Toric dynamical systems},\ }\href
  {https://doi.org/10.1016/j.jsc.2008.08.006} {\bibfield  {journal} {\bibinfo
  {journal} {Journal of Symbolic Computation}\ }\bibinfo {series} {In
  {{Memoriam Karin Gatermann}}},\ \textbf {\bibinfo {volume} {44}},\ \bibinfo
  {pages} {1551} (\bibinfo {year} {2009})}\BibitemShut {NoStop}%
\bibitem [{\citenamefont {Kobayashi}\ \emph {et~al.}(2021)\citenamefont
  {Kobayashi}, \citenamefont {Loutchko}, \citenamefont {Kamimura},\ and\
  \citenamefont {Sughiyama}}]{kobayashi2021ArXiv211214910Phys.}%
  \BibitemOpen
  \bibfield  {author} {\bibinfo {author} {\bibfnamefont {T.~J.}\ \bibnamefont
  {Kobayashi}}, \bibinfo {author} {\bibfnamefont {D.}~\bibnamefont {Loutchko}},
  \bibinfo {author} {\bibfnamefont {A.}~\bibnamefont {Kamimura}},\ and\
  \bibinfo {author} {\bibfnamefont {Y.}~\bibnamefont {Sughiyama}},\ }\bibfield
  {title} {\bibinfo {title} {Kinetic {{Derivation}} of the {{Hessian Geometric
  Structure}} in {{Chemical Reaction Systems}}},\ }\href@noop {} {\bibfield
  {journal} {\bibinfo  {journal} {ArXiv211214910 Phys.}\ } (\bibinfo {year}
  {2021})},\ \Eprint {https://arxiv.org/abs/2112.14910} {arxiv:2112.14910
  [physics]} \BibitemShut {NoStop}%
\bibitem [{\citenamefont {Pachter}\ and\ \citenamefont
  {Sturmfels}(2005)}]{pachter2005}%
  \BibitemOpen
  \bibinfo {editor} {\bibfnamefont {L.}~\bibnamefont {Pachter}}\ and\ \bibinfo
  {editor} {\bibfnamefont {B.}~\bibnamefont {Sturmfels}},\ eds.,\ \href
  {https://doi.org/10.1017/CBO9780511610684} {\emph {\bibinfo {title}
  {Algebraic {{Statistics}} for {{Computational Biology}}}}}\ (\bibinfo
  {publisher} {{Cambridge University Press}},\ \bibinfo {address}
  {{Cambridge}},\ \bibinfo {year} {2005})\BibitemShut {NoStop}%
\bibitem [{\citenamefont {Rapallo}(2007)}]{rapallo2007AISM}%
  \BibitemOpen
  \bibfield  {author} {\bibinfo {author} {\bibfnamefont {F.}~\bibnamefont
  {Rapallo}},\ }\bibfield  {title} {\bibinfo {title} {Toric statistical models:
  Parametric and binomial representations},\ }\href
  {https://doi.org/10.1007/s10463-006-0079-z} {\bibfield  {journal} {\bibinfo
  {journal} {AISM}\ }\textbf {\bibinfo {volume} {59}},\ \bibinfo {pages} {727}
  (\bibinfo {year} {2007})}\BibitemShut {NoStop}%
\bibitem [{\citenamefont {M{\"u}ller}\ and\ \citenamefont
  {Regensburger}(2012)}]{muller2012SIAMJ.Appl.Math.}%
  \BibitemOpen
  \bibfield  {author} {\bibinfo {author} {\bibfnamefont {S.}~\bibnamefont
  {M{\"u}ller}}\ and\ \bibinfo {author} {\bibfnamefont {G.}~\bibnamefont
  {Regensburger}},\ }\bibfield  {title} {\bibinfo {title} {Generalized {{Mass
  Action Systems}}: {{Complex Balancing Equilibria}} and {{Sign Vectors}} of
  the {{Stoichiometric}} and {{Kinetic-Order Subspaces}}},\ }\href
  {https://doi.org/10.1137/110847056} {\bibfield  {journal} {\bibinfo
  {journal} {SIAM J. Appl. Math.}\ }\textbf {\bibinfo {volume} {72}},\ \bibinfo
  {pages} {1926} (\bibinfo {year} {2012})}\BibitemShut {NoStop}%
\bibitem [{\citenamefont {Noor}\ \emph {et~al.}(2013)\citenamefont {Noor},
  \citenamefont {Flamholz}, \citenamefont {Liebermeister}, \citenamefont
  {{Bar-Even}},\ and\ \citenamefont {Milo}}]{noor2013FEBSLetters}%
  \BibitemOpen
  \bibfield  {author} {\bibinfo {author} {\bibfnamefont {E.}~\bibnamefont
  {Noor}}, \bibinfo {author} {\bibfnamefont {A.}~\bibnamefont {Flamholz}},
  \bibinfo {author} {\bibfnamefont {W.}~\bibnamefont {Liebermeister}}, \bibinfo
  {author} {\bibfnamefont {A.}~\bibnamefont {{Bar-Even}}},\ and\ \bibinfo
  {author} {\bibfnamefont {R.}~\bibnamefont {Milo}},\ }\bibfield  {title}
  {\bibinfo {title} {A note on the kinetics of enzyme action: {{A}}
  decomposition that highlights thermodynamic effects},\ }\href
  {https://doi.org/10.1016/j.febslet.2013.07.028} {\bibfield  {journal}
  {\bibinfo  {journal} {FEBS Letters}\ }\bibinfo {series} {A Century of
  {{Michaelis}} - {{Menten}} Kinetics},\ \textbf {\bibinfo {volume} {587}},\
  \bibinfo {pages} {2772} (\bibinfo {year} {2013})}\BibitemShut {NoStop}%
\bibitem [{\citenamefont {Yoshimura}\ \emph {et~al.}(2023)\citenamefont
  {Yoshimura}, \citenamefont {Kolchinsky}, \citenamefont {Dechant},\ and\
  \citenamefont {Ito}}]{yoshimura2022}%
  \BibitemOpen
  \bibfield  {author} {\bibinfo {author} {\bibfnamefont {K.}~\bibnamefont
  {Yoshimura}}, \bibinfo {author} {\bibfnamefont {A.}~\bibnamefont
  {Kolchinsky}}, \bibinfo {author} {\bibfnamefont {A.}~\bibnamefont
  {Dechant}},\ and\ \bibinfo {author} {\bibfnamefont {S.}~\bibnamefont {Ito}},\
  }\bibfield  {title} {\bibinfo {title} {Housekeeping and excess entropy
  production for general nonlinear dynamics},\ }\href
  {https://doi.org/10.1103/PhysRevResearch.5.013017} {\bibfield  {journal}
  {\bibinfo  {journal} {Phys. Rev. Res.}\ }\textbf {\bibinfo {volume} {5}},\
  \bibinfo {pages} {013017} (\bibinfo {year} {2023})}\BibitemShut {NoStop}%
\bibitem [{\citenamefont {Lim}(2020)}]{lim2020SIAMRev.}%
  \BibitemOpen
  \bibfield  {author} {\bibinfo {author} {\bibfnamefont {L.-H.}\ \bibnamefont
  {Lim}},\ }\bibfield  {title} {\bibinfo {title} {Hodge {{Laplacians}} on
  {{Graphs}}},\ }\href {https://doi.org/10.1137/18M1223101} {\bibfield
  {journal} {\bibinfo  {journal} {SIAM Rev.}\ }\textbf {\bibinfo {volume}
  {62}},\ \bibinfo {pages} {685} (\bibinfo {year} {2020})}\BibitemShut
  {NoStop}%
\bibitem [{\citenamefont {Sunada}(2012)}]{sunada2012}%
  \BibitemOpen
  \bibfield  {author} {\bibinfo {author} {\bibfnamefont {T.}~\bibnamefont
  {Sunada}},\ }\href@noop {} {\emph {\bibinfo {title} {Topological
  {{Crystallography}}: {{With}} a {{View Towards Discrete Geometric
  Analysis}}}}}\ (\bibinfo  {publisher} {{Springer Science \& Business
  Media}},\ \bibinfo {year} {2012})\BibitemShut {NoStop}%
\bibitem [{\citenamefont {Desbrun}\ \emph {et~al.}(2005)\citenamefont
  {Desbrun}, \citenamefont {Hirani}, \citenamefont {Leok},\ and\ \citenamefont
  {Marsden}}]{desbrun2005}%
  \BibitemOpen
  \bibfield  {author} {\bibinfo {author} {\bibfnamefont {M.}~\bibnamefont
  {Desbrun}}, \bibinfo {author} {\bibfnamefont {A.~N.}\ \bibnamefont {Hirani}},
  \bibinfo {author} {\bibfnamefont {M.}~\bibnamefont {Leok}},\ and\ \bibinfo
  {author} {\bibfnamefont {J.~E.}\ \bibnamefont {Marsden}},\ }\href
  {https://doi.org/10.48550/arXiv.math/0508341} {\bibinfo {title} {Discrete
  {{Exterior Calculus}}}} (\bibinfo {year} {2005}),\ \Eprint
  {https://arxiv.org/abs/math/0508341} {arxiv:math/0508341} \BibitemShut
  {NoStop}%
\bibitem [{\citenamefont {Hirani}(2003)}]{hirani2003}%
  \BibitemOpen
  \bibfield  {author} {\bibinfo {author} {\bibfnamefont {A.~N.}\ \bibnamefont
  {Hirani}},\ }\emph {\bibinfo {title} {Discrete {{Exterior Calculus}}}},\
  \href {https://doi.org/10.7907/ZHY8-V329} {Ph.D. thesis},\ \bibinfo  {school}
  {California Institute of Technology} (\bibinfo {year} {2003})\BibitemShut
  {NoStop}%
\bibitem [{\citenamefont {Knauer}(2011)}]{knauer2011}%
  \BibitemOpen
  \bibfield  {author} {\bibinfo {author} {\bibfnamefont {U.}~\bibnamefont
  {Knauer}},\ }\href@noop {} {\emph {\bibinfo {title} {Algebraic {{Graph
  Theory}}: {{Morphisms}}, {{Monoids}} and {{Matrices}}}}}\ (\bibinfo
  {publisher} {{Walter de Gruyter}},\ \bibinfo {year} {2011})\BibitemShut
  {NoStop}%
\bibitem [{\citenamefont {Chen}(2012)}]{chen2012a}%
  \BibitemOpen
  \bibfield  {author} {\bibinfo {author} {\bibfnamefont {W.-K.}\ \bibnamefont
  {Chen}},\ }\href@noop {} {\emph {\bibinfo {title} {Applied {{Graph
  Theory}}}}}\ (\bibinfo  {publisher} {{Elsevier}},\ \bibinfo {year}
  {2012})\BibitemShut {NoStop}%
\bibitem [{\citenamefont {Craciun}\ \emph {et~al.}(2013)\citenamefont
  {Craciun}, \citenamefont {Nazarov},\ and\ \citenamefont
  {Pantea}}]{craciun2013SIAMJ.Appl.Math.}%
  \BibitemOpen
  \bibfield  {author} {\bibinfo {author} {\bibfnamefont {G.}~\bibnamefont
  {Craciun}}, \bibinfo {author} {\bibfnamefont {F.}~\bibnamefont {Nazarov}},\
  and\ \bibinfo {author} {\bibfnamefont {C.}~\bibnamefont {Pantea}},\
  }\bibfield  {title} {\bibinfo {title} {Persistence and {{Permanence}} of
  {{Mass-Action}} and {{Power-Law Dynamical Systems}}},\ }\href
  {https://doi.org/10.1137/100812355} {\bibfield  {journal} {\bibinfo
  {journal} {SIAM J. Appl. Math.}\ }\textbf {\bibinfo {volume} {73}},\ \bibinfo
  {pages} {305} (\bibinfo {year} {2013})}\BibitemShut {NoStop}%
\bibitem [{\citenamefont {Craciun}(2016)}]{craciun2016}%
  \BibitemOpen
  \bibfield  {author} {\bibinfo {author} {\bibfnamefont {G.}~\bibnamefont
  {Craciun}},\ }\href {https://doi.org/10.48550/arXiv.1501.02860} {\bibinfo
  {title} {Toric {{Differential Inclusions}} and a {{Proof}} of the {{Global
  Attractor Conjecture}}}} (\bibinfo {year} {2016}),\ \Eprint
  {https://arxiv.org/abs/1501.02860} {arxiv:1501.02860 [math]} \BibitemShut
  {NoStop}%
\bibitem [{\citenamefont
  {Bregman}(1967)}]{bregman1967USSRComputationalMathematicsandMathematicalPhysics}%
  \BibitemOpen
  \bibfield  {author} {\bibinfo {author} {\bibfnamefont {L.~M.}\ \bibnamefont
  {Bregman}},\ }\bibfield  {title} {\bibinfo {title} {The relaxation method of
  finding the common point of convex sets and its application to the solution
  of problems in convex programming},\ }\href
  {https://doi.org/10.1016/0041-5553(67)90040-7} {\bibfield  {journal}
  {\bibinfo  {journal} {USSR Computational Mathematics and Mathematical
  Physics}\ }\textbf {\bibinfo {volume} {7}},\ \bibinfo {pages} {200} (\bibinfo
  {year} {1967})}\BibitemShut {NoStop}%
\bibitem [{\citenamefont {Rockafellar}(1997)}]{rockafellar1997}%
  \BibitemOpen
  \bibfield  {author} {\bibinfo {author} {\bibfnamefont {R.~T.}\ \bibnamefont
  {Rockafellar}},\ }\href@noop {} {\emph {\bibinfo {title} {Convex
  {{Analysis}}}}}\ (\bibinfo  {publisher} {{Princeton University Press}},\
  \bibinfo {year} {1997})\BibitemShut {NoStop}%
\bibitem [{\citenamefont {Mitroi}\ and\ \citenamefont
  {Niculescu}(2011)}]{mitroi2011Abstr.Appl.Anal.}%
  \BibitemOpen
  \bibfield  {author} {\bibinfo {author} {\bibfnamefont {F.-C.}\ \bibnamefont
  {Mitroi}}\ and\ \bibinfo {author} {\bibfnamefont {C.~P.}\ \bibnamefont
  {Niculescu}},\ }\bibfield  {title} {\bibinfo {title} {An {{Extension}} of
  {{Young}}'s {{Inequality}}},\ }\href {https://doi.org/10.1155/2011/162049}
  {\bibfield  {journal} {\bibinfo  {journal} {Abstr. Appl. Anal.}\ }\textbf
  {\bibinfo {volume} {2011}},\ \bibinfo {pages} {e162049} (\bibinfo {year}
  {2011})}\BibitemShut {NoStop}%
\bibitem [{\citenamefont
  {Nielsen}(2021)}]{nielsen2021ProgressinInformationGeometry:TheoryandApplications}%
  \BibitemOpen
  \bibfield  {author} {\bibinfo {author} {\bibfnamefont {F.}~\bibnamefont
  {Nielsen}},\ }\bibfield  {title} {\bibinfo {title} {On {{Geodesic Triangles}}
  with {{Right Angles}} in a {{Dually Flat Space}}},\ }in\ \href
  {https://doi.org/10.1007/978-3-030-65459-7_7} {\emph {\bibinfo {booktitle}
  {Progress in {{Information Geometry}}: {{Theory}} and {{Applications}}}}},\
  \bibinfo {series and number} {Signals and {{Communication Technology}}},\
  \bibinfo {editor} {edited by\ \bibinfo {editor} {\bibfnamefont
  {F.}~\bibnamefont {Nielsen}}}\ (\bibinfo  {publisher} {{Springer
  International Publishing}},\ \bibinfo {address} {{Cham}},\ \bibinfo {year}
  {2021})\ pp.\ \bibinfo {pages} {153--190}\BibitemShut {NoStop}%
\bibitem [{\citenamefont {Callen}\ and\ \citenamefont
  {Callen}(1985)}]{callen1985}%
  \BibitemOpen
  \bibfield  {author} {\bibinfo {author} {\bibfnamefont {H.~B.}\ \bibnamefont
  {Callen}}\ and\ \bibinfo {author} {\bibfnamefont {H.~B.}\ \bibnamefont
  {Callen}},\ }\href@noop {} {\emph {\bibinfo {title} {Thermodynamics and an
  {{Introduction}} to {{Thermostatistics}}}}}\ (\bibinfo  {publisher}
  {{Wiley}},\ \bibinfo {year} {1985})\BibitemShut {NoStop}%
\bibitem [{\citenamefont {Sughiyama}\ \emph
  {et~al.}(2021{\natexlab{b}})\citenamefont {Sughiyama}, \citenamefont
  {Loutchko}, \citenamefont {Kamimura},\ and\ \citenamefont
  {Kobayashi}}]{sughiyama2021ArXiv211212403Cond-MatPhysicsphysicsa}%
  \BibitemOpen
  \bibfield  {author} {\bibinfo {author} {\bibfnamefont {Y.}~\bibnamefont
  {Sughiyama}}, \bibinfo {author} {\bibfnamefont {D.}~\bibnamefont {Loutchko}},
  \bibinfo {author} {\bibfnamefont {A.}~\bibnamefont {Kamimura}},\ and\
  \bibinfo {author} {\bibfnamefont {T.~J.}\ \bibnamefont {Kobayashi}},\
  }\bibfield  {title} {\bibinfo {title} {A {{Hessian Geometric Structure}} of
  {{Chemical Thermodynamic Systems}} with {{Stoichiometric Constraints}}},\
  }\href@noop {} {\bibfield  {journal} {\bibinfo  {journal} {ArXiv211212403
  Cond-Mat Physicsphysics}\ } (\bibinfo {year} {2021}{\natexlab{b}})},\ \Eprint
  {https://arxiv.org/abs/2112.12403} {arxiv:2112.12403 [cond-mat,
  physics:physics]} \BibitemShut {NoStop}%
\bibitem [{\citenamefont {{Hiriart-Urruty}}\ and\ \citenamefont
  {Lemarechal}(1996)}]{hiriart-urruty1996}%
  \BibitemOpen
  \bibfield  {author} {\bibinfo {author} {\bibfnamefont {J.-B.}\ \bibnamefont
  {{Hiriart-Urruty}}}\ and\ \bibinfo {author} {\bibfnamefont {C.}~\bibnamefont
  {Lemarechal}},\ }\href@noop {} {\emph {\bibinfo {title} {Convex {{Analysis}}
  and {{Minimization Algorithms II}}: {{Advanced Theory}} and {{Bundle
  Methods}}}}}\ (\bibinfo  {publisher} {{Springer Science \& Business Media}},\
  \bibinfo {year} {1996})\BibitemShut {NoStop}%
\bibitem [{\citenamefont {Krasnosel'skij}\ and\ \citenamefont
  {Rutickij}(1962)}]{krasnoselskij1962}%
  \BibitemOpen
  \bibfield  {author} {\bibinfo {author} {\bibfnamefont {M.~A.}\ \bibnamefont
  {Krasnosel'skij}}\ and\ \bibinfo {author} {\bibfnamefont {J.~B.}\
  \bibnamefont {Rutickij}},\ }\href@noop {} {\emph {\bibinfo {title} {Convex
  {{Functions}} and {{Orlicz Spaces}}}}}\ (\bibinfo  {publisher} {{Hindustan
  Publ.}},\ \bibinfo {year} {1962})\BibitemShut {NoStop}%
\bibitem [{\citenamefont {Lods}\ and\ \citenamefont
  {Pistone}(2015)}]{lods2015Entropy}%
  \BibitemOpen
  \bibfield  {author} {\bibinfo {author} {\bibfnamefont {B.}~\bibnamefont
  {Lods}}\ and\ \bibinfo {author} {\bibfnamefont {G.}~\bibnamefont {Pistone}},\
  }\bibfield  {title} {\bibinfo {title} {Information {{Geometry Formalism}} for
  the {{Spatially Homogeneous Boltzmann Equation}}},\ }\href
  {https://doi.org/10.3390/e17064323} {\bibfield  {journal} {\bibinfo
  {journal} {Entropy}\ }\textbf {\bibinfo {volume} {17}},\ \bibinfo {pages}
  {4323} (\bibinfo {year} {2015})}\BibitemShut {NoStop}%
\bibitem [{\citenamefont {Pistone}(2018)}]{pistone2018Inf.Geom.ItsAppl.}%
  \BibitemOpen
  \bibfield  {author} {\bibinfo {author} {\bibfnamefont {G.}~\bibnamefont
  {Pistone}},\ }\bibfield  {title} {\bibinfo {title} {Information {{Geometry}}
  of the {{Gaussian Space}}},\ }in\ \href
  {https://doi.org/10.1007/978-3-319-97798-0_5} {\emph {\bibinfo {booktitle}
  {Inf. {{Geom}}. {{Its Appl}}.}}},\ \bibinfo {series and number} {Springer
  {{Proceedings}} in {{Mathematics}} \& {{Statistics}}},\ \bibinfo {editor}
  {edited by\ \bibinfo {editor} {\bibfnamefont {N.}~\bibnamefont {Ay}},
  \bibinfo {editor} {\bibfnamefont {P.}~\bibnamefont {Gibilisco}},\ and\
  \bibinfo {editor} {\bibfnamefont {F.}~\bibnamefont {Mat{\'u}{\v s}}}}\
  (\bibinfo  {publisher} {{Springer International Publishing}},\ \bibinfo
  {address} {{Cham}},\ \bibinfo {year} {2018})\ pp.\ \bibinfo {pages}
  {119--155}\BibitemShut {NoStop}%
\bibitem [{\citenamefont {Ambrosio}\ \emph {et~al.}(2006)\citenamefont
  {Ambrosio}, \citenamefont {Gigli},\ and\ \citenamefont
  {Savare}}]{ambrosio2006a}%
  \BibitemOpen
  \bibfield  {author} {\bibinfo {author} {\bibfnamefont {L.}~\bibnamefont
  {Ambrosio}}, \bibinfo {author} {\bibfnamefont {N.}~\bibnamefont {Gigli}},\
  and\ \bibinfo {author} {\bibfnamefont {G.}~\bibnamefont {Savare}},\
  }\href@noop {} {\emph {\bibinfo {title} {Gradient {{Flows}}: {{In Metric
  Spaces}} and in the {{Space}} of {{Probability Measures}}}}}\ (\bibinfo
  {publisher} {{Springer Science \& Business Media}},\ \bibinfo {year}
  {2006})\BibitemShut {NoStop}%
\bibitem [{\citenamefont {Liero}\ \emph {et~al.}(2016)\citenamefont {Liero},
  \citenamefont {Mielke},\ and\ \citenamefont
  {Savar{\'e}}}]{liero2016SIAMJ.Math.Anal.}%
  \BibitemOpen
  \bibfield  {author} {\bibinfo {author} {\bibfnamefont {M.}~\bibnamefont
  {Liero}}, \bibinfo {author} {\bibfnamefont {A.}~\bibnamefont {Mielke}},\ and\
  \bibinfo {author} {\bibfnamefont {G.}~\bibnamefont {Savar{\'e}}},\ }\bibfield
   {title} {\bibinfo {title} {Optimal {{Transport}} in {{Competition}} with
  {{Reaction}}: {{The Hellinger--Kantorovich Distance}} and {{Geodesic
  Curves}}},\ }\href {https://doi.org/10.1137/15M1041420} {\bibfield  {journal}
  {\bibinfo  {journal} {SIAM J. Math. Anal.}\ }\textbf {\bibinfo {volume}
  {48}},\ \bibinfo {pages} {2869} (\bibinfo {year} {2016})}\BibitemShut
  {NoStop}%
\bibitem [{\citenamefont {Onsager}(1931{\natexlab{a}})}]{onsager1931Phys.Rev.}%
  \BibitemOpen
  \bibfield  {author} {\bibinfo {author} {\bibfnamefont {L.}~\bibnamefont
  {Onsager}},\ }\bibfield  {title} {\bibinfo {title} {Reciprocal {{Relations}}
  in {{Irreversible Processes}}. {{I}}.},\ }\href
  {https://doi.org/10.1103/PhysRev.37.405} {\bibfield  {journal} {\bibinfo
  {journal} {Phys. Rev.}\ }\textbf {\bibinfo {volume} {37}},\ \bibinfo {pages}
  {405} (\bibinfo {year} {1931}{\natexlab{a}})}\BibitemShut {NoStop}%
\bibitem [{\citenamefont
  {Onsager}(1931{\natexlab{b}})}]{onsager1931Phys.Rev.a}%
  \BibitemOpen
  \bibfield  {author} {\bibinfo {author} {\bibfnamefont {L.}~\bibnamefont
  {Onsager}},\ }\bibfield  {title} {\bibinfo {title} {Reciprocal {{Relations}}
  in {{Irreversible Processes}}. {{II}}.},\ }\href
  {https://doi.org/10.1103/PhysRev.38.2265} {\bibfield  {journal} {\bibinfo
  {journal} {Phys. Rev.}\ }\textbf {\bibinfo {volume} {38}},\ \bibinfo {pages}
  {2265} (\bibinfo {year} {1931}{\natexlab{b}})}\BibitemShut {NoStop}%
\bibitem [{\citenamefont {Machlup}\ and\ \citenamefont
  {Onsager}(1953)}]{machlup1953Phys.Rev.}%
  \BibitemOpen
  \bibfield  {author} {\bibinfo {author} {\bibfnamefont {S.}~\bibnamefont
  {Machlup}}\ and\ \bibinfo {author} {\bibfnamefont {L.}~\bibnamefont
  {Onsager}},\ }\bibfield  {title} {\bibinfo {title} {Fluctuations and
  {{Irreversible Process}}. {{II}}. {{Systems}} with {{Kinetic Energy}}},\
  }\href {https://doi.org/10.1103/PhysRev.91.1512} {\bibfield  {journal}
  {\bibinfo  {journal} {Phys. Rev.}\ }\textbf {\bibinfo {volume} {91}},\
  \bibinfo {pages} {1512} (\bibinfo {year} {1953})}\BibitemShut {NoStop}%
\bibitem [{\citenamefont {Lisini}(2009)}]{lisini2009ESAIM:COCV}%
  \BibitemOpen
  \bibfield  {author} {\bibinfo {author} {\bibfnamefont {S.}~\bibnamefont
  {Lisini}},\ }\bibfield  {title} {\bibinfo {title} {Nonlinear diffusion
  equations with variable coefficients as gradient flows in {{Wasserstein}}
  spaces},\ }\href {https://doi.org/10.1051/cocv:2008044} {\bibfield  {journal}
  {\bibinfo  {journal} {ESAIM: COCV}\ }\textbf {\bibinfo {volume} {15}},\
  \bibinfo {pages} {712} (\bibinfo {year} {2009})}\BibitemShut {NoStop}%
\bibitem [{\citenamefont {Peletier}(2014)}]{peletier2014}%
  \BibitemOpen
  \bibfield  {author} {\bibinfo {author} {\bibfnamefont {M.~A.}\ \bibnamefont
  {Peletier}},\ }\href {https://doi.org/10.48550/arXiv.1402.1990} {\bibinfo
  {title} {Variational {{Modelling}}: {{Energies}}, gradient flows, and large
  deviations}} (\bibinfo {year} {2014}),\ \Eprint
  {https://arxiv.org/abs/1402.1990} {arxiv:1402.1990 [math-ph]} \BibitemShut
  {NoStop}%
\bibitem [{\citenamefont {Truesdell}\ \emph {et~al.}(2004)\citenamefont
  {Truesdell}, \citenamefont {Truesdell}, \citenamefont {Noll}, \citenamefont
  {Antman},\ and\ \citenamefont {Noll}}]{truesdell2004}%
  \BibitemOpen
  \bibfield  {author} {\bibinfo {author} {\bibfnamefont {C.~A.}\ \bibnamefont
  {Truesdell}}, \bibinfo {author} {\bibfnamefont {C.}~\bibnamefont
  {Truesdell}}, \bibinfo {author} {\bibfnamefont {W.}~\bibnamefont {Noll}},
  \bibinfo {author} {\bibfnamefont {S.}~\bibnamefont {Antman}},\ and\ \bibinfo
  {author} {\bibfnamefont {W.}~\bibnamefont {Noll}},\ }\href@noop {} {\emph
  {\bibinfo {title} {The {{Non-Linear Field Theories}} of {{Mechanics}}}}}\
  (\bibinfo  {publisher} {{Springer Science \& Business Media}},\ \bibinfo
  {year} {2004})\BibitemShut {NoStop}%
\bibitem [{\citenamefont {Bergmann}\ and\ \citenamefont
  {Lebowitz}(1955)}]{bergmann1955Phys.Rev.}%
  \BibitemOpen
  \bibfield  {author} {\bibinfo {author} {\bibfnamefont {P.~G.}\ \bibnamefont
  {Bergmann}}\ and\ \bibinfo {author} {\bibfnamefont {J.~L.}\ \bibnamefont
  {Lebowitz}},\ }\bibfield  {title} {\bibinfo {title} {New {{Approach}} to
  {{Nonequilibrium Processes}}},\ }\href
  {https://doi.org/10.1103/PhysRev.99.578} {\bibfield  {journal} {\bibinfo
  {journal} {Phys. Rev.}\ }\textbf {\bibinfo {volume} {99}},\ \bibinfo {pages}
  {578} (\bibinfo {year} {1955})}\BibitemShut {NoStop}%
\bibitem [{\citenamefont {Maes}(2021)}]{maes2021SciPostPhys.Lect.Notes}%
  \BibitemOpen
  \bibfield  {author} {\bibinfo {author} {\bibfnamefont {C.}~\bibnamefont
  {Maes}},\ }\bibfield  {title} {\bibinfo {title} {Local detailed balance},\
  }\href {https://doi.org/10.21468/SciPostPhysLectNotes.32} {\bibfield
  {journal} {\bibinfo  {journal} {SciPost Phys. Lect. Notes}\ ,\ \bibinfo
  {pages} {032}} (\bibinfo {year} {2021})}\BibitemShut {NoStop}%
\bibitem [{\citenamefont {Maes}(2017)}]{maes2017}%
  \BibitemOpen
  \bibfield  {author} {\bibinfo {author} {\bibfnamefont {C.}~\bibnamefont
  {Maes}},\ }\href@noop {} {\emph {\bibinfo {title} {Non-{{Dissipative
  Effects}} in {{Nonequilibrium Systems}}}}}\ (\bibinfo  {publisher}
  {{Springer}},\ \bibinfo {year} {2017})\BibitemShut {NoStop}%
\bibitem [{\citenamefont {Patterson}\ and\ \citenamefont
  {Renger}(2019)}]{patterson2019MathPhysAnalGeom}%
  \BibitemOpen
  \bibfield  {author} {\bibinfo {author} {\bibfnamefont {R.}~\bibnamefont
  {Patterson}}\ and\ \bibinfo {author} {\bibfnamefont {M.}~\bibnamefont
  {Renger}},\ }\bibfield  {title} {\bibinfo {title} {Large deviations of
  reaction fluxes},\ }\href {https://doi.org/10.1007/s11040-019-9318-4}
  {\bibfield  {journal} {\bibinfo  {journal} {Math Phys Anal Geom}\ }\textbf
  {\bibinfo {volume} {22}},\ \bibinfo {pages} {21} (\bibinfo {year} {2019})},\
  \Eprint {https://arxiv.org/abs/1802.02512} {arxiv:1802.02512 [math]}
  \BibitemShut {NoStop}%
\bibitem [{\citenamefont
  {Wegscheider}(1902)}]{wegscheider1902Z.FuerPhys.Chem.}%
  \BibitemOpen
  \bibfield  {author} {\bibinfo {author} {\bibfnamefont {R.}~\bibnamefont
  {Wegscheider}},\ }\bibfield  {title} {\bibinfo {title} {{\"Uber simultane
  Gleichgewichte und die Beziehungen zwischen Thermodynamik und
  Reaktionskinetik homogener Systeme}},\ }\href
  {https://doi.org/10.1515/zpch-1902-3919} {\bibfield  {journal} {\bibinfo
  {journal} {Z. F\"ur Phys. Chem.}\ }\textbf {\bibinfo {volume} {39U}},\
  \bibinfo {pages} {257} (\bibinfo {year} {1902})}\BibitemShut {NoStop}%
\bibitem [{\citenamefont {Pistone}\ and\ \citenamefont
  {Sempi}(1995)}]{pistone1995Ann.Stat.}%
  \BibitemOpen
  \bibfield  {author} {\bibinfo {author} {\bibfnamefont {G.}~\bibnamefont
  {Pistone}}\ and\ \bibinfo {author} {\bibfnamefont {C.}~\bibnamefont
  {Sempi}},\ }\bibfield  {title} {\bibinfo {title} {An {{Infinite-Dimensional
  Geometric Structure}} on the {{Space}} of all the {{Probability Measures
  Equivalent}} to a {{Given One}}},\ }\href
  {https://doi.org/10.1214/aos/1176324311} {\bibfield  {journal} {\bibinfo
  {journal} {Ann. Stat.}\ }\textbf {\bibinfo {volume} {23}},\ \bibinfo {pages}
  {1543} (\bibinfo {year} {1995})}\BibitemShut {NoStop}%
\bibitem [{\citenamefont {Pistone}(2013)}]{pistone2013Entropy}%
  \BibitemOpen
  \bibfield  {author} {\bibinfo {author} {\bibfnamefont {G.}~\bibnamefont
  {Pistone}},\ }\bibfield  {title} {\bibinfo {title} {Examples of the
  {{Application}} of {{Nonparametric Information Geometry}} to {{Statistical
  Physics}}},\ }\href {https://doi.org/10.3390/e15104042} {\bibfield  {journal}
  {\bibinfo  {journal} {Entropy}\ }\textbf {\bibinfo {volume} {15}},\ \bibinfo
  {pages} {4042} (\bibinfo {year} {2013})}\BibitemShut {NoStop}%
\bibitem [{\citenamefont {Shiraishi}(2021)}]{shiraishi2021JStatPhys}%
  \BibitemOpen
  \bibfield  {author} {\bibinfo {author} {\bibfnamefont {N.}~\bibnamefont
  {Shiraishi}},\ }\bibfield  {title} {\bibinfo {title} {Optimal {{Thermodynamic
  Uncertainty Relation}} in {{Markov Jump Processes}}},\ }\href
  {https://doi.org/10.1007/s10955-021-02829-8} {\bibfield  {journal} {\bibinfo
  {journal} {J Stat Phys}\ }\textbf {\bibinfo {volume} {185}},\ \bibinfo
  {pages} {19} (\bibinfo {year} {2021})}\BibitemShut {NoStop}%
\bibitem [{\citenamefont {Chow}\ \emph {et~al.}(2012)\citenamefont {Chow},
  \citenamefont {Huang}, \citenamefont {Li},\ and\ \citenamefont
  {Zhou}}]{chow2012ArchRationalMechAnal}%
  \BibitemOpen
  \bibfield  {author} {\bibinfo {author} {\bibfnamefont {S.-N.}\ \bibnamefont
  {Chow}}, \bibinfo {author} {\bibfnamefont {W.}~\bibnamefont {Huang}},
  \bibinfo {author} {\bibfnamefont {Y.}~\bibnamefont {Li}},\ and\ \bibinfo
  {author} {\bibfnamefont {H.}~\bibnamefont {Zhou}},\ }\bibfield  {title}
  {\bibinfo {title} {Fokker\textendash{{Planck Equations}} for a {{Free Energy
  Functional}} or {{Markov Process}} on a {{Graph}}},\ }\href
  {https://doi.org/10.1007/s00205-011-0471-6} {\bibfield  {journal} {\bibinfo
  {journal} {Arch Rational Mech Anal}\ }\textbf {\bibinfo {volume} {203}},\
  \bibinfo {pages} {969} (\bibinfo {year} {2012})}\BibitemShut {NoStop}%
\bibitem [{\citenamefont {Maas}(2011)}]{maas2011JournalofFunctionalAnalysis}%
  \BibitemOpen
  \bibfield  {author} {\bibinfo {author} {\bibfnamefont {J.}~\bibnamefont
  {Maas}},\ }\bibfield  {title} {\bibinfo {title} {Gradient flows of the
  entropy for finite {{Markov}} chains},\ }\href
  {https://doi.org/10.1016/j.jfa.2011.06.009} {\bibfield  {journal} {\bibinfo
  {journal} {Journal of Functional Analysis}\ }\textbf {\bibinfo {volume}
  {261}},\ \bibinfo {pages} {2250} (\bibinfo {year} {2011})}\BibitemShut
  {NoStop}%
\bibitem [{\citenamefont {Mielke}(2013)}]{mielke2013Calc.Var.}%
  \BibitemOpen
  \bibfield  {author} {\bibinfo {author} {\bibfnamefont {A.}~\bibnamefont
  {Mielke}},\ }\bibfield  {title} {\bibinfo {title} {Geodesic convexity of the
  relative entropy in reversible {{Markov}} chains},\ }\href
  {https://doi.org/10.1007/s00526-012-0538-8} {\bibfield  {journal} {\bibinfo
  {journal} {Calc. Var.}\ }\textbf {\bibinfo {volume} {48}},\ \bibinfo {pages}
  {1} (\bibinfo {year} {2013})}\BibitemShut {NoStop}%
\bibitem [{\citenamefont {Liero}\ and\ \citenamefont
  {Mielke}(2013)}]{liero2013Philos.Trans.R.Soc.Math.Phys.Eng.Sci.}%
  \BibitemOpen
  \bibfield  {author} {\bibinfo {author} {\bibfnamefont {M.}~\bibnamefont
  {Liero}}\ and\ \bibinfo {author} {\bibfnamefont {A.}~\bibnamefont {Mielke}},\
  }\bibfield  {title} {\bibinfo {title} {Gradient structures and geodesic
  convexity for reaction\textendash diffusion systems},\ }\href
  {https://doi.org/10.1098/rsta.2012.0346} {\bibfield  {journal} {\bibinfo
  {journal} {Philos. Trans. R. Soc. Math. Phys. Eng. Sci.}\ }\textbf {\bibinfo
  {volume} {371}},\ \bibinfo {pages} {20120346} (\bibinfo {year}
  {2013})}\BibitemShut {NoStop}%
\bibitem [{\citenamefont {Van~Vu}\ and\ \citenamefont
  {Saito}(2023)}]{vanvu2022}%
  \BibitemOpen
  \bibfield  {author} {\bibinfo {author} {\bibfnamefont {T.}~\bibnamefont
  {Van~Vu}}\ and\ \bibinfo {author} {\bibfnamefont {K.}~\bibnamefont {Saito}},\
  }\bibfield  {title} {\bibinfo {title} {Thermodynamic {{Unification}} of
  {{Optimal Transport}}: {{Thermodynamic Uncertainty Relation}}, {{Minimum
  Dissipation}}, and {{Thermodynamic Speed Limits}}},\ }\href
  {https://doi.org/10.1103/PhysRevX.13.011013} {\bibfield  {journal} {\bibinfo
  {journal} {Phys. Rev. X}\ }\textbf {\bibinfo {volume} {13}},\ \bibinfo
  {pages} {011013} (\bibinfo {year} {2023})}\BibitemShut {NoStop}%
\bibitem [{\citenamefont
  {Yamano}(2013{\natexlab{b}})}]{yamano2013J.Math.Phys.}%
  \BibitemOpen
  \bibfield  {author} {\bibinfo {author} {\bibfnamefont {T.}~\bibnamefont
  {Yamano}},\ }\bibfield  {title} {\bibinfo {title} {Phase space gradient of
  dissipated work and information: {{A}} role of relative {{Fisher}}
  information},\ }\href {https://doi.org/10.1063/1.4828855} {\bibfield
  {journal} {\bibinfo  {journal} {J. Math. Phys.}\ }\textbf {\bibinfo {volume}
  {54}},\ \bibinfo {pages} {113301} (\bibinfo {year}
  {2013}{\natexlab{b}})}\BibitemShut {NoStop}%
\bibitem [{\citenamefont
  {Hyv{\"a}rinen}(2005)}]{hyvarinen2005J.Mach.Learn.Res.}%
  \BibitemOpen
  \bibfield  {author} {\bibinfo {author} {\bibfnamefont {A.}~\bibnamefont
  {Hyv{\"a}rinen}},\ }\bibfield  {title} {\bibinfo {title} {Estimation of
  {{Non-Normalized Statistical Models}} by {{Score Matching}}},\ }\href@noop {}
  {\bibfield  {journal} {\bibinfo  {journal} {J. Mach. Learn. Res.}\ }\textbf
  {\bibinfo {volume} {6}},\ \bibinfo {pages} {695} (\bibinfo {year}
  {2005})}\BibitemShut {NoStop}%
\bibitem [{\citenamefont {Shima}(2007)}]{shima2007}%
  \BibitemOpen
  \bibfield  {author} {\bibinfo {author} {\bibfnamefont {H.}~\bibnamefont
  {Shima}},\ }\href@noop {} {\emph {\bibinfo {title} {The {{Geometry}} of
  {{Hessian Structures}}}}}\ (\bibinfo  {publisher} {{World Scientific}},\
  \bibinfo {year} {2007})\BibitemShut {NoStop}%
\bibitem [{\citenamefont {Wolsey}\ and\ \citenamefont
  {Nemhauser}(1999)}]{wolsey1999}%
  \BibitemOpen
  \bibfield  {author} {\bibinfo {author} {\bibfnamefont {L.~A.}\ \bibnamefont
  {Wolsey}}\ and\ \bibinfo {author} {\bibfnamefont {G.~L.}\ \bibnamefont
  {Nemhauser}},\ }\href@noop {} {\emph {\bibinfo {title} {Integer and
  {{Combinatorial Optimization}}}}}\ (\bibinfo  {publisher} {{John Wiley \&
  Sons}},\ \bibinfo {year} {1999})\BibitemShut {NoStop}%
\bibitem [{\citenamefont {Davis}(1959)}]{davis1959Math.Mag.}%
  \BibitemOpen
  \bibfield  {author} {\bibinfo {author} {\bibfnamefont {H.~F.}\ \bibnamefont
  {Davis}},\ }\bibfield  {title} {\bibinfo {title} {On {{Isosceles
  Orthogonality}}},\ }\href {https://doi.org/10.2307/3029494} {\bibfield
  {journal} {\bibinfo  {journal} {Math. Mag.}\ }\textbf {\bibinfo {volume}
  {32}},\ \bibinfo {pages} {129} (\bibinfo {year} {1959})},\ \Eprint
  {https://arxiv.org/abs/3029494} {3029494} \BibitemShut {NoStop}%
\bibitem [{\citenamefont
  {Amari}(1996)}]{amari1996Proc.9thInt.Conf.NeuralInf.Process.Syst.}%
  \BibitemOpen
  \bibfield  {author} {\bibinfo {author} {\bibfnamefont {S.-I.}\ \bibnamefont
  {Amari}},\ }\bibfield  {title} {\bibinfo {title} {Neural learning in
  structured parameter spaces: Natural {{Riemannian}} gradient},\ }in\
  \href@noop {} {\emph {\bibinfo {booktitle} {Proc. 9th {{Int}}. {{Conf}}.
  {{Neural Inf}}. {{Process}}. {{Syst}}.}}},\ \bibinfo {series and number}
  {{{NIPS}}'96}\ (\bibinfo  {publisher} {{MIT Press}},\ \bibinfo {address}
  {{Cambridge, MA, USA}},\ \bibinfo {year} {1996})\ pp.\ \bibinfo {pages}
  {127--133}\BibitemShut {NoStop}%
\bibitem [{\citenamefont {Gunasekar}\ \emph {et~al.}(2021)\citenamefont
  {Gunasekar}, \citenamefont {Woodworth},\ and\ \citenamefont
  {Srebro}}]{gunasekar2021}%
  \BibitemOpen
  \bibfield  {author} {\bibinfo {author} {\bibfnamefont {S.}~\bibnamefont
  {Gunasekar}}, \bibinfo {author} {\bibfnamefont {B.}~\bibnamefont
  {Woodworth}},\ and\ \bibinfo {author} {\bibfnamefont {N.}~\bibnamefont
  {Srebro}},\ }\bibfield  {title} {\bibinfo {title} {Mirrorless {{Mirror
  Descent}}: {{A Natural Derivation}} of {{Mirror Descent}}},\ }in\ \href@noop
  {} {\emph {\bibinfo {booktitle} {Proc. 24th {{Int}}. {{Conf}}. {{Artif}}.
  {{Intell}}. {{Stat}}.}}}\ (\bibinfo  {publisher} {{PMLR}},\ \bibinfo {year}
  {2021})\ pp.\ \bibinfo {pages} {2305--2313}\BibitemShut {NoStop}%
\bibitem [{\citenamefont {Li}\ and\ \citenamefont
  {Mont{\'u}far}(2018)}]{li2018Info.Geo.}%
  \BibitemOpen
  \bibfield  {author} {\bibinfo {author} {\bibfnamefont {W.}~\bibnamefont
  {Li}}\ and\ \bibinfo {author} {\bibfnamefont {G.}~\bibnamefont
  {Mont{\'u}far}},\ }\bibfield  {title} {\bibinfo {title} {Natural gradient via
  optimal transport},\ }\href {https://doi.org/10.1007/s41884-018-0015-3}
  {\bibfield  {journal} {\bibinfo  {journal} {Info. Geo.}\ }\textbf {\bibinfo
  {volume} {1}},\ \bibinfo {pages} {181} (\bibinfo {year} {2018})}\BibitemShut
  {NoStop}%
\bibitem [{\citenamefont {Amari}\ \emph {et~al.}(2018)\citenamefont {Amari},
  \citenamefont {Karakida},\ and\ \citenamefont {Oizumi}}]{amari2018Info.Geo.}%
  \BibitemOpen
  \bibfield  {author} {\bibinfo {author} {\bibfnamefont {S.-i.}\ \bibnamefont
  {Amari}}, \bibinfo {author} {\bibfnamefont {R.}~\bibnamefont {Karakida}},\
  and\ \bibinfo {author} {\bibfnamefont {M.}~\bibnamefont {Oizumi}},\
  }\bibfield  {title} {\bibinfo {title} {Information geometry connecting
  {{Wasserstein}} distance and {{Kullback}}\textendash{{Leibler}} divergence
  via the entropy-relaxed transportation problem},\ }\href
  {https://doi.org/10.1007/s41884-018-0002-8} {\bibfield  {journal} {\bibinfo
  {journal} {Info. Geo.}\ }\textbf {\bibinfo {volume} {1}},\ \bibinfo {pages}
  {13} (\bibinfo {year} {2018})}\BibitemShut {NoStop}%
\bibitem [{\citenamefont {Bain}\ and\ \citenamefont {Crisan}(2008)}]{bain2008}%
  \BibitemOpen
  \bibfield  {author} {\bibinfo {author} {\bibfnamefont {A.}~\bibnamefont
  {Bain}}\ and\ \bibinfo {author} {\bibfnamefont {D.}~\bibnamefont {Crisan}},\
  }\href@noop {} {\emph {\bibinfo {title} {Fundamentals of {{Stochastic
  Filtering}}}}}\ (\bibinfo  {publisher} {{Springer Science \& Business
  Media}},\ \bibinfo {year} {2008})\BibitemShut {NoStop}%
\bibitem [{\citenamefont {Brigo}\ \emph {et~al.}(1998)\citenamefont {Brigo},
  \citenamefont {Hanzon},\ and\ \citenamefont
  {LeGland}}]{brigo1998IEEETrans.Autom.Control}%
  \BibitemOpen
  \bibfield  {author} {\bibinfo {author} {\bibfnamefont {D.}~\bibnamefont
  {Brigo}}, \bibinfo {author} {\bibfnamefont {B.}~\bibnamefont {Hanzon}},\ and\
  \bibinfo {author} {\bibfnamefont {F.}~\bibnamefont {LeGland}},\ }\bibfield
  {title} {\bibinfo {title} {A differential geometric approach to nonlinear
  filtering: The projection filter},\ }\href {https://doi.org/10.1109/9.661075}
  {\bibfield  {journal} {\bibinfo  {journal} {IEEE Trans. Autom. Control}\
  }\textbf {\bibinfo {volume} {43}},\ \bibinfo {pages} {247} (\bibinfo {year}
  {1998})}\BibitemShut {NoStop}%
\bibitem [{\citenamefont {Li}\ \emph {et~al.}(2017)\citenamefont {Li},
  \citenamefont {Cheng}, \citenamefont {Li}, \citenamefont {Wang},
  \citenamefont {Hua},\ and\ \citenamefont
  {Qin}}]{li2017201736thChin.ControlConf.CCC}%
  \BibitemOpen
  \bibfield  {author} {\bibinfo {author} {\bibfnamefont {Y.}~\bibnamefont
  {Li}}, \bibinfo {author} {\bibfnamefont {Y.}~\bibnamefont {Cheng}}, \bibinfo
  {author} {\bibfnamefont {X.}~\bibnamefont {Li}}, \bibinfo {author}
  {\bibfnamefont {H.}~\bibnamefont {Wang}}, \bibinfo {author} {\bibfnamefont
  {X.}~\bibnamefont {Hua}},\ and\ \bibinfo {author} {\bibfnamefont
  {Y.}~\bibnamefont {Qin}},\ }\bibfield  {title} {\bibinfo {title} {Information
  geometric approach for nonlinear filtering},\ }in\ \href
  {https://doi.org/10.23919/ChiCC.2017.8027514} {\emph {\bibinfo {booktitle}
  {2017 36th {{Chin}}. {{Control Conf}}. {{CCC}}}}}\ (\bibinfo {year} {2017})\
  pp.\ \bibinfo {pages} {1211--1216}\BibitemShut {NoStop}%
\bibitem [{\citenamefont {Hill}(2005)}]{hill2005}%
  \BibitemOpen
  \bibfield  {author} {\bibinfo {author} {\bibfnamefont {T.~L.}\ \bibnamefont
  {Hill}},\ }\href@noop {} {\emph {\bibinfo {title} {Free {{Energy
  Transduction}} and {{Biochemical Cycle Kinetics}}}}}\ (\bibinfo  {publisher}
  {{Courier Corporation}},\ \bibinfo {year} {2005})\BibitemShut {NoStop}%
\bibitem [{\citenamefont {Altaner}\ \emph {et~al.}(2012)\citenamefont
  {Altaner}, \citenamefont {Grosskinsky}, \citenamefont {Herminghaus},
  \citenamefont {Katth{\"a}n}, \citenamefont {Timme},\ and\ \citenamefont
  {Vollmer}}]{altaner2012Phys.Rev.Ea}%
  \BibitemOpen
  \bibfield  {author} {\bibinfo {author} {\bibfnamefont {B.}~\bibnamefont
  {Altaner}}, \bibinfo {author} {\bibfnamefont {S.}~\bibnamefont
  {Grosskinsky}}, \bibinfo {author} {\bibfnamefont {S.}~\bibnamefont
  {Herminghaus}}, \bibinfo {author} {\bibfnamefont {L.}~\bibnamefont
  {Katth{\"a}n}}, \bibinfo {author} {\bibfnamefont {M.}~\bibnamefont {Timme}},\
  and\ \bibinfo {author} {\bibfnamefont {J.}~\bibnamefont {Vollmer}},\
  }\bibfield  {title} {\bibinfo {title} {Network representations of
  nonequilibrium steady states: {{Cycle}} decompositions, symmetries, and
  dominant paths},\ }\href {https://doi.org/10.1103/PhysRevE.85.041133}
  {\bibfield  {journal} {\bibinfo  {journal} {Phys. Rev. E}\ }\textbf {\bibinfo
  {volume} {85}},\ \bibinfo {pages} {041133} (\bibinfo {year}
  {2012})}\BibitemShut {NoStop}%
\bibitem [{\citenamefont {Polettini}\ and\ \citenamefont
  {Esposito}(2014)}]{polettini2014J.Chem.Phys.}%
  \BibitemOpen
  \bibfield  {author} {\bibinfo {author} {\bibfnamefont {M.}~\bibnamefont
  {Polettini}}\ and\ \bibinfo {author} {\bibfnamefont {M.}~\bibnamefont
  {Esposito}},\ }\bibfield  {title} {\bibinfo {title} {Irreversible
  thermodynamics of open chemical networks. {{I}}. {{Emergent}} cycles and
  broken conservation laws},\ }\href {https://doi.org/10.1063/1.4886396}
  {\bibfield  {journal} {\bibinfo  {journal} {J. Chem. Phys.}\ }\textbf
  {\bibinfo {volume} {141}},\ \bibinfo {pages} {024117} (\bibinfo {year}
  {2014})}\BibitemShut {NoStop}%
\bibitem [{\citenamefont {Strang}(2020)}]{strang2020}%
  \BibitemOpen
  \bibfield  {author} {\bibinfo {author} {\bibfnamefont {A.}~\bibnamefont
  {Strang}},\ }\emph {\bibinfo {title} {{Applications of the Helmholtz-Hodge
  Decomposition to Networks and Random Processes}}},\ \href@noop {} {Ph.D.
  thesis},\ \bibinfo {address} {{Ann Arbor, United States}} (\bibinfo {year}
  {2020})\BibitemShut {NoStop}%
\bibitem [{\citenamefont {Oono}\ and\ \citenamefont
  {Paniconi}(1998)}]{oono1998ProgressofTheoreticalPhysicsSupplement}%
  \BibitemOpen
  \bibfield  {author} {\bibinfo {author} {\bibfnamefont {Y.}~\bibnamefont
  {Oono}}\ and\ \bibinfo {author} {\bibfnamefont {M.}~\bibnamefont
  {Paniconi}},\ }\bibfield  {title} {\bibinfo {title} {Steady {{State
  Thermodynamics}}},\ }\href {https://doi.org/10.1143/PTPS.130.29} {\bibfield
  {journal} {\bibinfo  {journal} {Progress of Theoretical Physics Supplement}\
  }\textbf {\bibinfo {volume} {130}},\ \bibinfo {pages} {29} (\bibinfo {year}
  {1998})}\BibitemShut {NoStop}%
\bibitem [{\citenamefont {Komatsu}\ \emph {et~al.}(2008)\citenamefont
  {Komatsu}, \citenamefont {Nakagawa}, \citenamefont {Sasa},\ and\
  \citenamefont {Tasaki}}]{komatsu2008Phys.Rev.Lett.}%
  \BibitemOpen
  \bibfield  {author} {\bibinfo {author} {\bibfnamefont {T.~S.}\ \bibnamefont
  {Komatsu}}, \bibinfo {author} {\bibfnamefont {N.}~\bibnamefont {Nakagawa}},
  \bibinfo {author} {\bibfnamefont {S.-i.}\ \bibnamefont {Sasa}},\ and\
  \bibinfo {author} {\bibfnamefont {H.}~\bibnamefont {Tasaki}},\ }\bibfield
  {title} {\bibinfo {title} {Steady-{{State Thermodynamics}} for {{Heat
  Conduction}}: {{Microscopic Derivation}}},\ }\href
  {https://doi.org/10.1103/PhysRevLett.100.230602} {\bibfield  {journal}
  {\bibinfo  {journal} {Phys. Rev. Lett.}\ }\textbf {\bibinfo {volume} {100}},\
  \bibinfo {pages} {230602} (\bibinfo {year} {2008})}\BibitemShut {NoStop}%
\bibitem [{\citenamefont {Maes}\ and\ \citenamefont {Neto{\v
  c}n{\'y}}(2014)}]{maes2014JStatPhys}%
  \BibitemOpen
  \bibfield  {author} {\bibinfo {author} {\bibfnamefont {C.}~\bibnamefont
  {Maes}}\ and\ \bibinfo {author} {\bibfnamefont {K.}~\bibnamefont {Neto{\v
  c}n{\'y}}},\ }\bibfield  {title} {\bibinfo {title} {A {{Nonequilibrium
  Extension}} of the {{Clausius Heat Theorem}}},\ }\href
  {https://doi.org/10.1007/s10955-013-0822-9} {\bibfield  {journal} {\bibinfo
  {journal} {J Stat Phys}\ }\textbf {\bibinfo {volume} {154}},\ \bibinfo
  {pages} {188} (\bibinfo {year} {2014})}\BibitemShut {NoStop}%
\bibitem [{\citenamefont {Dechant}\ \emph {et~al.}(2022)\citenamefont
  {Dechant}, \citenamefont {Sasa},\ and\ \citenamefont
  {Ito}}]{dechant2022ArXiv210912817Cond-Mat}%
  \BibitemOpen
  \bibfield  {author} {\bibinfo {author} {\bibfnamefont {A.}~\bibnamefont
  {Dechant}}, \bibinfo {author} {\bibfnamefont {S.-i.}\ \bibnamefont {Sasa}},\
  and\ \bibinfo {author} {\bibfnamefont {S.}~\bibnamefont {Ito}},\ }\bibfield
  {title} {\bibinfo {title} {Geometric decomposition of entropy production in
  out-of-equilibrium systems},\ }\href@noop {} {\bibfield  {journal} {\bibinfo
  {journal} {ArXiv210912817 Cond-Mat}\ } (\bibinfo {year} {2022})},\ \Eprint
  {https://arxiv.org/abs/2109.12817} {arxiv:2109.12817 [cond-mat]} \BibitemShut
  {NoStop}%
\bibitem [{\citenamefont {Guo}\ \emph {et~al.}(2005)\citenamefont {Guo},
  \citenamefont {Shamai},\ and\ \citenamefont
  {Verdu}}]{guo2005IEEETrans.Inf.Theory}%
  \BibitemOpen
  \bibfield  {author} {\bibinfo {author} {\bibfnamefont {D.}~\bibnamefont
  {Guo}}, \bibinfo {author} {\bibfnamefont {S.}~\bibnamefont {Shamai}},\ and\
  \bibinfo {author} {\bibfnamefont {S.}~\bibnamefont {Verdu}},\ }\bibfield
  {title} {\bibinfo {title} {Mutual information and minimum mean-square error
  in {{Gaussian}} channels},\ }\href {https://doi.org/10.1109/TIT.2005.844072}
  {\bibfield  {journal} {\bibinfo  {journal} {IEEE Trans. Inf. Theory}\
  }\textbf {\bibinfo {volume} {51}},\ \bibinfo {pages} {1261} (\bibinfo {year}
  {2005})}\BibitemShut {NoStop}%
\bibitem [{\citenamefont {{Mayer-Wolf}}\ and\ \citenamefont
  {Zakai}(1984)}]{mayer-wolf1984Filter.ControlRandomProcess.}%
  \BibitemOpen
  \bibfield  {author} {\bibinfo {author} {\bibfnamefont {E.}~\bibnamefont
  {{Mayer-Wolf}}}\ and\ \bibinfo {author} {\bibfnamefont {M.}~\bibnamefont
  {Zakai}},\ }\bibfield  {title} {\bibinfo {title} {On a formula relating the
  {{Shannon}} information to the fisher information for the filtering
  problem},\ }in\ \href {https://doi.org/10.1007/BFb0006569} {\emph {\bibinfo
  {booktitle} {Filter. {{Control Random Process}}.}}},\ \bibinfo {series and
  number} {Lecture {{Notes}} in {{Control}} and {{Information Sciences}}},\
  \bibinfo {editor} {edited by\ \bibinfo {editor} {\bibfnamefont
  {H.}~\bibnamefont {Korezlioglu}}, \bibinfo {editor} {\bibfnamefont
  {G.}~\bibnamefont {Mazziotto}},\ and\ \bibinfo {editor} {\bibfnamefont
  {J.}~\bibnamefont {Szpirglas}}}\ (\bibinfo  {publisher} {{Springer}},\
  \bibinfo {address} {{Berlin, Heidelberg}},\ \bibinfo {year} {1984})\ pp.\
  \bibinfo {pages} {164--171}\BibitemShut {NoStop}%
\bibitem [{\citenamefont {{Mayer-Wolf}}\ and\ \citenamefont
  {Zakai}(2007)}]{mayer-wolf2007Ann.Appl.Probab.}%
  \BibitemOpen
  \bibfield  {author} {\bibinfo {author} {\bibfnamefont {E.}~\bibnamefont
  {{Mayer-Wolf}}}\ and\ \bibinfo {author} {\bibfnamefont {M.}~\bibnamefont
  {Zakai}},\ }\bibfield  {title} {\bibinfo {title} {Some relations between
  mutual information and estimation error in {{Wiener}} space},\ }\href
  {https://doi.org/10.1214/105051607000000131} {\bibfield  {journal} {\bibinfo
  {journal} {Ann. Appl. Probab.}\ }\textbf {\bibinfo {volume} {17}},\ \bibinfo
  {pages} {1102} (\bibinfo {year} {2007})}\BibitemShut {NoStop}%
\bibitem [{\citenamefont {Amari}\ \emph {et~al.}(2006)\citenamefont {Amari},
  \citenamefont {Park},\ and\ \citenamefont
  {Ozeki}}]{amari2006NeuralComputation}%
  \BibitemOpen
  \bibfield  {author} {\bibinfo {author} {\bibfnamefont {S.-i.}\ \bibnamefont
  {Amari}}, \bibinfo {author} {\bibfnamefont {H.}~\bibnamefont {Park}},\ and\
  \bibinfo {author} {\bibfnamefont {T.}~\bibnamefont {Ozeki}},\ }\bibfield
  {title} {\bibinfo {title} {Singularities {{Affect Dynamics}} of {{Learning}}
  in {{Neuromanifolds}}},\ }\href {https://doi.org/10.1162/neco.2006.18.5.1007}
  {\bibfield  {journal} {\bibinfo  {journal} {Neural Computation}\ }\textbf
  {\bibinfo {volume} {18}},\ \bibinfo {pages} {1007} (\bibinfo {year}
  {2006})}\BibitemShut {NoStop}%
\bibitem [{\citenamefont {Watanabe}(2009)}]{watanabe2009}%
  \BibitemOpen
  \bibfield  {author} {\bibinfo {author} {\bibfnamefont {S.}~\bibnamefont
  {Watanabe}},\ }\href@noop {} {\emph {\bibinfo {title} {Algebraic {{Geometry}}
  and {{Statistical Learning Theory}}}}}\ (\bibinfo  {publisher} {{Cambridge
  University Press}},\ \bibinfo {year} {2009})\BibitemShut {NoStop}%
\bibitem [{\citenamefont {Murashita}\ \emph {et~al.}(2014)\citenamefont
  {Murashita}, \citenamefont {Funo},\ and\ \citenamefont
  {Ueda}}]{murashita2014Phys.Rev.E}%
  \BibitemOpen
  \bibfield  {author} {\bibinfo {author} {\bibfnamefont {Y.}~\bibnamefont
  {Murashita}}, \bibinfo {author} {\bibfnamefont {K.}~\bibnamefont {Funo}},\
  and\ \bibinfo {author} {\bibfnamefont {M.}~\bibnamefont {Ueda}},\ }\bibfield
  {title} {\bibinfo {title} {Nonequilibrium equalities in absolutely
  irreversible processes},\ }\href {https://doi.org/10.1103/PhysRevE.90.042110}
  {\bibfield  {journal} {\bibinfo  {journal} {Phys. Rev. E}\ }\textbf {\bibinfo
  {volume} {90}},\ \bibinfo {pages} {042110} (\bibinfo {year}
  {2014})}\BibitemShut {NoStop}%
\bibitem [{\citenamefont {{\"O}ttinger}(2005)}]{ottinger2005}%
  \BibitemOpen
  \bibfield  {author} {\bibinfo {author} {\bibfnamefont {H.~C.}\ \bibnamefont
  {{\"O}ttinger}},\ }\href@noop {} {\emph {\bibinfo {title} {Beyond
  {{Equilibrium Thermodynamics}}}}}\ (\bibinfo  {publisher} {{John Wiley \&
  Sons}},\ \bibinfo {year} {2005})\BibitemShut {NoStop}%
\bibitem [{\citenamefont {Wang}\ and\ \citenamefont
  {Li}(2021)}]{wang2021JSciComput}%
  \BibitemOpen
  \bibfield  {author} {\bibinfo {author} {\bibfnamefont {Y.}~\bibnamefont
  {Wang}}\ and\ \bibinfo {author} {\bibfnamefont {W.}~\bibnamefont {Li}},\
  }\bibfield  {title} {\bibinfo {title} {Accelerated {{Information Gradient
  Flow}}},\ }\href {https://doi.org/10.1007/s10915-021-01709-3} {\bibfield
  {journal} {\bibinfo  {journal} {J Sci Comput}\ }\textbf {\bibinfo {volume}
  {90}},\ \bibinfo {pages} {11} (\bibinfo {year} {2021})}\BibitemShut {NoStop}%
\bibitem [{\citenamefont {Li}\ \emph {et~al.}(2023)\citenamefont {Li},
  \citenamefont {Liu},\ and\ \citenamefont {Osher}}]{li2023J.Comput.Phys.}%
  \BibitemOpen
  \bibfield  {author} {\bibinfo {author} {\bibfnamefont {W.}~\bibnamefont
  {Li}}, \bibinfo {author} {\bibfnamefont {S.}~\bibnamefont {Liu}},\ and\
  \bibinfo {author} {\bibfnamefont {S.}~\bibnamefont {Osher}},\ }\bibfield
  {title} {\bibinfo {title} {Controlling conservation laws {{I}}:
  {{Entropy}}\textendash entropy flux},\ }\bibfield  {journal} {\bibinfo
  {journal} {J. Comput. Phys.}\ }\textbf {\bibinfo {volume} {480}},\ \href
  {https://doi.org/10.1016/j.jcp.2023.112019} {10.1016/j.jcp.2023.112019}
  (\bibinfo {year} {2023})\BibitemShut {NoStop}%
\bibitem [{\citenamefont {Gao}\ \emph {et~al.}(2022)\citenamefont {Gao},
  \citenamefont {Li},\ and\ \citenamefont {Liu}}]{gao2022}%
  \BibitemOpen
  \bibfield  {author} {\bibinfo {author} {\bibfnamefont {Y.}~\bibnamefont
  {Gao}}, \bibinfo {author} {\bibfnamefont {W.}~\bibnamefont {Li}},\ and\
  \bibinfo {author} {\bibfnamefont {J.-G.}\ \bibnamefont {Liu}},\ }\href
  {https://doi.org/10.48550/arXiv.2212.05675} {\bibinfo {title} {Master
  equations for finite state mean field games with nonlinear activations}}
  (\bibinfo {year} {2022}),\ \Eprint {https://arxiv.org/abs/2212.05675}
  {arxiv:2212.05675 [math]} \BibitemShut {NoStop}%
\end{thebibliography}

\end{document}